\documentclass[10pt]{article}

\usepackage{amsmath,amsfonts,amsthm,amssymb,mathrsfs}
\usepackage[T1]{fontenc}
\usepackage[utf8]{inputenc}
\usepackage[colorlinks,backref]{hyperref}
\usepackage{xcolor}
\usepackage{fullpage}
\usepackage{mathtools}
\usepackage{authblk}
\usepackage{algorithm}
\usepackage{algpseudocode}
\usepackage{float, subcaption}
\usepackage{float, xcolor, colortbl}
\usepackage{natbib}

\RequirePackage{tikz}
\usetikzlibrary{fit,positioning,arrows,automata,calc}
\tikzset{
   main/.style={circle, minimum size = 10mm, thick, 
        draw =black!80, node distance = 10mm},
   box/.style={rectangle, draw=black!100}
}
\usetikzlibrary{positioning}
\usepackage{enumitem}
\setlist{  
  listparindent=\parindent,
  parsep=0pt
}
\usepackage{thmtools}
\usepackage{thm-restate}
\usepackage{bbm}

\usepackage[capitalize]{cleveref}
\crefname{assumption}{Assumption}{Assumptions}
\usepackage{titlesec}
\titleformat*{\section}{\large\bfseries}

\titleformat*{\subsection}{\normalfont\bfseries}

\DeclareMathOperator{\tr}{tr}
\usepackage{authblk}
\DeclareMathOperator{\Diag}{Diag}
\DeclareMathOperator{\Var}{Var}
\DeclareMathOperator{\Cov}{Cov}
\DeclareMathOperator{\Ber}{Ber}
\DeclareMathOperator{\rank}{rank}
\DeclareMathOperator{\col}{col}

\newcommand{\perpeq}{%
  \mathrel{%
    \vcenter{%
      \offinterlineskip
      \halign{##\cr
        \hbox to 1.2ex{\hfil\rule{0.5pt}{1.1ex}\hfil}\cr 
        \noalign{\vskip0.2ex}
        \hbox{\rule{1.2ex}{0.4pt}}\cr 
        \noalign{\vskip0.3ex}
        \hbox{\rule{1.2ex}{0.4pt}}\cr 
      }%
    }%
  }%
}
\newcommand{\notperpeq}{%
  \mathrel{%
    \ooalign{%
      $\perpeq$\cr
      \hidewidth$\mkern2mu/\mkern2mu$\hidewidth\cr
    }%
  }%
}

\renewcommand{\t}{\t} 
\newcommand{\diag}{\mathrm{diag}} 
\newcommand{\inv}{^{-1}}
\newcommand{\E}{\mathbb{E}}

\renewcommand{\hat}{\widehat}

\renewcommand{\tilde}{\widetilde}
\newtheorem{proposition}{Proposition}
\newtheorem{theorem}{Theorem}
\newtheorem{corollary}{Corollary}
\newtheorem{lemma}{Lemma}

\theoremstyle{definition} 
\newtheorem{definition}{Definition}

\newtheorem{remark}{Remark}
\newtheorem{assumption}{Assumption}
\hypersetup{%
  colorlinks=true,
  linkcolor=blue,
  citecolor = blue,
  urlcolor = blue
}

\usepackage{mathtools}
\usepackage{bbm}
\usepackage{algpseudocode}
\usepackage{algorithm}

\usepackage{bm}

\usepackage{stackengine}

\usepackage{centernot}

\renewcommand{\t}{^{\top}}

\renewcommand{\Pr}{\mathbb{P}}
\usepackage{multicol,caption}
\usepackage{comment}
\title{Two-Sample Hypothesis Testing for Subspace Equality in Network Data}
\author{Rajdeep Brahma, Joshua Agterberg, and Yuguo Chen}
\affil{University of Illinois Urbana-Champaign}
\date{\today}
\begin{document}

\maketitle

\begin{abstract}
In many settings one is often interested in determining whether two networks share some joint structural connectivity patterns such as communities.  However, while communities may be shared across networks, edge probabilities may differ significantly.  Therefore, in this paper we consider testing a general null hypothesis that two networks have the same underlying subspace, which in particular includes the setting that communities are the same for either stochastic blockmodels or mixed-membership stochastic blockmodels (even if edge probabilities are different).  We propose a test statistic based on the Frobenius norm of the difference of the leading subspace projection matrices, and we prove that our test statistic, after appropriate centering and scaling, converges in distribution to a Gaussian random variable as long as the average expected degree grows at least logarithmically in the number of vertices.  We then provide estimators for the asymptotic mean and variance and show consistency under a stronger signal condition, and we give the local power of our test when the networks are sufficiently dense.  Our theoretical results are based on a limit theorem for the projection difference of empirical and true eigenvectors which can also be viewed as the one-sample version of our test statistic, and this result may be of independent interest.   We demonstrate our results through numerical simulations and an application to US Flight data.  
\end{abstract}

\vspace{\baselineskip}
\noindent
\textbf{Keywords:} Low-rank matrices, Subspace perturbation, Network data, Signal-plus-noise, Random matrix theory.

\tableofcontents

\section{Introduction}

In the modern era, network-valued data has become central to scientific inquiry across diverse fields, such as  social science\citep{lazega2001collegial, girvan2002community} or neuroscience\citep{bullmore2009complex, kolaczyk2009statistical}. In these settings, the primary analytical goal is often to understand the underlying latent structure (such as communities) that dictates network behavior. A critical challenge in these settings is determining whether two networks share this structural geometry, even when their overall edge probabilities differ.  For example, communities may remain the same between social networking platforms, but edge formations can differ due to the particular platform.

To address this, we propose a spectral two-sample test for the equality of underlying population-level subspaces, assuming only that the edge probability matrices are low-rank.  In the case of networks with community structure such as stochastic blockmodels (\cref{def:SBM}) or mixed-membership stochastic blockmodels (\cref{def:MMSBM}), this hypothesis is equivalent to testing whether the communities are equal across networks.  Our theoretical results rely on a novel limit theorem for the projection distance of empirical eigenvectors, a result that may be of independent interest.

\subsection{Key Contributions}

The primary contributions of this work are threefold:

\begin{enumerate}
    
    \item \textbf{A Two-Sample Test for Common Subspaces:} 
 We propose a test statistic for the hypothesis that two networks arise from probability matrices with the same column space. We prove that, after appropriate centering and scaling, our test statistic converges in distribution to a standard Gaussian random variable. These limit theorems are valid over a wide parameter space regarding the network dimension $n$ and the sparsity parameter $\rho_n$. We also provide consistent plug-in estimators under a slightly stronger sparsity assumption.

    \item \textbf{Application to Synthetic and Real-World Data:} 
    We apply our test procedure to synthetic data, showing that the test exhibits strong power in distinguishing stochastic blockmodels with different community structures and is capable of detecting subtle shifts in mixed-membership models. Furthermore, we apply our methodology to United States domestic flight data. Our analysis reveals statistically significant shifts in the structural connectivity of the airport network corresponding to the onset of the COVID-19 pandemic.
    \item \textbf{A Limit Theorem for Projection Distances:} Our main technical results for our test are based on a corresponding limit theorem for the projection distance between subspaces, which can be viewed as the ``one-sample'' analogue of our test.  
This limit theorem, with associated plug-in estimators, can be used for confidence intervals     for the $\sin \Theta$ distance between the observed and population subspaces.  This result generalizes several previous results in the literature to Bernoulli noise, and may be of independent interest.  
    
\end{enumerate}

\subsection{Related Work}

The problem of hypothesis testing for network data has received considerable attention in recent years. Our work resides at the intersection of two-sample inference for random graphs and high-dimensional spectral perturbation theory.

\paragraph{Two-Sample Network Hypothesis Testing.}
A foundational nonparametric framework for determining if two networks are generated from the same random dot product graph (RDPG) latent positions was established by \citet{tang_nonparametric_2017}, later generalized by \citet{agterberg_nonparametric_2020}. Other nonparametric approaches include assumption-lean inference \citep{li_assumption-lean_2025} and tests for arbitrary vertices \citep{auerbach_testing_2022, wu_two-sample_2023}. While these tests are robust, they generally test for the equality of the precise generating parameters rather than the underlying  subspace up to rotation or scaling. 

Another significant strand of literature has focused on moment-based, subsampling, and bootstrap approaches. \citet{chakraborty_scalable_2025}, \citet{lunde_subsampling_2023}, \citet{deng_subsampling-based_2024}, and \citet{zu_local_2025} leverage subsampling techniques to construct test statistics, while others utilize spectral moments, graph cumulants, or U-statistics \citep{bravo-hermsdorff_quantifying_2023, shao_higher-order_nodate, shao_u-statistic_2023, shen_combinatorial-probabilistic_2020, zhang_edgeworth_2022, qi_multivariate_2024}, or employ bootstrap methods \citep{bhadra_bootstrap-based_2025, levin_bootstrapping_2025}. Alternatively, \citet{chatterjee_two-sample_2023} and \citet{ghoshdastidar_two-sample_2020} propose tests based on matrix norms.
In contrast to these strategies, our approach is based on direct spectral analysis. Instead of relying on empirical subsamples or higher-order moment matching, we establish the asymptotic normality of our test statistic directly via different technical matrix analysis tools. Our test statistic admits a closed-form limiting distribution, obviating the need for computationally intensive resampling or moment-computation procedures.

Related work has addressed hypothesis testing problems similar to ours, though primarily under the assumption of a Stochastic Blockmodel (SBM). For instance, \citet{bhattacharjee_change_2020} and \citet{wang_optimal_2021} address change-point detection in dynamic networks, which can be viewed as a sequential extension of two-sample testing. However, these approaches typically monitor global deviations in model parameters via least-squares objectives or cumulative sums, rather than directly testing the stability of the invariant subspace. Consequently, our statistic is distinct in its specific focus on the distance between spectral projectors, ensuring robustness to nuisance fluctuations in edge density that might otherwise be identified as change points. Methodological differences are also pronounced in comparison to \citet{fu_two-sample_2022}, who propose an extreme-value test based on the maximum standardized deviation of edges from estimated block probabilities. While their strategy is contingent on consistent community recovery, our statistic utilizes the Frobenius norm of the difference between spectral projectors. This $L_2$-type functional captures global discrepancies in the underlying subspace geometry, avoiding reliance on local fluctuations or the explicit alignment of community labels.

Beyond these specific formulations, recent literature has expanded network inference to address varied constraints. \citet{nguen_network_2024} and \citet{jin_two-sample_2024, jin_optimal_2025} develop tests for scenarios involving unknown vertex correspondence, unequal graph sizes, or degree-corrected mixed-membership assumptions. Other complementary contributions focus on rank estimation \citep{han_universal_2023}, mesoscale comparisons \citep{macdonald_mesoscale_2024}, and the sensitivity of testing power to label misalignment \citep{saxena_lost_2025}. Distinct from these specialized focuses, our work establishes a comprehensive asymptotic theory for the global projection distance itself. Our work is closely related to the semiparametric framework of \citet{tang_semiparametric_2017}, which motivates a similar spectral statistic but relies on permutation-based validity arguments. We, conversely, derive the exact Gaussian limiting distribution of the test statistic, enabling a direct asymptotic pivot for rigorous Type I error control.

The work most relevant to our work is \citet{li2018two}, which develops a two-sample test specifically for community memberships in weighted SBMs. However, their approach is restricted to a specific class of blockmodels. In contrast, our subspace-based formulation is significantly more general: by testing the null hypothesis of invariant column spaces, our method remains valid for both SBMs and Mixed Membership Stochastic Blockmodels (MMSBMs) and is robust to nuisance variation in edge probabilities.

\paragraph{Spectral Theory and Subspace Analysis.}
Our theoretical analysis relies on the behavior of empirical eigenvectors; see \citet{athreya_statistical_2018, agterberg_overview_2025}, and \citet{chen_spectral_2021} for overviews. Recent advances have established entrywise perturbation bounds under asymmetric noise \citep{chen_asymmetry_2021}, small eigengaps \citep{li_minimax_2025, wang_analysis_2024}, and incoherence \citep{yan_coherence-free_2024}. \citet{agterberg_joint_2025} provides perturbation bounds results for a slightly different model that support our technical analysis.

While many results provide perturbative bounds, our inference tools require distributional theory. The works \citet{fan_asymptotic_2022}, \citet{zheng_limit_2024}, and \citet{liu_asymptotic_2025} provide limit theorems for individual eigenvectors or eigenvalues. \citet{fan_asymptotic_2025} extends this to Laplacian matrices, and \citet{du_hypothesis_2023} and \citet{fan_simple-rc_2022} apply these results to testing membership equality of two individual vertices. \citet{agterberg_distributional_2024} and \citet{cheng_tackling_2021} address small eigengaps and asymmetry.  

In contrast to row-wise or ``local'' inference, our statistic depends on the global projection matrix. The primary theoretical antecedents are \citet{xia_normal_2021} and the closely related \citet{bao_singular_2021}, which derive normal approximations for singular subspaces. We also note complementary results for $\ell_{2,\infty}$ norms \citep{chang_extreme_2025} and tensor methods \citep{xia_inference_2022, agterberg_statistical_2024, xia_confidence_2019}. \citet{xia_normal_2021} provides explicit representation formulas for empirical spectral projectors. However, that work assumes i.i.d. Gaussian noise. In contrast, our technical analysis is tailored to heteroskedastic Bernoulli noise, and we leverage our results to develop our two-sample test.

\subsection{Notation}
We write \(J_n\) for the \(n\times n\) all-ones matrix, \(I_n\) for the \(n\times n\) identity matrix, and  \( \mathbf{1}_n\) for the all-ones vector of length \(n\).
For a vector \(\bm v \in \mathbb{R}^n\), the operator \(\Diag(\bm v)\) denotes the diagonal matrix with diagonal entries given by the components of \(\bm v\).  
For \(M \in \mathbb{R}^{n\times n}\), the operator \(\diag(M)\) returns the vector of diagonal entries of \(M\) and the operator \(\Diag(M)\) returns a diagonal matrix with only the diagonal entries of \(M\).  

For a matrix $M$, $M_{i \cdot}$ represents the $i$-th row and $M_{\cdot i}$ the $i$-th column of $M$. The $\ell_{2,\infty}$-norm is defined as  
$$\|M\|_{2,\infty} = \max_{i} \|M_{i\cdot}\|.$$
The Frobenius norm is denoted by $\|M\|_F$, the spectral norm by $\|M\|$, the infinity norm by $\|M\|_{\infty}=\max_{i,j} |M_{ij}|$, and the $\ell_0$-norm by $\|M\|_0$ (which denotes the number of non-zero elements in the matrix, and similarly for a vector $v$, $\|v\|_0$ denotes the number of non-zero elements in that vector). The column space of a matrix \(M\) is denoted by $\col(M)$, the trace by $\tr(M)$, and the least absolute eigenvalue by $\lambda_{\min}(M)$.  
For two matrices $M_1$ and $M_2$, we denote their Frobenius inner product by $\langle M_1, M_2 \rangle$, defined as $\langle M_1, M_2 \rangle = \tr(M_1 M_2^{\top})$.
Their Hadamard (entrywise) product is denoted by $M_1 \circ M_2$.
We write $A \sim \Ber(P)$ to indicate that
$$
A_{ij} =
\begin{cases}
\mathrm{Bernoulli}(P_{ij}), & i \le j,\\
A_{ji}, & i > j,
\end{cases}
$$
with the collection $\{A_{ij} : i \le j\}$ being mutually independent.

We write \(f(n) \lesssim g(n)\) to indicate that there exists a sufficiently large constant \(C>0\) such that  
\(f(n) \le C\, g(n)\) for all sufficiently large \(n\).  
We write \(f(n) \ll g(n)\) if \(f(n)/g(n)\to 0\) as \(n\to\infty\).  
Similarly, \(f(n) = O(g(n))\) denotes \(f(n)\lesssim g(n)\), and \(f(n) = o(g(n))\) denotes \(f(n)\ll g(n)\).  
We use \(f(n) \asymp g(n)\) to mean both \(f(n) = O(g(n))\) and \(g(n) = O(f(n))\).

If \(X\) is a random variable, then \(X = O_p(f(n))\) means  
\(
\mathbb{P}\left( X \lesssim f(n) \right) > 1 - O(n^{-c})
\quad \text{for some constant } c>1,
\)
and \(X=o_p(f(n))\) means  
\(
\mathbb{P}\left( X \ll f(n) \right) > 1 - O(n^{-c})
\quad \text{for some constant } c>0.
\) For a sequence of random variables $\{X_n\}_{n=1}^{\infty}$, we say $X_n \xrightarrow{D} X$ if $X_n$ converges to $X$ in distribution. 
Finally, we say that a random variable \(X\) satisfies \(X \asymp f(n)\) if for every \(\epsilon >0\) there exist constants \(a,b>0\), independent of \(n\), such that
\(
\mathbb{P}\left( a\,f(n) \le X \le b\,f(n) \right) \ge 1-\epsilon .
\)
\subsection{Organization of the Paper}
The remainder of this article is organized as follows. \cref{sec:two_sample} formulates the two-sample hypothesis testing problem, defines the test statistic, and presents the main asymptotic normality result and associated estimators. \cref{sec:simulations} validates the method through numerical simulations, while \cref{sec:airport} applies it to US flight network data. \cref{sec:one_samplepf} provides the theoretical foundation by establishing a limit theorem for the projection distance between empirical and population eigenspaces. Finally, \cref{sec:discussion} concludes with a discussion, and proofs are deferred to the Appendices.

\section{Two-Sample Hypothesis Testing for Subspaces}\label{sec:two_sample}
Suppose one observes two independent adjacency matrices \({A^{(1)}}\) and \({A^{(2)}}\), with ${A^{(i)}} \in \{0,1\}^{n\times n}$ assumed to be symmetric.  We consider the case where \({A^{(i)}} \sim \Ber(P^{(i)})\) for \(i = 1,2\), and each \(P^{(i)}\) is a symmetric, low-rank probability matrix of dimension \(n \times n\) and rank \(k\). Let $P^{(i)}$ have spectral decomposition \(P^{(i)} = {V^{(i)}} \Lambda^{(i)} {V^{(i)}}\t\), where \({V^{(i)}} \in \mathbb{R}^{n \times k}\) contains the $k$ leading orthonormal eigenvectors associated with the \(k\) eigenvalues of \(P^{(i)}\), collected in the diagonal matrix \(\Lambda^{(i)}\).  For simplicity of analysis, we allow self-loops (i.e., ${A^{(i)}}$ need not be hollow), but the results do not materially change is self-loops are disallowed. 

 Formally, we consider the hypotheses
\begin{align} \label{hypo2}
    H_0 : {V^{(1)}} {V^{(1)}}\t = {V^{(2)}} {V^{(2)}}\t 
    \quad \text{versus} \quad 
    H_1 : {V^{(1)}} {V^{(1)}}\t \neq {V^{(2)}} {V^{(2)}}\t.
\end{align}
In words, this hypothesis test seeks to determine whether the subspaces spanned by the columns of \({V^{(1)}}\) and \({V^{(2)}}\) coincide, which is equivalent to \eqref{hypo2}. To contextualize our hypothesis test within the framework of block models and latent position models, we provide formal definitions for the generalized random dot product graph, the stochastic blockmodel, and the mixed-membership stochastic blockmodel.

\begin{definition}[Generalized Random Dot Product Graph (GRDPG)]\label{def:GRDPG}
Let \(n\in\mathbb{N}\) be the number of vertices and \(k\ge 1\) be the latent dimension.
Let \(I_{p,q} \in \mathbb{R}^{k \times k}\) be a diagonal matrix with \(p\) ones and \(q\) minus ones on the diagonal, where \(p+q=k\).
Let \(\Gamma\in\mathbb{R}^{n\times k}\) denote the matrix of latent positions.
Suppose the rows of \(\Gamma\) satisfy the constraint that for all \(i,j \in \{1,\dots,n\}\), \(0 \le ( \Gamma I_{p,q} \Gamma\t )_{ij} \le 1\).
We say \(A\) is an instantiation of a \emph{generalized random dot product graph}, denoted \(A \sim \text{GRDPG}_{p,q}(\Gamma)\), if \(A \sim \Ber(P)\) with \(P = \Gamma I_{p,q} \Gamma\t\).
\end{definition}

\begin{definition}[Stochastic Blockmodel (SBM)]\label{def:SBM}
Let \(n\in\mathbb{N}\) be the number of vertices and let \(k\ge 1\) be the number of communities.  
Let \(B\in[0,1]^{k\times k}\) denote the symmetric connectivity matrix.  
Suppose each vertex $i$ belongs to some community $C(i) \in \{1, \dots, k\}$. 
Define the matrix $Z \in \{0,1\}^{n \times k}$ via $Z_{ik} = 1$ if vertex $i$ belongs to community $k$, and $Z_{ik} = 0$ otherwise.  
We say $A$ is an instantiation of a \emph{stochastic blockmodel} if \(A \sim \Ber(  P)\) with \(P=Z B Z\t.\)
\end{definition}

\begin{definition}[Mixed membership stochastic blockmodel (MMSBM)]\label{def:MMSBM}
Let \(n\in\mathbb{N}\) be the number of vertices and \(k\ge 1\) be the number of communities. Let \(B\in[0,1]^{k\times k}\) be the symmetric connectivity matrix. For each vertex \(i\in\{1,\dots,n\}\), let
\begin{align} \label{simplex}
    Z_{i\cdot}=(Z_{i1},\dots,Z_{ik})\in\Delta^{k-1}:=\bigg\{x\in[0,1]^k:\sum_{j=1}^k x_j=1\bigg\}
\end{align}
be a membership vector representing the fractional affiliation of vertex \(i\) to each of the \(k\) communities. We say $A$ is an instantiation of a \emph{mixed membership stochastic blockmodel} if \(A \sim \Ber(  P)\) with \(P=Z B Z\t.\)
\end{definition}

All the models described above satisfy $\mathbb{E} \left(A\right) = P$, where $P$ has rank at most $k$. The following propositions examine how the general null hypothesis \eqref{hypo2} applies to specific model cases.  We begin with the GRDPG,
where the test admits a geometric interpretation regarding the column space of the latent positions.

\begin{proposition}[Equivalence of Latent Subspaces in GRDPG] \label{prop:GRDPG-equivalence}
Let \(A^{(1)}\) and \(A^{(2)}\) be independent instantiations of generalized random dot product graphs with latent position matrices \(\Gamma^{(1)}, \Gamma^{(2)} \in \mathbb{R}^{n \times k}\).
Assume \(\Gamma^{(1)}\) and \(\Gamma^{(2)}\) have full column rank \(k\).
Let the associated diagonal matrices be \(I_{p_1,q_1}\) and \(I_{p_2,q_2}\), where \(p_i+q_i = k\), such that \(P^{(i)} = \Gamma^{(i)} I_{p_i,q_i} {\Gamma^{(i)}}\t\) for \(i=1,2\).
Let \(V^{(i)} \in \mathbb{R}^{n \times k}\) be the matrix of \(k\) orthonormal eigenvectors corresponding to the non-zero eigenvalues of \(P^{(i)}\).
Then
\[
    V^{(1)} {V^{(1)}}\t = V^{(2)} {V^{(2)}}\t
\]
if and only if there exists a non-singular matrix \(W \in \mathbb{R}^{k \times k}\) such that
\(
    \Gamma^{(1)} = \Gamma^{(2)} W.
\)
\end{proposition}

\begin{proof}
    See \cref{prop1}.
\end{proof}

\noindent In essence, \cref{prop:GRDPG-equivalence} states that for GRDPGs, the hypothesis test determines whether the latent positions of the two graphs span the same feature space. The test is invariant to the linear transformation \(W\), meaning it detects fundamental changes in the geometric configuration of the latent positions rather than simple rotations or scaling.

When we restrict our attention to community based models, this geometric equivalence forces a stricter structural correspondence. For two matrices $M_1,M_2$, we write $M_1 \perpeq M_2$ if there exists a permutation matrix $\Pi$ such that $M_1 = M_2\Pi$.

\begin{proposition}\label{prop:subspace-vs-membership}
Let \(P^{(i)}\) be one of the blockmodels with a full rank symmetric connectivity matrix as defined in \cref{def:MMSBM,def:SBM}. Let \({V^{(i)}}\in\mathbb{R}^{n\times k}\) be an orthonormal matrix whose columns span the column space of \(P^{(i)}\) (the nonzero population eigenvectors). We further assume that there is a pure node for each community for both $Z^{(1)}$ and $Z^{(2)}$. Then
\[
{V^{(1)}}{V^{(1)}}\t = {V^{(2)}}{V^{(2)}}\t
\]
if and only if \(Z^{(1)} \perpeq Z^{(2)}\).
\end{proposition}

\begin{proof}
    See \cref{prop2}.
\end{proof}

Thus, it follows from \cref{prop:subspace-vs-membership} that \eqref{hypo2} is equivalent to
\begin{align} \label{hypo1}
    H_0 : Z^{(1)} \perpeq Z^{(2)}
\quad \text{versus} \quad
H_1 : Z^{(1)} \notperpeq Z^{(2)}
\end{align}
under both the SBM and MMSBM. Testing \({V^{(1)}}{V^{(1)}}\t = {V^{(2)}}{V^{(2)}}\t\) is, in effect, a test of whether the two networks encode the same underlying community structure even if the probability matrices \(P_1\) and \(P_2\) are different. The airport network analysis in \cref{sec:airport} offers a concrete illustration: our procedure successfully distinguishes between periods of stable community structure and periods of severe structural disruption, such as the COVID-19 pandemic. Full empirical results and their interpretation are presented in Section~\ref{sec:airport}. Hence, a test of \eqref{hypo2} furnishes a principled and interpretable mechanism for detecting structural changes in community organization while remaining robust to transient edge-level fluctuations.

\subsection{Test Statistic Asymptotics}
Before proposing our test statistics and studying their asymptotics, we impose several mild regularity conditions on the parameter space. 
For notational simplicity, we state the assumptions using a single adjacency matrix \(A\) and its population counterpart \(P=\mathbb{E}(A)\). We assume that \(P\) admits the spectral decomposition \(P = V \Lambda V^\top\), where \(V \in \mathbb{R}^{n \times k}\) is the matrix of eigenvectors and \(\Lambda\) is the diagonal matrix of eigenvalues with the \(i\)-th eigenvalue given by \(\lambda_i\).

Our first assumption concerns the asymptotic regime and sparsity of each network.
\begin{assumption}[Asymptotic regime]\label{ass:asymptotic1}
The edge probabilities are uniformly of order \(\rho_n\); that is, there exist constants \(0<a<b<\infty\) such that
\(
a\,\rho_n \le P_{ij} \le b\,\rho_n\) for all \(1\le i,j\le n.
\)
The number of communities \(k\) remains fixed as \(n\to\infty\), and the sparsity parameter satisfies \(\rho_n\to 0\).  
\end{assumption}

\begin{assumption}[Sparse regime]\label{ass:sparse}
The sparsity parameter satisfies 
\(
n\rho_n \gg \log n
\)
for sufficiently large \(n\).
\end{assumption}
It is well-known that if $n\rho_n \ll \log n$, then the network is disconnected with high probability, and the Frobenius distance between the empirical and population subspace projection matrices fails to converge to zero \citep{agterberg_overview_2025,lei2015consistency}.  Thus, our assumption that $n\rho_n \gg \log n$ is relatively mild, as we focus on \emph{inference}, not estimation. For simplicity we assume that $\rho_n$ is of the same order for both networks.

Our next assumption places conditions on the leading $k$ eigenvalues of each matrix.
\begin{assumption}[Eigenvalue scaling]\label{ass:eigen-scaling}
There exist positive constants \(C_i,D_i\) (independent of \(n\)) such that, for each \(i=1,\dots,k\),
\(
C_i < \frac{|\lambda_i|}{n\rho_n} < D_i
\).
\end{assumption}
This assumption imposes only a mild regularity condition on the eigenvalues of the probability matrix \(P\), and is satisfied in virtually all of the canonical instances of the blockmodel family. Its validity is formalized in the proposition below. We first include the following definition.
\begin{definition}[Balanced Community Sizes]\label{def:balanced}
We say that the community sizes are balanced if the following conditions hold:
\begin{itemize}
    \item \textbf{For SBM:} Let \(n_l\) denote the size of community \(l\), such that \(n_l = n\pi_l\), where proportions \(\pi_l\) satisfy
    \begin{align*}
        C_{1} \le \min_{l}\pi_{l} \le \max_{l}\pi_{l} \le C_{2},
    \end{align*}
    for some fixed constants \(C_{1}, C_{2} > 0\).
    \item \textbf{For MMSB:} The membership matrix \(Z\) satisfies
    \[
        \lambda_{\min}(Z^\top Z) \ge c\,\frac{n}{k},
    \]
    for some constant \(c>0\).
\end{itemize}
\end{definition}
\begin{proposition}\label{prop:eig-scaling-sbm-mmsbm}
Let \(P\) be one of the blockmodels as defined in \cref{def:MMSBM,def:SBM}. Define \(\widetilde{B} = \frac{1}{\rho_n} B\) where we assume that the eigenvalues of $\widetilde{B}$ are bounded; specifically, there exist constants $0 < a_1 < a_2 < \infty$ such that
\(
    a_1 < |\lambda_{\min}(\widetilde{B})| \le |\lambda_{\max}(\widetilde{B})| < a_2.
\)
Furthermore, assume that the community sizes are balanced in the sense of \cref{def:balanced} and that Assumptions \ref{ass:asymptotic1} and \ref{ass:sparse} are satisfied. Then the \(k\) nonzero eigenvalues \(\lambda_{1},\dots,\lambda_{k}\) of \(P\) satisfy \cref{ass:eigen-scaling}.
\end{proposition}
\begin{proof}
    See \cref{prop3}.
\end{proof}
Finally, the following assumption imposes the condition that the rows of $V$ (denoted by \(V_{1 \cdot},\dots,V_{n \cdot}\)) are uniformly \emph{spread}, a concept known as \emph{incoherence} in the literature.  In the particular case of the SBM or MMSBM, this assumption follows if we assume that the community sizes are balanced as we will demonstrate in \cref{lem:balanced-incoherence,lem:balanced-mmsb-incoherence}.
\begin{assumption}[Incoherence]\label{ass:incoherence}
The matrix of population eigenvectors \(V\) satisfies the following incoherence  condition: there exist a constant \(\psi>0\) (independent of \(n\)) such that for all \(j=1, \dots,n,\)
\[  
\|V_{j\cdot}\|_{2} 
\le \sqrt{\frac{\psi k}{n}}.
\]
\end{assumption}
Let \(\hat V^{(i)}\in\mathbb{R}^{n\times k}\) denote the matrix of leading \(k\) eigenvectors of ${A^{(i)}}$ with \(\hat{\Lambda}^{(i)}\)  the corresponding eigenvalue matrix.  We consider the test statistic
\begin{align} \label{2sampleTS}
    T_n := \bigl\|(\hat V^{(1)})(\hat V^{(1)})\t - (\hat V^{(2)})(\hat V^{(2)})\t\bigr\|_F^2.
\end{align}
The following theorem characterizes its null distribution after appropriate centering and scaling. The full proof can be found in \cref{sec:mainthm2}. Here, \(\tilde\mu_2\) and \(\tilde\sigma_2\) denote the sample mean and sample standard deviation of \eqref{2sampleTS}, respectively, with the subscript indicating the two-sample setting. Subsequently, we will consider a one-sample analogue of the following result; see \cref{thm:onesample}. In that case, \(\mu_1\) and \(\sigma_1\) will denote the corresponding one-sample mean and standard deviation.

\begin{theorem}
    \label{thm:twosample}
    Suppose \cref{ass:asymptotic1,ass:sparse,ass:eigen-scaling,ass:incoherence} hold. Then, under the null as defined in \eqref{hypo2},
    \begin{equation}
    \label{eq:clt2}
    \frac{T_n - \tilde \mu_{2}}{\tilde \sigma_2}
\xrightarrow{D} \mathcal{N}(0,1),
    \end{equation}
    where the asymptotic centering and scaling terms satisfy
    \begin{align}
    \tilde \mu_{2}
    &= \underbrace{2\sum_{i=1}^{2}\,\operatorname{tr}\biggl[\beta_i^{-2}\Bigl({\Sigma^{(i)}}\circ\beta_i^\perp 
      + \operatorname{Diag}\bigl({\Sigma^{(i)}}\cdot d_i - \operatorname{diag}({\Sigma^{(i)}}) \circ d_i\bigr)\Bigr)\biggr]}_{\mu_2}
      + O_p\left(\frac{k}{n^2 \rho_n^2}\right),
    \label{param:1} \\[1em]
    \tilde \sigma_2^2
    &= \sigma_2^2 + o\left(\frac{k^2}{n^3 \rho_n^2}\right) \nonumber \\[0.5em]
    &= \begin{aligned}[t]
        &\sum_{i=1}^{2} \biggl( 8 \Bigl\langle (J_n - I_n)\circ(\beta_i^{-2}\circ\beta_i^{-2}),({\Sigma^{(i)}})^2\Bigr\rangle 
      + 4 \,\alpha_i^\top K^{(i)}\,\alpha_i \biggr) \\
        &+ 4 \Bigl\langle G \circ G,{\Sigma^{(1)}}^\top{\Sigma^{(2)}} 
      + \left(J_n-I_n\right) \circ \left({\Sigma^{(1)}} \circ {\Sigma^{(2)}}\right)\Bigr\rangle + o\left(\frac{k^2}{n^3 \rho_n^2}\right).
    \end{aligned}
    \label{param:2}
\end{align}  Here we define
    \begin{align}\label{last00}
    &\beta_i^k = {V^{(i)}}({\Lambda^{(i)}})^k {V^{(i)}}\t,
    \quad \beta_i^\perp = I_n - {V^{(i)}} {V^{(i)}}\t, 
    \quad d_i = \diag(\beta_i^\perp),\notag \\
    &G=\beta_1^{-1}\beta_2^{-1} + \beta_2^{-1}\beta_1^{-1},
    \quad {\Sigma^{(i)}} = P^{(i)}\circ(J_n-P^{(i)}),
    \quad K^{(i)} = 2\,{\Sigma^{(i)}}\circ(J_n - 2P^{(i)}),
    \quad \alpha_i = \diag(\beta_i^{-2}).
    \end{align}
    
\end{theorem}
\begin{proof}
    See \cref{sec:mainproof}.
\end{proof}
The residual contributions \(O_p\bigl(\tfrac{k}{n^{2}\rho_n^{2}}\bigr)\) and 
\(o\bigl(\tfrac{k^{2}}{n^{3}\rho_n^{2}}\bigr)\) appearing in \cref{param:1} and 
\cref{param:2} arise from the accumulation of all higher–order terms in the expansion underlying our statistic.  
These terms remain controlled under the asymptotic regime of 
\cref{ass:asymptotic1,ass:sparse}, but they do not vanish automatically when $n\rho_n \gg \log n$.  
In particular, when the graph is moderately sparse, their magnitude is 
asymptotically negligible relative to the leading terms, but can still influence the centering and scaling terms. If, however, we strengthen the sparsity condition to a denser regime, then all 
such residual higher–order terms become negligible.  
This phenomenon is rigorously verified in \cref{thm:twosample_dense}, and motivates the additional assumption stated below.

\newcounter{saveAssumptionCount}
\setcounter{saveAssumptionCount}{\value{assumption}}

\setcounter{assumption}{0}

\begingroup
\renewcommand{\theassumption}{2$'$}

\renewcommand{\theHassumption}{1prime} 

\begin{assumption}[Dense Regime]\label{ass:dense}
The sparsity parameter satisfies
\[
n\rho_n \gg \sqrt n .
\]
\end{assumption}
\endgroup

\setcounter{assumption}{\value{saveAssumptionCount}}

\begin{theorem}\label{thm:twosample_dense}
Suppose \cref{ass:asymptotic1,ass:dense,ass:eigen-scaling,ass:incoherence} hold.  
Then, under the null as defined in \eqref{hypo2},
\begin{equation}\label{eq:clt2_dense}
\frac{T_n -  \mu_{2}}{ \sigma_2}
\xrightarrow{D} \mathcal{N}(0,1),
\end{equation}
where the asymptotic centering and scaling terms are given by
\begin{align}
 \mu_{2}
&= 2\sum_{i=1}^{2}\,
   \tr\Bigl[\beta_i^{-2}\bigl({\Sigma^{(i)}}\circ\beta_i^\perp 
   + \Diag\bigl({\Sigma^{(i)}}\cdot d_i - \diag({\Sigma^{(i)}})\circ d_i\bigr)\bigr)\Bigr],\\[4pt]
 \sigma_2^2
&= \sum_{i=1}^{2} \Bigl(
      8\left\langle (J_n-I_n)\circ(\beta_i^{-2}\circ\beta_i^{-2}),{\Sigma^{(i)}}^{2}\right\rangle
      + 4\,\alpha_i\t K^{(i)} \alpha_i
    \Bigr) \\
&\qquad\qquad
    + 4\Bigl\langle G\circ G,
       {\Sigma^{(1)}}\t{\Sigma^{(2)}} + (J_n-I_n)\circ({\Sigma^{(1)}}\circ {\Sigma^{(2)}})
      \Bigr\rangle ,
\end{align}
where all notation coincides with that of \cref{thm:twosample}.
\end{theorem}
\begin{proof}
See \cref{thm:twosample_densepf}.
\end{proof}

\noindent
We emphasize that the validity of our two–sample limit theorem does not require \cref{ass:dense}; \cref{thm:twosample} already establishes the asymptotic 
normality under the weaker sparsity conditions in \cref{ass:sparse}.  
The role of \cref{ass:dense} is instead methodological: in the dense regime, the 
dominant components in \cref{param:1} and \cref{param:2} can be estimated with 
sufficient accuracy so that the resulting data–driven test remains consistent.  
This observation motivates the construction of the estimators in the next 
section, which are designed specifically to recover the leading terms appearing 
in the asymptotic centering and scaling expressions.


\subsection{Data-Driven Estimators}
\label{sec:data-driven}

To carry out valid inference using our proposed two-sample test,  
we require consistent, data-driven estimators of the dominant mean and variance
terms \(\mu_{2}\) and \(\sigma_{2}\) appearing in the Gaussian limits established in
\cref{thm:twosample_dense}. For the plug-in procedure to remain valid, the estimators \(\hat{\mu}_2\) and \(\hat{\sigma}_2\) must satisfy: 
\begin{align} \label{rates}
    \frac{\hat{\mu}_2 - \mu_2}{\hat{\sigma}_2} \xrightarrow{p} 0, \quad \text{and} \quad \frac{\hat{\sigma}_2}{\sigma_2} \xrightarrow{p} 1.
\end{align}
These requirements ensure that the plug-in centering and scaling do not distort
the limiting distribution, thereby yielding a fully data-driven inference
procedure whose Type-I error and power remain asymptotically valid.

Below we describe a principled plug-in procedure that satisfies \cref{rates}.  The parameters \(\mu_2\) and \(\sigma_2\) are functions of the population eigen-structure \(({V^{(i)}},\Lambda^{(i)})\) for \(i=1,2\).  Since \(({V^{(i)}},\Lambda^{(i)})\) are not observed, we estimate them from the data; for the eigenvectors we use the empirical eigenvectors \(\hat V^{(i)}\), while for the eigenvalues we use the empirical eigenvalues.

Motivated by the explicit expressions for \(\mu_{2}\) and \(\sigma_{2}\) in \cref{thm:twosample_dense}, we next define the following data-driven estimators by plug-in:
\begin{equation}\label{eq:estimators}
\begin{aligned}
\hat{\mu}_2 
&= 2\sum_{i=1}^{2}\,\tr\Bigl[\hat{\beta}_i^{-2}\Bigl((\hat{\Sigma}^{(i)})\circ \hat\beta_i^\perp 
      + \Diag\bigl((\hat{\Sigma}^{(i)})\cdot \hat d_i \bigr)\Bigr)\Bigr], \\[4pt]
\hat{\sigma}^2_2 
&= \sum_{i=1}^{2} \left(8 \bigl\langle (J_n - I_n)\circ(\hat{\beta}_i^{-2}\circ\hat{\beta}_i^{-2}),(\hat{\Sigma}^{(i)})^2\bigr\rangle
      + 4 \, \hat{\alpha}_i\t {\hat{K}^{(i)}}\,\hat{\alpha}_i \right)\\[3pt]
&\qquad\qquad+ 4 \bigl\langle \hat{G}\circ\hat{G},
      (\hat{\Sigma}^{(1)})\t(\hat{\Sigma}^{(2)}) + (J_n-I_n)\circ((\hat{\Sigma}^{(1)})\circ(\hat{\Sigma}^{(2)}))\bigr\rangle,
\end{aligned}
\end{equation}
where, for \(i=1,2\),
\begin{equation}
\begin{aligned}\label{thm3.0}
    {\hat{P}^{(i)}} = (\hat V^{(i)}) (\hat{\Lambda}^{(i)}) (\hat V^{(i)})\t, 
\qquad
\hat{\beta}_i^k = (\hat V^{(i)}) (\hat{\Lambda}^{(i)})^k (\hat V^{(i)})\t,
\qquad
\hat\beta_i^\perp = I_n - (\hat V^{(i)})(\hat V^{(i)})\t, 
\qquad
\hat d_i = \diag(\hat\beta_i^\perp),\\
\hat{\Sigma}^{(i)} = {\hat{P}^{(i)}}\circ(J_n-{\hat{P}^{(i)}}),
\qquad 
{\hat{K}^{(i)}} = 2\,(\hat{\Sigma}^{(i)})\circ(J_n - 2{\hat{P}^{(i)}}),
\qquad 
\hat{\alpha}_i = \diag(\hat{\beta}_i^{-2}), \quad \hat{G} = \hat{\beta}_1^{-1}\hat{\beta}_2^{-1} + \hat{\beta}_2^{-1}\hat{\beta}_1^{-1}.
\end{aligned}
\end{equation}
\cref{thm:2sampleconsis} establishes the consistency properties of the estimators \(\hat{\mu}_{2}\) and \(\hat{\sigma}_{2}\). 
\begin{theorem}[Mean and Variance consistency]\label{thm:2sampleconsis}
Let \(\mu_2\) and \(\sigma_2^2\) denote the theoretical mean and variance of the two-sample test statistic as defined in \cref{thm:twosample}.
Define the plug-in mean estimator \(\hat\mu_2\) and variance estimator \(\hat\sigma_2^2\) as in \cref{eq:estimators}.  Then under \cref{ass:asymptotic1,ass:dense,ass:eigen-scaling,ass:incoherence},
\[
|\hat{\mu}_2 - \mu_2| =O_p \left(\max \left(\frac{k}{n^2 \rho_n^2},\frac{k \sqrt{\log n}}{n^2 \rho_n^{3/2}}\right)\right),
\qquad
|\hat{\sigma}_2^2 - \sigma_2^2| = o_p \left(\frac{k^2}{n^3 \rho_n^2}\right).
\]
In particular, \(
\frac{\hat{\mu}_2 - \mu_2}{\hat{\sigma}_2} \xrightarrow{p} 0,
\) and \(
\frac{\hat{\sigma}_2}{\sigma_2} \xrightarrow{p} 1
\) as required.
\end{theorem}
\begin{proof}
    See \cref{sec:estproof}.
\end{proof}
\begin{remark}[Proof Strategy]
The proof of \cref{thm:2sampleconsis} proceeds via a systematic decomposition: (i) expanding the plug-in estimation error into sums involving deviations of empirical eigenvectors and eigenvalues from their population counterparts; (ii) controlling the empirical eigenvector and eigenvalue fluctuations; and (iii) verifying that all remainder terms are asymptotically negligible in the dense regime.
While this high-level roadmap is standard, the novelty of our analysis lies in the intricate handling of the higher-order terms required by the Bernoulli noise model. 
Unlike homoskedastic Gaussian settings where perturbation terms are often amenable to simpler bounds, the adjacency matrix entails heteroskedasticity and discreteness. 
Our analysis necessitates a delicate higher-order expansion to ensure that the complex variance structure of the Bernoulli entries does not invalidate the consistency rates.
Crucially, despite these theoretical challenges, the resulting estimators remain algorithmically attractive: they are closed-form functions of the observed empirical eigen-structure, permitting valid inference without requiring complex auxiliary optimization or regularization.
\end{remark}
\noindent
Our full testing procedure is summarized in \cref{alg:hypo_test}.

\begin{algorithm}[H]
\caption{Hypothesis Test}\label{alg:hypo_test}
\begin{algorithmic}[1]
  \Require Two adjacency matrices \( A^{(1)}, A^{(2)} \in \{0,1\}^{n \times n} \); significance level \(\alpha>0\); rank of \(P^{(i)} = k\)
  \Ensure Decision regarding hypothesis \eqref{hypo2} (``Accept'' or ``Reject'')
  
  \State \textbf{Step 1 (Estimators).}  
  Compute the estimators \(\widehat{\mu}_{2}\) and \(\widehat{\sigma}_{2}\) defined in \cref{eq:estimators}, where \(\hat{V}^{(i)}\) denotes the matrix of eigenvectors corresponding to the \(k\) eigenvalues of \(A_{i}\) with the largest absolute values, and \(\hat{\Lambda}_{i}\) is the diagonal matrix containing these \(k\) eigenvalues.
  
  \State \textbf{Step 2 (Test statistic).}  
  Compute
  \[
    Z
    =
    \frac{\bigl\|(\widehat{V}^{(1)})(\widehat{V}^{(1)})^{\top}
           -(\widehat{V}^{(2)})(\widehat{V}^{(2)})^{\top}\bigr\|_{F}^{2}
          -\widehat{\mu}_{2}}
         {\widehat{\sigma}_{2}}.
  \]
  
  \State \textbf{Step 3 (Decision).}  
  Reject the null hypothesis in \eqref{hypo2} if
  \[
    \bigl|\Phi^{-1}(Z)\bigr| > \frac{\alpha}{2}.
  \]
  
\end{algorithmic}
\end{algorithm}

\subsection{Test Consistency and Application to Blockmodel Families} \label{TestConsistency}
Under the null hypothesis $H_0: {V^{(1)}} {V^{(1)}}\t= {V^{(2)}} {V^{(2)}}\t$, \cref{thm:twosample_dense} guarantees that the test statistic is asymptotically standard normal, and the estimator consistency results in the previous section guarantee asymptotic Type I error control at any desired level $\alpha$.  Therefore, in this section we consider the consistency of our full testing procedure under general alternatives and the important case of blockmodel families.   

We study the local alternative hypothesis $H_1:{V^{(1)}} {V^{(1)}}\t \neq {V^{(2)}} {V^{(2)}}\t$. The following result establishes the power of our test under local alternatives.  
\begin{theorem}\label{thm:testcon}
Suppose \cref{ass:asymptotic1,ass:eigen-scaling,ass:incoherence,ass:dense} hold. Let the two–sample test statistic be as in \cref{2sampleTS} and consider the testing problem stated in \eqref{hypo2}. If the alternative satisfies the condition
\begin{equation}\label{eq:lowerbound_Frobenius}
\|{V^{(1)}}{V^{(1)}}\t - {V^{(2)}}{V^{(2)}}\t\|_F^2 \gg \frac{k}{n^{3/2}\,\rho_n},
\end{equation}
then the test is consistent in the sense that the power converges to one as \(n\to\infty\). 
\end{theorem}
\begin{proof}
    See \cref{sec:testconpf}.  
\end{proof}

We now proceed to investigate what this separation condition entails in the context of specific blockmodel families. Our first result establishes the local power for SBM as in \cref{def:SBM}.

\begin{corollary} \label{cor:SBM}
Let \({A^{(i)}}\), \(i = 1, 2\),  denote the adjacency matrices of two SBMs with $k$ communities and balanced community sizes satisfying \cref{def:balanced}. Suppose that \(P^{(i)} = {Z^{(i)}} B^{(i)} {Z^{(i)}}\t\) for \(i = 1, 2\), where \({Z^{(i)}} \in \{0,1\}^{n \times k}\) is the community membership matrix, and \(B^{(i)}\) is the corresponding block matrix under model \(i\). Let \({Z_{\cdot l}^{(i)}} \in \{0,1\}^n\) be as in \cref{def:SBM}. Then, under \cref{ass:asymptotic1,ass:dense}, the proposed two-sample test statistic in \eqref{2sampleTS} has power converging to one for testing
\[
H_0 : Z^{(1)} \perpeq Z^{(2)} \quad \text{versus} \quad H_1 : Z^{(1)} \notperpeq Z^{(2)},
\]
provided that the following condition holds:
\( \quad
\max_{1 \leq l \leq k} \sum_{j=1}^{n} \left| Z_{jl}^{(1)} - Z_{jl}^{(2)} \right| \gg \frac{k}{n^{1/2}\rho_n}.
\)
\end{corollary}
\begin{proof}
    See \cref{sec:testconpf}.  
\end{proof}

The above corollary reveals that even if we swap the community memberships of as few as \(n_0 = k\) vertices, the null hypothesis \(H_0: Z^{(1)} \perpeq Z^{(2)}\) will be rejected with high probability, indicating strong power of the proposed test. This highlights the sensitivity of the test in detecting small but structured perturbations in community assignments. 

We will now consider MMSBMs.
\begin{corollary} \label{cor:MMSBM}
Let \({A^{(i)}}\) denote the adjacency matrices of two MMSBMs for \(i=1,2.\) Suppose that the membership matrices satisfy
\begin{align*}
    \lambda_{\min}({Z^{(i)}}\t {Z^{(i)}}) \ge c\,\frac{n}{k},
\end{align*}
where $c$ is some positive constant. Suppose that \(P^{(i)} = {Z^{(i)}} B^{(i)} {Z^{(i)}}\t\) for \(i = 1, 2\), where \({Z^{(i)}} \in [0,1]^{n \times k}\) is the mixed-membership matrix as in \cref{def:MMSBM}, and \(B^{(i)}\) is the corresponding block matrix under model \(i\). Then, under \cref{ass:asymptotic1,ass:dense}, the proposed two-sample test statistic in \eqref{2sampleTS} has power converging to one for testing
\[
H_0 : Z^{(1)} \perpeq Z^{(2)} \quad \text{versus} \quad H_1 : Z^{(1)} \notperpeq Z^{(2)},
\]
provided that the following condition holds:
\( \quad
\|Z^{(1)} - Z^{(2)}\|_F^2 \gg \frac{k}{n^{1/2}\rho_n}.
\)
\end{corollary}
\begin{proof}
    See \cref{sec:testconpf}.  
\end{proof}
The corollary states that, in the MMSBM setting, detectability of a difference between two network models is governed by the squared Frobenius distance between their membership matrices, with the scaling determined by the sample size \(n\), the rank \(k\), and the sparsity parameter \(\rho_n\). This condition naturally extends the SBM result to the soft–assignment regime: larger deviations of the entire membership profile (measured in Frobenius norm) are required for reliable discrimination when the networks are sparser.

\section{Numerical Simulations}\label{sec:simulations}
In this section we consider simulations under both the SBM and the MMSBM.  

\subsection{Stochastic Blockmodels}
First, we apply \cref{alg:hypo_test} to SBMs.  For each specification of \(k\)  (the number of communities), we construct two block matrices \(B^{(1)}\) and \(B^{(2)}\). The diagonal entries of these matrices are independently sampled from \([0.8,1]\), while the off-diagonal entries of \(B^{(1)}\) are sampled from \([0.3,0.5]\) and those of \(B^{(2)}\) are sampled from \([0.1,0.3]\).  

For \(k=3\), the two block matrices (rounded to two decimal places) are given below:

\[
B^{(1)} =
\begin{bmatrix}
0.98 & 0.36 & 0.46 \\
0.36 & 0.99 & 0.38 \\
0.46 & 0.38 & 0.81
\end{bmatrix},
\qquad
B^{(2)} =
\begin{bmatrix}
0.97 & 0.12 & 0.14 \\
0.12 & 0.96 & 0.25 \\
0.14 & 0.25 & 0.87
\end{bmatrix}.
\]
For \(k=5\), the two block matrices (rounded to two decimal places) are given below:
\[
B^{(1)} =
\begin{bmatrix}
0.99 & 0.36 & 0.46 & 0.48 & 0.41 \\
0.36 & 0.89 & 0.38 & 0.49 & 0.48 \\
0.46 & 0.38 & 0.94 & 0.31 & 0.41 \\
0.48 & 0.49 & 0.31 & 0.91 & 0.39 \\
0.41 & 0.48 & 0.41 & 0.39 & 0.82
\end{bmatrix},
\qquad
B^{(2)} =
\begin{bmatrix}
0.87 & 0.12 & 0.14 & 0.27 & 0.12 \\
0.12 & 0.84 & 0.25 & 0.26 & 0.16 \\
0.14 & 0.25 & 0.95 & 0.17 & 0.15 \\
0.27 & 0.26 & 0.17 & 0.96 & 0.18 \\
0.12 & 0.16 & 0.15 & 0.18 & 0.92
\end{bmatrix}.
\]

To construct the community assignment for $Z^{(1)}$, we allocate the $n$ nodes uniformly across the $k$ communities, such that each community contains exactly $n/k$ nodes. The subsequent assignment $Z^{(2)}$ is generated by initializing $Z^{(2)} = Z^{(1)}$, selecting $n_0$ nodes uniformly at random, and reassigning their community labels to a different class, while leaving the memberships of the remaining $n - n_0$ nodes unchanged. We then form the two probability matrices as \(P^{(i)} = {Z^{(i)}} B^{(i)} {Z^{(i)}}\t\), and generate the observed adjacency matrices according to \({A^{(i)}} \sim \text{Ber}(P^{(i)})\). For each fixed combination of \(n,k,\rho_n\), and \(n_0\), we simulate the data set 100 times and apply the test at \(\alpha=0.05\) level of significance, recording the number of rejections. The empirical power of the test at those parameters is then estimated by the fraction of rejections. \cref{tab:mmsbm-k3-rho04-09,tab:mmsbm-k5-rho04-09} report these empirical powers, with rows corresponding to different values of \(\rho_n\) and columns corresponding to different values of \(n_0\). The first row of each table ($n_0 = 0$) reflects the performance of the proposed test under the null hypothesis in \eqref{hypo2}. While finite-sample empirical sizes may exhibit slight deviations from the exact nominal level $\alpha = 0.05$, they remain closely aligned. Furthermore, as anticipated by our theoretical analysis, this size calibration steadily improves, converging toward 0.05, as the density parameter $\rho_n$ increases toward 1 for a fixed $k$. Conversely, under the alternative hypothesis ($n_0 > 0$), the empirical power exhibits a sharp transition, converging rapidly to $1$ as the perturbation size $n_0$ increases down the rows, demonstrating the high sensitivity of the procedure. While we do not keep track of the dependence on $k$ in our theoretical results, we see that larger $k$ has weaker detection ability.

\begin{table}[H]
\centering
\caption{Level and power of \cref{alg:hypo_test} for SBMs with $k=3$}
\label{tab:mmsbm-k3-rho04-09}
\begin{subtable}[t]{0.49\textwidth}
\centering
\caption{$n=2400,\; k=3$}
\label{tab:mmsbm-k3-n2400}
\renewcommand{\arraystretch}{1.05}
\setlength{\tabcolsep}{3.5pt}
\resizebox{\linewidth}{!}{%
\begin{tabular}{|c||c|c|c|c|c|c|}
\hline
$n_0$ & $\rho_n=0.4$ & $\rho_n=0.5$ & $\rho_n=0.6$ & $\rho_n=0.7$ & $\rho_n=0.8$ & $\rho=0.9$ \\
\hline \hline
0  & 0.05 & 0.04 & 0.06 & 0.05 & 0.05 & 0.06 \\
\hline
2  & 1.00 & 1.00 & 1.00 & 1.00 & 1.00 & 1.00 \\
\hline
4  & 1.00 & 1.00 & 1.00 & 1.00 & 1.00 & 1.00 \\
\hline
6  & 1.00 & 1.00 & 1.00 & 1.00 & 1.00 & 1.00 \\
\hline
8  & 1.00 & 1.00 & 1.00 & 1.00 & 1.00 & 1.00 \\
\hline
10 & 1.00 & 1.00 & 1.00 & 1.00 & 1.00 & 1.00 \\
\hline
\end{tabular}}
\end{subtable}
\hfill
\begin{subtable}[t]{0.49\textwidth}
\centering
\caption{$n=3000,\; k=3$}
\label{tab:mmsbm-k3-n3000-rho04-09}
\renewcommand{\arraystretch}{1.05}
\setlength{\tabcolsep}{3.5pt}
\resizebox{\linewidth}{!}{%
\begin{tabular}{|c||c|c|c|c|c|c|}
\hline
$n_0$ & $\rho_n=0.4$ & $\rho_n=0.5$ & $\rho_n=0.6$ & $\rho_n=0.7$ & $\rho_n=0.8$ & $\rho_n=0.9$ \\
\hline \hline
0  & 0.03 & 0.03 & 0.05 & 0.08 & 0.03 & 0.04 \\
\hline
2  & 1.00 & 1.00 & 1.00 & 1.00 & 1.00 & 1.00 \\
\hline
4  & 1.00 & 1.00 & 1.00 & 1.00 & 1.00 & 1.00 \\
\hline
6  & 1.00 & 1.00 & 1.00 & 1.00 & 1.00 & 1.00 \\
\hline
8  & 1.00 & 1.00 & 1.00 & 1.00 & 1.00 & 1.00 \\
\hline
10 & 1.00 & 1.00 & 1.00 & 1.00 & 1.00 & 1.00 \\
\hline
\end{tabular}}
\end{subtable}
\end{table}

\begin{table}[H]
\centering
\caption{Level and power of \cref{alg:hypo_test} for SBMs with $k=5$}
\label{tab:mmsbm-k5-rho04-09}
\begin{subtable}[t]{0.49\textwidth}
\centering
\caption{$n=2500,\; k=5$}
\label{tab:mmsbm-k5-n2500}
\renewcommand{\arraystretch}{1.05}
\setlength{\tabcolsep}{3.5pt}
\resizebox{\linewidth}{!}{%
\begin{tabular}{|c||c|c|c|c|c|c|}
\hline
$n_0$ & $\rho_n=0.4$ & $\rho_n=0.5$ & $\rho_n=0.6$ & $\rho_n=0.7$ & $\rho_n=0.8$ & $\rho_n=0.9$ \\
\hline \hline
0  & 0.08 & 0.13 & 0.13 & 0.12 & 0.12 & 0.06 \\
\hline
2  & 0.36 & 0.53 & 0.77 & 0.93 & 1.00 & 1.00 \\
\hline
4  & 0.90 & 0.98 & 1.00 & 1.00 & 1.00 & 1.00 \\
\hline
6  & 0.99 & 1.00 & 1.00 & 1.00 & 1.00 & 1.00 \\
\hline
8  & 1.00 & 1.00 & 1.00 & 1.00 & 1.00 & 1.00 \\
\hline
10 & 1.00 & 1.00 & 1.00 & 1.00 & 1.00 & 1.00 \\
\hline
\end{tabular}}
\end{subtable}
\hfill
\begin{subtable}[t]{0.49\textwidth}
\centering
\caption{$n=3000,\; k=5$}
\label{tab:mmsbm-k5-n3000-rho04-09}
\renewcommand{\arraystretch}{1.05}
\setlength{\tabcolsep}{3.5pt}
\resizebox{\linewidth}{!}{%
\begin{tabular}{|c||c|c|c|c|c|c|}
\hline
$n_0$ & $\rho_n=0.4$ & $\rho_n=0.5$ & $\rho_n=0.6$ & $\rho_n=0.7$ & $\rho_n=0.8$ & $\rho_n=0.9$ \\
\hline \hline
0  & 0.08 & 0.06 & 0.05 & 0.08 & 0.12 & 0.07 \\
\hline
2  & 0.35 & 0.58 & 0.84 & 0.98 & 1.00 & 1.00 \\
\hline
4  & 0.97 & 1.00 & 1.00 & 1.00 & 1.00 & 1.00 \\
\hline
6  & 1.00 & 1.00 & 1.00 & 1.00 & 1.00 & 1.00 \\
\hline
8  & 1.00 & 1.00 & 1.00 & 1.00 & 1.00 & 1.00 \\
\hline
10 & 1.00 & 1.00 & 1.00 & 1.00 & 1.00 & 1.00 \\
\hline
\end{tabular}}
\end{subtable}
\end{table}

\subsection{Mixed Membership Stochastic Blockmodels}
Next, we apply \cref{alg:hypo_test} to MMSBMs. For each specification of the number of communities \(k\), we construct two block matrices \(B^{(1)}\) and \(B^{(2)}\). The diagonal entries of these matrices are independently sampled from \([0.8,1]\), while the off-diagonal entries of \(B^{(1)}\) are sampled from \([0.3,0.5]\) and those of \(B^{(2)}\) are sampled from \([0.1,0.3]\) 

For \(k=3\), the block matrices are given below:
\[
B^{(1)} =
\begin{bmatrix}
0.98 & 0.53 & 0.58 \\
0.53 & 0.99 & 0.54 \\
0.58 & 0.54 & 0.81
\end{bmatrix},
\qquad
B^{(2)} =
\begin{bmatrix}
0.97 & 0.11 & 0.12 \\
0.11 & 0.96 & 0.17 \\
0.12 & 0.17 & 0.87
\end{bmatrix}.
\]
For \(k=5\), the two block matrices are given below:
\[
B^{(1)} =
\begin{bmatrix}
0.99 & 0.53 & 0.58 & 0.59 & 0.55 \\
0.53 & 0.89 & 0.54 & 0.59 & 0.59 \\
0.58 & 0.54 & 0.94 & 0.50 & 0.56 \\
0.59 & 0.59 & 0.50 & 0.91 & 0.55 \\
0.55 & 0.59 & 0.56 & 0.55 & 0.82
\end{bmatrix},
\qquad
B^{(2)} =
\begin{bmatrix}
0.87 & 0.11 & 0.12 & 0.19 & 0.11 \\
0.11 & 0.84 & 0.17 & 0.18 & 0.13 \\
0.12 & 0.17 & 0.95 & 0.13 & 0.12 \\
0.19 & 0.18 & 0.13 & 0.96 & 0.14 \\
0.11 & 0.13 & 0.12 & 0.14 & 0.92
\end{bmatrix}.
\]

To construct the membership matrices, we proceed as follows. For \(Z^{(1)}\), we randomly select \(n_0\) nodes and assign their community membership vectors to be 
\(
\left( \tfrac{1}{k-1}, \tfrac{1}{k-1}, \ldots, \tfrac{1}{k-1}, 0 \right),
\)
while the membership vectors of all remaining nodes are independently sampled from a Dirichlet distribution with parameter vector 
\(
\left( \tfrac{1}{k}, \tfrac{1}{k}, \ldots, \tfrac{1}{k} \right).
\)
Similarly, for \(Z^{(2)}\), we choose the same \(n_0\) nodes as before and assign their community membership vectors to be
\(
\left( \tfrac{1}{k}, \tfrac{1}{k}, \ldots, \tfrac{1}{k} \right),
\)
while the remaining nodes are kept identical to those in \(Z^{(1)}\).  As before, we form the two probability matrices
\(
P^{(i)} = {Z^{(i)}} B^{(i)} {Z^{(i)}}\t\) for \(i=1,2,
\)
and generate the observed adjacency matrices according to 
\(
{A^{(i)}} \sim \text{Ber}(P^{(i)}).
\)

For each fixed combination of \(n,k,\rho_n\), and \(n_0\), we simulate the data set 100 times and apply the test at \(\alpha=0.05\) level of significance, recording the number of rejections. The empirical power at those parameters is then estimated by the fraction of rejections. \cref{tab:mmsbm-k3-power,tab:mmsbm-k5-power} report these empirical powers, with rows corresponding to different values of \(\rho_n\) and columns corresponding to different values of \(n_0\). Consistent with the results for the SBM, we observe analogous trends under the MMSBM. The first row ($n_0 = 0$) reflects the performance of the test under the null hypothesis in \eqref{hypo2}. The finite-sample empirical sizes remain closely aligned with the nominal level $\alpha = 0.05$, and this calibration steadily improves as the density parameter $\rho_n$ increases for a fixed $k$. Furthermore, under the alternative hypothesis ($n_0 > 0$), the empirical power again exhibits a stable convergence to $1$ as the perturbation size $n_0$ increases, aligning with our theoretical expectations.  The power with $k = 5$ tends to one more slowly than in the previous subsection, which suggests that detecting changes in this setup is more difficult than the case where networks are discrete stochastic blockmodels.

\begin{table}[H]
\centering
\caption{Level and power results for MMSBMs with $k=3$}
\label{tab:mmsbm-k3-power}
\begin{subtable}[t]{0.49\textwidth}
\centering
\caption{$n=2700,\; k=3$}
\label{tab:mmsbm-k3-n2700}
\renewcommand{\arraystretch}{1.05}
\setlength{\tabcolsep}{3.5pt}
\resizebox{\linewidth}{!}{%
\begin{tabular}{|c||c|c|c|c|c|c|}
\hline
$n_0$  & $\rho_n=0.75$ & $\rho_n=0.8$ & $\rho_n=0.85$ & $\rho_n=0.9$ & $\rho_n=0.95$ & $\rho_n=1$ \\
\hline \hline
0   & 0.03 & 0.06 & 0.13 & 0.07 & 0.09 & 0.06 \\
\hline
20  & 0.81 & 0.86 & 0.88 & 0.98 & 1.00 & 1.00 \\
\hline
40  & 1.00 & 1.00 & 1.00 & 1.00 & 1.00 & 1.00 \\
\hline
60  & 1.00 & 1.00 & 1.00 & 1.00 & 1.00 & 1.00 \\
\hline
80  & 1.00 & 1.00 & 1.00 & 1.00 & 1.00 & 1.00 \\
\hline
100 & 1.00 & 1.00 & 1.00 & 1.00 & 1.00 & 1.00 \\
\hline
150 & 1.00 & 1.00 & 1.00 & 1.00 & 1.00 & 1.00 \\
\hline
\end{tabular}}
\end{subtable}
\hfill
\begin{subtable}[t]{0.49\textwidth}
\centering
\caption{$n=3000,\; k=3$}
\label{tab:mmsbm-k3-n3000}
\renewcommand{\arraystretch}{1.05}
\setlength{\tabcolsep}{3.5pt}
\resizebox{\linewidth}{!}{%
\begin{tabular}{|c||c|c|c|c|c|c|}
\hline
$n_0$  & $\rho_n=0.75$ & $\rho_n=0.8$ & $\rho_n=0.85$ & $\rho_n=0.9$ & $\rho_n=0.95$ & $\rho_n=1$ \\
\hline \hline
0   & 0.12 & 0.09 & 0.10 & 0.06 & 0.06 & 0.04 \\
\hline
20  & 0.80 & 0.87 & 0.93 & 1.00 & 1.00 & 1.00 \\
\hline
40  & 1.00 & 1.00 & 1.00 & 1.00 & 1.00 & 1.00 \\
\hline
60  & 1.00 & 1.00 & 1.00 & 1.00 & 1.00 & 1.00 \\
\hline
80  & 1.00 & 1.00 & 1.00 & 1.00 & 1.00 & 1.00 \\
\hline
100 & 1.00 & 1.00 & 1.00 & 1.00 & 1.00 & 1.00 \\
\hline
150 & 1.00 & 1.00 & 1.00 & 1.00 & 1.00 & 1.00 \\
\hline
\end{tabular}}
\end{subtable}
\end{table}

\begin{table}[H]
\centering
\caption{Level and Power results for MMSBMs with $k=5$}
\label{tab:mmsbm-k5-power}
\begin{subtable}[t]{0.49\textwidth}
\centering
\caption{$n=3000,\; k=5$}
\label{tab:mmsbm-k5-n3000}
\renewcommand{\arraystretch}{1.05}
\setlength{\tabcolsep}{3.5pt}
\resizebox{\linewidth}{!}{%
\begin{tabular}{|c||c|c|c|c|c|c|}
\hline
$n_0$  & $\rho_n=0.75$ & $\rho_n=0.8$ & $\rho_n=0.85$ & $\rho_n=0.9$ & $\rho_n=0.95$ & $\rho_n=1$ \\
\hline \hline
0   & 0.87 & 0.71 & 0.37 & 0.09 & 0.08 & 0.02 \\
\hline
20  & 0.92 & 0.86 & 0.51 & 0.26 & 0.13 & 0.11 \\
\hline
40  & 0.95 & 0.93 & 0.66 & 0.41 & 0.35 & 0.30 \\
\hline
60  & 0.99 & 0.95 & 0.73 & 0.64 & 0.62 & 0.55 \\
\hline
80  & 1.00 & 0.98 & 0.88 & 0.83 & 0.78 & 0.71 \\
\hline
100 & 1.00 & 0.97 & 0.95 & 0.93 & 0.90 & 0.90 \\
\hline
150 & 1.00 & 1.00 & 1.00 & 0.99 & 1.00 & 1.00 \\
\hline
\end{tabular}}
\end{subtable}
\hfill
\begin{subtable}[t]{0.49\textwidth}
\centering
\caption{$n=3500,\; k=5$}
\label{tab:mmsbm-k5-n3500}
\renewcommand{\arraystretch}{1.05}
\setlength{\tabcolsep}{3.5pt}
\resizebox{\linewidth}{!}{%
\begin{tabular}{|c||c|c|c|c|c|c|}
\hline
$n_0$  & $\rho_n=0.75$ & $\rho_n=0.8$ & $\rho_n=0.85$ & $\rho_n=0.9$ & $\rho_n=0.95$ & $\rho_n=1$ \\
\hline \hline
0   & 0.70 & 0.43 & 0.21 & 0.06 & 0.09 & 0.07 \\
\hline
20  & 0.77 & 0.62 & 0.37 & 0.18 & 0.19 & 0.09 \\
\hline
40  & 0.87 & 0.80 & 0.56 & 0.38 & 0.36 & 0.27 \\
\hline
60  & 0.93 & 0.89 & 0.71 & 0.66 & 0.56 & 0.51 \\
\hline
80  & 0.96 & 0.94 & 0.81 & 0.80 & 0.78 & 0.74 \\
\hline
100 & 0.98 & 0.97 & 0.92 & 0.92 & 0.90 & 0.89 \\
\hline
150 & 0.99 & 0.99 & 1.00 & 0.99 & 1.00 & 1.00 \\
\hline
\end{tabular}}
\end{subtable}
\end{table}

\section{Application to Airport Data}{\label{sec:airport}}

We also apply \cref{alg:hypo_test} to U.S.\ domestic flight data publicly available from the Bureau of Transportation Statistics \citep{bts_flight_data} and previously analyzed in \citet{agterberg_joint_2025} and \citet{agterberg_estimating_2025}.. After restricting to the largest connected component, we obtain \(n=343\) airports and data for \(T=69\) months (January 2016 through September 2021). For each month \(t=1,\dots,T\), we have a count matrix
\(
F_t=(F_t(i,j))_{i,j=1}^n\in\mathbb{N}^{n\times n},
\)
where \(F_t(i,j)\) denotes the number of flights from airport \(i\) to airport \(j\) during month \(t\). To obtain a binary network representation for each month, we symmetrize and threshold the counts: define the adjacency matrix \(A_t=(A_t(i,j))_{i,j=1}^n\in\{0,1\}^{n\times n}\) by
\[
A_t(i,j) = \mathbbm{1}\{\,F_t(i,j)+F_t(j,i)\ge 1\,\},\qquad i\neq j,
\]
and set \(A_t(i,i)=0\) for all \(i\). In words, \(A_t(i,j)=1\) if there exists at least one flight between airports \(i\) and \(j\) (in either direction) during month \(t\), and \(A_t(i,j)=0\) otherwise. For detecting the number of communities \(k\), we select the embedding dimension using the ``elbow'' method of \citet{zhu2006automatic}; the resulting choice is \(k=5\).

The months of January, June, and November provide a particularly informative contrast for assessing structural stability in the U.S. flight network. June 2020 corresponds to the period of most pronounced disruption in national air-traffic patterns due to the COVID shock, and thus serves as a natural stress point at which one might expect substantive alterations in the community structure.  
By comparison, January and November exhibit far more regular seasonal behavior across years, and consequently their underlying connectivity patterns are expected to display a high degree of structural persistence.  
This separation of regimes allows for a clear evaluation of the sensitivity of our methodology to both stable and significantly perturbed network environments.

In \cref{tab:jan-pvals,tab:jun-pvals,tab:nov-pvals}, the number of rows $m$ in each table corresponds to the frequency of that specific month within the study period. Consequently, there are $m=6$ occurrences for January and June, and $m=5$ occurrences for November. The \((i,j)\)-th entry, \(i\neq j\), reports the \(p\)-value from a hypothesis test comparing the adjacency matrices of that month in year \(i\) and year \(j\). Because multiple pairwise comparisons are performed (there are \(m(m-1)/2\) tests per month), we apply a Bonferroni correction: at nominal level \(\alpha=0.05\), the significance threshold is
\(
{2\alpha}/({m(m-1))}.
\)
Cells with \(p\)-values at or below this threshold are displayed in red, and cells with \(p\)-values above the threshold are displayed in green.

\begin{table}[H]
\centering
\caption{January (matrix dimension $m=6$, Bonferroni threshold $2\alpha/(m(m-1))=0.003$).}
\label{tab:jan-pvals}
\begin{tabular}{|c||c|c|c|c|c|c|}
\hline
 & 2016-01 & 2017-01 & 2018-01 & 2019-01 & 2020-01 & 2021-01 \\
\hline\hline
2016-01 &
 & \cellcolor{green!15}0.0726 & \cellcolor{green!15}0.0875 & \cellcolor{green!15}0.2437 & \cellcolor{green!15}0.8976 & \cellcolor{green!15}0.2951 \\
\hline
2017-01 &
\cellcolor{green!15}0.0726 &  & \cellcolor{green!15}0.0615 & \cellcolor{green!15}0.1810 & \cellcolor{green!15}0.4708 & \cellcolor{green!15}0.3407 \\
\hline
2018-01 &
\cellcolor{green!15}0.0875 & \cellcolor{green!15}0.0615 &  & \cellcolor{green!15}0.0634 & \cellcolor{green!15}0.2615 & \cellcolor{green!15}0.8945 \\
\hline
2019-01 &
\cellcolor{green!15}0.2437 & \cellcolor{green!15}0.1810 & \cellcolor{green!15}0.0634 &  & \cellcolor{green!15}0.0787 & \cellcolor{green!15}0.8612 \\
\hline
2020-01 &
\cellcolor{green!15}0.8976 & \cellcolor{green!15}0.4708 & \cellcolor{green!15}0.2615 & \cellcolor{green!15}0.0787 &  & \cellcolor{green!15}0.8358 \\
\hline
2021-01 &
\cellcolor{green!15}0.2951 & \cellcolor{green!15}0.3407 & \cellcolor{green!15}0.8945 & \cellcolor{green!15}0.8612 & \cellcolor{green!15}0.8358 &  \\
\hline
\end{tabular}
\end{table}

\begin{table}[H]
\centering
\caption{June (matrix dimension $m=6$, Bonferroni threshold $2\alpha/(m(m-1))=0.003$).}
\label{tab:jun-pvals}
\begin{tabular}{|c||c|c|c|c|c|c|}
\hline
 & 2016-06 & 2017-06 & 2018-06 & 2019-06 & 2020-06 & 2021-06 \\
\hline\hline
2016-06 &
 & \cellcolor{green!15}0.1815 & \cellcolor{green!15}0.5459 & \cellcolor{green!15}0.9343 & \cellcolor{red!15}0.0002 & \cellcolor{green!15}0.0760 \\
\hline
2017-06 &
\cellcolor{green!15}0.1815 &  & \cellcolor{green!15}0.2123 & \cellcolor{green!15}0.4802 & \cellcolor{red!15}0.0005 & \cellcolor{green!15}0.1777 \\
\hline
2018-06 &
\cellcolor{green!15}0.5459 & \cellcolor{green!15}0.2123 &  & \cellcolor{green!15}0.0674 & \cellcolor{red!15}0.0001 & \cellcolor{green!15}0.0919 \\
\hline
2019-06 &
\cellcolor{green!15}0.9343 & \cellcolor{green!15}0.4802 & \cellcolor{green!15}0.0674 &  & \cellcolor{red!15}0.0000 & \cellcolor{green!15}0.3596 \\
\hline
2020-06 &
\cellcolor{red!15}0.0002 & \cellcolor{red!15}0.0005 & \cellcolor{red!15}0.0001 & \cellcolor{red!15}0.0000 &  & \cellcolor{green!15}0.2915 \\
\hline
2021-06 &
\cellcolor{green!15}0.0760 & \cellcolor{green!15}0.1777 & \cellcolor{green!15}0.0919 & \cellcolor{green!15}0.3596 & \cellcolor{green!15}0.2915 &  \\
\hline
\end{tabular}
\end{table}

\begin{table}[H]
\centering
\caption{November (matrix dimension $m=5$, Bonferroni threshold $2\alpha/(m(m-1))=0.005$).}
\label{tab:nov-pvals}
\begin{tabular}{|c||c|c|c|c|c|}
\hline
 & 2016-11 & 2017-11 & 2018-11 & 2019-11 & 2020-11 \\
\hline\hline
2016-11 &
 & \cellcolor{green!15}0.0101 & \cellcolor{green!15}0.2853 & \cellcolor{green!15}0.5370 & \cellcolor{green!15}0.4457 \\
\hline
2017-11 &
\cellcolor{green!15}0.0101 &  & \cellcolor{green!15}0.0891 & \cellcolor{green!15}0.1971 & \cellcolor{green!15}0.8067 \\
\hline
2018-11 &
\cellcolor{green!15}0.2853 & \cellcolor{green!15}0.0891 &  & \cellcolor{green!15}0.0876 & \cellcolor{green!15}0.6468 \\
\hline
2019-11 &
\cellcolor{green!15}0.5370 & \cellcolor{green!15}0.1971 & \cellcolor{green!15}0.0876 &  & \cellcolor{green!15}0.9973 \\
\hline
2020-11 &
\cellcolor{green!15}0.4457 & \cellcolor{green!15}0.8067 & \cellcolor{green!15}0.6468 & \cellcolor{green!15}0.9973 &  \\
\hline
\end{tabular}
\end{table}
For January and November, all cells are shaded green, indicating that the corresponding flight networks are not significantly different across years at the Bonferroni-corrected level.  
In contrast, for June the entries involving year 2020 are shaded red, while all other entries remain green.  
This pattern is attributable to the COVID--19 pandemic, which peaked in June 2020 and caused a severe disruption of U.S.\ domestic air traffic.  
During this period many airports experienced no incoming or outgoing flights, leading to a flight network structure that is markedly different from the corresponding networks in other years.  
These findings illustrate that our proposed testing methodology is capable of distinguishing between genuinely similar network structures (as in January and November) and substantive structural changes (as in June 2020).

\section{Projection Distance Asymptotics and a One-Sample Test}\label{sec:one_samplepf}
In order to understand the proof of \cref{thm:twosample} for the two‐sample setting, we find it beneficial to study its one‐sample analogue.  In this simpler setting, we observe a single adjacency matrix
\(
A \sim \Ber(P),
\)
with population spectral decomposition \(P = V \Lambda V\t\).  We wish to test whether the principal subspace of \(P\) coincides with a given reference subspace spanned by an orthonormal matrix \(V^*\in\mathbb{R}^{n\times k}\).  Formally, we consider
\begin{align}\label{hypo3}
    H_0: V V\t = V^* V^{*\t} 
\quad\text{versus}\quad 
H_1: V V\t \neq V^* V^{*\t}
\end{align}
Let \(\hat V\in\mathbb{R}^{n\times k}\) denote the matrix of the leading \(k\) empirical eigenvectors of \(A\).  We define the one‐sample test statistic
\[
T_{n}^{({\sf os})} := \bigl\|\hat V\,\hat V\t - V^* V^{*\t}\bigr\|_F^2.
\]
Note that we can equivalently reject the null hypothesis at level $\alpha$ if we can form an appropriate $1 - \alpha$ confidence interval. The following theorem establishes the Gaussian limit under the null hypothesis in \eqref{hypo3}.

\begin{theorem} \label{thm:onesample}
Suppose \cref{ass:asymptotic1,ass:sparse,ass:eigen-scaling,ass:incoherence} hold.  Then  under \(H_0\) of \eqref{hypo2},
\begin{equation}
\label{eq:clt1}
\frac{T_{n}^{({\sf os})} - \mu_1}{\sigma_1}
\xrightarrow{D} \mathcal{N}(0,1),
\end{equation}
where
\begin{align}
\mu_1 
&= 2\,\mathrm{tr}\Bigl(\beta^{-2}\left[\Sigma\circ\beta^\perp 
   + \Diag(\Sigma \cdot d-\diag(\Sigma) \circ d)\right]\Bigr)
   + O_p\left(\frac{k}{n^2 \rho_n^2}\right),
\label{param:11}\\
\sigma^2_1
    &= 8 \Bigl\langle (J_n - I_n)\circ(\beta^{-2}\circ\beta^{-2}),\Sigma^2\Bigr\rangle
      + 4 \,\alpha\t K\,\alpha+o \left(\frac{k^2}{n^3 \rho_n^2} \right),
\label{param:12}
\end{align}
with \(V V\t = V^* V^{*\t}\). Here we define
\[
P = V \Lambda V\t, 
\quad \beta^k = V\Lambda^kV\t,
\quad \beta^\perp = I_n - VV\t, 
\quad d = \diag(\beta^\perp),
\]
\[
\Sigma = P\circ(J_n-P),
\quad K = 2\,\Sigma\circ(J_n - 2P),
\quad \alpha = \diag(\beta^{-2}).
\]
\end{theorem}

\begin{proof}
    See \cref{sec:mainproof}.
\end{proof}
The above one–sample result serves as the cornerstone for deriving the more delicate two–sample asymptotic limit stated in Theorem~\ref{thm:twosample}. In particular, the expansion arguments and concentration inequalities developed here form the analytical basis for the proof of the two–sample result.
\begin{remark}[Proof Overview]
We provide a brief outline of the proof of \cref{thm:onesample}. The argument begins with a series expansion of the projection distance 
\(\|\hat V \hat V\t - V^* {V^*}\t\|_F^2\) first developed in \citet{xia_normal_2021}. We show the statistic decomposes as
\[
T_{n}^{({\sf os})} = T_{1}^{(S)} + T_{1}^{(3)} + T^{(H)}_{1},
\]
where we denote the second–order term by \(T_{1}^{(S)}\), the third–order term by \(T_{1}^{(3)}\)  and the remaining higher–order terms collectively by \(T^{(H)}_{1}\).  The first order term vanishes.
The dominant component is \(T_{1}^{(S)}\), and the central limit behavior follows from a martingale central limit theorem (CLT), building on ideas from \citet{fan_asymptotic_2022}. In particular, we can show that
\[
\frac{T_{1}^{(S)} - \E \big(T_{1}^{(S)}\big)}{\sqrt{\Var\big(T_{1}^{(S)}\big)}} \xrightarrow{D} \mathcal{N}(0,1).
\]
The variance comparison is then crucial: we establish that
\(\Var( T_{1}^{(S)}) \gg \Var(T^{(3)}_{1})\). For the expectations, we verify that \(\E (T^{(3)}_{1})\) is of order \(o(\frac{k}{n^2 \rho_n^2})\). Furthermore, by a straightforward application of Chebyshev’s inequality, \(\big(T^{(3)}_{1}-\E \big(T^{(3)}_{1}\big)\big)\big/\sqrt{\Var\big(T_{1}^{(S)}\big)}\) vanish in probability. Finally from Theorem 1 of \cite{xia_normal_2021}, we show that \(\|T^{(H)}_1\|_F \lesssim \frac{k}{n^2 \rho_n^2}\) with high probability.

Putting these ingredients together, we arrive at the desired Gaussian limit under \cref{ass:sparse}:
\[
\frac{T_{n}^{({\sf os})} - \mu_1}{\sigma_1} \xrightarrow{D} \mathcal{N}(0,1),
\] where \(\mu_1=\E \left(T_{1}^{(S)}\right)+O_p \left(\frac{k}{n^2 \rho_n^2} \right)\) and \(\sigma_1^2= \Var \left(T_{1}^{(S)} \right)\).
\end{remark}

For the test implementation, it remains to consistently estimate the mean and variance appearing in \eqref{eq:clt1}--\eqref{param:12}. The estimation procedure mirrors that in the two–sample setting: population eigenvectors \(V\) are replaced by their empirical estimates \(\hat{V}\), and the population eigenvalues by their empirical counterparts \(\hat{\Lambda}\). 

Define the plug-in estimators
\begin{equation}\label{eq:estimatorsone}
\begin{aligned}
\hat{\mu}_1 
&=2\tr\Bigl[\hat{\beta}^{-2}\Bigl(\hat \Sigma\circ \hat\beta^\perp 
      + \Diag\bigl(\hat{\Sigma}\cdot \hat d \bigr)\Bigr)\Bigr], \\[6pt]
\hat{\sigma}_1^2 
&= 8 \Bigl\langle (J_n - I_n)\circ(\hat{\beta}^{-2}\circ\hat{\beta}^{-2}),\hat{\Sigma}^2\Bigr\rangle
      + 4 \,\hat{\alpha}\t \hat{K}\,\hat{\alpha},
\end{aligned}
\end{equation}
where 
\begin{align}\label{thm6.1}
    \hat{P}=\hat V\,\hat{\Lambda}\,\hat V\t, \qquad 
\hat{\beta}^k=\hat V\,\hat{\Lambda}^{k}\,\hat V\t, \qquad
\hat\beta^\perp=I_n-\hat V\hat V\t, \qquad 
\hat d=\diag(\hat\beta^\perp),
\end{align}
and
\begin{align}\label{thm6.2}
    \hat{\Sigma}=\hat{P}\circ(J_n-\hat{P}), \qquad
\hat{K}=2\,\hat{\Sigma}\circ(J_n-2\hat{P}), \qquad
\hat{\alpha}=\diag(\hat{\beta}^{-2}).
\end{align}

\begin{theorem}[Mean and variance estimators for the one–sample test]\label{thm:1sampleconsis}
The estimators in \cref{eq:estimatorsone} satisfy the following rate under \cref{ass:asymptotic1,ass:dense,ass:eigen-scaling,ass:incoherence}:
\[
|\hat{\mu}_1 - \mu_1| =O_p \left(\max \left(\frac{k}{n^2 \rho_n^2},\frac{k \sqrt{\log n}}{n^2 \rho_n^{3/2}}\right)\right),
\qquad
|\hat{\sigma}_1^2 - \sigma_1^2| = o_p \left(\frac{k^2}{n^3 \rho_n^2}\right).
\]
\end{theorem}

\begin{proof}
    See \cref{prf:FinEst1}.
\end{proof}
The following proposition establishes that the upper bound on \(|\hat{\mu}_1 - \mu_1|\) as stated in \cref{thm:1sampleconsis} is sharp and therefore cannot be improved for plug-in estimators.

\begin{proposition}\label{rem:Minimax2}
Suppose \(P\) is generated from a SBM. Let \(\hat{\mu}_1\) be the estimator defined as in \cref{eq:estimatorsone}. Then, under \cref{ass:asymptotic1,ass:eigen-scaling,ass:incoherence,ass:dense},
\[
    \bigl|\hat{\mu}_1 - \mu_1\bigr|
    \asymp
    \max \left(\frac{k}{n^2 \rho_n^2},\frac{k \sqrt{\log n}}{n^2 \rho_n^{3/2}}\right).
\]
\end{proposition}

\begin{proof}
    See \cref{prop4}.
\end{proof}
An immediate implication is that the proposed inference procedure cannot remain valid outside the dense regime (\cref{ass:dense}).  In particular, when \(n\rho_{n} \lesssim \sqrt{n}\), the estimation error \(\frac{k}{n^{2}\rho_{n}^{2}}\) causes the test statistic to lose consistency.  This precisely motivates the necessity of \cref{ass:dense}.

\section{Discussion}\label{sec:discussion}

In this work, we have developed a rigorous and computationally efficient two-sample testing framework for assessing the equality of the underlying low-rank subspaces of two networks. At the core of our methodology are novel Gaussian limit theorems for the projection distance between estimated spectral subspaces. Notably, these limit theorems are the first of their kind established explicitly under a Bernoulli noise model. Supported by easily computable, data-driven plug-in estimators, our proposed test demonstrates strong empirical power in detecting subtle structural shifts. Furthermore, the generality of our approach extends beyond detecting community structure changes in SBMs. For instance, within the framework of GRDPGs, our test inherently determines whether the latent positions of two graphs span the exact same feature space, rendering the procedure robust to arbitrary linear transformations of the latent geometry.

While our findings establish a  theoretically grounded foundation for network hypothesis testing, they simultaneously open several broader avenues for future research. For instance, one could extend the test statistic to incorporate information beyond the principal angles between the estimated subspaces. Our current angle-based procedure provides robustness against global density fluctuations; however, incorporating the spectral distance between the underlying connectivity matrices could further enhance statistical power.

Furthermore, while our theoretical framework establishes rigorous guarantees under bounded, independent Bernoulli noise,  extending these limit theorems to accommodate more complex noise structures presents an exciting frontier. Two immediate avenues emerge:
\begin{itemize}
    \item \textbf{Heavy-Tailed Distributions:} Developing limit theorems for spectral projectors under heavy-tailed noise would broaden the applicability of spectral testing methods to weighted networks.
    \item \textbf{Dependent Edges:} Investigating the robustness of the test statistic under weak, local edge dependencies (e.g., transitivity or reciprocity) poses a challenging and valuable open problem.
\end{itemize}

Finally, our current analysis assumes the rank $k$ of the probability matrices is known \emph{a priori}. While several consistent estimators for $k$ are readily available in the literature, formally incorporating the theoretical uncertainty of rank selection into the asymptotic distribution of the test statistic is a mathematically non-trivial task. Developing a fully adaptive testing procedure that remains robust to an estimated rank, or constructing a test intrinsically agnostic to the specific choice of $k$ would be another interesting avenue for future research.

\appendix
\section{Proof of \cref{thm:onesample}} \label{sec:mainproof}
In this section we give the full proof of \cref{thm:onesample}.  Our results are based on three steps.  In \cref{sec:seriesexpansion} we apply the matrix series expansion of \citet{xia_normal_2021} to identify the leading-order and higher-order terms.  In \cref{sec:secondorderasymptotics} we study the asymptotic behavior of the leading-order term, and in \cref{sec:higherorder} we show that the higher-order terms are asymptotically negligible relative to the leading-order term.  Combining all these ingredients gives the final proof of \cref{thm:onesample}.

\subsection{Matrix Series Expansion and Second-Order Term} \label{sec:seriesexpansion}
In order to write the test statistic \(\|\hat{V}\hat{V}\t - VV\t\|_F^2\) using Theorem~1 of \citet{xia_normal_2021}, we first define the following good event and show it holds with high probability.
\begin{lemma}\label{lem:Good Set}
Suppose $A \sim \Ber(P)$, and let $X = A - P$. 
Define the event
\[
\mathcal{E}_{{\sf good}} = \{\|X\|\lesssim \sqrt{n\rho_n}\,\}.
\]
Then under \cref{ass:asymptotic1,ass:eigen-scaling,ass:sparse},
\(
\mathbb{P}(\mathcal{E}_{{\sf good}})\ge1-O\left(n^{-19}\right).
\)
Moreover, on the event \(\mathcal{E}_{{\sf good}}\) the following bound holds:
\[
\big\|\hat V\hat V\t - V V\t\big\|_F^2 \lesssim \frac{k}{n\rho_n},
\]
where \(\hat V\) and \(V\) are the matrices of estimated and population top-\(k\) eigenvectors, respectively. Thus, we have 
\[\big\|\hat V\hat V\t - V V\t\big\|_F^2 = O_p \left(\frac{k}{n\rho_n}\right) \quad \text{and} \quad \|X\|=O_p( \sqrt{n \rho_n}).\]
\end{lemma}
\begin{proof}
    See \cref{sec:goodsetproof}.
\end{proof}
Define 
\begin{align}
   \beta^0 = \beta^\perp &= I - VV\t; \\
    \beta^{-l} &= V \Lambda^{-l} V\t,
\end{align}
and for each integer $l \geq 1$,
\begin{align} \label{Sk_1sample}
    S_{l}(X)
&=\sum_{\substack{s=(s_1,\dots,s_{l+1}):\\ s_1+\cdots+s_{l+1}=l}}
(-1)^{1+\tau(s)}\;
\beta^{-s_1}\, X\, \beta^{-s_2}\, X\cdots \beta^{-s_{l}} X\, \beta^{-s_{l+1}}, \\
\tau(s)&=\sum_{j=1}^{l+1} 1\{s_j>0\},
\end{align}
where we recall \(X:=A-P\) denotes the perturbation matrix.  On the event $\mathcal{E}_{{\sf good}}$, \cref{ass:eigen-scaling} impies that $\| X \| \ll \lambda_k$, and hence the conditions of Theorem 1 of \citet{xia_normal_2021} are satisfied.  Applying this result, we have that 
\begin{equation}\label{eq1}
\begin{aligned}
\|\hat{V}\hat{V}\t - VV\t\|_F^2
&=\,-2\,\sum_{l\ge2}\,\langle VV\t,\,S_l(X)\rangle_F \\
&=2\,\|\beta^{\perp}X\beta^{-1}\|_F^2
\;-\;2\,\sum_{l\ge3}\,\langle VV\t,\,S_l(X)\rangle_F.
\end{aligned}
\end{equation}
The leading nonvanishing contribution to the eigenspace projection error thus arises from the second–order term \(2\,\|\beta^{\perp}X\beta^{-1}\|_F^2\).  We derive its asymptotic distribution in \cref{sec:secondorderasymptotics}, and study the higher-order terms in \cref{sec:higherorder}.

\subsection{Second-Order Term Asymptotics}  \label{sec:secondorderasymptotics}
We begin this subsection by deriving the mean and variance of the leading contribution to the test statistic. Define
\begin{align*}
T_{1}^{(S)} =
2\big\|\beta^{\perp} X \beta^{-1}\big\|_F^2
= 2\tr\big(\beta^{-1}X\t\beta^{\perp} X\beta^{-1}\big).
\end{align*}
Our first result demonstrates that the variance of this term can be approximated by the variance of a similar term with $\beta^{\perp} $ taken to be the identity.
\label{sec:secondorder}
\begin{lemma}\label{lem:beta_perp1one}
Suppose \cref{ass:asymptotic1,ass:sparse,ass:eigen-scaling,ass:incoherence} hold. Then
\[
\frac{\Var\big(\|\beta^{\perp} X \beta^{-1}\|_F^2\big)}
     {\Var\big(\|X \beta^{-1}\|_F^2\big)}
\longrightarrow 1,
\]
\end{lemma}
\begin{proof}
    See \cref{sec:betaperp1proof}.
\end{proof}
The following result computes this mean and variance.
\begin{lemma}[Second--order term: mean and variance of one-sample test]\label{lem:second-order1}
Suppose \cref{ass:asymptotic1,ass:sparse,ass:eigen-scaling,ass:incoherence} hold. Define
\[
\Sigma := P\circ(J_n-P),\qquad d:=\diag(\beta^\perp),\qquad
\alpha:=\diag(\beta^{-2}),\qquad K:=2\,\Sigma\circ(J_n-2P).
\]
The mean and variance of \(T_{1}^{(S)}\) admit the following expressions:
\[
\mathbb{E}\bigl[T_{1}^{(S)}\bigr]
= 2\,\tr\Bigl[\beta^{-2}\bigl(\Sigma\circ\beta^\perp + \mathrm{Diag}(\Sigma\cdot d-\diag(\Sigma) \circ d)\bigr)\Bigr],
\]
and 
\[
\Var\bigl(T_{1}^{(S)}\bigr)
= 8\big\langle (J_n-I_n)\circ(\beta^{-2}\circ\beta^{-2}),\Sigma^2 \big\rangle
+ 4\,\alpha\t K \alpha + o \left( \frac{k^2}{n^3 \rho_n^2} \right).
\]
The quantities coincide with the leading mean and variance contributions stated in \cref{param:11,param:12}. In particular, we also have
\(
\Var\bigl(T_{1}^{(S)}\bigr)
\asymp
\frac{k^{2}}{n^{3}\rho_n^{2}}.
\)

\end{lemma}
\begin{proof}
    See \cref{sec:second-order1proof}.
\end{proof}

Having calculated these moments, we now study the asymptotic distribution of the first-order approximation $T_1^{(S)}$.  The following result uses the martingale central limit theorem to establish the Gaussian limit for the centered and scaled second-order term.
\begin{lemma}[Second--order asymptotic distribution of one-sample test]\label{lem:second-order-clt1}
Suppose \cref{ass:asymptotic1,ass:sparse,ass:eigen-scaling,ass:incoherence} hold.  The quantity
\(
T_{1}^{(S)} \)
satisfies
\[
\frac{
  T_{1}^{(S)} - \,\mathbb{E}[T_{1}^{(S)}]
}{
  \sqrt{\Var(T_{1}^{(S)})}
}
\xrightarrow{D} \mathcal{N}(0,1)
\]
as $n \to \infty$.
\end{lemma}
\begin{proof}
    See \cref{sec:second-orderclt1proof}.
\end{proof}
This completes the derivation of the Gaussian limit for the second–order term and provides the principal probabilistic ingredient in the proof of \cref{thm:onesample}.

\subsection{Concentration of Higher-Order Terms} \label{sec:higherorder}
In this section we study the higher-order terms from the decomposition \cref{eq1}. For $l \ge 3$, each term in the expansion \cref{eq1} takes the form \(\langle V V\t, S_{l}(X) \rangle_F\). Using the definition of \(S_l(X)\), this can be written as
\begin{equation}
\begin{aligned}\label{bigvar}
    \langle V V\t, S_{l}(X) \rangle_F &= \sum_{\substack{s=(s_1,\dots,s_{l+1}):\\ s_1+\cdots+s_{l+1}=l}}
(-1)^{1+\tau(s)}\;
\langle V V\t, \beta^{-s_1}\, X\, \beta^{-s_2}\, X\cdots \beta^{-s_{l}} X\, \beta^{-s_{l+1}}\rangle.
\end{aligned}
\end{equation}

We separate the expansion into the third-order term and higher order terms.  We first analyze the contribution of $T_1^{(3)}= \langle V V\t, S_{3}(X) \rangle_F$.
\begin{lemma}\label{lem:3rdmean1}
Assume that \cref{ass:asymptotic1,ass:sparse,ass:eigen-scaling,ass:incoherence} hold.  Let \(S_{3}(X)\) be as specified in \cref{Sk_1sample}.  Then, we have 
\[\E \langle V V\t,S_3(X)\rangle = o \left(\frac{k}{n^2 \rho_n^2}\right).\]
\end{lemma}
\begin{proof}
    See \cref{sec:3rd}.
\end{proof}
We next consider the variance of $\langle VV\t, S_3(X)\rangle$.
\begin{lemma}\label{lem:3rdvar1}
Assume that \cref{ass:asymptotic1,ass:sparse,ass:eigen-scaling,ass:incoherence} hold. Then \(\Var(\langle V V\t,S_3(X)\rangle) \asymp \frac{k^2}{n^4 \rho_n^3}\).
\end{lemma}
\begin{proof}
    See \cref{sec:3rdproof}.
\end{proof}
A direct application of Chebyshev’s inequality gives
\begin{equation*}
\begin{aligned}
    \Pr\left[\,
    \frac{\big|T_{1}^{(3)} - \E \left(T_{1}^{(3)}\right)\big|}
    {\sqrt{\Var\left(T_{1}^{(S)}\right)}} \;\ge\; \delta
    \right]
    &\le \frac{\Var\left(T_{1}^{(3)}\right)}{\delta^2 \Var(T_{1}^{(S)})}\\
    &\lesssim \frac{1}{\delta^2 n \rho_n},
\end{aligned}
\end{equation*}
for any $\delta > 0$. Therefore by \cref{ass:sparse},
\begin{align} \label{eq:cheby1}
\frac{\left|T_{1}^{(3)} - \E \left(T_{1}^{(3)}\right)\right|}
{\sqrt{\Var\left(T_{1}^{(S)}\right)}} 
\;\to\; 0
\qquad \text{in probability}.
\end{align}
The higher-order terms are significantly simpler to analyze.  We note that whenever \(A\in\mathcal{E}_{{ \sf good}}\), we have 
\(
\|S_l(X)\|
\le
\Big(\tfrac{c}{n\rho_n}\Big)^{l/2}\) and
\( 
\rank\big(S_l(X)\big)\le{2l\choose l}k\;\le\;4^{l}k,
\)
for some constant \(c>0\).  
Consequently,
\[
\|S_l(X)\|_{F}
\;\le\;
\sqrt{\rank(S_l(X))}\,\|S_l(X)\|
\;\le\;
\big(4^{l}k\big)^{1/2}\Big(\tfrac{c}{n\rho_n}\Big)^{l/2}
= \sqrt{k}
\Big(\tfrac{4c}{n\rho_n}\Big)^{l/2}.
\]

\medskip\noindent
Using this bound termwise in the higher order terms,
\(
T_1^{(H)}
=
\big\langle 
VV\t,
\sum_{l\ge4} S_l(X)
\big\rangle,
\)
we obtain
\[
|T_1^{(H)}|
\;\le\;
k\,\Big\|\sum_{l\ge4} S_l(X)\Big\|_{F}
\;\le\;
\sum_{l\ge4} k\,\|S_l(X)\|_{F}
\;\le\;
k\sum_{l\ge4}\Big(\tfrac{4c}{n\rho_n}\Big)^{l/2} \frac{16 k c^{2}}{n^{2}\rho_n^{2}}.
\]
Therefore
\(
T_1^{(H)}
=
O_{p}\left(\frac{k}{n^{2}\rho_n^{2}}\right).
\)
Combining this result with \cref{lem:second-order-clt1} and \cref{eq:cheby1} yields the claim of \cref{thm:onesample}.  
This completes the proof.

\section{Proof of \cref{thm:twosample}}\label{sec:mainthm2}
We now turn to the proof of \cref{thm:twosample}. Following the same methodology as in the proof of \cref{thm:onesample} in the previous section, we begin by developing the matrix expansion of the two-sample statistic. The analysis proceeds in several steps. First, we derive the explicit form of the expansion and isolate the second-order contribution, which captures the leading fluctuation behavior of the statistic. We then compute the mean and variance of this second-order term and establish its asymptotic Gaussianity under the stated assumptions. Finally, we demonstrate that the higher-order terms possess variances of smaller order compared to that of the second-order component. Consequently, these higher-order terms affect only the mean of the test statistic while leaving its asymptotic variance unchanged.

\subsection{Matrix Series Expansion and Second-Order Term}

The two–sample projection–distance statistic considered in \cref{thm:twosample} is
\(
\|(\hat V^{(1)}) (\hat V^{(1)})\t - (\hat V^{(2)}) (\hat V^{(2)})\t\|_F^2
\).
Define
\begin{align*}
\beta_i^k = {V^{(i)}} (\Lambda^{(i)})^k {V^{(i)}}\t, 
\quad \beta^{\perp} = I_n - {V^{(i)}} {V^{(i)}}\t.
\end{align*}
Applying the matrix expansion of \citet{xia_normal_2021} to each empirical projector and collecting terms gives the decomposition
\begin{equation}
\begin{aligned}\label{eq2}
\|(\hat V^{(1)}) (\hat V^{(1)})\t - (\hat V^{(2)}) (\hat V^{(2)})\t\|_F^2
&= \big\|\hat{V}_{1} \hat{V}_{1}\t - {V^{(1)}} {V^{(1)}}\t\big\|_F^2
  + \big\|\hat{V}_{2} \hat{V}_{2}\t - {V^{(2)}} {V^{(2)}}\t\big\|_F^2 \notag \\
&\qquad - 2\, \tr\Big(\big(\hat{V}_{1} \hat{V}_{1}\t - {V^{(1)}} {V^{(1)}}\t\big)\big(\hat{V}_{2} \hat{V}_{2}\t - {V^{(2)}} {V^{(2)}}\t\big)\Big) \\
&= 2\big\|\beta_1^{\perp} {X^{(1)}} \beta_1^{-1}\big\|_F^2 - 2\,\sum_{j\ge3}\,\langle {V^{(1)}} {V^{(1)}}\t,\,S_{1,j}({X^{(1)}})\rangle_F
+ 2\big\|\beta_2^{\perp} {X^{(2)}} \beta_2^{-1}\big\|_F^2 \\
&\quad - 2\,\sum_{j\ge3}\,\langle {V^{(1)}} {V^{(1)}}\t,\,S_{2,j}({X^{(2)}})\rangle_F - 2\, \tr\bigl(S_{1,1}({X^{(1)}})\, S_{2,1}({X^{(2)}})\bigr)\\
&\quad - 2\sum_{l \geq 4} \sum_{l_1+l_2=l} \, \tr\bigl(S_{1,1}({X^{(1)}})\, S_{2,1}({X^{(2)}})\bigr),
\end{aligned}
\end{equation}
where \({X^{(1)}}={A^{(1)}}-{P^{(1)}}\), \({X^{(2)}}={A^{(2)}}-{P^{(2)}}\) and \begin{align} \label{Sk_2sample}
    S_{i,k}(X^{(i)})
&=\sum_{\substack{s=(s_1,\dots,s_{k+1}):\\ s_1+\cdots+s_{k+1}=k}}
(-1)^{1+\tau(s)}\;
\beta_i^{-s_1}\, {X^{(i)}}\, \beta_i^{-s_2}\, {X^{(i)}}\cdots \beta_i^{-s_{k}}{X^{(i)}}\, \beta_i^{-s_{k+1}}, \\
\tau(s)&=\sum_{j=1}^{k+1} 1\{s_j>0\} \notag.
\end{align}
Thus, the second-order term takes the form
\begin{equation}
\begin{aligned} \label{T_2}
T_2^{(S)}&= 2\big\|\beta_1^{\perp} {X^{(1)}} \beta_1^{-1}\big\|_F^2 
+ 2\big\|\beta_2^{\perp} {X^{(2)}} \beta_2^{-1}\big\|_F^2 
- 2\, \tr\bigl(S_{1,1}({X^{(1)}})\, S_{2,1}({X^{(2)}})\bigr) \\[4pt]
&= 2\big\|\beta^{\perp} {X^{(1)}} \beta_1^{-1}\big\|_F^2 
+ 2\big\|\beta^{\perp} {X^{(2)}} \beta_2^{-1}\big\|_F^2 
- 2\, \tr\big(\beta^{\perp} {X^{(1)}} \beta_1^{-1} \beta_2^{-1} {X^{(2)}}\big) 
- 2\, \tr\big(\beta_1^{-1} {X^{(1)}} \beta^{\perp} {X^{(2)}} \beta_2^{-1} \big) \\[4pt]
&= 2\bigl\|\beta^{\perp}\bigl( {X^{(1)}} \beta_1^{-1} - {X^{(2)}} \beta_2^{-1} \bigr)\bigr\|_F^2 \\[4pt]
&= 2\,\tr\left( 
U\t 
\begin{pmatrix}
(Y^{(1)})\t (Y^{(1)}) & -(Y^{(1)})\t (Y^{(2)})\\[4pt]
-(Y^{(2)})\t (Y^{(1)}) & (Y^{(2)})\t (Y^{(2)})
\end{pmatrix}
U
\right),
\end{aligned}
\end{equation}
where \(Y^{(i)} = \beta_i^{\perp}{X^{(i)}}\) for \(i=1,2\) and
\(
U = \begin{pmatrix} \beta_1^{-1} & \beta_2^{-1} \end{pmatrix}
\)
is the corresponding block matrix (whose columns we denote by \(U_{1 \cdot},\dots,U_{2n \cdot}\in\mathbb{R}^n\)).  The representation in \cref{T_2} exhibits the two–sample second–order contribution as a quadratic form in the block matrix of vectors \(U_{j \cdot}\), and provides the starting point for the subsequent mean/variance calculations and the martingale central limit argument.

\subsection{Second-Order Term Asymptotics}
As in the one-sample case, we first derive the mean and variance of the leading contribution to the two–sample projection–distance statistic.  Recall the definition:
\[
T_{2}^{(S)} = 2\bigl\|\beta_1^{\perp} {X^{(1)}} \beta_1^{-1}\bigr\|_F^2 
     + 2\bigl\|\beta_2^{\perp} {X^{(2)}} \beta_2^{-1}\bigr\|_F^2 
     - 2\, \tr\bigl(S_{1,1}({X^{(1)}})\, S_{2,1}({X^{(2)}})\bigr),
\]
equivalently written in the block form appearing in \cref{T_2}. Just like in the one sample setting, we have a similar version of \cref{lem:beta_perp1one}
for the two-sample setting under the same set of assumptions which is stated as follows.
\begin{lemma}\label{lem:beta_perp1two}
In the context of \cref{thm:2sampleconsis}, it holds that
\[
\frac{\Var\Big(2\,\tr\Big( 
U\t 
\begin{pmatrix}
(Y^{(1)})\t (Y^{(1)}) & -(Y^{(1)})\t (Y^{(2)})\\[4pt]
-(Y^{(2)})\t (Y^{(1)}) & (Y^{(2)})\t (Y^{(2)})
\end{pmatrix}
U
\Big)\Big)}
     {\Var\Big(2\,\tr\Big( 
U\t 
\begin{pmatrix}
{X^{(1)}}\t {X^{(1)}} & -{X^{(1)}}\t {X^{(2)}}\\[4pt]
-{X^{(2)}}\t {X^{(1)}} & {X^{(2)}}\t {X^{(2)}}
\end{pmatrix}
U
\Big)\Big)}
\;\longrightarrow\; 1,
\]
where \(U = \begin{pmatrix} \beta_1^{-1} & \beta_2^{-1} \end{pmatrix}\) and \(Y^{(i)} = \beta_i^{\perp} {X^{(i)}}\) for \(i = 1,2\).
\end{lemma}
\begin{proof}
    See \cref{sec:betaperp1proof2}.
\end{proof}
We similarly calculate the mean and variance of the second-order contribution.
\begin{lemma}[Second--order term: mean and variance of the two-sample test statistic]\label{lem:second-order2}
Define the second--order contribution
\begin{equation}
\begin{aligned} 
    T_{2}^{(S)} 
    &= 2\bigl\|\beta^{\perp} {X^{(1)}} \beta_1^{-1}\bigr\|_F^2 
     + 2\bigl\|\beta^{\perp} {X^{(2)}} \beta_2^{-1}\bigr\|_F^2 
     - 2\, \tr\bigl(S_{1,1}({X^{(1)}})\, S_{2,1}({X^{(2)}})\bigr) \\
    &= 2\,\tr\left( 
    U\t
    \begin{pmatrix}
        (Y^{(1)})\t (Y^{(1)}) & -(Y^{(1)})\t (Y^{(2)})\\[4pt]
        -(Y^{(2)})\t (Y^{(1)}) & (Y^{(2)})\t (Y^{(2)})
    \end{pmatrix}
    U
    \right),
\end{aligned}
\end{equation}
where \(Y^{(i)} = \beta^{\perp} {X^{(i)}}\) and \(U = \begin{pmatrix} \beta_1^{-1} & \beta_2^{-1} \end{pmatrix}\).
Assume that \cref{ass:asymptotic1,ass:sparse,ass:eigen-scaling,ass:incoherence} hold. Recall the definitions from \cref{last00}
\begin{align*}
&\quad d_i = \diag(\beta_i^{\perp}),
\quad G = \beta_1^{-1}\beta_2^{-1} + \beta_2^{-1}\beta_1^{-1}, 
\quad {\Sigma^{(i)}} = P^{(i)} \circ (J_n - P^{(i)}), \\
&\quad K^{(i)} = 2\,{\Sigma^{(i)}} \circ (J_n - 2P^{(i)}),
\quad \alpha_i = \diag(\beta_i^{-2}).
\end{align*}
Then the mean and variance of \(T_{2}^{(S)}\) satisfy
\[
\mathbb{E}\bigl[T_{2}^{(S)}\bigr]
= 2\sum_{i=1}^{2} 
  \tr\Bigl[
    \beta_i^{-2}
    \Bigl(
      {\Sigma^{(i)}} \circ \beta^{\perp} 
      + \Diag({\Sigma^{(i)}} \cdot d_i - \diag({\Sigma^{(i)}}) \circ d_i)
    \Bigr)
  \Bigr],
\]
and
\[
\Var\bigl(T_{2}^{(S)}\bigr)
= \sum_{i=1}^{2} 
  \Bigl(
    8\,\bigl\langle (J_n - I_n)\circ(\beta_i^{-2}\circ\beta_i^{-2}),\;({\Sigma^{(i)}})^2 \bigr\rangle
    + 4\,\alpha_i\t K^{(i)} \alpha_i
  \Bigr)
  + 4\,\bigl\langle 
      G \circ G,\;
      {\Sigma^{(1)}}\t {\Sigma^{(2)}} + (J_n - I_n) \circ ({\Sigma^{(1)}} \circ {\Sigma^{(2)}})
    \bigr\rangle + o \left( \frac{k^2}{n^3 \rho_n^2}\right).
\]
In particular, we also have
\(
\Var\bigl(T_{2}^{(S)}\bigr)
\;\asymp\;
\frac{k^{2}}{n^{3}\rho_n^{2}}.
\)
\end{lemma}
\begin{proof}
    See \cref{sec:second-order2proof}.
\end{proof}

Combining these moment formulae with the martingale central limit argument yields the Gaussian limit for the centered and scaled second–order term.
\begin{lemma}[Second--order asymptotic distribution of the two-sample test]\label{lem:second-order-clt2}
Suppose \cref{ass:asymptotic1,ass:sparse,ass:eigen-scaling,ass:incoherence} hold.  The quantity
\(
T_{2}^{(S)} \)
satisfies
\[
\frac{
  T_{2}^{(S)} - \mathbb{E}[T_{2}^{(S)}]
}{
  \sqrt{\Var(T_{2}^{(S)})}
}
\xrightarrow{D} \mathcal{N}(0,1).
\]
\end{lemma}
\begin{proof}
    See \cref{sec:second-orderclt2proof}.
\end{proof}

This completes the derivation of the Gaussian limit for the two–sample second–order contribution and provides the principal stochastic ingredient in the proof of \cref{thm:twosample}.

\subsection{Concentration of Higher-Order Terms}

Like in the proof of \cref{thm:onesample}, let \(T_{2}^{(H)}\) denote the higher--order contributions in the asymptotic expansion of the test statistic:
\[
T_{2}^{(H)} = \sum_{l \geq 3} T_{2}^{(l)} = \sum_{l \geq 3} \left( 
    \langle {V^{(1)}} {V^{(1)}}\t, S_{1,l}({X^{(1)}}) \rangle 
    + \langle {V^{(2)}} {V^{(2)}}\t, S_{2,l}({X^{(2)}}) \rangle 
    + \sum_{l_1 + l_2 = l} 
      \langle S_{1,l_1}({X^{(1)}}), S_{2,l_2}({X^{(2)}}) \rangle 
\right).
\]
As in the one-sample setting, we first separate out the third-order contribution.  The following result bounds its mean.\begin{lemma}\label{lem:3rdmean2}
Assume \cref{ass:asymptotic1,ass:sparse,ass:eigen-scaling,ass:incoherence}.  Let \(S_{i,3}\) denote the third–order remainder terms in \cref{Sk_2sample}.   Then we have 
\[\E \langle V V\t,S_{i,3}(X)\rangle = o \left(\frac{k}{n^2 \rho_n^2}\right).\]
\end{lemma}
\begin{proof}
    See \cref{sec:3rdm2}.
\end{proof}
Next we bound the variance of the third–order contribution.  By \cref{lem:3rdvar1} we already have tight control of the pure projection term $\Var\big(\langle VV\t,S_{i,3}(X)\rangle\big)$, so it suffices to control the mixed or cross–term variances of the form
\(
\Var\big(\langle S_{1,l_1}(X),S_{2,l_2}(X)\rangle\big)\) for \(l_1,l_2\ge 1,
\)
because these are the only remaining contributions appearing in the expansion of the third–order part.  The next lemma provides the required uniform bound on these cross–term variances, which together with \cref{lem:3rdvar1}, yields the desired bound on the overall third–order variance.

\begin{lemma}\label{lem:3rdvar2}
Assume that \cref{ass:asymptotic1,ass:sparse,ass:eigen-scaling,ass:incoherence} hold. Then \(\Var(\langle S_{1,l_1}(X),S_{2,l_2}(X)\rangle) \asymp \frac{k^2}{n^4 \rho_n^3}\) where \(l_1+l_2=3\) and \(l_1,l_2 \in \mathbb{Z}^{+}.\)
\end{lemma}
\begin{proof}
    See \cref{sec:3rdproofvar2}.
\end{proof}
A simple application of Chebyshev's inequality, as in \cref{eq:cheby1}, yields the following result in probability:
\[\frac{\left|T_{2}^{(3)} - \E \left(T_{2}^{(3)}\right)\right|}
{\sqrt{\Var\left(T_{2}^{(S)}\right)}} 
\to 0.\]
For the higher order terms, we follow the same proof as in one-sample case, which together with the asymptotic normality established in \cref{lem:second-order-clt2}, yields the conclusion of \cref{thm:twosample}.

\section{Proof of Estimator Consistency (\cref{thm:1sampleconsis,thm:2sampleconsis})}\label{sec:estproof} 
To prove these result we start with the following lemma that shows the empirical eigenvectors are sufficiently incoherent.
\begin{lemma}\label{lem:Very Good Set}
Let $\mathcal{E}_{{\sf good}}$ be the event from \cref{lem:Good Set} and define the event
\[
\mathcal{E}_{{\sf very \ good}} = \mathcal{E}_{{\sf good}} \cap \bigg\{\\ \|\hat{V}\|_{2,\infty} \lesssim \sqrt{\frac{k}{n}}\,\bigg\}.
\]
Then $\mathbb{P}(\mathcal{E}_{{\sf very \ good}})\ge1 - O(n^{-19}).$
\end{lemma}
\begin{proof}
    See \cref{sec:verygoodproof}.
\end{proof}
The event \(\mathcal{E}_{{\sf very \ good}}\) defines a regime where both the spectral norm bound \(\|X\| \lesssim \sqrt{n \rho_n}\) and the incoherence condition on \(\hat V\) hold simultaneously. These are crucial for controlling higher-order error terms and ensuring the consistency of the proposed estimator.

\subsection{Entrywise Decomposition for Estimated Adjacency Matrices}
Our estimator consistency results are built upon the following series expansion for the difference ${\hat P} - P$, where ${\hat P}$ is the rank $k$ approximation of $A$.
\begin{lemma}[Series expansion for the low-rank estimator]
\label{lem:series-expansion}
Let \({\widehat P}\) denote the best (with respect to Frobenius norm) rank-\(k\) approximation of \(A\).  Then under \cref{ass:asymptotic1,ass:sparse,ass:eigen-scaling,ass:incoherence}, \({\widehat P}-P\) admits the following series expansion:
\[
{\widehat P} - P = \sum_{l=1}^{\infty} T_{l}(X),
\]
where the first-order term is
\[
T_{1}(X)
= V V\t X (I - V V\t)+(I-V V\t) X V V\t+V V\t X V V\t,
\]
and, for each integer \(l\ge 2\),
\begin{align*}
T_{l}(X)
&= S_{l}(X)\,P + P\,S_{l}(X) + \sum_{l_1+l_2=l} S_{l_1}(X)\,P\,S_{l_2}(X)\\
&\qquad + S_{l-1}(X)\,X\,V V\t +  V V\t X S_{l-1}(X)
+ \sum_{l_1+l_2=l-1} S_{l_1}(X)\,X\,S_{l_2}(X).
\end{align*}
Here \(S_{l}(X)\) are the polynomial operators defined as in \cref{Sk_1sample}.  Moreover, the higher-order terms satisfy the operator-norm bound
\[
\|T_{l}(X)\| \lesssim c_l \left(\frac{\|X\|}{n\rho_n}\right)^{\!l}\,(n\rho_n),
\]
for constants \(c_l>0\) such that \(\log c_l = O(1)\).
\end{lemma}
\begin{proof}
    See \cref{sec:series-expansionproof}.
\end{proof}

\subsection{Consistency of Intermediate Quantities}
We now study the plug-in estimates we use in the definition of our mean estimator.  
\begin{lemma}\label{inter1}
Define $\Sigma, \hat{\Sigma}$ as in \cref{thm:onesample,thm6.2}. Define the matrices
$$
    M := \Diag(\Sigma \mathbf{1}_n), \quad 
    \hat{M} := \Diag(\hat{\Sigma} \mathbf{1}_n), \quad 
    \Delta M := M - \hat{M}, \quad 
    \Delta\Diag(\Sigma) := \Diag(\Sigma) - \Diag(\hat{\Sigma}).
$$
Then, under \cref{ass:asymptotic1,ass:dense,ass:eigen-scaling,ass:incoherence}, the following bounds hold:
\begin{align}
    \big| \tr\bigl(VV\t\,\Delta M\bigr) \big| &= O_p\bigl(\max (k,k \sqrt{\rho_n \log n})\bigr); \label{inter-a} \\
    \big| \tr\bigl(VV\t\,\Delta\Diag(\Sigma)\bigr)\big| &= O_p\left(\max \left(\frac{k}{n}, \frac{k \sqrt{\rho_n \log n}}{n}\right)\right). \label{inter-b}
\end{align}
\end{lemma}
\begin{proof}
    See \cref{sec:interproof1}.
\end{proof}

\begin{lemma}\label{inter3}
Defining the quantities as in \cref{inter1}, the following bounds hold:
\begin{align}
\big| \tr\bigl((V \Lambda^{-2} V\t-\widehat V \widehat \Lambda^{-2} \widehat V\t)\,\Diag(\hat{\Sigma}\bigr)\big|&=O_p \left(\max \left(\frac{k}{n^3 \rho_n^2},\frac{k \sqrt{\log n}}{n^3 \rho_n^{3/2}}\right)\right);\label{inter-d1}\\ 
\big| \tr\bigl((V \Lambda^{-2} V\t-\widehat V \widehat \Lambda^{-2} \widehat V\t)\,\Diag(\hat{\Sigma}\cdot \mathbf{1}_n)\bigr)\big|&=O_p \left(\max \left(\frac{k}{n^2 \rho_n^2},\frac{k \sqrt{\log n}}{n^2 \rho_n^{3/2}}\right)\right).\label{inter-d2}
\end{align}
\end{lemma}
\begin{proof}
    See \cref{sec:interproof3}.
\end{proof}

Next, we will study the plug-in estimates we use in the definition of our variance estimator.  
\begin{lemma}\label{inter2}
Define the vectors
\(
\alpha := \diag(\beta^{-2})
\)
and
\(
\hat{\alpha} := \diag(\hat{\beta}^{-2})
\). Then, under \cref{ass:asymptotic1,ass:dense,ass:eigen-scaling,ass:incoherence}, the following bounds hold:
\begin{align}
\|\hat{\alpha}-\alpha\| &=O_p\left( \frac{k}{n^{3}\rho_n^{2.5}}\right);\label{inter-c}\\
\big\|\hat{G}\circ\hat{G}-G\circ G\big\|_F &=O_p\left( \frac{k^{3/2}}{n^{5.5}\rho_n^{4.5}}\right);\label{inter-e}\\
\left\| \hat \beta^{-2} - \beta^{-2} \right\|_F &=O_p\left(\frac{k^{1/2}}{n^{2.5}\rho_n^{2.5}}\right).\label{inter-g}
\end{align}
where \(G\) and \(\hat{G}\) are as defined in \cref{thm:twosample,thm3.0}.
\end{lemma}
\begin{proof}
    See \cref{sec:interproof2}.
\end{proof}

\begin{lemma}\label{inter4}
Defining the quantities as in \cref{lem:series-expansion}, the following bounds hold:
\begin{align}
\left\| \hat{P}^2 - P^2 \right\|_F &=O_p\left( k^{1/2} n^{3/2} \rho_n^{3/2}\right);\label{inter-f1}\\
\left\| \hat{P} ({\hat P} \circ {\hat P}) - P ( P \circ P) \right\|_F &=O_p\left( k^{1/2} n^{3/2} \rho_n^{5/2}\right);\label{inter-f2}\\
\left\| ({\hat P} \circ \hat{P})^2 - ( P \circ P)^2 \right\|_F &=O_p\left( k^{1/2} n^{1.5} \rho_n^{3.5}\right).\label{inter-f3}
\end{align}
\end{lemma}
\begin{proof}
    See \cref{sec:interproof4}.
\end{proof}

\subsection{Consistency of Final Estimators}
We first establish the consistency of the one–sample estimators, which will subsequently be used to derive the consistency of the two–sample estimators. Throughout the proof, we work on the event $\mathcal{E}_{{\sf very \ good}}$ which has been defined in \cref{lem:Very Good Set}.
\subsubsection{Proof of  \cref{thm:1sampleconsis}}\label{prf:FinEst1}
\begin{proof}
We now establish \cref{thm:1sampleconsis}, which concerns the consistency of the mean and variance estimators in the one–sample setting under \cref{ass:dense}.
\\ \ \\
\textbf{Analyzing the estimated mean.} From \cref{param:11}, the population mean admits the expansion
\[
\mu_1
=\underbrace{2\tr\bigl[\beta^{-2}\bigl(\Sigma \circ\beta^\perp\bigr)\bigr]}_{\mu_1^{(1)}}
+\underbrace{2\tr\bigl[\beta^{-2}\bigl(\Diag(\Sigma\cdot d)\bigr)\bigr]}_{\mu_1^{(2)}}
+ O_p\left(\frac{k}{n^2\rho_n^2}\right),
\]
where the remainder term vanishes in \cref{ass:dense}. Note that the original mean expression in \cref{param:11} contained an additional term \(2\,\mathrm{tr}\Bigl(\beta^{-2}\left[\Diag(\diag(\Sigma)\circ d)\right]\Bigr)\), which is absent in the above expression. This is because
\begin{equation*}
\begin{aligned}
\mathrm{tr}\Bigl(\beta^{-2}\left[\Diag(\diag(\Sigma)\circ d)\right]\Bigr)
&\lesssim \mathrm{tr}\Bigl(\beta^{-2}\left[\Diag(\diag(\Sigma)\right]\Bigr) + \mathrm{tr}\Bigl(\beta^{-2}\left[\Diag\left(\diag(\Sigma)\circ \diag(VV\t)\right)\right]\Bigr)\\
& \lesssim \rho_n\,\mathrm{tr}(\beta^{-2}) \\
& \lesssim \frac{k}{n^2\rho_n},
\end{aligned}
\end{equation*}
where the second line follows from \cref{ass:incoherence,ass:asymptotic1}. Hence, this term is absorbed within the remainder term \(O_p\left(\tfrac{k}{n^2\rho_n^2}\right)\).
 The plug–in estimator of $\mu_1$, defined in \cref{eq:estimatorsone}, is given by
\[
\hat{\mu}_1
=\underbrace{2\tr\Bigl[\hat{\beta}^{-2}\bigl(\hat{\Sigma}\circ\hat{\beta}^\perp\bigr)\Bigr]}_{\hat{\mu}_1^{(1)}}
+\underbrace{2\tr\Bigl[\hat{\beta}^{-2}\Diag\bigl(\hat{\Sigma}\cdot\hat{d}\bigr)\Bigr]}_{\hat{\mu}_1^{(2)}},
\]
where all notation follows that of \cref{thm:1sampleconsis}. To control the estimation error, we first consider the leading component $\mu_1^{(1)}$ and decompose the difference as
\begin{align*}
\mu_1^{(1)} - \hat{\mu}_1^{(1)}
&= \tr\left( V \Lambda^{-2} V\t \cdot \Sigma\circ\beta^\perp
      - V \Lambda^{-2} V\t\Diag(\Sigma) \right) \\
&\quad + \tr\left( V \Lambda^{-2} V\t\Diag(\Sigma)
      - V \Lambda^{-2} V\t\Diag(\hat{\Sigma}) \right) \\
&\quad + \tr\left( V \Lambda^{-2} V\t\Diag(\hat{\Sigma})
      - \hat{V} \hat{\Lambda}^{-2} \hat{V}\t\Diag(\hat{\Sigma}) \right) \\
&\quad + \tr\left( \hat{V} \hat{\Lambda}^{-2} \hat{V}\t\Diag(\hat{\Sigma})
      - \hat{V} \hat{\Lambda}^{-2} \hat{V}\t \cdot \hat{\Sigma}\circ\hat{\beta}^\perp \right).
\end{align*}
An analogous decomposition applies to $\mu_1^{(2)}$:
\begin{align*}
\mu_1^{(2)} - \hat{\mu}_1^{(2)} 
&= \tr\left( V \Lambda^{-2} V\t\Diag(\Sigma \cdot d) - V \Lambda^{-2} V\t\Diag(\Sigma \cdot \mathbf{1}_n) \right) \\
&\quad + \tr\left( V \Lambda^{-2} V\t\Diag(\Sigma \cdot \mathbf{1}_n) - V \Lambda^{-2} V\t \Diag(\hat{\Sigma} \cdot \mathbf{1}_n) \right) \\
&\quad + \tr\left(  V \Lambda^{-2} V\t \Diag(\hat{\Sigma} \cdot \mathbf{1}_n) -  \hat{V} \hat{\Lambda}^{-2} \hat{V}\t \Diag(\hat{\Sigma} \cdot \mathbf{1}_n) \right)\\
&\quad + \tr\left(  \hat{V} \hat{\Lambda}^{-2} \hat{V}\t \Diag(\hat{\Sigma} \cdot \mathbf{1}_n) -  \hat{V} \hat{\Lambda}^{-2} \hat{V}\t \Diag(\hat{\Sigma} \cdot \hat{d}) \right),
\end{align*}
where \(d=\diag\left(I-V V\t\right)\) and \(\hat{d}=\diag\left(I-\hat{V} \hat{V}\t\right)\).
\begin{itemize}
    \item \textbf{Bounding $|\mu_1^{(1)} - \hat \mu_1^{(1)} |$}. We first establish that 
\begin{align}\label{mu1bound}
|\mu_1^{(1)} - \hat{\mu}_1^{(1)} |=o_p\left(\max \left(\frac{k}{n^2 \rho_n^2}, \frac{k \sqrt{\rho_n \log n}}{n^2 \rho_n^2}\right)\right).
\end{align}
Consider
\begin{equation}
\begin{aligned} \label{eq:ab1}
    \big|\tr\left( V \Lambda^{-2} V\t  \Sigma \circ \beta^\perp
      - V \Lambda^{-2} V\t \Diag(\Sigma) \right)\big|
      &= |\tr\left( V \Lambda^{-2} V\t  \Sigma \circ (V V\t) \right)| \\
      &\lesssim \frac{1}{n^2 \rho_n^2}\, |\tr\left( V V\t \cdot \Sigma \circ (V V\t) \right)|\\
      &\lesssim \frac{k}{n^2 \rho_n^2}\, \|V V\t\|_F\, \|\Sigma \circ (V V\t)\|_F\\
      &\lesssim \frac{k}{n^3 \rho_n^2}\, \|V V\t\|_F\, \|\Sigma\|_F\\
      &\lesssim \frac{k^{3/2}}{n^2 \rho_n} \\
      &=o_p\left(\max \left(\frac{k}{n^2 \rho_n^2}, \frac{k \sqrt{\rho_n \log n}}{n^2 \rho_n^2}\right)\right),
\end{aligned}
\end{equation}
where the fourth line follows from \cref{ass:incoherence} and the final line follows from \cref{ass:asymptotic1}.

Next, we bound the term involving the deviation between $\Sigma$ and $\hat{\Sigma}$:
\begin{align*}
\big|\tr\left( V \Lambda^{-2} V\t \Diag(\Sigma - \hat{\Sigma}) \right)\big|
&\lesssim \frac{1}{(n\rho_n)^2}\, \big|\tr\left( V V\t \Diag(\Sigma - \hat{\Sigma}) \right)\big|\\
&=O_p\left(\max \left(\frac{k}{n^3 \rho_n^2}, \frac{k \sqrt{\rho_n \log n}}{n^3 \rho_n^2}\right)\right)\\
&=o_p\left(\max \left(\frac{k}{n^2 \rho_n^2}, \frac{k \sqrt{\rho_n \log n}}{n^2 \rho_n^2}\right)\right),
\end{align*}
where the second line uses \cref{inter-b} of \cref{inter1}.

We then control the eigenvalue and eigenvector perturbation component directly using \cref{inter-d1} of \cref{inter3} via
\begin{align*}
\big|\tr\left( V \Lambda^{-2} V\t \Diag(\hat{\Sigma}) - \hat V \hat \Lambda^{-2} \hat V\t \Diag(\hat{\Sigma}) \right)\big|
&=O_p \left(\max \left(\frac{k}{n^3 \rho_n^2},\frac{k \log n}{n^3 \rho_n^{3/2}}\right)\right)=o_p\left(\max \left(\frac{k}{n^2 \rho_n^2}, \frac{k \sqrt{\rho_n \log n}}{n^2 \rho_n^2}\right)\right).
\end{align*}
Finally, for the term
\[
\tr\left( V \Lambda^{-2} V\t\Diag(\hat{\Sigma}) - V \Lambda^{-2} V\t \cdot \hat{\Sigma}\circ\beta^\perp \right)=o_p\left(\max \left(\frac{k}{n^2 \rho_n^2}, \frac{k \sqrt{\rho_n \log n}}{n^2 \rho_n^2}\right)\right),
\]
the same reasoning as in \cref{eq:ab1} follows, with the incoherence of $\hat{V}$ (guaranteed by \cref{lem:Very Good Set}) replacing that of $V$.  

Combining all four bounds above yields the desired rate in \cref{mu1bound}.
\item \textbf{Bounding $| \mu_1^{(2)} - \hat \mu_1^{(2)} |$.}
We next establish that 
\begin{align}\label{mu12bound}
|\mu_1^{(2)} - \hat{\mu}_1^{(2)}| =O_p\left(\max \left(\frac{k}{n^2 \rho_n^2}, \frac{k \sqrt{\rho_n\log n} }{n^2 \rho_n^2}\right)\right).
\end{align}
Consider first
\begin{equation}
\begin{aligned}\label{eq:ab2}
    \big | \tr\left( V \Lambda^{-2} V\t\Diag(\Sigma \cdot d) - V \Lambda^{-2} V\t\Diag(\Sigma \cdot \mathbf{1}_n) \right) \big | 
    &= \tr\left( V \Lambda^{-2} V\t\Diag(\Sigma \cdot \diag(V V\t)) \right) \\
    &\lesssim \frac{1}{n^2 \rho_n^2} \tr\left( V V\t\Diag(\Sigma \cdot \diag(V V\t)) \right) \\
    &\lesssim \frac{1}{n^2 \rho_n^2} \left( \frac{k}{n} \right) \tr\Big( \Diag(\Sigma \cdot \diag(V V\t)) \Big) \\
    &\lesssim \frac{k}{n^3 \rho_n^2}\, \big( \mathbf{1}_n\t \cdot \Sigma \cdot \diag(V V\t) \big) \\
    &\lesssim \frac{k}{n^3 \rho_n^2} \left( \frac{k}{n} \right) \big( \mathbf{1}_n\t \cdot \Sigma \cdot \mathbf{1}_n \big) \\
    &\lesssim \frac{k^2}{n^4 \rho_n^2} (n^2 \rho_n) \\
    &= \frac{k^2}{n^2 \rho_n}\\
    &= o_p\left(\max \left(\frac{k}{n^2 \rho_n^2}, \frac{k \sqrt{\rho_n \log n}}{n^2 \rho_n^2}\right)\right),
\end{aligned}
\end{equation}
where we have used \cref{ass:incoherence} and the final line follows from \cref{ass:asymptotic1}.

Next, for the perturbation in $\Sigma$:
\begin{align*}
\big| \tr\left( V \Lambda^{-2} V\t (\Diag(\Sigma \cdot \mathbf{1}_n) - \Diag(\hat{\Sigma} \cdot \mathbf{1}_n)) \right) \big|
&\lesssim \frac{1}{(n\rho_n)^2} \big| \tr\left( V V\t (\Diag(\Sigma \cdot \mathbf{1}_n) - \Diag(\hat{\Sigma} \cdot \mathbf{1}_n)) \right) \big| \\
&=O_p\left(\max \left(\frac{k}{n^2 \rho_n^2}, \frac{k \sqrt{\rho_n\log n} }{n^2 \rho_n^2}\right)\right),
\end{align*}
where the last line follows from \cref{inter-a} of \cref{inter1}.

We then control the eigenvalue and eigenvector perturbation component directly using \cref{inter-d2} of \cref{inter3}:
\begin{align*}
    &\Big|\tr\big( V\Lambda^{-2}V\t\Diag(\hat\Sigma \cdot \mathbf{1}_n)-\hat V\hat\Lambda^{-2}\hat V\t\Diag(\hat\Sigma \cdot \mathbf{1}_n)\big)\Big|  =O_p\left(\max \left(\frac{k}{n^2 \rho_n^2}, \frac{k \sqrt{\rho_n\log n} }{n^2 \rho_n^2}\right)\right).
\end{align*}
Finally, for the term 
\(
\tr\left(  \hat{V} \hat{\Lambda}^{-2} \hat{V}\t \Diag(\hat{\Sigma} \cdot \mathbf{1}_n) -  \hat{V} \hat{\Lambda}^{-2} \hat{V}\t \Diag(\hat{\Sigma} \cdot \hat{d}) \right),
\)
we have the same argument as in \cref{eq:ab2}, using incoherence of \(\hat V\) which is guaranteed on the event $\mathcal{E}_{{\sf very \ good}}$ as per \cref{lem:Very Good Set}. Combining all four bounds yields the desired rate in \cref{mu12bound}.
\end{itemize}
Thus, from \cref{mu1bound} and \cref{mu12bound} we have
\begin{align}\label{mubd}
    \big|\hat{\mu}_1 - \mu_1 \big| =O_p\left(\max \left(\frac{k}{n^2 \rho_n^2}, \frac{k \sqrt{\rho_n\log n} }{n^2 \rho_n^2}\right)\right),
\end{align}
under \cref{ass:asymptotic1,ass:dense,ass:eigen-scaling,ass:incoherence}. 
\\ \ \\
\noindent
\textbf{Analyzing the Estimated Variance.}
From \cref{param:12}, the population variance admits the expansion
\begin{align*}
\sigma_1^2
    = 8\,\underbrace{\Bigl\langle (J_n - I_n)\circ(\beta^{-2}\circ\beta^{-2}),\Sigma^2\Bigr\rangle}_{(\sigma^{11})^{2}}
      + 4\,\underbrace{\alpha\t K\,\alpha}_{(\sigma^{12})^{2}}
      + o\left(\frac{k^2}{n^3 \rho_n^2}\right).
\end{align*}
The corresponding plug–in estimator, as defined in \cref{eq:estimatorsone}, takes the form
\[
\hat{\sigma}_1^2 
= 8\,\underbrace{\Bigl\langle (J_n - I_n)\circ(\hat{\beta}^{-2}\circ\hat{\beta}^{-2}),\hat{\Sigma}^2\Bigr\rangle}_{(\hat \sigma^{11})^{2}}
      + 4\,\underbrace{\hat{\alpha}\t \hat{K}\,\hat{\alpha}}_{(\hat \sigma^{12})^{2}},
\]
where all notation follows that of \cref{thm:1sampleconsis}. To analyze the estimation error, we first focus on the first component $\left(\sigma^{11}\right)^{2}$ and decompose the difference as
\begin{align*}
    (\hat \sigma^{11})^{2} - ( \sigma^{11})^{2} 
    &= \underbrace{\left\langle (J_n - I_n)\circ(\hat{\beta}^{-2}\circ\hat{\beta}^{-2}),\,\hat{\Sigma}^2 - \Sigma^2 \right\rangle}_{I_1} 
     + \underbrace{\left\langle (J_n - I_n)\circ(\hat{\beta}^{-2}\circ\hat{\beta}^{-2} - \beta^{-2}\circ\beta^{-2}),\,\Sigma^2 \right\rangle}_{I_2}.
\end{align*}
An analogous decomposition applies to $( \sigma^{12})^{2}$:
\begin{align}
(\hat \sigma^{12})^{2} - ( \sigma^{12})^{2} 
&= \alpha\t \hat{K} \alpha - \hat{\alpha}\t K \hat{\alpha} \nonumber \\
&= \alpha\t (\hat{K} - K)\alpha
   + 2\,\alpha\t \hat{K} (\hat{\alpha} - \alpha)
   + (\hat{\alpha} - \alpha)\t \hat{K} (\hat{\alpha} - \alpha). \label{sigma12decomposition}
\end{align}
\begin{itemize}
    \item \textbf{Bounding $|(\hat \sigma^{11})^{2} - ( \sigma^{11})^{2}|$}. 
We will establish that 
\begin{equation}
\begin{aligned}\label{sigma11bound}
\big|(\hat \sigma^{11})^{2} - ( \sigma^{11})^{2}\big| =o_p \left( \frac{k^2}{n^3 \rho_n^2}\right).
\end{aligned}
\end{equation}
Applying the Cauchy–Schwarz inequality and noting that $\|(J_n-I_n)\circ M\|_F \le \|M\|_F$, we obtain
\begin{equation}
\begin{aligned} \label{eq:i1}
\big| I_1 \big| 
&= \left\langle (J_n - I_n)\circ(\hat{\beta}^{-2}\circ\hat{\beta}^{-2}),\,\hat{\Sigma}^2 - \Sigma^2 \right\rangle \\
&\le \left\| (J_n - I_n)\circ(\hat{\beta}^{-2}\circ\hat{\beta}^{-2}) \right\|_F 
      \left\| \hat{\Sigma}^2 - \Sigma^2 \right\|_F \\
&\lesssim \left\| \hat{\beta}^{-2}\circ\hat{\beta}^{-2} \right\|_F 
        \left\| \hat{\Sigma}^2 - \Sigma^2 \right\|_F \\
&\lesssim \left\| \hat{\beta}^{-2}\circ\hat{\beta}^{-2} \right\|_F \left(
        \left\| \hat{P}^2 - P^2 \right\|_F + 2 \left\| \hat{P} ({\hat P} \circ {\hat P}) - P ( P \circ P) \right\|_F + \left\| ({\hat P} \circ \hat{P})^2 - ( P \circ P)^2 \right\|_F\right)\\
&\lesssim \frac{1}{n^4 \rho_n^4}
         \left\| ((\hat {V}^{(i)}) (\hat {V}^{(i)})\t) \circ ((\hat {V}^{(i)}) (\hat {V}^{(i)})\t) \right\|_F 
         \left(
        \left\| \hat{P}^2 - P^2 \right\|_F\right) \\
&=O_p \left( \frac{(k/n)}{n^4 \rho_n^4}\, k^{1/2} (n\rho_n)^{1.5} \right)\\
&= o_p \left( \frac{k^{3/2}}{n^3 \rho_n^2} \right).
\end{aligned}
\end{equation}
Here, the third to last line and the second to last line follow from \cref{inter-f1,inter-f2,inter-f3} in \cref{inter4}. The bound 
$\|((\hat {V}^{(i)}) (\hat {V}^{(i)})\t) \circ ((\hat {V}^{(i)}) (\hat {V}^{(i)})\t)\|_F \asymp k/n$ in the second to last line 
is ensured by \cref{ass:incoherence}.

Note the following bounds:
\begin{align}\label{note01}
\|\hat{\beta}^{-2}+\beta^{-2}\|_{\infty} \le \|\hat{\beta}^{-2}\|_{\infty} + \|\beta^{-2}\|_{\infty}=O_p \left(\frac{1}{n^3 \rho_n^2} \right),
\end{align}
\begin{equation}
\begin{aligned}\label{note02}
\|\Sigma^2\|_F
&\le
\|\Sigma^2-P^2\|_F+\|P^2\|_F
\\
&\le \|(\Sigma-P)\Sigma+P(\Sigma-P) \|_F+\|P^2\|_F\\
&\lesssim
n\rho_n\|\Sigma-P\|_F+k n^2 \rho_n^2\\
&\lesssim
n\rho_n\|P\circ P\|_F+k n^2 \rho_n^2\\
&\lesssim
n\rho_n^2\|P\|_F+k n^2 \rho_n^2\\
&\lesssim k n^2 \rho_n^2.
\end{aligned}
\end{equation}
Now, consider the term \(I_2\). We have
\begin{equation}
\begin{aligned}\label{eq:i2}
\big| I_2 \big|
&=
\Big|
\big\langle
(J_n-I_n)\circ(\hat{\beta}^{-2}\circ\hat{\beta}^{-2}-\beta^{-2}\circ\beta^{-2}),
\,\Sigma^2
\big\rangle
\Big| \\[4pt]
&\lesssim
\Big|
\big\langle
(\hat{\beta}^{-2}-\beta^{-2})\circ(\hat{\beta}^{-2}+\beta^{-2}),
\,\Sigma^2
\big\rangle
\Big| \\[4pt]
&\lesssim
\|
(\hat{\beta}^{-2}-\beta^{-2})\circ(\hat{\beta}^{-2}+\beta^{-2})
\|_F\,
\|\Sigma^2\|_F \\[4pt]
&\lesssim
\|\Sigma^2\|_F\,
\|\hat{\beta}^{-2}+\beta^{-2}\|_{\infty}\,
\|\hat{\beta}^{-2}-\beta^{-2}\|_F \\[4pt]
&\lesssim
k n^2 \rho_n^2 \cdot
\|\hat{\beta}^{-2}+\beta^{-2}\|_{\infty}\,
\|\hat{\beta}^{-2}-\beta^{-2}\|_F \\[4pt]
&=O_p \left(
\frac{k^{3/2}}{n^{7/2}\rho_n^{5/2}} \right)\\
&=
o_p \left(\frac{k^{3/2}}{n^3\rho_n^2}\right).
\end{aligned}
\end{equation}
The third to last line and the second to last line follow from \cref{inter-g} in \cref{inter2} and \cref{note01,note02}.
Combining \cref{eq:i1,eq:i2} yields \cref{sigma11bound}.

\item \textbf{Bounding $|(\hat \sigma^{12})^{2} - ( \sigma^{12})^{2}|$}. 
We will now establish that 
\begin{align}
|\hat{\sigma}_{12}^2 - \sigma_{12}^2 |=o_p \left( \frac{k^2}{n^3 \rho_n^2}\right). \label{sigma12bound}
\end{align}
Recalling that $\alpha = \diag(\beta^{-2})$, we have that 
\(
\|\alpha\|_{\infty} \asymp \frac{k}{n^3 \rho_n^{2}}, \) and 
\(\|\alpha\|_{2} \asymp \frac{k}{n^{2.5} \rho_n^{2}}.
\) For the first term on the right hand side of \cref{sigma12decomposition}, we have
\begin{equation}
\begin{aligned}\label{kbd01}
|{\alpha}\t (\hat{K} - K) {\alpha}| 
&\lesssim \|\alpha\|^2 \cdot \|\hat{K} - K\| \\
&\lesssim \left( \frac{k^2}{n^5 \rho_n^4} \right) \cdot \|P - {\hat P}\| \\
&=O_p \left( \frac{k^2}{n^{4.5} \rho_n^{3.5}} \right) = o_p \left( \frac{k^2}{n^3 \rho_n^2} \right),
\end{aligned}
\end{equation}
which holds on the event $\mathcal{E}_{{\sf good}}$ as explained in \cref{lem:Good Set}. To establish the bound $\|\widehat K-K\| \lesssim \|P-{\widehat P}\|$, we first observe that
$$
    \|\widehat K-K\| \lesssim \|\widehat\Sigma-\Sigma\| + \|\widehat\Sigma\circ\widehat P-\Sigma\circ P\|.
$$
By the triangle inequality, the second term on the right-hand side satisfies
\begin{align}\label{note03}
    \big\|\widehat\Sigma\circ\widehat P-\Sigma\circ P \big\| \le \big\|\widehat\Sigma \circ (\widehat P- P)\big\| + \big\| P \circ (\widehat \Sigma- \Sigma)\big\|.
\end{align}
For the second term on the right hand side of \cref{note03}, \cref{ass:incoherence,ass:eigen-scaling,ass:asymptotic1} imply
\begin{equation}\label{schur01}
\begin{aligned} 
    \big\| P \circ (\widehat \Sigma- \Sigma)\big\| &= \bigg\|\sum_{l=1}^{k} \lambda_l \Diag(V_{l.}) (\widehat \Sigma- \Sigma) \Diag(V_{l.})\bigg\| \\
    &\le \left(\sum_{l=1}^{k} \lambda_l \|V_{l.}\|_{\infty}^2 \right) \big\|\widehat \Sigma- \Sigma \big\| \\
    &\le k^2 \rho_n \big\|\widehat \Sigma- \Sigma \big\| \ll \big\|\widehat \Sigma- \Sigma \big\|.
\end{aligned}
\end{equation}
Similarly, the first term on the right hand side of \cref{note03} can be bounded as
\begin{equation}\label{schur02}
\begin{aligned} 
    \big\|\widehat\Sigma \circ (\widehat P- P)\big\| &\le \big\|\widehat P \circ (\widehat P- P)\big\| + \big\|(\widehat P \circ \widehat P) \circ (\widehat P- P)\big\| \\
    &\ll \big\|\widehat P- P \big\|
\end{aligned}
\end{equation}
with probability $1-O(n^{-19})$, where the final step follows from \cref{lem:Very Good Set}, \cref{ass:eigen-scaling,ass:asymptotic1}, and the fact that $\rank(\widehat P \circ \widehat P) \le k^2$. Under the same assumptions, an analogous argument yields $\|\widehat\Sigma-\Sigma\| \lesssim \|P-{\widehat P}\|$. Combining this with \cref{schur01} and \cref{schur02} establishes the desired bound $\|\widehat K-K\| \lesssim \|P-{\widehat P}\|$.

For the cross term in \cref{sigma12decomposition}, it follows that
\begin{equation}
\begin{aligned}\label{kbd02}
|{\alpha}\t \hat{K} (\hat{\alpha} - {\alpha})| 
&\leq \|\alpha\| \, \|\hat{K}\| \, \|\hat{\alpha} - {\alpha}\| \\
&=O_p \left( \frac{k^{2}}{n^{4.5} \rho_n^{3.5}} \right) = o_p \left( \frac{k^{2}}{n^3 \rho_n^2}\right),
\end{aligned}
\end{equation}
where the second line uses \cref{inter-c} of \cref{inter2}.

For the quadratic term in \cref{sigma12decomposition}, again by \cref{inter-c},
\begin{equation}
\begin{aligned}\label{kbd03}
|(\hat{\alpha} - {\alpha})\t \hat{K} (\hat{\alpha} - {\alpha})| 
&\leq \|\hat{K}\| \cdot \|\hat{\alpha} - {\alpha}\|^2 \\
&=O_p \left( n \rho_n \cdot \left( \frac{k}{n^{3} \rho_n^{2.5}} \right)^2 \right)  =o_p \left( \frac{k^2}{n^3 \rho_n^2}\right).
\end{aligned}
\end{equation}
Combining \cref{kbd01,kbd02,kbd03,sigma12decomposition} yields
\begin{align}\label{sigma12bound}
    \big|(\hat \sigma^{12})^{2} - ( \sigma^{12})^{2} \big|
= o_p\left( \frac{k^{2}}{n^3 {\rho_n}^2} \right).
\end{align}
\end{itemize}
\noindent
Consequently, \cref{sigma12bound} and \cref{sigma11bound} imply that 
\begin{align}\label{sigbd}
    |\hat{\sigma}_1^2 - \sigma_1^2| = o_p \left( \frac{k^2}{n^3 \rho_n^2} \right).
\end{align}

\noindent
\Cref{lem:second-order1} implies that \(\hat{\sigma}_1^2 \asymp \frac{k^{2}}{n^{3}\rho_n^{2}}\). Combining this with \cref{sigbd}, we obtain
\(
    \frac{\hat{\sigma}_1}{\sigma_1} \xrightarrow{p} 1.
\)
Furthermore, \Cref{ass:dense}, \cref{mubd} and \Cref{lem:second-order1} imply that
\(
    \frac{\hat{\mu}_1 - \mu_1}{\hat{\sigma}_1} \xrightarrow{p} 0,
\)
which verifies the conditions required for the consistency of the estimator.
\end{proof}

\subsubsection{Proof of \cref{thm:2sampleconsis}}
We will next establish \cref{thm:2sampleconsis}.

\begin{proof}
\noindent\textbf{Computing the mean}. The population mean for the two–sample statistic admits the decomposition
\[
\mu_2=2\sum_{i=1}^{2}\mu_2^{(i)}
=2\sum_{i=1}^{2}\tr\Bigl[\beta_i^{-2}\Bigl({\Sigma^{(i)}}\circ\beta_i^\perp 
      + \Diag({\Sigma^{(i)}}\cdot d_i)\Bigr)\Bigr],
\]
where \(\mu_2^{(i)}:=\tr\Bigl[\beta_i^{-2}\Bigl({\Sigma^{(i)}}\circ\beta_i^\perp 
      + \Diag({\Sigma^{(i)}}\cdot d_i)\Bigr)\Bigr].\)
The corresponding plug–in estimator introduced in \cref{eq:estimators} is
\[
\hat{\mu}_2 
= \sum_{i=1}^{2}\hat{\mu}_2^{(i)}
= 2\sum_{i=1}^{2}\tr\Bigl[\hat{\beta}_i^{-2}\Bigl({\hat{P}^{(i)}}\circ(J_n-{\hat{P}^{(i)}})\circ\hat\beta_i^\perp
      + \Diag\bigl(({\hat{P}^{(i)}}\circ(J_n-{\hat{P}^{(i)}}))\cdot\hat d_i\bigr)\Bigr)\Bigr],
\]
with notation as in \cref{thm:2sampleconsis}.

Since the two–sample mean decomposes as the sum of the two one–sample contributions, the consistency and rate established in \cref{thm:1sampleconsis} apply componentwise. In particular, for each \(i\in\{1,2\}\) we have
\(
\big|\hat{\mu}_2^{(i)}-\mu_2^{(i)}\big| = O_p\Bigl(\frac{k}{n^2\rho_n^2}\Bigr),
\)
and hence
\begin{align}\label{mu2bd}
    \big|\hat{\mu}_2-\mu_2\big| = O_p \left(\max \left(\frac{k}{n^2 \rho_n^2},\frac{k \sqrt{\log n}}{n^2 \rho_n^{3/2}}\right)\right),
\end{align}
which establishes the claimed mean consistency for the two–sample statistic and thus proves the first part of \cref{thm:2sampleconsis}.
\\ \ \\
\noindent\textbf{Computing the Variance}. From \cref{thm:twosample_dense}, the population variance for the two–sample statistic admits the decomposition
\begin{align*}
\sigma^2_2 
&= \sum_{i=1}^{2} \bigl(\sigma_2^{(i)}\bigr)^2 + \sigma_2^{(1,2)} \\
&= \underbrace{\sum_{i=1}^{2} \!\left( 
8 \,\bigl\langle (J_n - I_n)\!\circ\!(\beta_i^{-2}\!\circ\!\beta_i^{-2}),\,({\Sigma^{(i)}})^2 \bigr\rangle
+ 4\,\alpha_i\t K^{(i)}\,\alpha_i \right)}_{\sum_{i=1}^{2} (\sigma_2^{(i)})^2}
+ \underbrace{4 \,\bigl\langle G\!\circ\!G,
{\Sigma^{(1)}}\t{\Sigma^{(2)}} + (J_n-I_n)\!\circ\!({\Sigma^{(1)}}\!\circ\!{\Sigma^{(2)}}) \bigr\rangle}_{\sigma_2^{(1,2)}}.
\end{align*}
The corresponding plug–in estimator introduced in \cref{eq:estimators} takes the form
\begin{equation*}
\begin{aligned}
    \hat{\sigma}^2_2 
    &= \sum_{i=1}^{2} \bigl(\hat{\sigma}_2^{(i)}\bigr)^2 + \hat{\sigma}_2^{(1,2)} \\
    &= \underbrace{
        \begin{aligned}
            \sum_{i=1}^{2} \biggl[ \,
            & 8\,\bigl\langle (J_n - I_n)\circ(\hat{\beta}_i^{-2}\circ\hat{\beta}_i^{-2}),
            \,(\hat{\Sigma}^{(i)})^2 \bigr\rangle  + 4\,\hat{\alpha}_i^\top {\hat{K}^{(i)}}\,\hat{\alpha}_i \biggr]
        \end{aligned}
       }_{\sum_{i=1}^{2} (\hat{\sigma}_2^{(i)})^2} \\
    &\quad + \underbrace{
        \begin{aligned}
            4\,\bigl\langle \hat{G}\circ\hat{G},
            (\hat{\Sigma}^{(1)})^\top(\hat{\Sigma}^{(2)}) 
            + (J_n-I_n)\circ((\hat{\Sigma}^{(1)})\circ(\hat{\Sigma}^{(2)}))\bigr\rangle
        \end{aligned}
       }_{\hat{\sigma}_2^{(1,2)}}
\end{aligned}
\end{equation*}
with notation as in \cref{thm:2sampleconsis}.

From the one–sample variance consistency result as in \cref{sigbd}, it follows that 
\(|
(\hat{\sigma}_2^{(i)})^2 - (\sigma_2^{(i)})^2| \ll \frac{k^2}{n^3\rho_n^2},\) for \(i = 1,2\).
To establish the full variance consistency of the two–sample statistic, it remains to verify that 
\(
|\hat{\sigma}_2^{(1,2)} - \sigma_2^{(1,2)}| =\; o_p \left(\frac{k^2}{n^3\rho_n^2}\right).
\) We decompose the difference as
\begin{align*}
    \hat{\sigma}_2^{(1,2)} - \sigma_2^{(1,2)} 
    &= \Big\langle \hat{G}\!\circ\!\hat{G},\;
    (\hat{\Sigma}^{(1)})\t(\hat{\Sigma}^{(2)}) 
    + (J_n-I_n)\!\circ\!((\hat{\Sigma}^{(1)})\!\circ\!(\hat{\Sigma}^{(2)}))\Big\rangle 
    - \Big\langle G\!\circ\!G,\;
    {\Sigma^{(1)}}\t{\Sigma^{(2)}} + (J_n-I_n)\!\circ\!({\Sigma^{(1)}}\!\circ\!{\Sigma^{(2)}}) \Big\rangle \\
    &= \underbrace{\Big\langle \hat{G}\!\circ\!\hat{G} - G \circ G,\; 
    (\hat{\Sigma}^{(1)})\t(\hat{\Sigma}^{(2)}) +
    (J_n-I_n)\!\circ\!((\hat{\Sigma}^{(1)})\!\circ\!(\hat{\Sigma}^{(2)}))\Big\rangle}_{Q_2} \\
    &-
    \underbrace{\Big\langle G\!\circ\!G,\;
    ({\Sigma^{(1)}}\t{\Sigma^{(2)}} - (\hat{\Sigma}^{(1)})\t(\hat{\Sigma}^{(2)}))
    + (J_n-I_n)\!\circ\!({\Sigma^{(1)}}\!\circ\!{\Sigma^{(2)}} - (\hat{\Sigma}^{(1)})\!\circ\!(\hat{\Sigma}^{(2)}))
    \Big\rangle}_{Q_1}.
\end{align*}
We first bound \(Q_2\). Observe that
\begin{equation}
\begin{aligned}\label{newG1}
    \big| \langle \hat{G}\!\circ\!\hat{G} - G \circ G,\; 
    (\hat{\Sigma}^{(1)})\t(\hat{\Sigma}^{(2)}) +
    (J_n-I_n)\!\circ\!((\hat{\Sigma}^{(1)})\!\circ\!(\hat{\Sigma}^{(2)}))\rangle \big|
    &\lesssim 
    \big| \langle \hat{G}\!\circ\!\hat{G} - G \circ G,\; 
    (\hat{\Sigma}^{(1)})\t(\hat{\Sigma}^{(2)}) \rangle \big| \\
    &\lesssim 
     \| \hat{G}\!\circ\!\hat{G} - G \circ G\|_F \|(\hat{\Sigma}^{(1)})\t(\hat{\Sigma}^{(2)})\|_F \\
     &\lesssim \sqrt{k}
     \| \hat{G}\!\circ\!\hat{G} - G \circ G\|_F \|(\hat{\Sigma}^{(1)})\|_2 \|\t(\hat{\Sigma}^{(2)})\|_2 \\
    &=O_p \left(\sqrt{k} \cdot \frac{k^{3/2}}{n^{5.5}\rho_n^{4.5}} \cdot n^2 \rho_n^2\right) \\
    &=O_p \left( \frac{k^{2}}{n^{3.5}\rho_n^{2.5}} \right)\\
    &= o_p \left(\frac{k^2}{n^3 \rho_n^2}\right),
\end{aligned}
\end{equation}
where the fourth line follows from \cref{inter-e} of \cref{inter2}.

Next, we bound \(Q_1\). Observe that
\begin{equation}
\begin{aligned}\label{newG}
& \big|
\langle G\!\circ\!G,\;
({\Sigma^{(1)}}{\Sigma^{(2)}} - (\hat{\Sigma}^{(1)})(\hat{\Sigma}^{(2)}))
+ (J_n-I_n)\!\circ\!({\Sigma^{(1)}}\!\circ\!{\Sigma^{(2)}} - (\hat{\Sigma}^{(1)})\!\circ\!(\hat{\Sigma}^{(2)}))
\rangle
\big|\\
\lesssim
& \big|
\langle G\!\circ\!G,\;
{\Sigma^{(1)}}{\Sigma^{(2)}} - (\hat{\Sigma}^{(1)})(\hat{\Sigma}^{(2)})
\rangle
\big|  + \big|
\langle G\!\circ\!G,
({\Sigma^{(1)}}\!\circ\!{\Sigma^{(2)}} - (\hat{\Sigma}^{(1)})\!\circ\!(\hat{\Sigma}^{(2)}))
\rangle
\big|. 
\end{aligned}
\end{equation}
By Hölder’s inequality for Frobenius inner products,
\[
\big|
\langle G\!\circ\!G,\;
{\Sigma^{(1)}}{\Sigma^{(2)}} - (\hat{\Sigma}^{(1)})(\hat{\Sigma}^{(2)})
\rangle
\big|
\le
\|G\!\circ\!G\|_F\,
\|{\hat{P}^{(1)}}{\hat{P}^{(2)}} - {P^{(1)}}{P^{(2)}}\|_F .
\]
Using the identity
\(
{\hat{P}^{(1)}}{\hat{P}^{(2)}} - {P^{(1)}}{P^{(2)}}
=
{\hat{P}^{(1)}}({\hat{P}^{(2)}}-{P^{(2)}}) - ({\hat{P}^{(1)}}-{P^{(1)}}){P^{(2)}} ,
\)
together with \(\|G\!\circ\!G\|_F \le \|G\|_\infty \|G\|_F\), we obtain
\begin{align} \label{right1}
\|G\!\circ\!G\|_F\,
\|{\hat{P}^{(1)}}{\hat{P}^{(2)}} - {P^{(1)}}{P^{(2)}}\|_F
&\lesssim
\|G\|_\infty \|G\|_F
\Big(
\|{\hat{P}^{(1)}}\|_F \|{X^{(2)}}\|
+
\|{X^{(1)}}\| \|{P^{(2)}}\|_F
\Big).
\end{align}
Invoking \cref{ass:asymptotic1,ass:eigen-scaling,ass:incoherence} and \cref{lem:Very Good Set}, we have
\[\|G\|_\infty \le \|\beta_1^{-1}\beta_2^{-1}\|_{\infty} + \|\widehat \beta_1^{-1} \widehat \beta_2^{-1}\|_{\infty} \lesssim \frac{k}{n^3\rho_n^2}, \qquad \|G\|_F \le \|\beta_1^{-1}\beta_2^{-1}\|_F + \|\widehat \beta_1^{-1} \widehat \beta_2^{-1}\|_F \le \frac{\sqrt{k}}{n^2\rho_n^2}.\]
Thus, \cref{right1} is bounded by
\[
\frac{k}{n^3\rho_n^2}\cdot\frac{\sqrt{k}}{n^2\rho_n^2}
\Big(
\|{\hat{P}^{(1)}}\|_F \|{X^{(2)}}\|
+
\|{X^{(1)}}\| \|{P^{(2)}}\|_F
\Big)
=
O_p \left(\frac{k^2}{n^{3.5}\rho_n^{2.5}}\right)
=
o_p\!\left(\frac{k^2}{n^3\rho_n^2}\right),
\]
which completes the bound for the first term on the right hand side of \cref{newG}.
Before bounding the second term on the right hand side of \cref{newG}, we observe that \(({\Sigma^{(1)}}\!\circ\!{\Sigma^{(2)}} - (\hat{\Sigma}^{(1)})\!\circ\!(\hat{\Sigma}^{(2)}))\) is an \(n \times n\) matrix whose elements are \(O(\rho_n^2).\)
Therefore,
\begin{equation}
\begin{aligned}\label{right2}
    \big|
\langle G\!\circ\!G,\;
({\Sigma^{(1)}}\!\circ\!{\Sigma^{(2)}} - (\hat{\Sigma}^{(1)})\!\circ\!(\hat{\Sigma}^{(2)}))
\rangle
\big| & \lesssim 
\| G\!\circ\!G \|_F \|{\Sigma^{(1)}}\!\circ\!{\Sigma^{(2)}} - (\hat{\Sigma}^{(1)})\!\circ\!(\hat{\Sigma}^{(2)}) \|_F \\
&= O_p\left( \frac{k^{3/2}}{n^5 \rho_n^4} \cdot \sqrt{k} n \rho_n \right) = o_p\!\left(\frac{k^2}{n^3\rho_n^2}\right).
\end{aligned}
\end{equation}
Thus, \cref{right2,newG} give the following bound on \(Q_1\) 
\[\big|
\langle G\!\circ\!G,\;
({\Sigma^{(1)}}{\Sigma^{(2)}} - (\hat{\Sigma}^{(1)})(\hat{\Sigma}^{(2)}))
+ (J_n-I_n)\!\circ\!({\Sigma^{(1)}}\!\circ\!{\Sigma^{(2)}} - (\hat{\Sigma}^{(1)})\!\circ\!(\hat{\Sigma}^{(2)}))
\rangle
\big|=o_p\!\left(\frac{k^2}{n^3\rho_n^2}\right).\]
Combining it with \cref{newG1} where we bounded \(Q_2\), we conclude that
\(
\big| \hat{\sigma}_2^{(1,2)} - \sigma_2^{(1,2)} \big| \ll \frac{k^2}{n^3\rho_n^2},
\) implying
\begin{align}\label{sig2bd}
    \big| \hat{\sigma}_2^2 - \sigma_2^2 \big| \ll \frac{k^2}{n^3\rho_n^2}.
\end{align}

\noindent
\cref{lem:second-order2} implies that \(\hat{\sigma}_2^2 \asymp \frac{k^{2}}{n^{3}\rho_n^{2}}\). Combining this with \cref{sig2bd}, we obtain
\(
    \frac{\hat{\sigma}_2}{\sigma_2} \xrightarrow{p} 1.
\)
Furthermore, \Cref{ass:dense}, \cref{mu2bd} and \cref{lem:second-order2} imply that
\(
    \frac{\hat{\mu}_2 - \mu_2}{\hat{\sigma}_2} \xrightarrow{p} 0,
\)
which completes the proof of \cref{thm:2sampleconsis}.
\end{proof}

\section{Proofs of Corollaries from \texorpdfstring{\cref{TestConsistency}}{Section~\ref{TestConsistency}} and \cref{thm:testcon}}\label{sec:testconpf}
We first state some lemmas which will be used to prove the corollaries for the consistency of our test.
\subsection{Auxiliary Lemmas for Consistency of Our Test}
\begin{lemma}[Cross--term moments]\label{lem:cross-term-moments}
Suppose \cref{ass:asymptotic1,ass:eigen-scaling,ass:incoherence,ass:dense} hold.  Then \[\frac{\langle {V^{(1)}}{V^{(1)}}\t - {V^{(2)}}{V^{(2)}}\t,\; \hat {V^{(2)}}\hat {V^{(2)}}\t - {V^{(2)}}{V^{(2)}}\t\rangle}{\hat\sigma_2} \asymp 1,\]
where $\hat\sigma_2$ is the variance estimator defined in \cref{eq:estimators}.
\end{lemma}
\begin{proof}
    See \cref{sec:cross-term-momentsproof}.
\end{proof}

\begin{lemma}\label{lem:delta-linearization}
Suppose \cref{ass:asymptotic1,ass:eigen-scaling,ass:incoherence,ass:dense} hold.  Let $\Delta= {V^{(1)}} {V^{(1)}}\t - {V^{(2)}} {V^{(2)}}\t.$ Then the quantity
\[
L(\Delta):=\big\langle (\hat V^{(1)}) (\hat V^{(1)})\t - {V^{(1)}} {V^{(1)}}\t,\;\Delta\big\rangle
      -\big\langle (\hat V^{(2)}) (\hat V^{(2)})\t - {V^{(2)}} {V^{(2)}}\t,\Delta\big\rangle,
\]
satisfies
\(
L(\Delta)
= O_p\left(\max \left(\frac{\|\Delta\|_F}{n\rho_n},\frac{\|\Delta\|_F \sqrt{\rho_n\log n}}{n \rho_n}\right)\right).
\)
\end{lemma}
\begin{proof}
    See \cref{sec:delta-linearization-proof}.
\end{proof}

\begin{lemma}[Balanced communities imply incoherence]\label{lem:balanced-incoherence}
Let \(P=ZBZ\t\) denote the population probability matrix of a SBM, where \(Z\in\{0,1\}^{n\times k}\) is the membership matrix with exactly one unit entry per row and \(B\in\mathbb{R}^{k\times k}\) is a symmetric block matrix. Let \(n_a\) denote the size of community \(a\) and set \(N=\mathrm{diag}(n_1,\dots,n_k)=Z\t Z\). Assume \(\rank(P)=k\). If the community sizes are balanced in the sense that
\(
 n_a \asymp \frac{n}{k},
\)
then \cref{ass:incoherence} holds.
\end{lemma}
\begin{proof}
    See \cref{sec:balanced-incoherenceproof}.
\end{proof}

\begin{lemma}[Balanced mixed memberships imply incoherence]\label{lem:balanced-mmsb-incoherence}
Let \(P=ZBZ\t\) denote the population probability matrix of a MMSBM, where \(Z\in[0,1]^{n\times k}\) has rows \({Z^{(i)}}\t\) lying in the probability simplex (\(\sum_{a=1}^k z_{i,a}=1\) for each \(i\)), \(B\in\mathbb{R}^{k\times k}\) is symmetric, and \(\rank(P)=k\). Define the Gram matrix \(N:=Z\t Z\in\mathbb{R}^{k\times k}\) and let \(V\in\mathbb{R}^{n\times k}\) be an orthonormal basis of \(\col(P)\) (the population eigenvectors associated to the nonzero eigenvalues). Suppose there exist constants \(c_1,c_2>0\) (independent of \(n\)) such that
\(
\lambda_{\min}(N)\ge c_1\,\frac{n}{k}
\).
Then \cref{ass:incoherence} holds.
\end{lemma}
\begin{proof}
    See \cref{sec:balanced-mmsbm-incoherenceproof}.
\end{proof}

\subsection{Proof of \cref{thm:testcon}} \label{sec:testconproof}

\noindent
\begin{proof}
We begin by decomposing the test statistic in terms of its population and estimation components. Using the expression of $\hat{\mu}_2$ from \cref{eq:estimators}, we obtain
\begin{equation}
\begin{aligned} \label{eq:Power1}
\frac{\|\hat{V}_{1}\hat{V}_{1}\t - \hat{V}_{2}\hat{V}_{2}\t\|_F^2 - \hat{\mu}_2}{\hat{\sigma}_2} 
&= \frac{\|\hat{V}_{1}\hat{V}_{1}\t - {V^{(1)}}{V^{(1)}}\t\|_F^2 - \hat{\mu}_2^{(1)}}{\hat{\sigma}_2} 
+ \frac{\|\hat{V}_{2}\hat{V}_{2}\t - {V^{(2)}}{V^{(2)}}\t\|_F^2 - \hat{\mu}_2^{(2)}}{\hat{\sigma}_2} \\
&\quad + \frac{\|{V^{(1)}}{V^{(1)}}\t - {V^{(2)}}{V^{(2)}}\t\|_F^2}{\hat{\sigma}_2} 
- 2\,\frac{\langle \hat{V}_{1}\hat{V}_{1}\t - {V^{(1)}}{V^{(1)}}\t,\, \hat{V}_{2}\hat{V}_{2}\t - {V^{(2)}}{V^{(2)}}\t \rangle}{\hat{\sigma}_2}
+ 2\,\frac{L(\Delta)}{\hat{\sigma}_2},
\end{aligned}
\end{equation}
where \(L(\Delta)\) is defined in \cref{lem:delta-linearization}.
From \cref{thm:onesample}, the first two terms on the right hand side of \cref{eq:Power1}
\(
\frac{\|\hat{V}_{i}\hat{V}_{i}\t - {V^{(i)}}{V^{(i)}}\t\|_F^2 - \hat{\mu}_2^{(i)}}{\hat{\sigma}_2}\asymp 1\) 
for \(i = 1, 2\). Furthermore, by \cref{lem:cross-term-moments}, the cross–sample interaction term satisfies
\[
\frac{\langle \hat{V}_{1}\hat{V}_{1}\t - {V^{(1)}}{V^{(1)}}\t,\, \hat{V}_{2}\hat{V}_{2}\t - {V^{(2)}}{V^{(2)}}\t \rangle}{\hat{\sigma}_2} \asymp 1.
\]
Turning to the component involving $L(\Delta)$, \cref{lem:delta-linearization} yields
\(L(\Delta)
= O_p\left(\max \left(\frac{\|\Delta\|_F}{n\rho_n},\frac{\|\Delta\|_F \sqrt{\rho_n \log n}}{n \rho_n}\right)\right).
\)
Finally, recalling from \cref{lem:second-order2} that \(\hat{\sigma}_2 \asymp \dfrac{k}{n^{3/2}\rho_n}\), the decomposition in \cref{eq:Power1} reduces to the asymptotic expansion
\begin{equation}\label{eq:Power2}
\frac{\|\hat{V}_{1}\hat{V}_{1}\t - \hat{V}_{2}\hat{V}_{2}\t\|_F^2 - \hat{\mu}_2}{\hat{\sigma}_2}
\asymp
\zeta_1
+
\frac{n^{3/2}\rho_n}{k}\,\|\Delta\|_F^2
+
\zeta_2,
\end{equation}
where \(\zeta_1\asymp 1\) and \(\zeta_2=O_p\left(\max \left(\frac{\sqrt n\|\Delta\|_F}{k},\frac{\|\Delta\|_F \sqrt{n\rho_n \log n}}{k}\right)\right)\).

Therefore, for the left–hand side of \cref{eq:Power2} to diverge it suffices that  \(\dfrac{n^{3/2}\rho_n}{k}\|\Delta\|_F^2\) diverges; equivalently,
\(
\|\Delta\|_F^2 \gg \frac{k}{n^{3/2}\rho_n}.
\) We note that forcing divergence via \(\zeta_2\) would require \(\|\Delta\|_F^2 \gg \min\left(\dfrac{k^2}{n},\dfrac{k^2}{n \rho_n \log n}\right)\) and under \cref{ass:dense}, we have \(\min \left(\dfrac{k^2}{n},\dfrac{k^2}{n \rho_n \log n} \right)\gg \dfrac{k}{n^{3/2}\rho_n}\).

Thus, in the dense regime (\cref{ass:dense}) the test attains power tending to one whenever
\[
\|{V^{(1)}}{V^{(1)}}\t - {V^{(2)}}{V^{(2)}}\t\|_F^2 \gg \frac{k}{n^{3/2}\rho_n},
\]
which verifies the signal–strength condition stated in \cref{eq:lowerbound_Frobenius} and completes the argument.
\end{proof}

\subsection{Proof of \cref{cor:SBM}}

\noindent
\begin{proof}
Using Lemma 1 from \cite{agterberg_overview_2025}, in the SBM setting with balanced communities, the population eigenspace projector admits the explicit representation
\(
Q_i = {V^{(i)}} {V^{(i)}}\t = {Z^{(i)}} ({Z^{(i)}}\t {Z^{(i)}})^{-1} {Z^{(i)}}\t
= \sum_{l=1}^k \frac{1}{{n^{(i)}_{l}}}\,({Z^{(i)}})_{\cdot l}({Z^{(i)}})_{\cdot l}\t\) for \(i=1,2\),
where ${n^{(i)}_{l}}=n{\pi^{(i)}_{l}}$ denotes the size of community $l$ under model $i$. Since each $Q_i$ is an orthogonal projector of rank $k$, we may write
\begin{align} \label{eq:Q1}
    \|Q_1-Q_2\|_F^2 = \tr(Q_1) + \tr(Q_2) - 2\tr(Q_1Q_2) = 2k - 2\tr(Q_1Q_2).
\end{align}
Hence it suffices to analyse $\tr(Q_1Q_2)$ using the block expansions above. Under the balanced–community assumption (so that ${\pi^{(i)}_{l}}$ are bounded away from $0$ and $\infty$) we have
\begin{equation}
\begin{aligned}\label{trq1q2}
\tr(Q_1Q_2) &= \sum_{l_1,l_2=1}^k \frac{((Z^{(1)})_{\cdot l_1}\t(Z^{(2)})_{\cdot l_2})^2}{{n^{(1)}_{l}}{n^{(2)}_{l}}} \\
&= \frac{1}{n^2}\sum_{l_1,l_2=1}^k \frac{((Z^{(1)})_{\cdot l_1}\t(Z^{(2)})_{\cdot l_2})^2}{{\pi_{l_1}^{(1)}}{\pi_{l_2}^{(2)}}}\\
&=\frac{1}{n^2}\sum_{l=1}^k \frac{((Z^{(1)})_{\cdot l}\t(Z^{(2)})_{\cdot l})^2}{{\pi^{(1)}_{l}}{\pi^{(2)}_{l}}} + O \left(\frac{\|Z^{(1)}-Z^{(2)}\|_0^2}{n^2}\right).
\end{aligned} 
\end{equation}
The last big-O term arises from the fact that \(((Z^{(1)})_{\cdot l_1}\t(Z^{(2)})_{\cdot l_2})^2 \lesssim \|(Z^{(1)})_{\cdot l_1}-(Z^{(2)})_{\cdot l_2}\|_0^2\) if \(l_1 \neq l_2\) and corresponds to the cross term contributions. These cross-terms count nodes that moved from community $l_1$ to community $l_2$.

Let $w_l := (Z^{(1)})_{\cdot l}-(Z^{(2)})_{\cdot l}$. Then
\[
\|w_l\|^2 = (Z^{(1)})_{\cdot l}\t(Z^{(1)})_{\cdot l}+(Z^{(2)})_{\cdot l}\t(Z^{(2)})_{\cdot l}-2\,(Z^{(1)})_{\cdot l}\t(Z^{(2)})_{\cdot l}
= n({\pi^{(1)}_{l}}+{\pi^{(2)}_{l}}) - 2\,(Z^{(1)})_{\cdot l}\t(Z^{(2)})_{\cdot l},
\]
so that
\(
(Z^{(1)})_{\cdot l}\t(Z^{(2)})_{\cdot l} = \frac{{\pi^{(1)}_{l}}+{\pi^{(2)}_{l}}}{2}\,n - \tfrac12\|w_l\|^2.
\)
Squaring gives the expansion
\[
((Z^{(1)})_{\cdot l}\t(Z^{(2)})_{\cdot l})^2
= \Big(\frac{{\pi^{(1)}_{l}}+{\pi^{(2)}_{l}}}{2}\Big)^2 n^2
- \Big(\frac{{\pi^{(1)}_{l}}+{\pi^{(2)}_{l}}}{2}\Big)n\,\|w_l\|^2
+ \tfrac14\|w_l\|^4.
\]
Substituting this into the expression for $\tr(Q_1Q_2)$ in \cref{trq1q2} and using ${n^{(i)}_{l}}=n{\pi^{(i)}_{l}}$ yields
\begin{align*}
    \tr(Q_1Q_2) &=\sum_{l=1}^k \frac{\Big(\frac{{\pi^{(1)}_{l}}+{\pi^{(2)}_{l}}}{2}\Big)^2}{{\pi^{(1)}_{l}}{\pi^{(2)}_{l}}} + \frac{1}{n^2}\sum_{l=1}^k \frac{
- \Big(\frac{{\pi^{(1)}_{l}}+{\pi^{(2)}_{l}}}{2}\Big)n\,\|w_l\|^2
+ \tfrac14\|w_l\|^4}{{\pi^{(1)}_{l}}{\pi^{(2)}_{l}}} + O \left(\frac{\|Z^{(1)}-Z^{(2)}\|_0^2}{n^2}\right).\\ 
&= k + \sum_{l=1}^{k} \frac{({\pi^{(1)}_{l}}-{\pi^{(2)}_{l}})^2}{4 {\pi^{(1)}_{l}} {\pi^{(2)}_{l}}} + \sum_{l=1}^k \Bigg[
- \frac{{\pi^{(1)}_{l}}+{\pi^{(2)}_{l}}}{2n{\pi^{(1)}_{l}}{\pi^{(2)}_{l}}}\,\|w_l\|^2
+ \frac{1}{4n^2{\pi^{(1)}_{l}}{\pi^{(2)}_{l}}}\,\|w_l\|^4\Bigg] + O \left(\frac{\|Z^{(1)}-Z^{(2)}\|_0^2}{n^2} \right)
\\
&= k
+ \sum_{l=1}^k \Bigg[
- \frac{{\pi^{(1)}_{l}}+{\pi^{(2)}_{l}}}{2n{\pi^{(1)}_{l}}{\pi^{(2)}_{l}}}\,\|w_l\|^2
+ \frac{1}{4n^2{\pi^{(1)}_{l}}{\pi^{(2)}_{l}}}\,\|w_l\|^4\Bigg] + O \left(\frac{\|Z^{(1)}-Z^{(2)}\|_0^2}{n^2} \right).
\end{align*}
The leading $k$–term of the above expression cancels with the $2k$ appearing in \cref{eq:Q1}. Hence
\begin{align}\label{q1mq2}
    \|Q_1-Q_2\|_F^2
&= 2k - 2\tr(Q_1Q_2)\\ \notag
&=
2\sum_{l=1}^k \Bigg[
\frac{{\pi^{(1)}_{l}}+{\pi^{(2)}_{l}}}{2n{\pi^{(1)}_{l}}{\pi^{(2)}_{l}}}\,\|w_l\|^2
- \frac{1}{4n^2{\pi^{(1)}_{l}}{\pi^{(2)}_{l}}}\,\|w_l\|^4
\Bigg] + O \left(\frac{\|Z^{(1)}-Z^{(2)}\|_0^2}{n^2} \right).
\end{align}
Now, if \(\frac{\|Z^{(1)}-Z^{(2)}\|_0^2}{n^2} \asymp 1\), then the condition \(
\max_{1 \leq l \leq k} \sum_{j=1}^{n} | Z_{jl}^{(1)} - Z_{jl}^{(2)} | \gtrsim \frac{k}{n^{1/2}\rho_n}
\) is already satisfied trivially. So, we will assume \(\frac{\|Z^{(1)}-Z^{(2)}\|_0}{n} \ll 1\).
From the definition of \(w_l\), we have
\begin{align*}
    &\|w_l\|^2 \leq n_l^{(1)} + n_l^{(2)} = n{\pi^{(1)}_{l}} + n{\pi^{(2)}_{l}} = n({\pi^{(1)}_{l}} + {\pi^{(2)}_{l}}) 
\end{align*}
which implies that
\begin{align*}
 & \frac{1}{4n^2{\pi^{(1)}_{l}}{\pi^{(2)}_{l}}}\,\|w_l\|^4  \lesssim  \frac{{\pi^{(1)}_{l}}+{\pi^{(2)}_{l}}}{2n{\pi^{(1)}_{l}}{\pi^{(2)}_{l}}}\,\|w_l\|^2
\end{align*}
under the balanced–community assumption that ${\pi^{(i)}_{l}}$ are bounded above and below by constants. Thus, discarding the smaller quartic correction and the big-O term gives 
\[
\|Q_1-Q_2\|_F^2 \asymp \frac{1}{n}\sum_{l=1}^k \|(Z^{(1)})_{\cdot l}-(Z^{(2)})_{\cdot l}\|^2.
\]
Combining this representation with the condition from \cref{thm:testcon},
\(
\|Q_1-Q_2\|_F^2 \gg \frac{k}{n^{3/2}\rho_n},
\)
we obtain the necessary lower bound
\(
\sum_{l=1}^k \|(Z^{(1)})_{\cdot l}-(Z^{(2)})_{\cdot l}\|^2 \gg \frac{k}{n^{1/2}\rho_n},
\)
and in particular since community sizes are balanced,
\[
\max_{1 \leq l \leq k} \sum_{j=1}^{n} \left| Z_{jl}^{(1)} - Z_{jl}^{(2)} \right| \gg \frac{k}{n^{1/2}\rho_n},
\]
and each $\|(Z^{(1)})_{\cdot l}-(Z^{(2)})_{\cdot l}\|^2$ equals twice the Hamming distance for binary indicator vectors. This completes the proof.
\end{proof}

\subsection{Proof of \cref{cor:MMSBM}}

\begin{proof}
Adopting the notation from the statement, for \(i=1,2\) write
\(
Q_i = {V^{(i)}} {V^{(i)}}\t = {Z^{(i)}} ({Z^{(i)}}\t {Z^{(i)}})^{-1} {Z^{(i)}}\t,
\)
where the second equality follows from the mixed–membership representation \({V^{(i)}} = {Z^{(i)}} ({Z^{(i)}}\t {Z^{(i)}})^{-1/2}T\) with \(T\) orthogonal. 

We begin with the decomposition
\begin{equation}
\begin{aligned} \label{MMSBM_a_final}
\|Q_1-Q_2\|_F
&=\|Z^{(1)} ({Z^{(1)}}\t Z^{(1)})^{-1}{Z^{(1)}}\t - Z^{(2)} ({Z^{(2)}}\t Z^{(2)})^{-1}{Z^{(2)}}\t \|_F \\
&\le \|(Z^{(1)}-Z^{(2)})({Z^{(1)}}\t Z^{(1)})^{-1}{Z^{(1)}}\t \|_F
      + \|Z^{(2)}\bigl(({Z^{(1)}}\t Z^{(1)})^{-1}-({Z^{(2)}}\t Z^{(2)})^{-1}\bigr){Z^{(1)}}\t \|_F \\
&\qquad\qquad + \|Z^{(2)}({Z^{(2)}}\t Z^{(2)})^{-1}(Z^{(1)}-Z^{(2)})\t\|_F.
\end{aligned}
\end{equation}
Under the spectral gap hypothesis \(\lambda_{\min}({Z^{(i)}}\t {Z^{(i)}})\ge c\,\frac{n}{k}\), we have
\begin{align}\label{bdneed1}
    \|({Z^{(i)}}\t {Z^{(i)}})^{-1}\| \lesssim \frac{k}{n} \quad \text{and} \quad \|{Z^{(i)}}\|_F \lesssim \sqrt{n}
\end{align}
for \(i=1,2.
\)
Hence the first and third terms on the right hand side of \cref{MMSBM_a_final} are immediately bounded by
\begin{align*}
\|(Z^{(1)}-Z^{(2)})({Z^{(1)}}\t Z^{(1)})^{-1}{Z^{(1)}}\t \|_F &\lesssim \frac{k\,\|Z^{(1)}-Z^{(2)}\|_F}{\sqrt{n}}, \\
\|Z^{(2)}({Z^{(2)}}\t Z^{(2)})^{-1}(Z^{(1)}-Z^{(2)})\t \|_F &\lesssim \frac{k\,\|Z^{(1)}-Z^{(2)}\|_F}{\sqrt{n}}.
\end{align*}
To handle the middle term in \eqref{MMSBM_a_final} set \(A={Z^{(1)}}\t Z^{(1)}\) and \(B={Z^{(2)}}\t Z^{(2)}\). Using the identity \(A^{-1}-B^{-1}=A^{-1}(B-A)B^{-1}\) we obtain
\begin{equation}
\begin{aligned} \label{cross_MMSBM_final}
\|Z^{(2)}(A^{-1}-B^{-1}){Z^{(1)}}\t\|_F
&\le \|Z^{(2)}\|_F\,\|A^{-1}(B-A)B^{-1}\|_F\,\|Z^{(1)}\|_F \\
&\le \|Z^{(2)}\|_F \|A^{-1}\| \|B-A\|_F \|B^{-1}\| \|Z^{(1)}\|_F.
\end{aligned}
\end{equation}
A direct expansion gives
\begin{equation}
\begin{aligned} \label{cross_MMSBM_diff}
\|B-A\|_F &= \|{Z^{(2)}}\t Z^{(2)} - {Z^{(1)}}\t Z^{(1)}\|_F
= \|{Z^{(2)}}\t(Z^{(2)}-Z^{(1)})\|_F + \|(Z^{(2)}-Z^{(1)})\t Z^{(1)}\|_F \\
&\lesssim \|Z^{(2)}\|_F\|Z^{(2)}-Z^{(1)}\|_F + \|Z^{(1)}\|_F\|Z^{(2)}-Z^{(1)}\|_F
\lesssim \sqrt{n}\,\|Z^{(1)}-Z^{(2)}\|_F.
\end{aligned}
\end{equation}
Combining \cref{cross_MMSBM_diff} with \cref{cross_MMSBM_final,bdneed1} yields
\(
\|Z^{(2)}(A^{-1}-B^{-1}){Z^{(1)}}\t\|_F \lesssim \frac{k^2\,\|Z^{(1)}-Z^{(2)}\|_F}{\sqrt{n}}.
\)
Collecting the bounds for the three terms on the right hand side of \cref{MMSBM_a_final} we obtain \(
\|Q_1-Q_2\|_F \lesssim \frac{k^2}{\sqrt{n}}\,\|Z^{(1)}-Z^{(2)}\|_F.
\)
Recalling the signal requirement from \cref{thm:testcon},
\(
\|Q_1-Q_2\|_F^2 \gg \frac{k}{n^{3/2}\rho_n},
\)
and substituting the preceding inequality gives the condition on the membership matrices
\(
\|Z^{(1)}-Z^{(2)}\|_F^2 \gg \frac{1}{n^{1/2}\rho_n},
\)
as claimed.
\end{proof}

\section{Proofs of Lemmas from \cref{sec:mainproof}}
In this section we prove all of the technical results stated in \cref{sec:mainproof}.

\subsection{Proof of \cref{lem:Good Set}} \label{sec:goodsetproof}
\begin{proof}
The probability bound for \(\mathcal{E}_{{\sf good}}\) follows from a direct application of Corollary~3.12 of \citet{bandeira2016sharp}.  
Specifically, for a random matrix \(X=A-P\) with independent, mean-zero entries, we have
\(
\mathbb{P}\left(\|X\|\ge (1+\epsilon)\,2\tilde{\sigma}+t\right)
\le
\exp\left(\log n - \frac{t^2}{c_{\epsilon}^2\,\tilde{\sigma}_*^{\,2}}\right),
\)
where \(\tilde{\sigma}=\max_i\sqrt{\sum_j \E({X_{ij}}^2)}\) and 
\(\tilde{\sigma}_*=\max_{i,j}\|X_{ij}\|_{\infty}\).  
In our setting, since the entries of \(A\) are Bernoulli with mean \(P_{ij} \asymp \rho_n\), we have
\(\tilde{\sigma}\asymp \sqrt{n\rho_n}\) and \(\tilde{\sigma}_*\le 1\).  
Choosing \(t=\sqrt{n\rho_n}\) gives
\(
\mathbb{P}\left(\|X\|\gtrsim \sqrt{n\rho_n}\right)
\le
\exp\left(\log n - \frac{n\rho_n}{C^2}\right),
\)
for some universal constant \(C>0\).  
Under \cref{ass:asymptotic1,ass:sparse}, the right-hand side can be shown as \(O(n^{-19})\) .

The Davis–Kahan perturbation bound on the event \(\mathcal{E}_{{\sf good}}\) follows from Corollary~2.8 of \citet{chen_spectral_2021}.  
Since \(\|X\|\lesssim \sqrt{n\rho_n}\) on \(\mathcal{E}_{{\sf good}}\),  
and by \cref{ass:eigen-scaling} we have
\(\|X\|<(1-\frac{1}{\sqrt{2}})\lambda_k\),  
the eigengap condition required by Davis–Kahan holds.  
Consequently,
\(
\|\hat V\hat V\t - VV\t\|_F^2
\lesssim
\frac{k\,\|X\|^2}{(n\rho_n)^2}.
\)
Substituting \(\|X\|\lesssim \sqrt{n\rho_n}\) yields the desired bound
\(
\big\|\hat V\hat V\t - VV\t\big\|_F^2
\lesssim
\frac{k}{n\rho_n},
\)
which holds on the event \(\mathcal{E}_{{\sf good}}\).
\end{proof}

\subsection{Proof of \cref{lem:beta_perp1one}}\label{sec:betaperp1proof}

\begin{proof}
Define
$$M(X) = X^\top X, \qquad M(Y) = Y^\top Y,$$
where $Y = \beta^{\perp} X$. It therefore suffices to analyze
$$\frac{\Var\big(\|\beta^{\perp} X \beta^{-1}\|_F^2-\|X \beta^{-1}\|_F^2\big)}
     {\Var\big(\|X \beta^{-1}\|_F^2\big)}=\frac{\Var \big( \tr \big( \beta^{-1} (M(Y) - M(X)) \beta^{-1} \big) \big)}
     {\Var \big( \tr \big( \beta^{-1} M(X) \beta^{-1} \big) \big)}.$$
We can write
$$M(X) - M(Y) = X^\top X - X^\top \beta^{\perp} \beta^{\perp} X = X^\top (I - \beta^{\perp}) X = X^\top V V^\top X. $$
Denote the inner matrix by $W = V V^\top$. Then
$$\Var \left( \tr \big( \beta^{-1} (M(Y) - M(X)) \beta^{-1} \big) \right)
= 
\Var \left( \tr \left( V V^\top X \beta^{-2} X^\top \right) \right).$$
Hence, to demonstrate that $\Var\big(\tr(\beta^{-1} (M(Y) - M(X)) \beta^{-1})\big)$ is asymptotically negligible compared to $\Var\big(\tr(\beta^{-1} M(X) \beta^{-1})\big)$, it suffices to establish that
$$\frac{
\Var\left(\tr\left(V V^\top X \beta^{-2} X^\top\right)\right)
}{
\Var\left(\tr\left(X \beta^{-2} X^\top\right)\right)
}
\to 0.$$
The trace appearing in the numerator can be expanded as
$$\tr\big(V V^\top X \beta^{-2} X\big) = \sum_{a,b,c} (V V^\top)_{a,b}\, (\beta^{-2})_{c,c}\, X_{a,c} X_{b,c}.$$
We partition the above sum into diagonal ($a=b$) and off-diagonal ($a \neq b$) terms. For the diagonal terms ($a=b$),
\begin{equation}
\begin{aligned} \label{need02}
    \Var\bigg(\sum_{a,c} (V V^\top)_{a,a} (\beta^{-2})_{c,c} X_{a,c}^2 \bigg) & = \sum_{a,c} (V V^\top)_{a,a}^2 (\beta^{-2})_{c,c}^2 \Var(X_{a,c}^2)\\
    & \asymp \rho_n \sum_{a,c} (V V^\top)_{a,a}^2 (\beta^{-2})_{c,c}^2.
\end{aligned}
\end{equation}
For the off-diagonal terms ($a \neq b$),
\begin{equation}
\begin{aligned} \label{need03}
    \Var\bigg(\sum_{a \neq b, c} (V V^\top)_{a,b} (\beta^{-2})_{c,c} X_{a,c} X_{b,c} \bigg) &= \sum_{a \neq b, c} (V V^\top)_{a,b}^2 (\beta^{-2})_{c,c}^2 \Var(X_{a,c} X_{b,c})\\
    & \asymp \rho_n \sum_{a,c} (V V^\top)_{a,a}^2 (\beta^{-2})_{c,c}^2 + \rho_n^2 \|V V^\top\|_F^2 \|\beta^{-2}\|_F^2.
\end{aligned}
\end{equation}
We compare the above to the variance of the unweighted trace in the denominator. By a perfectly analogous variance expansion for $\tr(X \beta^{-2} X) = \sum_{a,c} (\beta^{-2})_{c,c} X_{a,c}^2$, one obtains
\begin{equation}
\begin{aligned} \label{need04}
    \Var\big(\tr(X \beta^{-2} X^\top)\big) = \sum_{a,c} (\beta^{-2})_{c,c}^2 \Var(X_{a,c}^2) \asymp \rho_n\, n\, \|\beta^{-2}\|_F^2.
\end{aligned}
\end{equation} 
Combining the results from \cref{need02,need03,need04} yields
\begin{equation}
\begin{aligned}
    \frac{\Var\big(\tr(V V^\top X \beta^{-2} X^\top)\big)}
     {\Var\big(\tr(X \beta^{-2} X^\top)\big)}
&\lesssim
\frac{\rho_n \sum_{a} (V V^\top)_{a,a}^2 \|\beta^{-2}\|_F^2 + \rho_n^2 \|V V^\top\|_F^2 \|\beta^{-2}\|_F^2}{\rho_n\, n\, \|\beta^{-2}\|_F^2}\\
&= \frac{\sum_{a} (V V^\top)_{a,a}^2 + \rho_n \|V V^\top\|_F^2}{n}\\
&\lesssim \frac{k(1+\rho_n)}{n} \to 0,
\end{aligned}
\end{equation}
where the last line follows from the fact that \(VV\t\) is a projection matrix and \cref{ass:asymptotic1}. This completes the proof of the lemma.
\end{proof}

\subsection{Proof of \cref{lem:second-order1}}
\label{sec:second-order1proof}

\begin{proof} 
\noindent \textbf{Computing the mean}: The expectation of the second–order term can be written in a compact operator form.  Using linearity of trace and expectation, we have
\begin{equation} \label{eq:mean1}
\begin{aligned}
\mathbb{E}\bigl[\,2\|\beta^{\perp} X \beta^{-1}\|_F^2\bigr]
&= 2\,\mathbb{E}\bigl[\tr(\beta^{-2} X \beta^{\perp} X)\bigr]
= 2\,\tr\bigl[\beta^{-2}\,\mathbb{E}(X \beta^{\perp} X)\bigr].
\end{aligned}
\end{equation}
To evaluate \(\mathbb{E}(X \beta^{\perp} X)\) we use the independence and centering of the entries of \(X\).  Write \(\Sigma:=P\circ(J_n-P)\) for the entrywise variances of the Bernoulli observations.  A straightforward bookkeeping of the contributing index configurations yields the decomposition of \(\mathbb{E}(X \beta^{\perp} X)\) into an off–diagonal Hadamard term and a diagonal correction:
\(
\mathbb{E}(X \beta^{\perp} X) = \Sigma\circ \beta^{\perp} + \mathrm{Diag}(\Sigma \cdot d - \diag(\Sigma) \circ d),\)
where \(d:=\diag(\beta^{\perp}).
\)
Substituting this identity into \cref{eq:mean1} produces the final mean expression
\begin{align*}
\mathbb{E}\bigl[\,2\|\beta^{\perp} X \beta^{-1}\|_F^2\bigr]
&= 2\,\tr\Bigl[\beta^{-2}\bigl(\Sigma\circ \beta^{\perp} + \mathrm{Diag}(\Sigma \cdot d - \diag(\Sigma) \circ d))\bigr)\Bigr] \\
&=2\,\tr\Bigl[\beta^{-2}\bigl(\Sigma\circ\beta^\perp + \mathrm{Diag}(\Sigma\cdot d - \diag(\Sigma) \circ d))\bigr)\Bigr].
\end{align*}
This identity is the deterministic mean contribution that appears in \cref{param:11}.\\ \ \\
\noindent \textbf{Computing the variance:} We now compute the variance of the leading second--order contribution.  
Observe that since $\beta^{\perp} = I - VV\t$, by the cyclic property of trace we have that
\begin{align*}
\|\beta^{\perp}X\beta^{-1}\|_F^2
= \tr(\beta^{-1}X\beta^{\perp}X\beta^{-1}) 
= \tr(X\beta^{-2}X) - \tr(VV\t X\beta^{-2}X).
\end{align*}
From \cref{lem:beta_perp1one}, it suffices to analyze the variance of
\(
2\|X \beta^{-1}\|_F^2 = 2\,\tr(X \beta^{-2} X\t)
= \sum_{i=1}^n 2\,X_{i \cdot}\t \beta^{-2} X_{i \cdot},
\)
For a fixed \(i\), define the quadratic form
\begin{align*}
Q_i := 2\,X_{i \cdot}\t \beta^{-2} X_{i \cdot}
= 4\sum_{s < t} (\beta^{-2})_{st}\,{X_{is}}\,{X_{it}}
+ 2\sum_{s=1}^n (\beta^{-2})_{ss}\,{X_{is}}^2,
\end{align*}
so that \(2\|X \beta^{-1}\|_F^2 = \sum_{i=1}^n Q_i\). The variance then decomposes as
\begin{align}\label{eq:varq1}
\Var\left(2\|X \beta^{-1}\|_F^2\right)
= \sum_{i=1}^n \Var(Q_i) + 2\sum_{i<k} \Cov(Q_i,Q_k).
\end{align}
For the diagonal terms, using independence and the standard fourth moment identity for centered Bernoulli variables, we obtain
\begin{align*}
\Var(Q_i)
= 16 \sum_{s < t} (\beta^{-2})_{st}^2\,\Sigma_{is}\,\Sigma_{it}
+ 4 \sum_{s=1}^n (\beta^{-2})_{ss}^2\,K_{is},
\end{align*} where
\(
K_{is} := \mathbb{E}({X_{is}}^4)-\bigl(\mathbb{E}({X_{is}}^2)\bigr)^2
= 2\,\Sigma_{is}(1-2P_{is})
\)
is the fourth--moment correction term.  

For the off--diagonal contributions, note that \(Q_i\) and \(Q_k\) share only the symmetric variable \(x_{ik}=x_{ki}\). A direct calculation shows
\(
\Cov(Q_i,Q_k) = 4\,K_{ik}\,(\beta^{-2})_{ii}\,(\beta^{-2})_{kk}.
\)
Combining the two components and rewriting the double sums from \cref{eq:varq1} yields 
\begin{align*}
\Var\left(2\|X \beta^{-1}\|_F^2\right) &=\sum_{i=1}^{n} \left( 16 \sum_{s < t} (\beta^{-2})_{st}^2\,\Sigma_{is}\,\Sigma_{it}
+ 4 \sum_{s=1}^n (\beta^{-2})_{ss}^2\,K_{is} \right) +  8 \sum_{i<k} \,K_{ik}\,(\beta^{-2})_{ii}\,(\beta^{-2})_{kk}\\
&=8\sum_{i=1}^{n} \left( \sum_{s \ne t} (\beta^{-2})_{st}^2\,\Sigma_{is}\,\Sigma_{it} \right)
+ \left( 4 \sum_{s=1}^n (\beta^{-2})_{ss}^2\,K_{is} +  4 \sum_{i \ne k} \,K_{ik}\,(\beta^{-2})_{ii}\,(\beta^{-2})_{kk} \right)\\
&= 8\big\langle (J_n-I_n)\circ(\beta^{-2}\circ\beta^{-2}),\,\Sigma^2 \big\rangle
+ 4\,\alpha\t K \alpha.
\end{align*}
This expression coincides with the variance term in \cref{param:12}.

Finally, we compute the asymptotic order of the variance obtained above.  
From \cref{ass:asymptotic1,ass:eigen-scaling,ass:incoherence}, we have 
\begin{equation}
\begin{aligned}\label{newres001}
    \sum_{i \neq j} (\beta^{-2})_{ij}^2 &= \|\beta^{-2}\|_F^2 - \sum_{i=1}^{n}(\beta^{-2})_{ii}^2\\
    &\gtrsim \tr(\beta^{-4}) - \max_{i,j} (\beta^{-2})_{ij} \tr(\beta^{-2})\\
    &\gtrsim \frac{k}{n^4 \rho_n^4}.
\end{aligned}
\end{equation}
Similarly it can be shown that 
\begin{equation}
\begin{aligned}\label{newres002}
    \sum_{i \neq j} (\beta^{-2})_{ij}^2 \lesssim \frac{k}{n^4 \rho_n^4}
\end{aligned}
\end{equation}
which combined with  \cref{newres001} yields
\begin{equation}\label{newres01}
     \sum_{i\neq j} (\beta^{-2})_{ij}^2 \asymp \frac{k}{n^4 \rho_n^4}.
\end{equation}
Therefore it follows that
\begin{align*}
    8\big\langle (J_n-I_n)\circ(\beta^{-2}\circ\beta^{-2}),\,\Sigma^2 \big\rangle &= 8 \sum_{i\neq j} (\beta^{-2})_{ij}^2 (\Sigma^2)_{ij}\asymp n \rho_n^2 \cdot \frac{k}{n^4 \rho_n^4}\asymp \frac{k}{n^3 \rho_n^2},
\end{align*}
and
\begin{align*}
    \alpha\t K \alpha &= \rho_n (\sum_i \alpha_i)^2 =\rho_n \tr(\beta^{-2})^2= \frac{k^2}{n^4 \rho_n^3}.
\end{align*}
Thus, it follows that
\(
\Var\left(2\|X \beta^{-1}\|_F^{2}\right)
\asymp
\frac{k^{2}}{n^{3}\rho_n^{2}}.
\)
Hence, by \cref{lem:beta_perp1one}, we obtain
\[
\Var\left(2\|\beta^{\perp} X \beta^{-1}\|_F^{2}\right)
=
8\big\langle (J_n - I_n)\circ(\beta^{-2}\circ \beta^{-2}),\,\Sigma^{2}\big\rangle
+
4\,\alpha\t K \alpha
+
o\left(\frac{k^2}{n^{3}\rho_n^{2}}\right).
\]
\end{proof}

\subsection{Proof of \cref{lem:second-order-clt1}} \label{sec:second-orderclt1proof}
\begin{proof} 
Suppose that
\(
\|\beta^{\perp} X \beta^{-1}\|_{F}^{2}
=
\|X\beta^{-1}\|_{F}^{2} + R.
\)
By \cref{lem:beta_perp1one}, we have
\(
\Var(R) / \Var(T_{1}^{(S)}) \to 0,
\)
and consequently
\(
(R - \E R)/\sqrt{\Var(T_{1}^{(S)})} \to 0
\)
in probability. Hence the contribution of \(R\) is asymptotically negligible.  
Therefore, in establishing the limiting distribution of
\(
\|\beta^{\perp} X \beta^{-1}\|_{F}^{2},
\)
it suffices to derive the asymptotic distribution of
\(
\|X\beta^{-1}\|_{F}^{2}
\)
after appropriate centering and scaling.  
The desired result then follows by a direct application of Slutsky’s theorem.
For notational convenience set \(U=\beta^{-1}\).  Then
\(
\|X\beta^{-1}\|_F^2=\tr(\beta^{-1}X^2\beta^{-1})=\tr(UX^2U),
\)
and hence
\(
n\tr(UX^2U)
= n \sum_{j=1}^{n} \sum_{1 \leq i,r,q \leq n} {X_{ri}} {X_{qi}} U_{jr} U_{jq}.
\)
Following the decomposition strategy employed in the proof of Lemma 2 in \citet{fan_asymptotic_2022}, the right–hand side can be reorganized into the form used for a martingale central limit theorem (CLT).  Concretely, 
\begin{equation} \label{eq:martingale-decomp}
\begin{aligned}
\sum_{j=1}^{n} \sum_{1 \leq i,r,q \leq n} n\,{X_{ri}} {X_{qi}} U_{jr} U_{jq}
&= \sum_{j=1}^{n} \sum_{1 \leq r < i \leq n} 2 {X_{ri}} \Bigl( n\,U_{jr} \sum_{1 \leq q < r} {X_{iq}} U_{jq} + n\,U_{ji} \sum_{1 \leq q < i} {X_{rq}} U_{jq} \Bigr) \\
&\quad + \sum_{j=1}^{n} \sum_{1 \leq q < r \leq n} 2n \,{X_{rr}} {X_{rq}} \,U_{jq} U_{jr} \\
&\quad + \sum_{j=1}^{n} \sum_{1 \leq r < i \leq n} {X_{ri}}^2 \bigl(n\,U_{jr}^2 + n\,U_{ji}^2\bigr) + \sum_{j=1}^{n} \sum_{r = 1}^{n} {X_{rr}}^2 \,n\,U_{jr}^2 \\
&=  \sum_{1 \leq r < i \leq n} 2 {X_{ri}} \Bigl( \sum_{1 \leq q < r} {X_{iq}} \gamma_{rq} + \sum_{1 \leq q < i} {X_{rq}} \gamma_{iq} \Bigr) \\
&\quad + \sum_{1 \leq q < r \leq n} 2 \gamma_{rq} {X_{rr}} {X_{rq}} + \sum_{1 \leq r < i \leq n} {X_{ri}}^2 (\gamma_{rr} + \gamma_{ii}) + \sum_{r = 1}^{n} {X_{rr}}^2 \gamma_{rr},
\end{aligned}
\end{equation}
where we have set
\(
\gamma_{r,q} = \sum_{j=1}^{n} n\, U_{jr} U_{jq}.
\) After centering we obtain
\begin{equation} \label{eq:centered}
\begin{aligned}
\sum_{j=1}^{k} U_{j \cdot}\t (X^2 - \mathbb{E}X^2) U_{j \cdot}
&= \sum_{1 \leq r < i \leq n} 2 {X_{ri}} \Bigl( \sum_{1 \leq q < r} {X_{iq}} \gamma_{rq} + \sum_{1 \leq q < i} {X_{rq}} \gamma_{iq} \Bigr)  + \sum_{1 \leq q < r \leq n} 2 \gamma_{rq} {X_{rr}} {X_{rq}}  \\
&+ \sum_{1 \leq r < i \leq n} ({X_{ri}}^2-\nu_{ri}^2) (\gamma_{rr} + \gamma_{ii})  + \sum_{r = 1}^{n} ({X_{rr}}^2 - \sigma_{rr}^2)\,\gamma_{rr}.
\end{aligned}
\end{equation}
To set up a martingale difference array, we introduce the triangular filtration
\[
\mathscr{F}_t := \sigma(w_1, \ldots, w_t), 
\qquad t = k + \tfrac{l(l-1)}{2},\quad 1\le k\le l\le n,
\]
so that \(\mathscr{F}_t\) records the entries \({X_{ij}}\) revealed up to position \(t\) in the standard lexicographic ordering of the upper–triangular matrix.  Under this filtration \cref{eq:centered} can be written as a sum of martingale differences.  Each summand (the increment associated with the revealing of the entry \({X_{ri}}\)) admits the decomposition
\[
{X_{ri}} b_{ri} + ({X_{ri}}^2 - \nu_{ri}^2)\, c_{ri},
\]
with coefficients given by
\begin{equation} \label{eq:coefficients}
\begin{aligned}
b_{ri} &= 
\begin{cases}
2 \displaystyle\Bigl( \sum_{1 \leq q < r} {X_{iq}} \gamma_{rq} 
      + \sum_{1 \leq q < i} {X_{rq}} \gamma_{iq} \Bigr), & r<i, \\[8pt]
\displaystyle 2\sum_{1 \leq q < r}  {X_{rq}} \gamma_{rq} , & i=r ,
\end{cases} \\[10pt]
c_{ri} &= 
\begin{cases}
\gamma_{rr} + \gamma_{ii}, & r<i, \\[6pt]
\gamma_{rr}, & i=r .
\end{cases}
\end{aligned}
\end{equation}
The conditional variance (given the $\sigma$–field just prior to revealing \({X_{ri}}\)) takes the explicit form
\begin{equation} \label{eq:cond-var}
\sum_{1 \leq r \leq i \leq n} \nu_{ri}^2 b_{ri}^2 
+ 2 \sum_{1 \leq r \leq i \leq n} \theta_{ri} b_{ri} c_{ri} 
+ \sum_{1 \leq r \leq i \leq n} \kappa_{ri} c_{ri}^2,
\end{equation}
where we adopt the notation
\(
\nu_{ri}^2 = \E[{X_{ri}}^2],  
\theta_{ri} = \E[{X_{ri}}^3], 
\kappa_{ri} = \E\bigl[({X_{ri}}^2-\nu_{ri}^2)^2\bigr].
\)

\medskip

The verification of the martingale CLT requires control of the first four moments of the coefficients \(b_{ri}\) and the second and fourth moments of the conditional variance of \cref{eq:cond-var}.  We proceed to record the necessary moment estimates.  Using the definitions above, one checks that the leading order behavior of the coefficients is captured by the following asymptotic relations (the algebraic manipulations are elementary and follow from counting the contributing index configurations).

We calculate the necessary moment estimates. First, note that \(\E[b_{ri}]=0\) for all \(r,i\). Also, under \cref{ass:asymptotic1,ass:eigen-scaling,ass:incoherence}, just like \cref{newres01} we can show that 
\begin{align}\label{newres02}
    \sum_{i \neq j}\gamma_{ij}^2 \asymp \frac{k}{n^2 \rho_n^4}, \quad \gamma_{ij}^2 \lesssim \frac{k}{n^4 \rho_n^4}
\end{align} 
Thus, 
\[
\sum_{r\le i}\E[b_{ri}^2]
\asymp
\sum_{r \le i}\sum_{1\le q<r}\gamma_{rq}^2\,\nu_{iq}^2
+ \sum_{r \le i}
\sum_{1\le q<i}\gamma_{iq}^2\,\sigma_{rq}^2
\asymp \frac{k}{n \rho_n^2}
\]
where \(\gamma_{rq}^2=(\sum_{r=1}^{n} n U_{jr} U_{jq})^2 \). 
Apart from that, we also have \(\nu_{ri}^2, \, \theta_{ri}, \kappa_{ri} \asymp \rho_n\).

The fourth moment of \(b_{ri}\) admits the expansion
\[
\begin{aligned}
\E[b_{ri}^4]
&\lesssim
\sum_{1\le q<r}\gamma_{rq}^4\,\E[{X_{iq}}^4]
+
\sum_{1\le q<i}\gamma_{iq}^4\,\E[{X_{rq}}^4] \\
&\quad
{}+
\sum_{\substack{q_1\neq q_2\\1\le q_1,q_2<r}}
\gamma_{rq_1}^2\gamma_{rq_2}^2\,
\E[{X_{iq_1}}^2]\,\E[{X_{iq_2}}^2]
+
\sum_{\substack{q_1\neq q_2\\1\le q_1,q_2<i}}
\gamma_{iq_1}^2\gamma_{iq_2}^2\,
\E[{X_{rq_1}}^2]\,\E[{X_{rq_2}}^2] \\
&\lesssim \frac{k^2}{n^6 \rho_n^6} 
\end{aligned}
\]
Finally, for the \(c_{ri}\) terms we have the simpler bounds:
\begin{align}\label{newc}
\sum_{1 \leq r \leq i \leq n} \kappa_{ri} c_{ri}^2 &\lesssim \rho_n \left[n \sum_{i} \gamma_{ii}^2 + \left(\sum_{i=1}^{n}\gamma_{ii}\right)^2\right]\\ \notag
&\lesssim \frac{k^2}{n^2 \rho_n^3}\\
&=o \left(\frac{k^2}{n \rho_n^2} \right) \notag
\end{align}
With these moment estimates in hand we evaluate the first two moments of the conditional variance of \cref{eq:cond-var}.  Denote by \(s_V^2\) the expectation of the conditional variance and by \(\kappa_V^2\) its variance.  A straightforward aggregation over indices yields the asymptotic relations
\begin{align*}
s_V^2 
&= \E\Big(\sum_{1 \leq r \leq i \leq n}\nu_{ri}^2 b_{ri}^2 + 2 \sum_{1 \leq r \leq i \leq n} \theta_{ri} b_{ri} c_{ri} + \sum_{1 \leq r \leq i \leq n} \kappa_{ri} c_{ri}^2\Big)\\ 
&\asymp \frac{k}{n \rho_n^2},
\end{align*}
and
\begin{align*}
\kappa_V
&= \Var\Big(\sum_{1 \leq r \leq i \leq n}\nu_{ri}^2 b_{ri}^2 + 2 \sum_{1 \leq r \leq i \leq n} \theta_{ri} b_{ri} c_{ri} + \sum_{1 \leq r \leq i \leq n} \kappa_{ri} c_{ri}^2\Big) \\[4pt]
&\lesssim \sum_{1 \leq r \leq i \leq n}\Var(\nu_{ri}^2 b_{ri}^2)
      + \sum_{1 \leq r \leq i \leq n}\Var(2\theta_{ri}b_{ri}c_{ri})
      \lesssim \frac{k^2}{n^4 \rho_n^4}.
\end{align*}
Combining the above estimates we obtain the Lyapunov–type ratio used to verify the martingale CLT:
\[
\frac{s_V}{\kappa_V^{1/4}} 
= \frac{n \rho_n}{\sqrt n  \rho_n} \gg 1.
\]
Thus the Lyapunov type condition as used in Lemma 9.12 of \citet{bai2010spectral} is satisfied.

Therefore the martingale central limit theorem applies and we conclude that the appropriately centered and scaled second–order term converges in distribution to a standard normal:
\begin{equation}
\frac{
  2\bigl\|\beta^{\perp} X \beta^{-1}\bigr\|_F^2
  - \,\mathbb{E}\left[\,2\bigl\|\beta^{\perp} X \beta^{-1}\bigr\|_F^2\right]
}{
  \sqrt{\Var\left( 
    2\bigl\|\beta^{\perp} X \beta^{-1}\bigr\|_F^2
  \right)}
}
\xrightarrow{D} \mathcal{N}(0,1).
\end{equation}
\end{proof}

\subsection{Proof of \cref{lem:3rdmean1}}\label{sec:3rd}
\begin{proof}
We first decompose \(\langle VV\t,\,S_{3}(X)\rangle\):
\begin{align*}
\E\left(\langle VV\t,S_3(X)\rangle\right)
&= \E\left(\tr\big(\beta^{-1} X \beta^{\perp} X \beta^{-1} X \beta^{-1}\big)\right)
+\E\left(\tr\big(\beta^{-1} X \beta^{-1} X \beta^{\perp} X \beta^{-1}\big)\right) \notag\\
&\quad
-\E\left(\tr\big(\beta^{-2} X \beta^{\perp} X \beta^{\perp} X \beta^{-1}\big)\right)
-\E\left(\tr\big(\beta^{-1} X \beta^{\perp} X \beta^{\perp} X \beta^{-2}\big)\right).
\end{align*}
We compute the orders of 
\(\E\left(\tr(\beta^{-1} X \beta^{\perp} X \beta^{-1} X \beta^{-1})\right)\)  
and  
\(\E\left(\tr(\beta^{-2} X \beta^{\perp} X \beta^{\perp} X \beta^{-1})\right)\);  
the remaining two terms follow identically. Note that

\begin{equation}
\begin{aligned}
\label{eq:mean31}
\bigg|\E\left(\tr\big(\beta^{-1} X \beta^{\perp} X \beta^{-1} X \beta^{-1}\big)\right)\bigg|
&= \E \bigg|\sum_{i,j}
(\beta^{-2})_{ji}\,(\beta^{\perp})_{ji}\,(\beta^{-1})_{ji}\,(X_{ij})^{3}\bigg| \\[2pt]
&=O_p \left( \frac{k^{2}}{n^{4}\rho_{n}^{2}}\right),
\end{aligned}
\end{equation}
where the final line follows from
\cref{ass:asymptotic1,ass:eigen-scaling,ass:incoherence}. Also, we have

\begin{equation}
\begin{aligned}
\label{eq:mean32}
\bigg|\E\left(\tr\big(\beta^{-2} X \beta^{\perp} X \beta^{\perp} X \beta^{-1}\big)\right)\bigg|
&= \E \bigg|\sum_{i,j}
(\beta^{-3})_{ji}\,(\beta^{\perp})_{ji}^{2}\,(X_{ij})^{3}\bigg| \\[2pt]
&= O_p \left( \frac{k}{n^{3}\rho_{n}^{2}}\right).
\end{aligned}
\end{equation}

Combining \cref{eq:mean31,eq:mean32} establishes the claim and completes the proof of \cref{lem:3rdmean1}.
\end{proof}

\subsection{Proof of \cref{lem:3rdvar1}}\label{sec:3rdproof}

\begin{proof}
To evaluate the order of \(\Var\!\big(\langle VV^{\top}, S_3(X)\rangle\big)\), it suffices to show that
\[
\frac{\Var\!\big(\langle VV^{\top}, S_3(X)\rangle\big)}
{\Var\!\left(\tr\!\big(\beta^{-3} X^{3}\big)\right)}
\asymp 1 
\] and calculate the order of \(\Var\!\left(\tr\!\big(\beta^{-3} X^{3}\big)\right).\)
We proceed by comparing cyclically reduced representatives. Define
\[
R_1 := \tr\!\big(\beta^{-3} X^{3}\big),
\qquad
R_2 := \tr\!\big(\beta^{-3} X(\beta^{\perp}X)^{2}\big).
\]
We first establish that
\(
\frac{\Var(R_1)}{\Var(R_2)} \to  1 .
\)
To this end, insert the decomposition \(I = VV\t + \beta^{\perp}\) in each interior gap of \(R_1\) to obtain 
\begin{align} \label{R1}
    R_1
= \sum_{\sigma \in \{VV\t,\,\beta^{\perp}\}^{\,2}}
\tr\big(\beta^{-3} X M_{\sigma,1} X M_{\sigma,2} X\big),
\end{align}
where each \(M_{\sigma,r}\in\{VV\t,\beta^{\perp}\}\).  
The unique summand for which every \(M_{\sigma,r}=\beta^{\perp}\) is exactly \(R_2\); every other summand (henceforth a \emph{remainder}) contains at least one factor \(VV\t\).  
Since the total number of summands is finite, it suffices to prove that each remainder has variance \(o(\Var(R_2))\); the claim then follows by finite summation.

Fix a remainder monomial
\(
\mathcal{M} := \tr\big(\beta^{-3} X M_{1} X M_{2} X\big).
\)
Expanding the trace yields
\begin{equation}
\begin{aligned}
\label{eq:monomial1}
\mathcal{M}
&=
\sum_{i_{1},i_{2},i_{3} \atop j_{1},j_{2},j_{3}}
(\beta^{-3})_{j_{1},i_{1}}\,
X_{i_{1},j_{2}}\,
(M_{1})_{j_{2},i_{2}}\,
X_{i_{2},j_{3}}\,
(M_{2})_{j_{3},i_{3}}\,
X_{i_{3},j_{1}} \\
&=
\sum_{i_{1},i_{2},i_{3} \atop j_{1},j_{2},j_{3}}
w(\mathbf{i},\mathbf{j})\,
X_{i_{1},j_{2}}\, X_{i_{2},j_{3}}\, X_{i_{3},j_{1}},
\end{aligned}
\end{equation}
where
\(
w(\mathbf{i},\mathbf{j})
:= (\beta^{-3})_{j_{1},i_{1}} (M_{1})_{j_{2},i_{2}} (M_{2})_{j_{3},i_{3}}.
\)
In order for
\[
\Cov\big( w(\mathbf{i},\mathbf{j})\,
X_{i_{1},j_{2}}\, X_{i_{2},j_{3}}\, X_{i_{3},j_{1}}, 
w(\mathbf{k},\mathbf{l})\,
X_{k_{1},l_{2}}\, X_{k_{2},l_{3}}\, X_{k_{3},l_{1}} \big) \neq 0 \]
to hold, the underlying edge configurations must coincide. Without loss of generality, the edge \((k_{2},l_{3})\) must coincide with \((i_{2},j_{3})\), and the subgraph \(\{(k_{1},l_{2}), (k_{3},l_{1})\}\) must coincide with \(\{(i_{1},j_{2}), (i_{3},j_{1})\}\). Alternatively, the edge \((k_{1},l_{2})\) is identified with \((i_{1},j_{2})\), or \((k_{3},l_{1})\) with \((i_{3},j_{1})\), with the corresponding subgraph matching accordingly. Consequently, these constraints imply that the index tuples \((\mathbf{k},\mathbf{l})\) agree with \((\mathbf{i},\mathbf{j})\) up to the natural symmetries of the configuration, a relationship we denote by \((\mathbf{k},\mathbf{l}) \sim (\mathbf{i},\mathbf{j})\). Thus, we obtain
\[
\Var(\mathcal{M})
\asymp
\sum_{(\mathbf{k},\mathbf{l}) \sim (\mathbf{i},\mathbf{j})}
w(\mathbf{i},\mathbf{j}) w(\mathbf{k},\mathbf{l})\, \rho_{n}^{3}.
\]
Define the quantity
\[
\Sigma(\mathcal{M})
:= \sum_{(\mathbf{k},\mathbf{l}) \sim (\mathbf{i},\mathbf{j})}
w(\mathbf{i},\mathbf{j}) w(\mathbf{k},\mathbf{l})
= \sum_{(\mathbf{k},\mathbf{l}) \sim (\mathbf{i},\mathbf{j})}
\Big(
(\beta^{-3})_{j_{1},i_{1}}\,
(M_{1})_{j_{1},i_{2}}\,
(M_{2})_{j_{2},i_{3}}
\Big) \Big(
(\beta^{-3})_{l_{1},k_{1}}\,
(M_{1})_{l_{1},k_{2}}\,
(M_{2})_{l_{2},k_{3}}
\Big).
\]
Invoking \cref{ass:eigen-scaling,ass:asymptotic1}, we find that
\[
\Sigma(\mathcal{M})
\lesssim \frac{k^2}{n^8 \rho_n^6}\sum_{(\mathbf{k},\mathbf{l}) \sim (\mathbf{i},\mathbf{j})}
\Big(
(M_{1})_{j_{1},i_{2}}\,
(M_{2})_{j_{2},i_{3}}
\Big) \Big(
(M_{1})_{l_{1},k_{2}}\,
(M_{2})_{l_{2},k_{3}}
\Big).
\]
Consider the configuration where $M_1 = V V^\top$ and $M_2 = \beta^\perp$. Evaluating $\Sigma(\mathcal{M})$ requires summing over six structural indices: $i_1, j_1, i_2, j_2, i_3, \text{ and } j_3$. However, because $M_2 = \beta^\perp$, it enforces an index collapse (specifically, $j_2 = i_3$). This structural constraint reduces the effective degrees of freedom in the summation from six to five, yielding a factor of $O(n^5)$. Coupling this with the entrywise bound of $(k/n)^2$ for $M_1$, guaranteed by the incoherence condition in \cref{ass:incoherence}, we directly obtain the upper bound:
$$
\Sigma(\mathcal{M}) \lesssim \frac{k^2}{n^8 \rho_n^6}\cdot n^5 \cdot \left( \frac{k}{n} \right)^2 \lesssim \frac{k^4}{n^5 \rho_n^6}.
$$
Similarly, if \(M_1=M_2=V V\t\), we have $$
\Sigma(\mathcal{M}) \lesssim \frac{k^2}{n^8 \rho_n^6}\cdot n^6 \cdot \left( \frac{k}{n} \right)^4 \lesssim \frac{k^6}{n^6 \rho_n^6}.
$$ 
Thus, it follows that whenever at least one of \(M_{1}\) or \(M_{2}\) equals \(VV\t\), the variance of the monomial satisfies 
\begin{align} \label{mono1}
    \Var(\mathcal{M}) = o(\Var(R_2)).
\end{align}
By the same covariance calculation over matching edge configurations, under \cref{ass:eigen-scaling,ass:incoherence,ass:asymptotic1} we have
\begin{equation}
\begin{aligned}\label{ord1}
    \Var(R_1) &= \Var( \tr(\beta^{-3}X^3))\\
    &\asymp n^2 \rho_n^3 \|\beta^{-3}\|_F^2 \\
    &\asymp\frac{k^2}{n^4 \rho_n^3}
\end{aligned}
\end{equation}
Combining \cref{R1,ord1} with 
\cref{mono1} establishes the claim that
\begin{align} \label{later01}
    \frac{\Var(R_1)}{\Var(R_2)} = \frac{\Var (\sum_{\sigma \in \{VV\t,\,\beta^{\perp}\}^{\,2}}
\tr\big(\beta^{-3} X M_{\sigma,1} X M_{\sigma,2} X\big))}{\Var(R_2)} = \frac{\Var(R_2) + o(\Var(R_2))}{\Var(R_2)}\to 1.
\end{align}
In exactly the same manner, we can show that
\[
\frac{\Var\left(\tr \left(\beta^{-t} X^t \beta^{-(3-t)} X^{3-t}\right) \right)}
     {\Var\left(\tr \left( \beta^{-t} X (\beta^{\perp}X)^{t-1}\beta^{-(3-t)}X(\beta^{\perp}X)^{3-t-1}\right) \right)}
\to 1
\] for \(t=1,2\) and \[
\frac{\Var\left(\tr \left(\beta^{-1} X \beta^{-1} X \beta^{-1} X\right) \right)}
     {\Var\left(\tr \left(\beta^{-1} X \beta^{-1} (\beta^\perp X) \beta^{-1} (\beta^\perp X)\right) \right)}
\to 1.
\]
To establish that the dominant variance contributions in  
\(\langle VV\t,\, S_{3}(X)\rangle\) arise from terms of the form  
\(
\tr\big(\beta^{-s_{1}} X\,\beta^{\perp} X \cdots \beta^{\perp}X\,\beta^{-s_{4}}\big),
\)
it therefore suffices to show that for \(t =1,2\),
\[
\frac{\Var\left(\tr\left(\beta^{-t} X^{t}\, \beta^{-(3-t)} X^{3-t}\right)\right)}
     {\Var\left(\tr\left(\beta^{-3} X^{3}\right)\right)}
\to 0,
\qquad
\frac{\Var\left(\tr\left(\beta^{-1} X \beta^{-1} X \beta^{-1} X\right)\right)}
     {\Var\left(\tr\left(\beta^{-3} X^{3}\right)\right)}
\to 0.
\]
The proof technique for both ratios is identical; hence we concentrate on the first.

Write the numerator trace as  
\(
R_{1,\mathrm{num}}
:= \tr\big(\beta^{-t} X^{t}\, \beta^{-(3-t)} X^{3-t}\big).
\)
Proceeding as in the preceding argument, we obtain
\begin{align*}
\Var(R_{1,\mathrm{num}})
&\lesssim
\rho_{n}^{3}
\sum_{i_{1}, i_{2} \atop j_{1}, j_{2}}
\left[
(\beta^{-t})_{i_{1},j_{1}}\,
(\beta^{-(3-t)})_{i_{2},j_{2}}
\right]^{2} \\
&\lesssim
\rho_{n}^{3}
\sum_{i_{1}, i_{2} \atop j_{1}, j_{2}}
\frac{k^{4}}{n^{10}\rho_{n}^{6}}
\lesssim
\frac{k^{4}}{n^{6}\rho_{n}^{3}}.
\end{align*}
Using \cref{mono1,ord1} we will therefore have
\[
\frac{
\Var\left(\tr\left(\beta^{-t} X^{t}\, \beta^{-(3-t)} X^{3-t}\right)\right)
}{
\Var\left(\tr\left(\beta^{-3} X^{3}\right)\right)
} = \frac{\frac{k^{4}}{n^{6}\rho_{n}^{3}}}{\frac{k^{2}}{n^{4}\rho_{n}^{3}}}
\to 0,
\]
as required. This establishes that the dominant contributions to
\(
\Var\!\big(\langle VV^{\top}, S_3(X)\rangle\big)
\)
arise precisely from summands of the form
\(
\tr\big(\beta^{-s_{1}} X\,\beta^{\perp} X \cdots \beta^{\perp}X\,\beta^{-s_{4}}\big)\).
In particular, we have shown that
\[
\frac{\Var\!\big(\langle VV^{\top}, S_3(X)\rangle\big)}
{\Var\!\left(\tr\!\big(\beta^{-3} X^{3}\big)\right)}
\asymp 1 .
\]
Combining this with the variance order in \cref{ord1} completes the proof of the lemma.
\end{proof}

\section{Proof of Lemmas from \cref{sec:mainthm2}}
\label{sec:mainthm2lemmasproof}
In this section we prove the auxiliary lemmas required for the proof of \cref{thm:twosample} in \cref{sec:mainthm2}.

\subsection{Proof of \cref{lem:beta_perp1two}} \label{sec:betaperp1proof2}
\begin{proof}
Define
\[
M(X) = 
\begin{pmatrix}
{X^{(1)}}\t {X^{(1)}} & -{X^{(1)}}\t {X^{(2)}} \\
- {X^{(2)}}\t {X^{(1)}} & {X^{(2)}}\t {X^{(2)}}
\end{pmatrix},
\qquad
M(Y) =
\begin{pmatrix}
(Y^{(1)})\t (Y^{(1)}) & -(Y^{(1)})\t (Y^{(2)}) \\
- (Y^{(2)})\t (Y^{(1)}) & (Y^{(2)})\t (Y^{(2)})
\end{pmatrix}.
\]
It therefore suffices to analyze
\begin{align*}
\frac{\Var \left( 2\, \tr \big( U (M(Y) - M(X)) U\t \big) \right)}
     {\Var \left( 2\, \tr \big( U M(X) U\t \big) \right)}.
\end{align*}
Observe that
\begin{align} \label{Wijdef}
    M(X) - M(Y) = 
    \begin{pmatrix}
        {X^{(1)}}\t W_{11} {X^{(1)}} & {X^{(1)}}\t W_{12} {X^{(2)}} \\
        {X^{(2)}}\t W_{21} {X^{(1)}} & {X^{(2)}}\t W_{22} {X^{(2)}}
    \end{pmatrix}.
\end{align}
Denote the $(i,j)$-th block of $M(X) - M(Y)$ by ${X^{(i)}}\t W_{ij} {X^{(j)}}$, where
\begin{equation}
\begin{aligned}\label{note04}
    W_{11} &= {V^{(1)}} {V^{(1)}}\t, \\
    W_{12} = W_{21} &= -({V^{(1)}} {V^{(1)}}\t + {V^{(2)}} {V^{(2)}}\t - {V^{(1)}} {V^{(1)}}\t {V^{(2)}} {V^{(2)}}\t), \\
    W_{22} &= {V^{(2)}} {V^{(2)}}\t.
\end{aligned}
\end{equation}
Then,
\begin{align*}
\Var \left( \tr \big( U (M(Y) - M(X)) U\t \big) \right)
= 
\Var \left( 
\sum_{i,j=1}^{2}
\tr\left( W_{ij} {X^{(j)}} \beta_j^{-1} \beta_i^{-1} {X^{(i)}}\t \right)
\right).
\end{align*}
Hence, to demonstrate that 
\(\Var\big(\tr(U (M(Y) - M(X)) U\t)\big)\) 
is asymptotically negligible compared to 
\(\Var\big(\tr(U M(X) U\t)\big)\),
it suffices to consider any fixed pair \((i,j)\) and establish that
\begin{align} \label{need01}
    \frac{
\Var\left(\tr\left(W_{ij} {X^{(j)}} \beta_j^{-1} \beta_i^{-1} {X^{(i)}}\t\right)\right)
}{
\Var\left(\tr\left({X^{(j)}} \beta_j^{-1} \beta_i^{-1} {X^{(i)}}\t\right)\right)
}
\to 0.
\end{align}
By conditioning on \({X^{(i)}}\) (recall \({X^{(j)}}\) is independent of \({X^{(i)}}\) and has mean zero) we have
\begin{align*}
\Var\big(\tr\big(W_{ij} {X^{(j)}} \beta_j^{-1}\beta_i^{-1} {X^{(i)}}\t\big)\big)
=\mathbb{E}\Big[\Var\big(\tr\big(W_{ij} {X^{(j)}} \beta_j^{-1}\beta_i^{-1} {X^{(i)}}\t\big)\bigm| {X^{(i)}}\big)\Big].
\end{align*}
Writing the trace in entrywise form yields
\(\tr\big(W_{ij} {X^{(j)}} \beta_j^{-1}\beta_i^{-1} {X^{(i)}}\t\big)
=\sum_{b,c} ({X^{(j)}})_{b,c} s_{b c}\big(({X^{(i)}})\big),\)
where, for each pair \((b,c)\),
\(
s_{b c}\big(({X^{(i)}})\big)
:=\sum_{a,d} W_{a b}\,(\beta_j^{-1}\beta_i^{-1})_{c d}\,({X^{(i)}})_{a,d}.
\)
By independence (up to symmetry), there exist positive constants \(C_1,C_2\) (independent of \(n\)) such that for every fixed realization of \({X^{(i)}}\),
\begin{equation}\label{eq:cond-var-bounds}
C_1\sum_{b,c}\Var\big(({X^{(j)}})_{b,c}\big)\,s_{b c}\big(({X^{(i)}})\big)^2
\le
\Var\big(\tr(W_{ij} {X^{(j)}} \beta_j^{-1}\beta_i^{-1} {X^{(i)}}\t)\bigm|{X^{(i)}}\big)
\le
C_2\sum_{b,c}\Var\big(({X^{(j)}})_{b,c}\big)\,s_{b c}\big(({X^{(i)}})\big)^2.
\end{equation}
Taking expectation with respect to \({X^{(i)}}\), and using \(\Var(({X^{(j)}})_{b,c})\asymp \rho_n\) uniformly in \((b,c)\), we obtain
\begin{equation}\label{eq:var-W-sum}
\Var\big(\tr(W_{ij} {X^{(j)}} \beta_j^{-1}\beta_i^{-1} {X^{(i)}}\t)\big)
\asymp
\rho_n \sum_{b,c}\mathbb{E}\big[s_{b c}\big(({X^{(i)}})\big)^2\big].
\end{equation}
We now expand \(\mathbb{E}\big[s_{b c}(({X^{(i)}}))^2\big]\). Using independence (up to symmetry) of the entries of \({X^{(i)}}\) and the uniform variance bound \(\Var(({X^{(i)}})_{a,d})\asymp\rho_n\),
it holds that
\begin{align*}
\mathbb{E}\big[s_{b c}\big(({X^{(i)}})\big)^2\big]
= \sum_{a,d} \bigg(W_{a b}^2\,(\beta_j^{-1}\beta_i^{-1})_{c d}^2 + W_{d b}^2\,(\beta_j^{-1}\beta_i^{-1})_{c a}^2\bigg)\,\mathbb{E}\big[({X^{(i)}})_{a,d}^2\big],
\end{align*}
and hence
\begin{equation}\label{eq:s-square-expectation}
\mathbb{E}\big[s_{b c}\big(({X^{(i)}})\big)^2\big]\asymp
\rho_n \sum_{a,d} \bigg(W_{a b}^2\,(\beta_j^{-1}\beta_i^{-1})_{c d}^2 + W_{d b}^2\,(\beta_j^{-1}\beta_i^{-1})_{c a}^2\bigg).
\end{equation}
Summing \cref{eq:s-square-expectation} over \(b,c\) and reordering the sums yields
\begin{align*}
\sum_{b,c}\mathbb{E}\big[s_{b c}\big(({X^{(i)}})\big)^2\big]
\asymp
\rho_n\sum_{a,b}\sum_{c,d} W_{a b}^2\,(\beta_j^{-1}\beta_i^{-1})_{c d}^2
\asymp \rho_n\,\|W_{ij}\|_F^2\,\|\beta_j^{-1}\beta_i^{-1}\|_F^2.
\end{align*}
Combining this with \cref{eq:var-W-sum} shows that
\begin{equation}\label{eq:var-weighted-final}
\Var\big(\tr(W_{ij} {X^{(j)}} \beta_j^{-1}\beta_i^{-1} {X^{(i)}}\t)\big)
\asymp \rho_n^2\,\|W_{ij}\|_F^2\,\|\beta_j^{-1}\beta_i^{-1}\|_F^2.
\end{equation}
We compare the above to the variance of the unweighted trace. By an identical conditioning argument, now with the coefficient
\(
r_{m p}\big(({X^{(i)}})\big):=\sum_{q}(\beta_j^{-1}\beta_i^{-1})_{p q}\,({X^{(i)}})_{m,q},
\)
appearing in the expansion \(\tr({X^{(j)}} \beta_j^{-1}\beta_i^{-1} {X^{(i)}}\t)=\sum_{m,p} ({X^{(j)}})_{m,p}\,r_{m p}(({X^{(i)}}))\), one obtains
\begin{equation}\label{eq:var-unweighted-final}
\Var\big(\tr({X^{(j)}} \beta_j^{-1}\beta_i^{-1} {X^{(i)}}\t)\big)
\asymp \rho_n^2\,n\,\|\beta_j^{-1}\beta_i^{-1}\|_F^2.
\end{equation}
Dividing \cref{eq:var-weighted-final} by \cref{eq:var-unweighted-final} yields
\[
\frac{\Var\big(\tr(W_{ij} {X^{(j)}} \beta_j^{-1}\beta_i^{-1} {X^{(i)}}\t)\big)}
     {\Var\big(\tr({X^{(j)}} \beta_j^{-1}\beta_i^{-1} {X^{(i)}}\t)\big)}
\asymp
\frac{\|W_{ij}\|_F^2}{n}.
\] Recalling the definition of \(W_{ij}\) in \cref{note04}, we note that
\(\|{V^{(1)}}{V^{(1)}}^{\top}\|_F^2 = k\) and \(\|{V^{(2)}}{V^{(2)}}^{\top}\|_F^2 = k\), which follow immediately from the orthonormality of the columns of \({V^{(1)}}\) and \({V^{(2)}}\).
Moreover,
\[
\|{V^{(1)}}{V^{(1)}}^{\top}+{V^{(2)}}{V^{(2)}}^{\top}-{V^{(1)}}{V^{(1)}}^{\top}{V^{(2)}}{V^{(2)}}^{\top}\|_F^2
=
\|{V^{(1)}}{V^{(1)}}^{\top}\|_F^2
+
\|{V^{(2)}}{V^{(2)}}^{\top}-{V^{(1)}}{V^{(1)}}^{\top}{V^{(2)}}{V^{(2)}}^{\top}\|_F^2,
\]
which is bounded above by \(2k\) and below by \(k\). Consequently,
\(\|W_{ij}\|_F^2 \asymp k\).
It follows that
\[
\frac{\Var\!\big(\tr(W_{ij} {X^{(j)}} \beta_j^{-1}\beta_i^{-1} {X^{(i)}}^{\top})\big)}
     {\Var\!\big(\tr({X^{(j)}} \beta_j^{-1}\beta_i^{-1} {X^{(i)}}^{\top})\big)}
\asymp \frac{k}{n} \longrightarrow 0,
\]
which completes the argument for \cref{need01}, thereby proving the lemma.
\end{proof}

\subsection{Proof of \cref{lem:second-order2}}
\label{sec:second-order2proof}
\begin{proof}[Proof of \cref{lem:second-order2}]
\textbf{Step 1: computing the mean.}  For independent, mean-zero matrices \({X^{(1)}}\) and \({X^{(2)}}\),
 \[
\mathbb{E}\bigl\|\beta^{\perp}({X^{(1)}} \beta_1^{-1} - {X^{(2)}} \beta_2^{-1})\bigr\|_F^2 
= \mathbb{E}\,\tr\bigl(\beta^{\perp} {X^{(1)}} \beta_1^{-2} {X^{(1)}}\bigr) 
+ \mathbb{E}\,\tr\bigl(\beta^{\perp} {X^{(2)}} \beta_2^{-2} {X^{(2)}}\bigr),
\]
since the cross terms vanish by independence. Using the one-sample expectation result from \cref{lem:second-order1}, for each \(i=1,2\),
\begin{align*}
\mathbb{E}\,({X^{(i)}} \beta_i^{\perp} {X^{(i)}}) 
= {\Sigma^{(i)}} \circ \beta_i^{\perp} 
+ \Diag\bigl({\Sigma^{(i)}} \cdot d_i - \diag \left({\Sigma^{(i)}}\right) \circ d_i \bigr),
\end{align*}
where
\( d_i = \diag(\beta_i^{\perp})\).
Consequently, the leading mean term satisfies
\begin{align*}
\mu_2 
= 2\sum_{i=1}^{2} 
  \tr\left[
    \beta_i^{-2} 
    \Bigl(
      {\Sigma^{(i)}} \circ \beta_i^{\perp} 
+ \Diag\bigl({\Sigma^{(i)}} \cdot d_i - \diag \left({\Sigma^{(i)}}\right) \circ d_i \bigr)
    \Bigr)
  \right].
\end{align*}
\\ \ \\
\noindent \textbf{Step 2: Computing the variance}.
According to \cref{lem:beta_perp1two}, it suffices to compute 
\[
\Var\bigg( 
2\,\tr\left(
U\t
\begin{pmatrix}
{X^{(1)}}\t {X^{(1)}} & -{X^{(1)}}\t {X^{(2)}}\\[4pt]
-{X^{(2)}}\t {X^{(1)}} & {X^{(2)}}\t {X^{(2)}}
\end{pmatrix}
U
\right) \bigg).
\]
Let
\begin{align*}
Q_i = {(X^{(1)})_{i.}}\t \beta_1^{-2} {(X^{(1)})_{i.}}
       + {(X^{(2)})_{i.}}\t \beta_2^{-2} {(X^{(2)})_{i.}}
       - {(X^{(1)})_{i.}}\t \beta_1^{-1}\beta_2^{-1} {(X^{(2)})_{i.}}
       - {(X^{(2)})_{i.}}\t \beta_2^{-1}\beta_1^{-1} {(X^{(1)})_{i.}},
       \end{align*}
and define 
\(\tilde{T}_{2}^{(S)} = 2\sum_{i=1}^{n} Q_i.
\)
Then
\begin{align} \label{eq:ind1}
\Var(\tilde{T}_{2}^{(S)}) 
= 4\sum_{i=1}^{n} \Var(Q_i) 
  + 8\sum_{1\le i<k\le n} \Cov(Q_i,Q_k).
\end{align}
First, we will calculate \(\Var(Q_i).\) For each sample \(r\in\{1,2\}\), the variance of the individual quadratic form \(Q_i^{(r)} = {(X^{(r)})_{i.}}\t \beta_r^{-2} {(X^{(r)})_{i.}}\)
follows directly from the one-sample calculation (see \cref{sec:second-order1proof})
\begin{align}\label{eq:ind2}
\Var(Q_i^{(r)}) 
= 4 \sum_{s<t} (\beta_r^{-2})_{st}^2\, \Sigma_{is}^{(r)} \Sigma_{it}^{(r)}
  + \sum_{s=1}^n (\beta_r^{-2})_{ss}^2\, K_{is}^{(r)},
\end{align}
where \(\Sigma_{ij}^{(r)}=\Var({(X^{(r)})_{ij}})\) and \(K_{ij}^{(r)} = \E\left({(X^{(r)})_{ij}}^4\right) - \E\left({(X^{(r)})_{ij}}^2\right)^2\). Since the vectors \({(X^{(r)})_{i.}}\) are centered, it follows immediately that 
\({(X^{(r_1)})_{i.}}\t\,\beta_{1}^{-1}\beta_{2}^{-1}\,{(X^{(r_2)})_{i.}}\) and 
\({(X^{(r_1)})_{i.}}\t\,\beta_{1}^{-1}\beta_{2}^{-1}\,{(X^{(r_1)})_{i.}}\)
are independent whenever \((r_{1},r_{2})=(1,2)\).
Hence, for the purpose of evaluating \(\Var(Q_{i})\), it suffices to compute
\(
\Var\left({(X^{(1)})_{i.}}\t\,\beta_{1}^{-1}\beta_{2}^{-1}\,{(X^{(2)})_{i.}}
+ {(X^{(2)})_{i.}}\t\,\beta_{2}^{-1}\beta_{1}^{-1}\,{(X^{(1)})_{i.}}\right).
\)

A direct calculation shows that
\(\
\Var\left({(X^{(1)})_{i.}}\t\,\beta_{1}^{-1}\beta_{2}^{-1}\,{(X^{(2)})_{i.}}\right)
= \sum_{s,t} M_{s,t}^2 \Sigma_{is}^{(1)} \Sigma_{is}^{(2)},\)
where \(
M := \beta_{1}^{-1}\beta_{2}^{-1}.
\)
Therefore, the total variance contribution from the bilinear cross–sample term equals
\begin{align} \label{eq:ind3}
\Var\left({(X^{(1)})_{i.}}\t\,\beta_{1}^{-1}\beta_{2}^{-1}\,{(X^{(2)})_{i.}}
+ {(X^{(2)})_{i.}}\t\,\beta_{2}^{-1}\beta_{1}^{-1}\,{(X^{(1)})_{i.}}\right) =
\sum_{s,t} G_{st}^{2}\,\Sigma_{is}^{(1)}\Sigma_{it}^{(2)},
\end{align}where \(
G := M + M\t,
\)
which completes the calculation.

We now compute the covariance terms.  
Decompose \(Q_i=Q_i^{(1)}+Q_i^{(2)}-Q_i^{(1,2)}\) such that
\(
Q_i^{(1)} = {(X^{(1)})_{i.}}\t \beta_1^{-2} {(X^{(1)})_{i.}},
Q_i^{(2)} = {(X^{(2)})_{i.}}\t \beta_2^{-2} {(X^{(2)})_{i.}},
Q_i^{(1,2)} = {(X^{(1)})_{i.}}\t M {(X^{(2)})_{i.}} + {(X^{(2)})_{i.}}\t M\t {(X^{(1)})_{i.}}.
\)
By bilinearity of covariance,
\[
\begin{aligned}
\Cov(Q_i, Q_k)
&= \Cov(Q_i^{(1)}, Q_k^{(1)}) + \Cov(Q_i^{(2)}, Q_k^{(2)}) 
- \Cov(Q_i^{(1)}, Q_k^{(1,2)}) - \Cov(Q_i^{(2)}, Q_k^{(1,2)}) \\[3pt]
&\quad - \Cov(Q_i^{(1,2)}, Q_k^{(1)}) - \Cov(Q_i^{(1,2)}, Q_k^{(2)}) 
+ \Cov(Q_i^{(1,2)}, Q_k^{(1,2)}).
\end{aligned}
\]
Observe that \(\Cov(Q_i^{(1)},Q_k^{(1,2)})=0\) because \(Q_k^{(1,2)}\) is a linear combination of products \({X^{(1)}_{ks}}{X^{(2)}_{kt}}\),
and by independence of \({X^{(1)}}\) and \({X^{(2)}}\) together with \(\E[{X^{(2)}_{kt}}]=0\), every term in \(\E[Q_i^{(1)} Q_k^{(1,2)}]\) vanishes.  
The same reasoning yields \(\Cov(Q_i^{(2)},Q_k^{(1,2)})=\Cov(Q_i^{(1,2)},Q_i^{(1)})=\Cov(Q_i^{(1,2)},Q_k^{(2)})=0\),
and independence across samples gives \(\Cov(Q_i^{(1)},Q_k^{(2)})=\Cov(Q_i^{(2)},Q_k^{(1)})=0\).  
Hence,
\begin{align*}
\Cov(Q_i,Q_k)
= \Cov(Q_i^{(1)},Q_k^{(1)}) + \Cov(Q_i^{(2)},Q_k^{(2)}) + \Cov(Q_i^{(12)},Q_k^{(12)}).
\end{align*}
Expanding
\(
Q_i^{(1)} = \sum_{s,t} (\beta_1^{-2})_{st}\, {X^{(1)}_{is}} {X^{(1)}_{it}},\)
and 
\(Q_k^{(1)} = \sum_{u,v} (\beta_1^{-2})_{uv}\, {X^{(1)}_{ku}} {X^{(1)}_{kv}},
\)
the only shared random variable for \(i\neq k\) is \({X^{(1)}_{ik}}={X^{(1)}_{ki}}\).  
Thus all cross-terms vanish except the product involving \(\left({X^{(1)}_{ik}}\right)^2\),
leading to
\begin{align*}
\Cov(Q_i^{(1)},Q_k^{(1)})
= \left(\E\left[\left({X^{(1)}_{ik}}\right)^4\right] - \E\left[\left({X^{(1)}_{ik}}\right)^2\right]^2\right)\,
  \left(\beta_1^{-2}\right)_{ii} \left(\beta_1^{-2}\right)_{kk}
= K_{ik}^{(1)} (\beta_1^{-2})_{ii} (\beta_1^{-2})_{kk}.
\end{align*}
Analogously,
\(
\Cov(Q_i^{(2)},Q_k^{(2)})
= K_{ik}^{(2)} (\beta_2^{-2})_{ii} (\beta_2^{-2})_{kk}.
\) 
Writing
\(
Q_i^{(12)} = \sum_{s,t} G_{st}\, {X^{(1)}_{is}} {X^{(2)}_{it}},\)
and exploiting the independence of \({X^{(1)}}\) and \({X^{(2)}}\),
\begin{align*}
\E[Q_i^{(12)} Q_k^{(12)}]
= \sum_{s,t,u,v} G_{st} G_{uv}\,
  \E[{X^{(1)}_{is}}{X^{(1)}_{ku}}]\,\E[{X^{(2)}_{it}}{X^{(2)}_{kv}}].
\end{align*}
For \(i\neq k\), the only index combination contributing a nonzero expectation is \((s,u)=(k,i)\) and \((t,v)=(i,k)\),
yielding
\(
\Cov(Q_i^{(12)}, Q_k^{(12)})
= G_{ik}^2\, \Sigma_{ik}^{(1)} \Sigma_{ik}^{(2)}.
\) Combining all nonzero contributions, for \(i\neq k\),
\begin{align}\label{eq:ind4}
\Cov(Q_i,Q_k)
= K_{ik}^{(1)} (\beta_1^{-2})_{ii} (\beta_1^{-2})_{kk}
  + K_{ik}^{(2)} (\beta_2^{-2})_{ii} (\beta_2^{-2})_{kk}
  + G_{ik}^2\, \Sigma_{ik}^{(1)} \Sigma_{ik}^{(2)}.
\end{align}
Let \(\alpha_r=\diag(\beta_r^{-2})\).  
Summing the within and between index contributions from \cref{eq:ind1,eq:ind2,eq:ind3,eq:ind4}, and including the cross-term variance from the bilinear component, yields
\[
\begin{aligned}
\Var(\tilde{T}_{2}^{(S)})
&= 8 \sum_{r=1}^{2} 
   \Big\langle (J_n-I_n) \circ (\beta_r^{-2}\circ\beta_r^{-2}),\, (\Sigma^{(r)})^2 \Big\rangle
 + 4 \sum_{r=1}^{2} \alpha_r\t K^{(r)} \alpha_r \\
&\quad + 4 \Big\langle G\circ G,
   {\Sigma^{(1)}}\t {\Sigma^{(2)}} + (J_n-I_n)\circ({\Sigma^{(1)}}\circ{\Sigma^{(2)}}) \Big\rangle,
\end{aligned}
\]
where \(G = \beta_1^{-1}\beta_2^{-1} + \beta_2^{-1}\beta_1^{-1}\). From \cref{ass:asymptotic1,ass:eigen-scaling,ass:incoherence}, we have   
\(K^{(r)}_{ij} \asymp \rho_n,\) and \(\Sigma^{(r)}_{ij} \asymp \rho_n\) for \(r=1,2.\)  Plugging these values in our expression using the identity \cref{newres01} just like in the last paragraph of \cref{sec:second-order1proof}, it follows that \(\Var(\tilde{T}_{2}^{(S)}) \asymp \frac{k^2}{n^3 \rho_n^2}\). Hence by \cref{lem:beta_perp1two}, we have 
\[
\Var\bigl(T_{2}^{(S)}\bigr)
= \sum_{i=1}^{2} 
  \Bigl(
    8\,\bigl\langle (J_n - I_n)\circ(\beta_i^{-2}\circ\beta_i^{-2}),({\Sigma^{(i)}})^2 \bigr\rangle
    + 4\,\alpha_i\t K^{(i)} \alpha_i
  \Bigr)
  + 4\,\bigl\langle 
      G \circ G,
      {\Sigma^{(1)}}\t {\Sigma^{(2)}} + (J_n - I_n) \circ ({\Sigma^{(1)}} \circ {\Sigma^{(2)}})
    \bigr\rangle + o \left( \frac{k^2}{n^3 \rho_n^2}\right).
\]
\end{proof}

\subsection{Proof of \cref{lem:second-order-clt2}} \label{sec:second-orderclt2proof}

\begin{proof}[Proof of \cref{lem:second-order-clt2}]
We will prove this using a very similar approach to the proof of \cref{lem:second-order-clt1}. Suppose that
\[
\underbrace{2\,\tr\left(
U\t
\begin{pmatrix}
(Y^{(1)})\t (Y^{(1)}) & -(Y^{(1)})\t (Y^{(2)})\\[4pt]
-(Y^{(2)})\t (Y^{(1)}) & (Y^{(2)})\t (Y^{(2)})
\end{pmatrix}
U
\right)}_{T_2^{(S)}}
=
\underbrace{2\,\tr\left(
U\t
\begin{pmatrix}
{X^{(1)}}\t {X^{(1)}} & -{X^{(1)}}\t {X^{(2)}}\\[4pt]
-{X^{(2)}}\t {X^{(1)}} & {X^{(2)}}\t {X^{(2)}}
\end{pmatrix}
U
\right)}_{\widetilde T_2^{(S)}} + R,
\] with the same notation as in \cref{lem:second-order2}. By \cref{lem:beta_perp1one}, we have
\(
\Var(R) / \Var(T_{2}^{(S)}) \to 0,
\)
and consequently
\(
(R - \E (R))/\sqrt{\Var(T_{2}^{(S)})} \to 0
\)
in probability. Hence the contribution of \(R\) is asymptotically negligible.  
Therefore, in establishing the limiting distribution of
\(
T_2^{(S)},
\)
it suffices to derive the asymptotic distribution of
\(
\widetilde T_2^{(S)}
\)
after appropriate centering and scaling.  
The desired result then follows by a direct application of Slutsky’s theorem. We expand the leading-order term via
\begin{equation} \label{break_step1}
\begin{aligned}
n \mathrm{tr}\bigg( 
U\t &
\begin{pmatrix}
{X^{(1)}}\t {X^{(1)}} & -{X^{(1)}}\t {X^{(2)}} \\
-{X^{(2)}}\t {X^{(1)}} & {X^{(2)}}\t {X^{(2)}}
\end{pmatrix}
U
\bigg) \\
&= \sum_{i=1}^{n} \sum_{j=1}^{n} \sum_{k=1}^{n} \sum_{l=1}^{n}
n\Big[ U_{k j}U_{l j}(x_1)_{i k}(x_1)_{i l}  + U_{n+k, j}U_{n+l,j}(x_2)_{i k}(x_2)_{i l}  - 2\,U_{k j}U_{n+l,j}(x_1)_{i k}(x_2)_{i l} \Big].
\end{aligned}
\end{equation}
Regrouping terms using $\gamma_{a,b} := n\sum_{j=1}^n U_{a j}U_{b j}$ and rearranging, we obtain
\begin{align} 
\begin{split}
n \mathrm{tr}\bigg( 
U\t &
\begin{pmatrix}
{X^{(1)}}\t {X^{(1)}} & -{X^{(1)}}\t {X^{(2)}} \\
-{X^{(2)}}\t {X^{(1)}} & {X^{(2)}}\t {X^{(2)}}
\end{pmatrix}
U
\bigg) \\
&= \sum_{i=1}^{n} \sum_{k=1}^{n} \sum_{l=1}^{n}
\Big[ \gamma_{k,l}(x_1)_{i k}(x_1)_{i l}
+ \gamma_{n+k,n+l}(x_2)_{i k}(x_2)_{i l} 
- 2\,\gamma_{k,n+l}(x_1)_{i k}(x_2)_{i l} \Big]  \end{split} \notag\\
 \begin{split}
&= \sum_{i=1}^{n} \sum_{\substack{k,l=1 \\ k>l}}^{n}
\Big[ 2\gamma_{k,l}(x_1)_{i k}(x_1)_{i l}
+ 2\gamma_{n+k,n+l}(x_2)_{i k}(x_2)_{i l}  - 2\,\gamma_{k,n+l}(x_1)_{i l}(x_2)_{i k}
- 2\,\gamma_{l,n+k}(x_1)_{i k}(x_2)_{i l} \Big] \\
&\quad + \sum_{i=1}^{n} \sum_{k=1}^{n}
\Big[ \gamma_{k,k}(x_1)_{i k}^2
+ \gamma_{n+k,n+k}(x_2)_{i k}^2
- 2\,\gamma_{k,n+k}(x_1)_{i k}(x_2)_{i k} \Big]\end{split} \notag\\
\begin{split}
&= \sum_{\substack{i,k,l=1 \\ k>l,\, k<i}}^{n}
\Big[ 2\gamma_{k,l}(x_1)_{i k}(x_1)_{i l}
+ 2\gamma_{n+k,n+l}(x_2)_{i k}(x_2)_{i l} 
- 2\,\gamma_{k,n+l}(x_1)_{i l}(x_2)_{i k}
- 2\,\gamma_{l,n+k}(x_1)_{i k}(x_2)_{i l} \Big]  \\
&\quad + \sum_{\substack{i,k,l=1 \\ l<i,\, k<i}}^{n}
\Big[ 2\gamma_{i,l}(x_1)_{i k}(x_1)_{k l}
+ 2\gamma_{n+i,n+l}(x_2)_{i k}(x_2)_{k l} 
- 2\,\gamma_{i,n+l}(x_1)_{k l}(x_2)_{i k}
- 2\,\gamma_{l,n+i}(x_1)_{i k}(x_2)_{l k} \Big]  \\
&\quad + \sum_{\substack{k,l=1 \\ l<k}}^{n}
\Big[ 2\gamma_{k,l}(x_1)_{k k}(x_1)_{k l}
+ 2\gamma_{n+k,n+l}(x_2)_{k k}(x_2)_{k l} 
- 2\,\gamma_{k,n+l}(x_1)_{k l}(x_2)_{k k} 
- 2\,\gamma_{l,n+k}(x_1)_{k k}(x_2)_{l k} \Big]  \\
&\quad + \sum_{i,k=1}^{n}
\Big[ \gamma_{k,k}(x_1)_{i k}^2
+ \gamma_{n+k,n+k}(x_2)_{i k}^2 
- 2\gamma_{k,n+k}(x_1)_{i k}(x_2)_{i k}  \Big] \\
&= \sum_{1 \leq k < i \leq n}
\Bigg[
\sum_{1 \leq l < k \leq n}
\big( 2\gamma_{k,l}(x_1)_{i k}(x_1)_{i l}
+ 2\gamma_{n+k,n+l}(x_2)_{i k}(x_2)_{i l} 
- 2\,\gamma_{k,n+l}(x_1)_{i l}(x_2)_{i k}
- 2\,\gamma_{l,n+k}(x_1)_{i k}(x_2)_{i l} \big) \\
&\quad + \sum_{1 \leq l < i \leq n}
\big( 2\gamma_{i,l}(x_1)_{i k}(x_1)_{k l}
+ 2\gamma_{n+i,n+l}(x_2)_{i k}(x_2)_{k l} 
- 2\,\gamma_{i,n+l}(x_1)_{k l}(x_2)_{i k}
- 2\,\gamma_{l,n+i}(x_1)_{i k}(x_2)_{l k} \big)
\Bigg] \\
&\quad + \sum_{1 \leq l < k \leq n}
\big( 2\gamma_{k,l}(x_1)_{k k}(x_1)_{k l}
+ 2\gamma_{n+k,n+l}(x_2)_{k k}(x_2)_{k l} 
- 2\,\gamma_{k,n+l}(x_1)_{k l}(x_2)_{k k}
- 2\,\gamma_{l,n+k}(x_1)_{k k}(x_2)_{l k} \big) \\
&\quad + \sum_{1 \leq k < i \leq n}
\big[ (\gamma_{k,k} + \gamma_{i,i})(x_1)_{i k}^2
+ (\gamma_{n+k,n+k} + \gamma_{n+i,n+i})(x_2)_{i k}^2
- 2(\gamma_{k,n+k} + \gamma_{i,n+i})(x_1)_{i k}(x_2)_{i k} \big] \\
&\quad + \sum_{1 \leq k \leq n}
\big[ \gamma_{k,k} (x_1)_{k k}^2
+ \gamma_{n+k,n+k} (x_2)_{k k}^2
- 2 \gamma_{n+k,k} (x_1)_{k k} (x_2)_{k k} \big].
\end{split} \notag
\end{align}
Now consider the centered version of the quadratic form above. The centering ensures that each term has mean zero, which is essential for the martingale difference structure that follows. Explicitly, we write
\begin{equation}
\begin{aligned} \label{centered2}
&\sum_{1 \leq k < i \leq n}
\Bigg[
\sum_{1 \leq l < k \leq n}
\big( 2\gamma_{k,l}(x_1)_{i k}(x_1)_{i l}
+ 2\gamma_{n+k,n+l}(x_2)_{i k}(x_2)_{i l} 
- 2\,\gamma_{k,n+l}(x_1)_{i l}(x_2)_{i k}
- 2\,\gamma_{l,n+k}(x_1)_{i k}(x_2)_{i l} \big) \\
&\quad + \sum_{1 \leq l < i \leq n}
\big( 2\gamma_{i,l}(x_1)_{i k}(x_1)_{k l}
+ 2\gamma_{n+i,n+l}(x_2)_{i k}(x_2)_{k l} 
- 2\,\gamma_{i,n+l}(x_1)_{k l}(x_2)_{i k}
- 2\,\gamma_{l,n+i}(x_1)_{i k}(x_2)_{l k} \big)
\Bigg] \\
&\quad + \sum_{1 \leq l < k \leq n}
\big( 2\gamma_{k,l}(x_1)_{k k}(x_1)_{k l}
+ 2\gamma_{n+k,n+l}(x_2)_{k k}(x_2)_{k l} 
- 2\,\gamma_{k,n+l}(x_1)_{k l}(x_2)_{k k}
- 2\,\gamma_{l,n+k}(x_1)_{k k}(x_2)_{l k} \big) \\
&\quad + \sum_{1 \leq k < i \leq n}
\big[ (\gamma_{k,k} + \gamma_{i,i})\big((x_1)_{i k}^2 - (\sigma_1)_{i,k}^2\big)
+ (\gamma_{n+k,n+k} + \gamma_{n+i,n+i})\big((x_2)_{i k}^2 - (\sigma_2)_{i,k}^2\big) \\
&\quad\quad - 2(\gamma_{k,n+k} + \gamma_{i,n+i})(x_1)_{i k}(x_2)_{i k} \big] \\
&\quad + \sum_{1 \leq k \leq n}
\big[ \gamma_{k,k} \big((x_1)_{k k}^2 - (\sigma_1)_{k,k}^2\big)
+ \gamma_{n+k,n+k} \big((x_2)_{k k}^2 - (\sigma_2)_{k,k}^2\big)
- 2 \gamma_{n+k,k} (x_1)_{k k} (x_2)_{k k} \big],
\end{aligned}
\end{equation}
where \(\E \,(x_i)_{k_1,k_2}^2=(\sigma_i)_{k_1,k_2}^2\).
To employ the martingale central limit theorem, we introduce a filtration that orders the terms appropriately.  
Define
\(
\mathscr{F}_t := \sigma({w}_1, \ldots, {w}_t), \; \text{where } t = k + \frac{l(l - 1)}{2}, \; 1 \leq k \leq l \leq n,
\)
so that
\[
\mathscr{F}_t = \sigma\big( (x_m)_{ij} : 1 \leq i \leq j < l \ \text{or} \ 1 \leq k \leq j < l, \ m=1,2 \big).
\]
Under this ordering, each summand in \cref{centered2} can be expressed as a martingale difference with respect to \(\{\mathscr{F}_t\}\).  Indeed, for \(1 \leq k \leq i \leq n\) and \(t = k + \frac{i(i-1)}{2}\),
\begin{equation}
\begin{aligned} \label{single2}
&\E\Big[(x_1)_{i,k}(b_1)_{i,k} + (x_2)_{i,k}(b_2)_{i,k} 
+ \big((x_1)^2_{i,k}-(\sigma_1)^2_{i,k}\big)(c_1)_{i,k}  \\
&\quad + \big((x_2)^2_{i,k}-(\sigma_2)^2_{i,k}\big)(c_2)_{i,k} 
+ (x_1)_{i,k}(x_2)_{i,k}(c_3)_{i,k} \ \big| \ \mathscr{F}_{t-1} \Big] = 0,
\end{aligned}
\end{equation}
where  
\[
\begin{aligned}
(b_1)_{i,k} &= \sum_{1 \leq l < k } \big( 2 \gamma_{k,l} (x_1)_{i,l} - 2 \gamma_{l,n+k} (x_2)_{i,l} \big) 
+ \sum_{1 \leq l < i } \big( 2 \gamma_{i,l} (x_1)_{k,l} - 2 \gamma_{l,n+i}(x_2)_{l,k} \big), \\
(b_2)_{i,k} &= \sum_{1 \leq l < k } \big( 2 \gamma_{n+k,n+l}(x_2)_{i,l} - 2\gamma_{k,n+l}(x_1)_{i,l} \big) 
+ \sum_{1 \leq l < i } \big( 2 \gamma_{n+i,n+l}(x_2)_{k,l} - 2 \gamma_{i,n+l} (x_1)_{k,l} \big), \\
(c_1)_{i,k} &=  \gamma_{k,k} + \gamma_{i,i} , \\
(c_2)_{i,k} &=  \gamma_{n+k,n+k} + \gamma_{n+i,n+i} , \\
(c_3)_{i,k} &= - 2\gamma_{k,n+k} - 2\gamma_{i,n+i}.
\end{aligned}
\]
From \cref{single2}, the conditional variance of the sum in \cref{centered2} with respect to \(\{\mathscr{F}_t\}\) denoted by \(\eta\) is given by:
\begin{equation}
\begin{aligned}\label{cond_var2}
\eta = \sum_{1\leq k<i \leq n} &\Big[ (b_1)^2_{i,k} (\sigma_1)^2_{i,k} + (b_2)^2_{i,k} (\sigma_2)^2_{i,k} 
+ (c_1)^2_{i,k} (\kappa_1)_{i,k} + (c_2)^2_{i,k} (\kappa_2)_{i,k} \\
&\quad + (c_3)^2_{i,k} (\sigma_1)^2_{i,k} (\sigma_2)^2_{i,k} 
+ (c_1)_{i,k}(b_1)_{i,k} (\theta_1)_{i,k} 
+ (c_2)_{i,k}(b_2)_{i,k} (\theta_2)_{i,k} \Big],
\end{aligned}
\end{equation}
where \((\kappa_m)_{i,k} = \E\big[ (x_m)^2_{i,k} - (\sigma_m)^2_{i,k} \big]^2\) and \((\theta_m)_{i,k} = \E\big[ (x_m)_{i,k} \big( (x_m)^2_{i,k} - (\sigma_m)^2_{i,k} \big) \big]\).

This explicit martingale difference decomposition, together with the variance representation in \cref{cond_var2}, provides the necessary framework to invoke the martingale central limit theorem.
From the definition \( U = \begin{pmatrix} \beta_1^{-1} & \beta_2^{-1} \end{pmatrix} \). Like \cref{newres01}, we have
\begin{align}\label{newres03}
    \sum_{i \neq j}\gamma_{ij}^2 \asymp \frac{k}{n^2 \rho_n^4}, \quad \gamma_{ij}^2 \lesssim \frac{k}{n^4 \rho_n^4}.
\end{align} 
We first compute the order of \(\E(\eta)\). From the definitions of \(b_i\) and \(c_i\), it is straightforward to verify that
\begin{equation}
\begin{aligned}\label{newres004}
\sum_{1\le k<i\le n}
(\sigma_1)_{ik}^{2}\,
\E\!\left[(b_1)_{ik}^{2}\right]
&\gtrsim
\rho_n^2
\sum_{1\le k<i\le n}
\left(
\sum_{1\le l<k}\gamma_{k,l}^{2}
+
\sum_{1\le l<i}\gamma_{i,l}^{2}
\right) \\
&=
\rho_n^2
\left[
\sum_{1\le l<k\le n}(n-k)\gamma_{k,l}^{2}
+
\sum_{1\le l<i\le n}(i-1)\gamma_{i,l}^{2}
\right] \\
&=
(n-1)\rho_n^2
\sum_{1\le l<k\le n}
\gamma_{k,l}^{2} \\
&\asymp
n\rho_n^2
\cdot
\frac{k}{n^2\rho_n^4}
=
\frac{k}{n\rho_n^2}.
\end{aligned}
\end{equation}
Similarly we have 
\begin{align}\label{newres005}
    \sum_{1\le k<i\le n}
(\sigma_2)_{ik}^{2}\,
\E\!\left[(b_2)_{ik}^{2}\right] \gtrsim \frac{k}{n\rho_n^2}.
\end{align}
The other side of the bound follows trivially from \cref{ass:incoherence} yielding
\begin{align}\label{newres04}
    \sum_{1\le k<i\le n}
(\sigma_m)_{ik}^{2}\,
\E\!\left[(b_m)_{ik}^{2}\right] \asymp \frac{k}{n\rho_n^2} \quad \text{for} \quad m=1,2
\end{align}
Similar to \cref{newc}, we also have the bounds
\begin{align}
    \sum_{1 \le k < i \le n}(\kappa_1)_{i,k} (c_1)^2_{i,k} \lesssim \frac{k^2}{n^2 \rho_n^3}=o \left(\frac{k}{n\rho_n^2}\right), \quad \sum_{1 \le k < i \le n}(\sigma_1)^2_{i,k} (\sigma_2)^2_{i,k} (c_3)^2_{i,k} \lesssim \frac{k^2}{n^2 \rho_n^3}=o \left(\frac{k}{n\rho_n^2}\right)
\end{align}

In addition, we have the bounds
\(
(\kappa_m)_{i,k} \asymp \rho_n, 
\; (\theta_m)_{i,k} \asymp \rho_n, 
\; (\sigma_m)^2_{i,k} \asymp \rho_n.
\)
Combining these estimates, we find
\(
\E(\eta) \asymp \frac{k}{n\rho_n^2}.
\)

Next, we determine the order of \(\Var(\eta)\). Since the terms \((c_m)_{i,k}\) are deterministic, they do not contribute directly to the variance; only the stochastic terms involving \((b_m)_{i,k}\) enter into \(\Var(\eta)\). We obtain
\begin{equation}
\begin{aligned}
\E\!\left[(b_m)_{i,k}^{4}\right]
&\lesssim
\rho_n \max_{1\le a\le 2n}
\sum_{b=1}^{2n}\gamma_{a,b}^{4}
+
\rho_n^2
\left(
\max_{1\le a\le 2n}
\sum_{b=1}^{2n}\gamma_{a,b}^{2}
\right)^2 \\
&\lesssim
\rho_n
\left(\max_{a,b}\gamma_{a,b}^{2}\right)
\left(\max_a \sum_b \gamma_{a,b}^{2}\right)
+
\rho_n^2
\left(
\frac{k}{n^3\rho_n^4}
\right)^2 \\
&\lesssim
\frac{k^2}{n^6\rho_n^6}.
\end{aligned}
\end{equation}
and thus,
\begin{align*}
&\Var\big( (b_m)^2_{i,k} (\sigma_m)^2_{i,k} \big) 
\lesssim (\sigma_m)^4_{i,k} \E ((b_m)^4_{i,k}) \lesssim\frac{k^2}{n^6 \rho_n^4}; \\
&\Var\big( (c_m)_{i,k}(b_m)_{i,k} (\theta_m)_{i,k}\big) 
\lesssim \frac{k^3}{n^7 \rho_n^5}.
\end{align*}
Aggregating these contributions, we conclude that
\(
\Var(\eta) \asymp \frac{k^2}{n^4 \rho_n^4}.
\)

Consequently, we have
\begin{align*}
\frac{\E(\eta)}{\sqrt{\Var(\eta)}} 
=\sqrt{n} \gg 1.
\end{align*}
This verify the conditions required for the martingale CLT as stated in Lemma 9.12 of \citet{bai2010spectral}. Therefore, the appropriately scaled second-order approximation of our test statistic satisfies
\begin{align*}
\frac{
  T_{2}^{(S)} - \mathbb{E}[T_{2}^{(S)}]
}{\sqrt{\Var(T_{2}^{(S)})}} \xrightarrow{D} \mathcal{N}(0,1)
  \end{align*}
which establishes the result.
\end{proof}

\subsection{Proof of \cref{lem:3rdmean2}}\label{sec:3rdm2}
\begin{proof}
The third order term is 
\[\left( 
    \langle {V^{(1)}} {V^{(1)}}\t, S_{1,3}({X^{(1)}}) \rangle 
    + \langle {V^{(2)}} {V^{(2)}}\t, S_{2,3}({X^{(2)}}) \rangle 
    + \sum_{l_1 + l_2 = 3} 
      \langle S_{1,l_1}({X^{(1)}}), S_{2,l_2}({X^{(2)}}) \rangle 
\right).\]
The first two terms are directly \(o \left(\frac{k}{n^2 \rho_n^2} \right)\) using \cref{lem:3rdmean1}. For the mixed terms, we have that
\begin{equation}
\begin{aligned} \label{eq:van}
\E\!\left\langle S_{1,1}({X^{(1)}}), S_{2,2}({X^{(2)}}) \right\rangle
&= \E_{{X^{(2)}}}\!\left[
\E_{{X^{(1)}}}\!\left(
\left\langle S_{1,1}({X^{(1)}}), S_{2,2}({X^{(2)}}) \right\rangle
\,\big|\, {X^{(2)}}
\right)
\right] \\
&= \E_{{X^{(2)}}}\!\left[
\left\langle \E_{{X^{(1)}}}\!\left( S_{1,1}({X^{(1)}}) \right), S_{2,2}({X^{(2)}}) \right\rangle
\right] \\
&= 0 .
\end{aligned}
\end{equation}
By the same argument as in \cref{eq:van}, the remaining mixed term
\(\E\langle S_{1,2}({X^{(1)}}), S_{2,1}({X^{(2)}})\rangle\) also vanishes identically.
Indeed, since \({X^{(1)}}\) and \({X^{(2)}}\) are independent and centered, the inner
conditional expectation with respect to either random matrix is zero.
This completes the proof of \cref{lem:3rdmean2}.
\end{proof}

\subsection{Proof of \cref{lem:3rdvar2}}\label{sec:3rdproofvar2}
\begin{proof}
This proof follows by identical index–counting applied component-wise as in \cref{lem:3rdvar1}. Without loss of generality, take \(l_1 = 1\) and \(l_2 = 2\). First we expand 
\(
\langle S_{1,1}({X^{(1)}}),\, S_{2,2}({X^{(2)}})\rangle
\):
\begin{equation}
\begin{aligned} \label{eq:exp}
    \langle S_{1,1}({X^{(1)}}),\, S_{2,2}({X^{(2)}})\rangle & = \tr\left(\beta_1^{-1} {X^{(1)}} \beta^{\perp} {X^{(2)}} \beta^{\perp} {X^{(2)}} \beta_2^{-2} + \beta^{\perp} {X^{(1)}} \beta_1^{-1} \beta_2^{-2} {X^{(2)}} \beta^{\perp} {X^{(2)}} \right) \\
    & - \tr\left(\beta_1^{-1} {X^{(1)}} \beta^{\perp} {X^{(2)}} \beta_2^{-1} {X^{(2)}} \beta_2^{-1} + \beta^{\perp} {X^{(1)}} \beta_1^{-1} \beta_2^{-1} {X^{(2)}} \beta_2^{-1} {X^{(2)}} \right).
\end{aligned}
\end{equation}
From the expansion in \cref{eq:exp}, we will establish that the term whose variance dominates the other terms and hence is the dominant variance contributor to 
\(
\langle S_{1,1}({X^{(1)}}),\, S_{2,2}({X^{(2)}})\rangle,
\)
arises from the fully separated monomials
\(
\tr\left(\beta_1^{-1} {X^{(1)}} \beta^{\perp} {X^{(2)}} \beta^{\perp} {X^{(2)}} \beta_2^{-2}\right) \) and \( \tr \left(\beta^{\perp} {X^{(1)}} \beta_1^{-1} \beta_2^{-2} {X^{(2)}} \beta^{\perp} {X^{(2)}} \right)
\), and then calculate its order. 
We first show that
\[
\frac{
\Var\left(\tr\left(\beta_1^{-1} {X^{(1)}} (\beta^{\perp} {X^{(2)}})(\beta^{\perp} {X^{(2)}})\beta_2^{-2}\right)\right)
}{
\Var\left(\tr\left(\beta_1^{-1} {X^{(1)}} {X^{(2)}}^{2}\beta_2^{-2}\right)\right)
}
\to 1,
\]
exactly paralleling the one–sample setting as in \cref{sec:3rdproof}.

Consider the mixed trace
\(
R_{3}
:= \tr\big(\beta_1^{-1} {X^{(1)}} ({X^{(2)}})^{2}\beta_2^{-2}\big).
\)
Decompose \({X^{(2)}} = \left({V^{(2)}} {V^{(2)}}\t + \beta^{\perp}\right){X^{(2)}}\) in the definition of \(R_{3}\).  
The unique term in which all insertions equal \(\beta^{\perp}\) is the fully separated monomial
\[
R_{4}
:= \tr\left(\beta_1^{-1} {X^{(1)}} (\beta^{\perp} {X^{(2)}})(\beta^{\perp} {X^{(2)}})\beta_2^{-2}\right),
\]
the two–sample analogue of \(R_{2}\) in \cref{sec:3rdproof}.  
Every other summand (a remainder term) contains at least one factor \({V^{(2)}} {V^{(2)}}\t\).
Fix a remainder term
\(
\mathcal{N}
:= \tr\left(\beta_1^{-1} {X^{(1)}} (N_{1} {X^{(2)}})(N_{2} {X^{(2)}})\beta_2^{-2}\right),
\)
where each of \(N_{1}, N_{2}\) is either \({V^{(2)}}{V^{(2)}}\t\) or \(\beta^{\perp}\).  
Expanding the trace into its index form, exactly as in \cref{sec:3rdproof}, yields
\[
\mathcal{N}
= 
\sum_{i_1,i_2,i_3 \atop j_1,j_2,j_3}
c(\mathbf{i},\mathbf{j})\,
({X^{(1)}})_{i_{1},j_{2}}\,
({X^{(2)}})_{i_{2},j_{3}}\,
({X^{(2)}})_{i_{3},j_{1}},
\]
where
\(
c(\mathbf{i},\mathbf{j})
:= 
(\beta_2^{-2}\beta_1^{-1})_{j_1,i_1}\,
(N_{1})_{j_2,i_2}\,
(N_{2})_{j_3,i_3}.
\)
Using the same arguments as in \cref{sec:3rdproof}, we define
\(\Sigma(\mathcal{N})
:= 
\sum_{(\mathbf{k},\mathbf{l}) \sim (\mathbf{i},\mathbf{j})}
c(\mathbf{i},\mathbf{j}) c(\mathbf{k},\mathbf{l})
\) which enables us to write \(\Var(\mathcal{N})\asymp \Sigma(\mathcal{N}) \rho_n^3.\)
The exact same argument of \cref{sec:3rdproof} holds here as well. Consider the case where \(N_1=V V\t\) and \(N_2= \beta^{\perp}\). The entrywise bounds for \(N_1\) and \(N_2\) implied by \cref{ass:incoherence} yield \(\Sigma(\mathcal{N}) \lesssim \frac{k^4}{n^5 \rho_n^6}\). Similarly, if \(N_1=N_2=V V\t\), we have \(\Sigma(\mathcal{N}) \lesssim \frac{k^6}{n^6 \rho_n^6}\).
It follows that whenever at least one of \(N_{1}\) or \(N_{2}\) equals \(VV\t\), the variance of the monomial satisfies 
\(
\Var(\mathcal{N}) = o(\Var(R_4)),
\) as in \cref{mono1}.
\\ \ \\
In contrast, just like in \cref{ord1},
\[
\Var(R_{3}) \asymp n^2 \rho_n^3 \|\beta_2^{-2} \beta_1^{-1}\|_F^2
\asymp
\frac{k^2}{n^{4}\rho_n^{3}}.
\]
Using the same logic as in \cref{later01}, this immediately proves
\[
\frac{
\Var\left(\tr\left(\beta_1^{-1} {X^{(1)}} (\beta^{\perp} {X^{(2)}})(\beta^{\perp} {X^{(2)}})\beta_2^{-2}\right)\right)
}{
\Var\left(\tr\left(\beta_1^{-1} {X^{(1)}} {X^{(2)}}^{2}\beta_2^{-2}\right)\right)
}
\to 1.
\]
The same argument establishes
\[
\frac{
\Var\left(\tr\left(\beta_1^{-1} {X^{(1)}} \beta^\perp {X^{(2)}} \beta_2^{-1} {X^{(2)}} \beta_2^{-1}\right)\right)
}{
\Var\left(\tr\left(\beta_1^{-1} {X^{(1)}} {X^{(2)}} \beta_2^{-1} {X^{(2)}} \beta_2^{-1}\right)\right)
}
\to 1.
\] Thus, 
it remains to show
\[
\frac{
\Var\left(\tr\left(\beta_1^{-1} {X^{(1)}} {X^{(2)}} \beta_2^{-1} {X^{(2)}} \beta_2^{-1}\right)\right)
}{
\Var\left(\tr\big(\beta_1^{-1} {X^{(1)}} {X^{(2)}}^{2}\beta_2^{-2}\big)\right)
}
\to 0.
\]
Write the trace in the numerator as
\[
R_{2,\mathrm{num}}
:= 
\tr\left(\beta_1^{-1} {X^{(1)}} {X^{(2)}} \beta_2^{-1} {X^{(2)}} \beta_2^{-1}\right).
\]
Using the same argument employed in \cref{sec:3rdproof},
\begin{align*}
\Var(R_{2,\mathrm{num}})
&\lesssim
\rho_n^{3}
\sum_{i_1,i_2 \atop j_1,j_2}
\left[
(\beta_2^{-1}\beta_1^{-1})_{i_1,j_1}\,
(\beta_2^{-1})_{i_2,j_2}
\right]^{2}
\\
&\lesssim
\rho_n^{3}
\sum_{i_1,i_2 \atop j_1,j_2}
\frac{k^{6}}{n^{10}\rho_n^{6}}
\lesssim
\frac{k^{6}}{n^{6}\rho_n^{3}}.
\end{align*}
Thus,
\[
\frac{
\Var\left(\tr\left(\beta_1^{-1} {X^{(1)}} {X^{(2)}} \beta_2^{-1} {X^{(2)}} \beta_2^{-1}\right)\right)
}{
\Var\left(\tr\big(\beta_1^{-1} {X^{(1)}} {X^{(2)}}^{2}\beta_2^{-2}\big)\right)
} = \frac{\Var(R_{2,\mathrm{num}})}{\Var(R_{3})} \lesssim \frac{\frac{k^6}{n^6 \rho_n^3}}{\frac{k^2}{n^4 \rho_n^3}}
\to 0.
\]
Similarly we can show
\[\frac{
\Var \left(\tr\left( \beta^{\perp} {X^{(1)}} \beta_1^{-1} \beta_2^{-1} {X^{(2)}} \beta_2^{-1} {X^{(2)}} \right) \right)
}{
\Var\left(\tr\left(\beta_1^{-1} {X^{(1)}} (\beta^{\perp} {X^{(2)}})(\beta^{\perp} {X^{(2)}})\beta_2^{-2}\right)\right)
} \to 0.\]
So, in \cref{eq:exp}, the variance of the terms \(
\tr\left(\beta_1^{-1} {X^{(1)}} \beta^{\perp} {X^{(2)}} \beta^{\perp} {X^{(2)}} \beta_2^{-2}\right)\) and  \(\tr \left(\beta^{\perp} {X^{(1)}} \beta_1^{-1} \beta_2^{-2} {X^{(2)}} \beta^{\perp} {X^{(2)}} \right)
\) is of order \(k^2/(n^4 \rho_n^3)\). The remaining terms of
\(
\Var\!\big(\langle S_{1,1}(X), S_{2,2}(X)\rangle\big)
\)
have variance of order
\(o\!\left(k^2/(n^4 \rho_n^3)\right)\).
Since the number of summands in
\(\langle S_{1,1}(X), S_{2,2}(X)\rangle\)
is finite, the claim of the lemma follows.
\end{proof}

\section{Proofs of Lemmas from \cref{sec:estproof}}
In this section we prove the additional results from \cref{sec:estproof}.  

\subsection{Proof of \cref{lem:Very Good Set}} \label{sec:verygoodproof}

\begin{proof}[Proof of \cref{lem:Very Good Set}]
Write $\mathcal{E}_{22}=\{\, \|\hat{V}\|_{2,\infty} \lesssim \sqrt{k/n}\,\}$.  
By union bound,
\(
\mathbb{P}(\mathcal{E}_{{\sf very \ good}}^c)\le\mathbb{P}(\mathcal{E}_{{\sf good}}^c)+\mathbb{P}(\mathcal{E}_{22}^c).
\)
The first term was controlled in \cref{lem:Good Set}, namely
\(
\mathbb{P}(\mathcal{E}_{{\sf good}}^c)=O(n^{-19}).
\) To bound $\mathbb{P}(\mathcal{E}_{22}^c)$ we invoke Lemma~C.6 of \citet{agterberg_joint_2025} (where we have that $\theta_i \asymp \sqrt{\rho_n}$ for all $n$ as well as the fact that their result holds under the same assumptions we impose herein).  Under \cref{ass:asymptotic1} and \cref{ass:eigen-scaling} that lemma guarantees that, with probability at least $1-O(n^{-19})$, there exists an orthogonal alignment matrix $W_*$ such that
\(
\|\hat{V}-V W_*\|_{2,\infty}
\lesssim
\sqrt{\frac{\log n}{n\rho_n}}\|V\|_{2,\infty}.
\)
Combining this with the triangle inequality gives, with the same high probability,
\(
\|\hat{V}\|_{2,\infty}
\le \|\hat{V}-V W_*\|_{2,\infty} + \|V\|_{2,\infty}
\lesssim
\Big(1+\sqrt{\frac{\log n}{n\rho_n}}\Big)\|V\|_{2,\infty}.
\)
\cref{ass:incoherence} implies $\|V\|_{2,\infty}\lesssim\sqrt{k/n}$ and, since $\sqrt{\frac{\log n}{n\rho_n}}=o(1)$ under \cref{ass:sparse}, the right-hand side is $\lesssim\sqrt{k/n}$.  Hence
\(
\mathbb{P}(\mathcal{E}_{22}^c)=O(n^{-19}).
\)
Combining the two bounds yields
\(
\mathbb{P}(\mathcal{E}_{{\sf very \ good}})
\ge 1- O(n^{-19}),
\)
as claimed.
\end{proof}

\subsection{Proof of \cref{lem:series-expansion}}
\label{sec:series-expansionproof}
\begin{proof}[Proof of \cref{lem:series-expansion}]
We begin with the following identity for the difference between the true probability matrix and its rank-\(k\) approximation:
\begin{align*}
P - \widehat{A} = P - \widehat{V}\widehat{V}\t P \widehat{V}\widehat{V}\t - \widehat{V}\widehat{V}\t (A-P)\widehat{V}\widehat{V}\t.
\end{align*}
Invoking Theorem~1 of \citet{xia_normal_2021}, we can expand the projection operator \(\widehat{V}\widehat{V}\t\) as
\begin{align*}
\widehat{V}\widehat{V}\t = V V\t + \sum_{l \geq 1} S_{l}(X),
\end{align*}
where the operators \(S_{l}(X)\) are defined in \cref{sec:mainproof}. Substituting this expansion, we obtain
\begin{equation}
\begin{aligned}
\widehat{A} - P
&= \widehat{V}\widehat{V}\t P \widehat{V}\widehat{V}\t + \widehat{V}\widehat{V}\t X \widehat{V}\widehat{V}\t - P \\
&= \big(V V\t + \sum_{l \geq 1} S_{l}(X)\big)\,P\,\big(V V\t + \sum_{l \geq 1} S_{l}(X)\big)  + \big(V V\t + \sum_{l \geq 1} S_{l}(X)\big)\,X\,\big(V V\t + \sum_{l \geq 1} S_{l}(X)\big) - P.
\end{aligned}
\label{eq:series-expansion}
\end{equation}
Expanding and collecting terms, we obtain
\begin{align*}
\widehat{A} - P
&= \sum_{l \geq 1} S_{l}(X)\,P + \sum_{l \geq 1} P\,S_{l}(X) + \sum_{l_1,l_2 \geq 1} S_{l_1}(X)\,P\,S_{l_2}(X) \\
&\quad + V V\t X V V\t + \sum_{l \geq 1} S_{l}(X)\,X\,V V\t + \sum_{l \geq 1} V V\t X S_{l}(X) \\
&\quad + \sum_{l_1,l_2 \geq 1} S_{l_1}(X)\,X\,S_{l_2}(X).
\end{align*}
From this decomposition, it follows that the first-order term (i.e., the term linear in $X$) is
\begin{align*}
T_{1}(X) = V V\t X (I - V V\t) + (I-V V\t) X V V\t + V V\t X V V\t,
\end{align*}
while for general \(l \geq 2\),
\[
\begin{aligned}
T_{l}(X)
&= S_{l}(X)\,P + P\,S_{l}(X) + \sum_{l_1+l_2=l} S_{l_1}(X)\,P\,S_{l_2}(X) \\
&\quad + S_{P,l-1}(X)\,X\,V V\t + X V V\t S_{P,l-1}(X) \\
&\quad + \sum_{l_1+l_2=l-1} S_{l_1}(X)\,X\,S_{l_2}(X).
\end{aligned}
\]
By Theorem~1 of \cite{xia_normal_2021}, we have the operator norm bound
\(
\|S_{l}(X)\| \lesssim \left(\frac{\|X\|}{n\rho_n}\right)^{\!l}.
\)
Combining this with the decomposition above, we obtain
\begin{align*}
\|T_{l}(X)\|
&\lesssim (l+1)\,2^{\,l-1}\, (n\rho_n)\, \|S_{l}(X)\| + l\,2^{\,l-2}\,\|X\|\,\|S_{l}(X)\| \\
&\lesssim c_l \left(\frac{\|X\|}{n\rho_n}\right)^{\!l}\,(n\rho_n),
\end{align*}
for constants \(c_l>0\) such that \(\log c_l = O(1)\). This completes the proof.
\end{proof}

\subsection{Proof of \cref{inter1}}
\label{sec:interproof1}
\begin{proof}[Proof of \cref{inter1}]
In this proof, we will assume that \(A\) (in the one-sample case) and \({A^{(i)}}\) (in the two-sample case) both belong to $\mathcal{E}_{{\sf very  \ good}}$, as defined in \cref{lem:Very Good Set}, which happens with high probability as already proved.
\\ \ \\ \noindent
\begin{equation}
\begin{aligned} \label{intereq01}
\tr\bigl(VV\t\,\Delta M\bigr) &= \sum_{i,j} (V V\t)_{ii}(P_{ij}(1-P_{ij}) - {\hat P}_{ij}(1-{\hat P}_{ij}) )\\
&= \sum_{i,j} (V V\t)_{ii} (P_{ij}-{\hat P}_{ij}) + \sum_{i,j} (V V\t)_{ii} ({\hat P}_{ij}^2  - P_{ij}^2 ).
\end{aligned}
\end{equation}
We will start the proof by first proving that
\begin{align} \label{claim01}
    \bigg|\sum_{i,j} (V V\t)_{ii} (P_{ij}-{\hat P}_{ij})\bigg| = \bigg|\tr \left(VV \t \Diag(P - \hat P \right)\bigg| =O_p(\max (k,k \sqrt{\rho_n \log n})).
\end{align}
By \cref{lem:series-expansion}, we can expand $\Diag((\hat{P}-P)\cdot \mathbf{1}_n)$. We first use Bernstein’s inequality to bound the first–order term, and then control the combined contribution of higher–order terms.

Define
\(
u:=V V\t\cdot \mathbf{1}_n,  
y:=(I - V V\t)X u
\).
By \cref{lem:series-expansion}, the term of \(\Diag((\hat{P}-P)\cdot \mathbf{1}_n)\) that is linear in $X$, is  \(\Diag(((I-VV\t)XVV\t + X VV\t)\cdot \mathbf{1}_n)\). We will control the first term which is \( \Diag((I-VV\t)XVV\t\cdot \mathbf{1}_n)) = \Diag(y)\). The bound for the second term \( \Diag(XVV\t\cdot \mathbf{1}_n))\) follows in exactly the same way. Let $w:=\diag(V V\t)$ and set $\alpha:=w\t(I - V V\t)$. Then
\begin{equation}\label{T}
\begin{aligned}
\tau_1
&:=\tr\bigl(V V\t\Diag(y)\bigr)
=\sum_{i=1}^n (V V\t)_{ii} y_i
= w\t y
= w\t(I - V V\t)X u
= \alpha\t X u.
\end{aligned}
\end{equation}
Since $X$ is symmetric and centered with independent off–diagonal entries ${X_{ij}}$, $\tau_1$ can be rewritten as
\begin{equation}\label{T1}
\begin{aligned}
\tau_1
&=\sum_{i,j=1}^n \alpha_{i} {X_{ij}} u_j
=\sum_{i< j}\epsilon_{ij}{X_{ij}} + \sum_{i}\epsilon_{ii}x_{ii},
\end{aligned}
\end{equation}
where $\epsilon_{ii}:=\alpha_{i}u_i$ and $\epsilon_{ij}:=\alpha_{i}u_j+\alpha_{j}u_i$. To bound $\epsilon_{ij}$ and $\Var(\tau_1)$, note that
\begin{equation}\label{bds}
\|w\| \lesssim kn^{-1/2}, 
\quad
\|\alpha\|=\|w(I-V V\t)\| \lesssim kn^{-1/2}, 
\quad
\|u\|_{2}\lesssim \sqrt {kn},
\quad
\|u\|_{\infty}\lesssim \sqrt {k}.
\end{equation}

Next, we have $\|\alpha\|_\infty \le \|w\|_{\infty} + \|w\t V V\t\|_{\infty} = O\left( \frac{k}{n} \right)$ from \cref{ass:incoherence}. Hence
\(
|\epsilon_{ij}|
\le 2\|\alpha\|_\infty\|u\|_\infty
=O(k^2 n^{-1}).
\)
Therefore,
\begin{align*}
\Var(\tau_1)
=\sum_{i\le j}\epsilon_{ij}^2\,\Var({X_{ij}})
\lesssim\sum_{i\le j}(\alpha_{i}u_j+\alpha_{j}u_i)^2\rho_n
\lesssim k^2\rho_n,
\end{align*}
and moreover $\max_{i,j}|\epsilon_{ij}{X_{ij}}|\lesssim \frac{k^2}{n}$. By Bernstein’s inequality,
\begin{align*}
\Pr(|\tau_1|>t)
\lesssim \exp\Biggl(-\frac{t^2/2}{k^2\rho_n+k^2t/(3n)}\Biggr).
\end{align*}
So by choosing \(t \gtrsim k \sqrt{\rho_n \log n}\), we have
\begin{align}\label{claim02}
   \tau_1 = O_p \left( k\sqrt{\rho_n \log n}\right).
\end{align}
This establishes the bound on the first–order term for \(\tr(V V\t \Delta M)\).

The higher–order contributions are given by
\(
\sum_{l \ge 2} \tr\bigl(V V\t\,\Diag(T_{l}(X)\cdot \mathbf{1}_n)\bigr).
\)
For each \(l \ge 2\), we have
\(
\tr\bigl(V V\t\,\Diag(T_{l}(X)\cdot \mathbf{1}_n)\bigr)
  = \sum_{i=1}^{n} (V V\t)_{ii}\,\bigl(T_{l}(X)\cdot \mathbf{1}_n\bigr)_{i}
  \lesssim \sum_{i=1}^{n} \frac{k}{n}\,\bigl(T_{l}(X)\cdot \mathbf{1}_n\bigr)_{i},
\)
where the inequality follows from \cref{ass:incoherence}.  Now \(\sum_{i=1}^{n} \bigl(T_{l}(X)\cdot \mathbf{1}_n\bigr)_{i}\) is the sum of all entries in the vector $T_{l}(X)\cdot \mathbf{1}_n$. Thus \(\sum_{i=1}^{n} \bigl(T_{l}(X)\cdot \mathbf{1}_n\bigr)_{i} = \langle \mathbf{1}_n,(T_l(X) \mathbf{1}_n) \rangle \le n \|T_l(X)\|\).
Consequently,
\begin{align}\label{claim03}
    \sum_{l \ge 2} \tr\bigl(V V\t\,\Diag(T_{l}(X)\cdot \mathbf{1}_n)\bigr) \lesssim \frac{k}{n}\, \sum_{i=1}^{n} \bigl(T_{l}(X)\cdot \mathbf{1}_n\bigr)_{i}
\le
\sum_{l \ge 2} \|T_{l}(X)\|
= O_p(k),
\end{align}
where the final bound follows directly from \cref{lem:Good Set,lem:series-expansion}, which control the operator norms of the higher–order terms in the expansion. Thus \cref{claim02,claim03} together prove \cref{claim01}.

Next, we will bound the second term on the right hand side of \cref{intereq01} by proving that 
\begin{align}\label{claim04}
    \bigg|\sum_{i,j} (V V\t)_{ii} (P_{ij}^2-{\hat P}_{ij}^2)\bigg| = o_p(\max (k,k \sqrt{\rho_n \log n})).
\end{align}
We decompose the term further:
\begin{equation}
\begin{aligned} \label{intereq02}
    \bigg|\sum_{i,j} (V V\t)_{ii} (P_{ij}^2-{\hat P}_{ij}^2)\bigg| &= \bigg|\sum_{i,j} (V V\t)_{ii} P_{ij}(P_{ij} - \hat P_{ij}) - \sum_{i,j} (V V\t)_{ii}(P_{ij} - \hat P_{ij})^2\bigg| \\
    & \le \bigg|\sum_{i,j} (V V\t)_{ii} P_{ij}(P_{ij} - \hat P_{ij}) \bigg| + \bigg| \sum_{i,j} (V V\t)_{ii}(P_{ij} - \hat P_{ij})^2\bigg|.
\end{aligned}
\end{equation}
The second term on the right hand side of \cref{intereq02} can be bounded using \cref{lem:series-expansion}:
\begin{equation}
\begin{aligned}\label{intereq03}
    \bigg| \sum_{i,j} (V V\t)_{ii}(P_{ij} - \hat P_{ij})^2\bigg| &\le \max_i |(V V\t)_{ii}| \|P - \hat P\|_F^2 \\
    & \lesssim \frac{k^{3/2}}{n} \|P - \hat P\|_2^2 = O_p (k^{3/2} \rho_n) = o_p (k ).
\end{aligned}
\end{equation}
The first term on the right hand side of \cref{intereq02} can be expressed as $$ \bigg| \sum_{i,j} (V V\t)_{ii} P_{ij}(P_{ij} - \hat P_{ij}) \bigg| = \bigg| \tr(P \Diag(VV \t) (P - \hat P)) \bigg|.$$ We calculate its order in a manner very similar to what we did for the first term of \cref{intereq01} earlier in the proof. By \cref{lem:series-expansion}, we can expand $\Diag(\hat{P}-P)$. We first use Bernstein’s inequality to bound the first-order term, and then control the combined contribution of higher-order terms:
$$
    \tr(P \Diag(VV \t) (P - \hat P)) = \tr(P \Diag(VV \t) T_1(X)) + \sum_{l \ge 2} \tr(P \Diag(VV \t) T_l(X)).
$$
Substituting $T_1(X)$ from \cref{lem:series-expansion}, we see that
\begin{align} \label{intereq04}
    \tau_2 = \tr(P \Diag(VV \t) T_1(X)) = \tr(P \Diag(VV \t) X) = \sum_{i,j}(VV \t)_{ii} P_{ij} X_{ij}.
\end{align}
We can easily check that $\Var(\tau_2) \lesssim k^2 \rho_n^3$ from \cref{ass:asymptotic1,ass:incoherence} and $|(VV \t)_{ii} P_{ij} X_{ij}| \le \frac{k \rho_n}{n}$. By Bernstein’s inequality,
\begin{align*}
    \Pr(|\tau_2|>t) \lesssim \exp\Biggl(-\frac{t^2/2}{k^2\rho_n^3+k^2 t \rho_n/(3n)}\Biggr).
\end{align*}
Choosing $t \gtrsim k \sqrt{\rho_n^3 \log n}$, we obtain 
\begin{align}\label{claim05}
    \tau_2 = O_p \left( k\sqrt{\rho_n^3 \log n}\right) = o_p (k ).
\end{align}
The higher-order terms can be bounded using \cref{lem:series-expansion,lem:Good Set}:
\begin{equation}
\begin{aligned}\label{intereq04}
    \sum_{l \ge 2} \tr(P \Diag(VV \t) T_l(X)) &\le \sum_{l \ge 2} k^{3/2} \|P\| \|\Diag(VV \t)\| \|T_l(X)\| \\
    & =O_p \left( \sum_{l \ge 2} k^{5/2}  c_l \left(\frac{\|X\|}{n\rho_n}\right)^{\!l}\,n\rho_n^2 \right) = O_p \left( k^{5/2} \rho_n \right) = o_p(k ).
\end{aligned}
\end{equation}
Thus, \cref{claim05,intereq04,intereq03} together prove \cref{claim04}. Since we have proved both \cref{claim04,claim01}, \cref{inter-a} also follows.

An identical argument applies to \cref{inter-b}, where again we expand $\Delta\Sigma$ and proceed analogously to the steps used for \cref{inter-a}.
\end{proof}
\subsection{Proof of \cref{inter3}}
\label{sec:interproof3}
\begin{proof}[Proof of \cref{inter3}]
In this proof, we will assume that \(A\) (in the one-sample case) and \({A^{(i)}}\) (in the two-sample case) both belong to $\mathcal{E}_{{\sf very  \ good}},$ as defined in \cref{lem:Very Good Set}, which happens with high probability as already proved.
\\ \ \\ 
\noindent
\textbf{Proof of \cref{inter-d1}.} Just as in the proof of \cref{inter-a}, it suffices to show
\begin{align} \label{claim06}
    \bigg|\tr\Bigl((V\Lambda^{-2}V\t-\widehat V\widehat\Lambda^{-2}\widehat V\t)\,\Diag({\hat P})\Bigr)\bigg|
    = O_p \left(\max \left(\frac{k}{n^3 \rho_n^2},\frac{k \sqrt{\log n}}{n^3 \rho_n^{3/2}}\right)\right),
\end{align}
because then we can similarly show that 
\begin{align*}
    \bigg|\tr\Bigl((V\Lambda^{-2}V\t-\widehat V\widehat\Lambda^{-2}\widehat V\t)\,\Diag({\hat P} \circ {\hat P})\Bigr)\bigg|
    = o_p \left(\max \left(\frac{k}{n^3 \rho_n^2},\frac{k \sqrt{\log n}}{n^3 \rho_n^{3/2}}\right)\right).
\end{align*}
To that end we first work with $D:=\Diag(P)$ and prove
\[
\bigg|\tr\Bigl((V\Lambda^{-2}V\t-\widehat V\widehat\Lambda^{-2}\widehat V\t)\,\Diag(P)\Bigr)\bigg|
=O_p \left(\max \left(\frac{k}{n^3 \rho_n^2},\frac{k \sqrt{\log n}}{n^3 \rho_n^{3/2}}\right)\right).
\]
Using the algebraic identity
\begin{align}\label{dseries}
\widehat V\widehat\Lambda^{-2}\widehat V^\top - V\Lambda^{-2}V^\top
= \big(\widehat V\widehat V^\top - VV^\top\big)\widehat V\widehat\Lambda^{-2}\widehat V^\top
+ V\big(V^\top\widehat V\widehat\Lambda^{-2}-\Lambda^{-2}V^\top\widehat V\big)\widehat V^\top
+ V\Lambda^{-2}V^\top\big(\widehat V\widehat V^\top - VV^\top\big),
\end{align}
and observing that
\[
\begin{aligned}
V\big(V^\top\widehat V\widehat\Lambda^{-2}-\Lambda^{-2}V^\top\widehat V\big)\widehat V^\top
&=V\Lambda^{-2}\big(\Lambda^2V^\top\widehat V - V^\top\widehat V\widehat\Lambda^2\big)\widehat\Lambda^{-2}\widehat V^\top\\
&=V\Lambda^{-2}\big(V^\top P^2\widehat V - V^\top A^2\widehat V\big)\widehat\Lambda^{-2}\widehat V^\top\\
&=V\Lambda^{-2}V^\top\big(P^2-A^2\big)\widehat V\widehat\Lambda^{-2}\widehat V^\top\\
&=-V\Lambda^{-2}V^\top\big(PX+XP+X^2\big)\widehat V\widehat\Lambda^{-2}\widehat V^\top,
\end{aligned}
\]
we expand the leading (first–order) contributions using the series representation
\(
\widehat V\widehat V\t - VV\t = S_1(X) + \sum_{j\ge 2} S_j(X),
\)
with
\(
S_1(X)=\beta^{\perp}X\beta^{-1}+\beta^{-1}X\beta^{\perp}.
\)
Collecting the principal terms and using \cref{dseries} yields the following bound:
\begin{equation}\label{dseries1}
\begin{aligned}
\bigg|\tr&\Bigl((\widehat V\widehat\Lambda^{-2}\widehat V\t-V\Lambda^{-2}V\t)\,\Diag(P)\Bigr)\bigg| \\
&\le \Big|\tr\Big(\big[V\Lambda^{-2}V\t S_1(X)+S_1(X)V\Lambda^{-2}V\t -\,V\Lambda^{-2}V\t(XP+PX)V\Lambda^{-2}V\t\big]\Diag(P)\Big)\Big|\\[4pt]
&\quad+\Big|\tr\Big(S_1(X)\big(\widehat V\widehat\Lambda^{-2}\widehat V\t-V\Lambda^{-2}V\t\big)\Diag(P)\Big)\Big|\\[4pt]
&\quad+\Big|\tr\Big(V\Lambda^{-2}V\t(XP+PX)\big(\widehat V\widehat\Lambda^{-2}\widehat V\t-V\Lambda^{-2}V\t\big)\Diag(P)\Big)\Big|\\[4pt]
&\quad+\Big|\tr\Big(\big(\sum_{j\ge2}V\Lambda^{-2}V\t S_j(X)+\sum_{j\ge2}S_j(X)\widehat V\widehat\Lambda^{-2}\widehat V\t-\,V\Lambda^{-2}V\t X^2\widehat V\widehat\Lambda^{-2}\widehat V\t\big)\Diag(P)\Big)\Big|.
\end{aligned}
\end{equation}
The first term on the right hand side of \cref{dseries1} corresponds precisely to those components in the expansion \cref{dseries} that are linear in the noise matrix \(X\). We bound these contributions using a standard Bernstein–type concentration argument.

Consider first the term
\(
\tr\left(V\Lambda^{-2}V\t S_{1}(X)\,\Diag(P)\right).
\)
By direct expansion, this term can be written as
\[
\tr\left(V\Lambda^{-2}V\t S_{1}(X)\,\Diag(P)\right)
=\sum_{t,q=1}^{n} w_{tq}\,x_{tq}
=\sum_{t>q}(w_{tq}+w_{qt})x_{tq}+\sum_{t} w_{tt}x_{tt},
\]
where the deterministic weights \(w_{tq}\) are given by
\(
w_{tq}
=\sum_{s=1}^{k}\sum_{r=1}^{n}
\frac{1}{\lambda_{ss}^{3}}\,
v_{rs}v_{ts}\,(I-VV\t)_{qr}\,p_{rr}.
\)
Under \cref{ass:asymptotic1,ass:eigen-scaling,ass:incoherence}, these coefficients satisfy
\[|w_{tq}|  \lesssim  \bigg| \sum_{s=1}^{k}
\frac{1}{\lambda_{ss}^{3}}\,
v_{qs}v_{ts}\,p_{qq} \bigg| \lesssim \frac{k}{n^{4}\rho_{n}^{2}}.\]
Moreover, since the entries \(\{x_{tq}\}\) are independent (up to symmetry) with variance of order \(\rho_n\), we have
\(
\Var\left(\sum_{t,q=1}^{n} w_{tq}x_{tq}\right)
\lesssim\frac{k^{2}}{n^{6}\rho_{n}^{3}}.
\)
Applying Bernstein’s inequality yields, for any \(z>0\),
\[
\mathbb{P}\left(
\left|\sum_{t,q=1}^{n} w_{tq}x_{tq}\right|\ge z
\right)
\le
2\exp\left(
-\frac{z^{2}/2}{
\frac{k^{2}}{n^{6}\rho_{n}^{3}}+\frac{kz}{3n^{4}\rho_{n}^{2}}
}
\right).
\]
Taking \(z\asymp \frac{k \sqrt {\log n}}{n^{3}\rho_{n}^{3/2}}\) gives
\begin{align}\label{claim001}
    \tr\left(V\Lambda^{-2}V\t S_{1}(X)\,\Diag(P)\right)
=O_{p}\left(\frac{k\sqrt {\log n}}{n^{3}\rho_{n}^{3/2}}\right).
\end{align}
A completely analogous argument applies to
\(
\tr\left(V\Lambda^{-2}V\t X V\Lambda^{-1}V\t\,\Diag(P)\right),
\)
which can again be written in the form
\(
\sum_{t,q=1}^{n} w_{tq}x_{tq}
=\sum_{t>q}(w_{tq}+w_{qt})x_{tq}+\sum_{t} w_{tt}x_{tt},
\)
with coefficients
\[
w_{tq}
=\sum_{s,j=1}^{k}\sum_{r=1}^{n}
\frac{1}{\lambda_{ss}^{2}\lambda_{jj}}\,
v_{rs}v_{ts}v_{qj}v_{rj}\,p_{rr}.
\]
Using the same approach one obtains
\begin{align}\label{claim002}
   \tr\left(V\Lambda^{-2}V\t X V\Lambda^{-1}V\t\,\Diag(P)\right)
=O_{p}\left(\frac{k\sqrt {\log n}}{n^{3}\rho_{n}^{3/2}}\right). 
\end{align}
The remaining linear terms,
\(
\tr\left(S_{1}(X)V\Lambda^{-2}V\t\,\Diag(P)\right)\)
and \(
\tr\left(V\Lambda^{-1}V\t X V\Lambda^{-2}V\t\,\Diag(P)\right),
\)
are handled in exactly the same manner. Collecting all such contributions from \cref{claim001,claim002}, we conclude that the total first–order (linear in \(X\)) component in \cref{dseries1} is
\(
O_{p}\left(\frac{k\sqrt {\log n}}{n^{3}\rho_{n}^{3/2}}\right).
\)

Bounding the final term in \cref{dseries1} is straightforward. We first establish a high–probability bound for
\[
\Bigg|\tr\Bigg(\sum_{j\ge 2} V\Lambda^{-2}V\t S_j(X)\,\Diag(P)\Bigg)\Bigg|.
\]
By \cref{lem:Good Set}, on the corresponding high–probability event, we have
\(
\|S_j(X)\| \lesssim \frac{1}{(n\rho_n)^{j/2}},\) for \(j\ge 2.
\)
Consequently,
\begin{equation}
\begin{aligned}\label{dseries4}
\Bigg|\tr\Bigg(\sum_{j\ge 2} V\Lambda^{-2}V\t S_j(X)\,\Diag(P)\Bigg)\Bigg|
&\le \sum_{j\ge 2} \Big|\tr\big(V\Lambda^{-2}V\t S_j(X)\,\Diag(P)\big)\Big| \\
&\le k\sum_{j\ge 2} \|\Diag(P)V\Lambda^{-2}V\t\| \, \|S_j(X)\| \\
&=O_p \left( \sum_{j\ge 2}
k \frac{\rho_n}{n^2\rho_n^2}\,
\frac{4^{j/2}}{(n\rho_n)^{j/2}} \right) \\
&=O_p \left( \frac{k}{n^3\rho_n^2}\right).
\end{aligned}
\end{equation}
Applying an identical argument, we have
\[
\Bigg|\tr\Bigg(\sum_{j\ge 2} S_j(X) \hat V \hat \Lambda^{-2} \hat V\t\,\Diag(P)\Bigg)\Bigg| =O_p \left( \frac{k}{n^3 \rho_n^2} \right).
\]
Finally, consider the term
\(
\tr\Big(V\Lambda^{-2}V\t X^2\widehat V\widehat\Lambda^{-2}\widehat V\t\,\Diag(P)\Big).
\)
By the Cauchy–Schwarz inequality for the Frobenius inner product,
\begin{equation}
\begin{aligned}\label{dseries5}
\Big|
\tr\!\Big(
V\Lambda^{-2}V^{\top} X^{2}\,
\widehat V \widehat\Lambda^{-2}\widehat V^{\top}\,
\Diag(P)
\Big)
\Big|
&\le
k\,
\big\|V\Lambda^{-2}V^{\top} X^{2}\big\|\,
\big\|\widehat V \widehat\Lambda^{-2}\widehat V^{\top}\Diag(P)\big\| \\
&\le
k\,
\|X\|^{2}\,
\|V\Lambda^{-2}V^{\top}\|\,
\big\|\widehat V \widehat\Lambda^{-2}\widehat V^{\top}\Diag(P)\big\| \\
&=O_p \left(
k\,
(n\rho_n)\,
\frac{1}{n^{2}\rho_n^{2}}\,
\frac{\rho_n}{n^{2}\rho_n^{2}} \right) = O_p \left(
\frac{k^{2}}{n^{3}\rho_n^{2}} \right),
\end{aligned}
\end{equation}
where the final bound follows from \cref{lem:Very Good Set,ass:eigen-scaling,ass:asymptotic1}. Combining the above estimates, we conclude that the last term on the right hand side of \cref{dseries1} is of order
\(
O_p\left(\frac{k}{n^3\rho_n^2}\right).
\)

We now turn to bounding the second term on the right hand side of \cref{dseries1}. To this end, we once again decompose 
\(\widehat V\widehat\Lambda^{-2}\widehat V^\top - V\Lambda^{-2}V^\top\) using the expansion in \cref{dseries}. This yields
\begin{equation}
\begin{aligned}\label{dseries2}
\Big|
\tr\Big(
S_1(X)\big(\widehat V\widehat\Lambda^{-2}\widehat V^\top - V\Lambda^{-2}V^\top\big)\Diag(P)
\Big)
\Big|
&\le
\Big|
\tr\Big(
S_1(X)\big(\widehat V\widehat V^\top - VV^\top\big)
\widehat V\widehat\Lambda^{-2}\widehat V^\top \Diag(P)
\Big)
\Big| \\
&\quad +
\Big|
\tr\Big(
S_1(X) V\big(V^\top\widehat V\widehat\Lambda^{-2}-\Lambda^{-2}V^\top\widehat V\big)
\widehat V^\top \Diag(P)
\Big)
\Big| \\
&\quad +
\Big|
\tr\Big(
S_1(X) V\Lambda^{-2}V^\top
\big(\widehat V\widehat V^\top - VV^\top\big)\Diag(P)
\Big)
\Big| \\
&\le
\sum_{j\ge 1}
\Big|
\tr\Big(
S_1(X) S_j(X)\widehat V\widehat\Lambda^{-2}\widehat V^\top \Diag(P)
\Big)
\Big| \\
&\quad +
\Big|
\tr\Big(
S_1(X) V\Lambda^{-2}V^\top
\big(PX+XP+X^2\big)
\widehat V\widehat\Lambda^{-2}\widehat V^\top \Diag(P)
\Big)
\Big| \\
&\quad +
\sum_{j\ge 1}
\Big|
\tr\Big(
S_1(X) V\Lambda^{-2}V^\top S_j(X)\Diag(P)
\Big)
\Big|.
\end{aligned}
\end{equation}
All terms appearing on the right hand side of \cref{dseries2} can be controlled using
\cref{ass:asymptotic1,ass:dense,ass:eigen-scaling,ass:incoherence}, together with the high–probability bounds established in \cref{lem:Very Good Set}. Indeed,
\begin{align*}
\sum_{j\ge 1}
\Big|
\tr\Big(
S_1(X) S_j(X)\widehat V\widehat\Lambda^{-2}\widehat V^\top \Diag(P)
\Big)
\Big|
&\le
\sum_{j\ge 1}
k\|S_1(X)\| \, \|S_j(X)\| \,
\|\widehat V\widehat\Lambda^{-2}\widehat V^\top \Diag(P)\| \\
&= O_p \left(
\sum_{j\ge 1}
k \frac{\rho_n}{n^{2.5}\rho_n^{2.5}}\,
\frac{4^{j/2}}{(n\rho_n)^{j/2}}\right)\\
&=
O_p \left(\frac{k}{n^3\rho_n^2} \right),
\end{align*}
where the second line follows similarly as in \cref{dseries4}. An identical bound holds for
\[
\sum_{j\ge 1}
\Big|
\tr\Big(
S_1(X) V\Lambda^{-2}V^\top S_j(X)\Diag(P)
\Big)
\Big|
\]
by the same argument. Moreover, \cref{dseries5} yields the same order of magnitude for
\[
\Big|
\tr\Big(
S_1(X) V\Lambda^{-2}V^\top
\big(PX+XP+X^2\big)
\widehat V\widehat\Lambda^{-2}\widehat V^\top \Diag(P)
\Big)
\Big|.
\]
Combining these bounds, we conclude that the second term on the right hand side of \cref{dseries1} is of order
\(
O_p\left(\frac{k}{n^3\rho_n^2}\right).
\)

By an entirely analogous argument, one can show that the third term on the right hand side of \cref{dseries1} is also
\(
O_p\left(\frac{k}{n^3\rho_n^2}\right).
\)
Specifically, one again decomposes
\(\widehat V\widehat\Lambda^{-2}\widehat V^\top - V\Lambda^{-2}V^\top\) using \cref{dseries} and applies
\cref{ass:asymptotic1,ass:dense,ass:eigen-scaling,ass:incoherence} together with \cref{lem:Very Good Set}.
The resulting calculations mirror those used for the second term and therefore yield the stated bound.

\medskip
We now turn to the deviation term involving $\Diag(\widehat P)$. Observe that
\begin{equation}
\begin{aligned}\label{claim07}
\tr\bigl((\widehat V \widehat \Lambda^{-2} \widehat V^\top - V \Lambda^{-2} V^\top)\Diag(\widehat P)\bigr)
&=
\tr\bigl((\widehat V \widehat \Lambda^{-2} \widehat V^\top - V \Lambda^{-2} V^\top) \Diag(P)\bigr)
+
\tr\bigl((\widehat V \widehat \Lambda^{-2} \widehat V^\top - V \Lambda^{-2} V^\top)( \Diag(\widehat P - P)\bigr) \\
&=
O_p \left(\max \left(\frac{k}{n^3 \rho_n^2},\frac{k \sqrt{\log n}}{n^3 \rho_n^{3/2}}\right)\right)
+
\tr\bigl((\widehat V \widehat \Lambda ^{-2}\widehat V^\top - V \Lambda^{-2} V^\top)\Diag({\widehat P} - P)\bigr),
\end{aligned}
\end{equation}
where the bound for the first term follows from \cref{claim06}. It therefore suffices to control the second term. To this end, we invoke both the decomposition in \cref{dseries} and the series expansion in \cref{lem:series-expansion}, and obtain
\begin{equation}
\begin{aligned}\label{dseries6}
\Big|
\tr\bigl((\widehat V \widehat \Lambda^{-2} \widehat V^\top - V \Lambda^{-2} V^\top)\Diag({\widehat P} - P)\bigr)
\Big|
&\le
\Big|
\tr\bigl((\widehat V\widehat V^\top - VV^\top)
\widehat V\widehat\Lambda^{-2}\widehat V^\top
\Diag({\widehat P} - P)\bigr)
\Big| \\
&\quad +
\Big|
\tr\bigl(
V\big(V^\top\widehat V\widehat\Lambda^{-2}-\Lambda^{-2}V^\top\widehat V\big)
\widehat V^\top
\Diag({\widehat P} - P)\bigr)
\Big| \\
&\quad +
\Big|
\tr\bigl(
V\Lambda^{-2}V^\top(\widehat V\widehat V^\top - VV^\top)
\Diag({\widehat P} - P)\bigr)
\Big| \\
&\le
\sum_{l\ge 1}\sum_{j\ge 1}
\Big|
\tr\bigl(
S_j(X)\widehat V\widehat\Lambda^{-2}\widehat V^\top
\Diag(T_l(X))\bigr)
\Big| \\
&\quad +
\sum_{l\ge 1}
\Big|
\tr\bigl(
V\Lambda^{-2}V^\top(PX+XP+X^2)
\widehat V\widehat\Lambda^{-2}\widehat V^\top
\Diag(T_l(X))\bigr)
\Big| \\
&\quad +
\sum_{l\ge 1}\sum_{j\ge 1}
\Big|
\tr\bigl(
V\Lambda^{-2}V^\top S_j(X)\Diag(T_l(X))\bigr)
\Big|.
\end{aligned}
\end{equation}
We now derive high–probability bounds for the terms on the right hand side of \cref{dseries6}. Using \cref{ass:asymptotic1,ass:dense,ass:eigen-scaling,ass:incoherence} together with the event in \cref{lem:Very Good Set} and the series representation in \cref{lem:series-expansion}, we obtain
\begin{align*}
\sum_{l\ge 1}\sum_{j\ge 1}
\Big|
\tr\bigl(
S_j(X)\widehat V\widehat\Lambda^{-2}\widehat V^\top
\Diag(T_l(X))\bigr)
\Big|
&\le k
\sum_{l,j\ge 1}
\|S_j(X)\| \,
\|\widehat V\widehat\Lambda^{-2}\widehat V^\top\Diag(T_l(X))\| \\
&= O_p \left(
\sum_{l,j\ge 1}
\frac{k}{(n\rho_n)^2}
\frac{4^{j/2}}{(n\rho_n)^{j/2}}
\frac{4^{l/2}}{(n\rho_n)^{l/2}}
\rho_n \right)
= O_p \left(
\frac{k}{n^3\rho_n^2} \right).
\end{align*}
The final inequality follows from the bound
\(
(T_l(X))_{ij}\lesssim
\rho_n\Big(\frac{\|X\|}{n\rho_n}\Big)^{l},
\)
which holds by the definition of $T_l(X)$ in \cref{lem:series-expansion}. The remaining two terms in \cref{dseries6} can be bounded in an identical manner by applying the same norm inequalities and high–probability controls. Combining all these bounds for the four terms in \cref{dseries1} and \cref{claim07} completes the proof of \cref{claim06}.
\\ \ \\
\textbf{Proof of \cref{inter-d2}}
This bound follows directly from \cref{inter-d1}. Indeed,
\begin{align*}
\big|
\tr\bigl((V \Lambda^{-2} V\t - \widehat V \widehat \Lambda^{-2} \widehat V\t)\,
\Diag(\widehat{\Sigma}\cdot \mathbf{1}_n)\bigr)
\big|
&\lesssim
n\,\big|
\tr\bigl((V \Lambda^{-2} V\t - \widehat V \widehat \Lambda^{-2} \widehat V\t)\,
\Diag(\widehat{\Sigma}\bigr)
\big| \\
&=O_p \left(\max \left(\frac{k}{n^2 \rho_n^2},\frac{k \sqrt{\log n}}{n^2 \rho_n^{3/2}}\right)\right),
\end{align*}
which completes the proof.
\end{proof}
\subsection{Proof of \cref{inter2}}
\label{sec:interproof2}
\begin{proof}[Proof of \cref{inter2}]
Just like in \cref{sec:interproof1}, we will assume that \(A\) (in the one-sample case) and \({A^{(i)}}\) (in the two-sample case) both belong to  $\mathcal{E}_{{\sf very  \ good}},$ as defined in \cref{lem:Very Good Set}, which happens with high probability as already proved.
\\ \ \\ \noindent
\textbf{Proof of \cref{inter-c}}
We have
\(
\alpha=\diag(\beta^{-2})
\)
and
\(
\widehat{\alpha}=\diag(\widehat{\beta}^{-2}).
\)
Once again, we invoke \cref{dseries} to obtain the bound
\begin{equation}\label{dseriesnew1}
\begin{aligned}
\bigg\|&
\diag\Bigl(
\widehat V \widehat\Lambda^{-2}\widehat V^{\top}
-
V\Lambda^{-2}V^{\top}
\Bigr)
\bigg\| \\
&\le
\Big\|
\diag\Big(
V\Lambda^{-2}V^{\top} S_1(X)
+
S_1(X)V\Lambda^{-2}V^{\top}  
-
V\Lambda^{-2}V^{\top}(XP+PX)V\Lambda^{-2}V^{\top}
\Big)
\Big\| \\[4pt]
&\quad+
\Big\|
\diag\Big(
S_1(X)
\big(
\widehat V \widehat\Lambda^{-2}\widehat V^{\top}
-
V\Lambda^{-2}V^{\top}
\big)
\Big)
\Big\| \\[4pt]
&\quad+
\Big\|
\diag\Big(
V\Lambda^{-2}V^{\top}(XP+PX)
\big(
\widehat V \widehat\Lambda^{-2}\widehat V^{\top}
-
V\Lambda^{-2}V^{\top}
\big)
\Big)
\Big\| \\[4pt]
&\quad+
\Big\|
\diag\Big(
\sum_{j\ge2} V\Lambda^{-2}V^{\top} S_j(X)
+
\sum_{j\ge2} S_j(X)\widehat V \widehat\Lambda^{-2}\widehat V^{\top}  -
V\Lambda^{-2}V^{\top} X^2 \widehat V \widehat\Lambda^{-2}\widehat V^{\top}
\Big)
\Big\|.
\end{aligned}
\end{equation}
Consider first the vector
\(
r
=
\diag\left(
V\Lambda^{-2}V^{\top} S_{1}(X)
\right)
=
\diag\left(
\beta^{-3} X \beta^{\perp}
\right).
\)
For each \(i\),
\(
r_i
=
u_i^{\top} \Lambda^{-3} w_i,
\)
where
\(u_i = V^{\top} e_i\)
and
\(w_i = V^{\top} X \beta^{\perp} e_i\).
By the Cauchy--Schwarz inequality,
\(
|r_i|^2
\le
\|\Lambda^{-3}\|_2^2
\|u_i\|_2^2
\|w_i\|_2^2.
\)
Consequently,
\[
\|r\|_2^2
=
\sum_{i=1}^n r_i^2
\le
\|\Lambda^{-3}\|_2^2
\sum_{i=1}^n
\|u_i\|_2^2
\|w_i\|_2^2
\lesssim
\frac{k}{n}
\|\Lambda^{-3}\|_2^2
\,
\|V^{\top} X \beta^{\perp}\|_F^2,
\]
where the final inequality follows from the incoherence assumption in \cref{ass:incoherence}.

Using standard norm inequalities,
\(
\|V^{\top} X \beta^{\perp}\|_F
\le
\|V\|_F
\|X\|_2
\|\beta^{\perp}\|_2
\le
\sqrt{k}\,\|X\|_2.
\)
Therefore,
\(
\|r\|_2
\lesssim
\frac{k}{\sqrt{n}}
\|X\|_2
\|\Lambda^{-3}\|_2.
\)
On the event $\mathcal{E}_{{\sf good}}$, together with the eigen-scaling assumption in \cref{ass:eigen-scaling}, we conclude that
\(
\|r\|_2
= O_p \left(
\frac{k}{n^{3}\rho_n^{2.5}}\right).
\)
Similarly we can also show that \(\|\diag\left(S_{1}(X)
V\Lambda^{-2}V^{\top}
\right)\| =O_p \left( \frac{k}{n^{3}\rho_n^{2.5}}\right)\).

Next, consider the vector
\(
r=\diag\left(\beta^{-2} X \beta^{-1}\right),
\)
for which, for each \(i\),
\(
r_i
= (V^{\top} e_i)^{\top}\Lambda^{-2}(V^{\top}X\beta^{-1}e_i).
\)
Proceeding exactly as in the previous argument, using Cauchy-Schwarz inequality, the incoherence condition in \cref{ass:incoherence}, standard norm inequalities,  the bound on event $\mathcal{E}_{{\sf good}}$ as in \cref{lem:Good Set}, together with the eigenvalue scaling in \cref{ass:eigen-scaling}, we obtain
\(
\|\diag(V\Lambda^{-2}V^{\top} X P V\Lambda^{-2}V^{\top})\|
=O_p \left( \frac{k}{n^{3}\rho_n^{2.5}} \right).
\)
The same bound holds for
\(
\|\diag(V\Lambda^{-2}V^{\top} P X V\Lambda^{-2}V^{\top})\|.
\)
Consequently, the entire first term on the right-hand side of \cref{dseriesnew1} is bounded by
\(
O_p \left(\frac{k}{n^{3}\rho_n^{2.5}}\right) .
\)

We now turn to bounding the second term in \cref{dseriesnew1}. To this end, we once again decompose 
\(\widehat V\widehat\Lambda^{-2}\widehat V^\top - V\Lambda^{-2}V^\top\) using the expansion in \cref{dseries}. This yields
\begin{equation}
\begin{aligned}\label{dseriesnew2}
\Big|
\diag\Big(
S_1(X)\big(\widehat V\widehat\Lambda^{-2}\widehat V^\top - V\Lambda^{-2}V^\top\big)
\Big)
\Big|
&\le
\Big|
\diag\Big(
S_1(X)\big(\widehat V\widehat V^\top - VV^\top\big)
\widehat V\widehat\Lambda^{-2}\widehat V^\top 
\Big)
\Big| \\
&\quad +
\Big|
\diag\Big(
S_1(X) V\big(V^\top\widehat V\widehat\Lambda^{-2}-\Lambda^{-2}V^\top\widehat V\big)
\widehat V^\top 
\Big)
\Big| \\
&\quad +
\Big|
\diag\Big(
S_1(X) V\Lambda^{-2}V^\top
\big(\widehat V\widehat V^\top - VV^\top\big)
\Big)
\Big| \\
&\le
\sum_{j\ge 1}
\Big|
\diag\Big(
S_1(X) S_j(X)\widehat V\widehat\Lambda^{-2}\widehat V^\top 
\Big)
\Big| \\
&\quad +
\Big|
\diag\Big(
S_1(X) V\Lambda^{-2}V^\top
\big(PX+XP+X^2\big)
\widehat V\widehat\Lambda^{-2}\widehat V^\top 
\Big)
\Big| \\
&\quad +
\sum_{j\ge 1}
\Big|
\diag\Big(
S_1(X) V\Lambda^{-2}V^\top S_j(X)
\Big)
\Big|.
\end{aligned}
\end{equation}
We first focus on the first term
\(
r_j
=\diag\Big(
S_1(X)\, S_j(X)\,\widehat V\widehat\Lambda^{-2}\widehat V^{\top}
\Big).
\)
For each \(i\), we may write
\[
(r_{1j})_i
= u_i^{\top}\widehat\Lambda^{-2} w_i,
\qquad
u_i := e_i^{\top} S_1(X) S_j(X) \widehat V,
\quad
w_i := \widehat V^{\top} e_i .
\]
By Cauchy-Schwarz inequality,
\(
(r_{1j})_i^2 \le
\|\widehat\Lambda^{-2}\|_2^{\,2}\,
\|u_i\|_2^{\,2}\,
\|w_i\|_2^{\,2}.
\)
Summing over \(i\) and using the incoherence of \(\widehat V\) from \cref{lem:Very Good Set}, we obtain
\begin{equation}\label{side1}
\begin{aligned}
\|r_{1j}\|
&\lesssim
\frac{\sqrt{k}}{\sqrt{n}}\,
\|\widehat\Lambda^{-2}\|\,
\|S_1(X) S_j(X)\widehat V\|_F \\[4pt]
&\lesssim
\frac{k}{\sqrt{n}}\,
\|\widehat\Lambda^{-2}\|\,
\|S_1(X) S_j(X)\| \\[4pt]
&\lesssim
\frac{k}{\sqrt{n}}\,
\|\widehat\Lambda^{-2}\|\,
\|S_1(X)\|\,
\|S_j(X)\| \\[4pt]
&=O_p \left(
\frac{k}{\sqrt{n}}\,
\frac{1}{n^{2}\rho_n^{2}}
\left(\frac{4}{n\rho_n}\right)^{(j+1)/2}\right).
\end{aligned}
\end{equation}
Summing the bound in \cref{side1} over \(j\ge 1\) yields
\[
\sum_{j\ge 1}
\Big\|
\diag\Big(
S_1(X)\, S_j(X)\,\widehat V\widehat\Lambda^{-2}\widehat V^{\top}
\Big)
\Big\|
=O_p \left(
\frac{k}{n^{3.5}\rho_n^{3}}\right) .
\]
Similarly, consider
\(
r_{2j}
=\diag\Big(
S_1(X)\, V \Lambda^{-2} V^{\top} S_j(X)
\Big).
\)
For each \(i\), we may write
\[
(r_{2j})_i
=
\big(e_i\t S_1(X) \big)\,
\big(V \Lambda^{-2} V^{\top}\big)\,
\big(S_j(X) e_i\big).
\]
Proceeding analogously to the previous case and invoking Cauchy--Schwarz inequality together with the incoherence of \(V\), we obtain

\begin{equation}\label{side2}
\begin{aligned}
\|r_{2j}\| &\le \|\Lambda^{-2}\|\, \left( \max_{1 \le i \le n} \|e_i^\top S_1(X) V\|_2 \right)\, \|V^\top S_j(X)\|_F \\
&\le
\frac{1}{\sqrt{n}}\,
\|\Lambda^{-2}\|\,
\|V^{\top} S_j(X)\|_F\,
\|S_1(X) V\|_F \\[4pt]
&\le
\frac{k}{n^{2.5}\rho_n^{2}}\,
\|S_1(X)\|\,
\|S_j(X)\| \\[4pt]
&=O_p \left(
\frac{k}{n^{2.5}\rho_n}\,
\left(\frac{4}{n\rho_n}\right)^{(j+1)/2}\right).
\end{aligned}
\end{equation}
Summing the bound in \cref{side2} over \(j\ge 1\) yields
\[
\sum_{j\ge 1}
\Big\|
\diag\Big(
S_1(X)\, V \Lambda^{-2} V^{\top} S_j(X)
\Big)
\Big\|
=O_p \left(
\frac{k}{n^{3.5}\rho_n^{3}}\right) .
\]
A completely analogous argument applies to the middle term in \cref{dseriesnew2},
\[
\diag\Big(
S_1(X)\, V\Lambda^{-2}V^{\top}
\big(PX+XP+X^2\big)
\widehat V\widehat\Lambda^{-2}\widehat V^{\top}
\Big),
\]
and yields the same rate. Consequently, the entire second term in \cref{dseriesnew2} is bounded by
\(
O_p\left(\frac{k}{n^{3.5}\rho_n^{3}}\right).
\)

By an entirely analogous argument, one can show that the third term in \cref{dseriesnew1} is also
\(
O_p\left(\frac{k^{3/2}}{n^{3.5}\rho_n^{3}}\right).
\)
Specifically, one again decomposes
\(\widehat V\widehat\Lambda^{-2}\widehat V^\top - V\Lambda^{-2}V^\top\) using \cref{dseries} and applies
\cref{ass:asymptotic1,ass:dense,ass:eigen-scaling,ass:incoherence} together with \cref{lem:Very Good Set}.
The resulting calculations mirror those used for the second term and therefore yield the stated bound.

For the final term in \cref{dseriesnew1}, we first focus on
\(
r_{3j}
=
\diag\Big(
V\Lambda^{-2}V^{\top} S_j(X)
\Big).
\)
For each \(i\), we may write
\[
(r_{3j})_i
=
u_i^{\top}\Lambda^{-2} w_i,
\qquad
u_i = V^{\top} e_i,
\quad
w_i = V^{\top} S_j(X) e_i .
\]
Proceeding along the same lines as in \cref{side1}, and using Cauchy--Schwarz inequality together with the incoherence of \(V\), we obtain a bound on each \(\|r_{3j}\|\). Summing over \(j\ge 2\) then yields
\(
\sum_{j\ge 2} \|r_{3j}\|
=
O_p\!\left(\frac{k}{n^{3.5}\rho_n^{3}}\right).
\)

The same argument applies to the remaining two terms in the final line of \cref{dseriesnew1}, namely
\[
\sum_{j\ge 2}
\diag\Big(
S_j(X)\widehat V \widehat\Lambda^{-2}\widehat V^{\top}
\Big)
\quad\text{and}\quad
\diag\Big(
V\Lambda^{-2}V^{\top} X^2 \widehat V \widehat\Lambda^{-2}\widehat V^{\top}
\Big),
\]
and both admit the bound
\(
O_p\!\left(\frac{k}{n^{3.5}\rho_n^{3}}\right).
\)
Adding up the bounds obtained for each term in \cref{dseriesnew1}, we conclude that
\[
\|\hat{\alpha}-\alpha\| = O_p \left(\frac{k}{n^{3}\rho_n^{2.5}} \right),
\]
which completes the proof of \cref{inter-c}.
\\ \ \\
\textbf{Proof of \cref{inter-e}}. From \cref{eq:estimators},
\(
G = \beta_1^{-1}\beta_2^{-1} + \beta_2^{-1}\beta_1^{-1}\)
and 
\(\hat{G} = \hat{\beta}_1^{-1}\hat{\beta}_2^{-1} + \hat{\beta}_2^{-1}\hat{\beta}_1^{-1}.
\)
Hence,
\begin{align*}
\|\hat{G} \circ \hat{G} - G \circ G\|_F 
&= \biggl\|\sum_{i_1 \neq i_2}\sum_{j_1 \neq j_2}
\Bigl[\bigl(\hat{\beta}_{i_1}^{-1}\hat{\beta}_{i_2}^{-1}\bigr)\circ\bigl(\hat{\beta}_{j_1}^{-1}\hat{\beta}_{j_2}^{-1}\bigr)
-\bigl(\beta_{i_1}^{-1}\beta_{i_2}^{-1}\bigr)\circ\bigl(\beta_{j_1}^{-1}\beta_{j_2}^{-1}\bigr)\Bigr]\biggr\|_F \\
&\le \sum_{i_1 \neq i_2}\sum_{j_1 \neq j_2}
\bigl\|\bigl(\hat{\beta}_{i_1}^{-1}\hat{\beta}_{i_2}^{-1}\bigr)\circ\bigl(\hat{\beta}_{j_1}^{-1}\hat{\beta}_{j_2}^{-1}\bigr)
-\bigl(\beta_{i_1}^{-1}\beta_{i_2}^{-1}\bigr)\circ\bigl(\beta_{j_1}^{-1}\beta_{j_2}^{-1}\bigr)\bigr\|_F.
\end{align*}
We present the argument for one representative configuration \((i_1,i_2,j_1,j_2)=(1,2,1,2)\); all other cases follow identically. Then
\begin{equation}
\begin{aligned} \label{intereq}
\bigl\|\bigl(\hat{\beta}_1^{-1}\hat{\beta}_2^{-1}\bigr)\circ\bigl(\hat{\beta}_1^{-1}\hat{\beta}_2^{-1}\bigr)
-\bigl(\beta_1^{-1}\beta_2^{-1}\bigr)\circ\bigl(\beta_1^{-1}\beta_2^{-1}\bigr)\bigr\|_F
&\le 
\bigl\|\bigl(\hat{\beta}_1^{-1}\hat{\beta}_2^{-1}\bigr)\circ
\bigl(\hat{\beta}_1^{-1}\hat{\beta}_2^{-1}-\beta_1^{-1}\beta_2^{-1}\bigr)\bigr\|_F  \\
&\quad + \bigl\|\bigl(\hat{\beta}_1^{-1}\hat{\beta}_2^{-1}-\beta_1^{-1}\beta_2^{-1}\bigr)\circ
\bigl({\beta}_1^{-1}{\beta}_2^{-1}\bigr)\bigr\|_F.
\end{aligned}
\end{equation}
We begin by bounding the term
\(
\|\hat{\beta}_1^{-1}\hat{\beta}_2^{-1}-\beta_1^{-1}\beta_2^{-1}\|_F .
\)
By the triangle inequality,
\begin{align}\label{ds}
\|\hat \beta_1^{-1} \hat \beta_2^{-1} - \beta_1^{-1} \beta_2^{-1}\|_F
&\le
\|\hat \beta_1^{-1}(\hat \beta_2^{-1} - \beta_2^{-1})\|_F
+
\| (\hat \beta_1^{-1} - \beta_1^{-1})\beta_2^{-1}\|_F \notag\\
&\le \|\hat \beta_1^{-1} \|_F\|\hat \beta_2^{-1} - \beta_2^{-1}\|_F
+
\| \hat \beta_1^{-1} - \beta_1^{-1}\|_F \|\beta_2^{-1}\|_F .
\end{align}
Since we know $\|\beta^{-1}_2\|_2 \asymp \frac{1}{n\rho_n}$ and $\|\widehat \beta^{-1}_1\|_2 \asymp \frac{1}{n\rho_n}$ from \cref{ass:eigen-scaling}, it suffices to derive a bound for the generic quantity
\(
\|\hat \beta^{-1} - \beta^{-1}\|_F .
\)

To this end, we employ the identity
\begin{align}\label{dseriesmod}
\widehat V\widehat\Lambda^{-1}\widehat V^\top - V\Lambda^{-1}V^\top
&=
\big(\widehat V\widehat V^\top - VV^\top\big)\widehat V\widehat\Lambda^{-1}\widehat V^\top \\
&\quad
+ V\big(V^\top\widehat V\widehat\Lambda^{-1}-\Lambda^{-1}V^\top\widehat V\big)\widehat V^\top \nonumber\\
&\quad
+ V\Lambda^{-1}V^\top\big(\widehat V\widehat V^\top - VV^\top\big), \nonumber
\end{align}
which is a first–order analogue of \cref{dseries}. Observe that
\[
\begin{aligned}
V\big(V^\top\widehat V\widehat\Lambda^{-1}-\Lambda^{-1}V^\top\widehat V\big)\widehat V^\top
&=
V\Lambda^{-1}\big(\Lambda V^\top\widehat V - V^\top\widehat V\widehat\Lambda\big)
\widehat\Lambda^{-1}\widehat V^\top \\
&=
-\,V\Lambda^{-1}V^\top X \widehat V\widehat\Lambda^{-1}\widehat V^\top .
\end{aligned}
\]
Using \cref{dseriesmod}, we obtain
\begin{equation}\label{ds1}
\begin{aligned}
    \|\hat \beta^{-1} - \beta^{-1}\|_F
    &\le \|\widehat V\widehat V^\top - VV^\top\|_F \, \|\hat \beta^{-1}\|_2
    + \|\hat \beta^{-1}\|_2 \, \|X\|_2 \, \|\beta^{-1}\|_F \\
    &\qquad + \|\widehat V\widehat V^\top - VV^\top\|_F \, \|\beta^{-1}\|_2 \\
    &=O_p \left( \frac{k^{1/2}}{n^{3/2}\rho_n^{3/2}} \right),
\end{aligned}
\end{equation}
where the final line follows from the Davis--Kahan bound in \cref{lem:Good Set} together with \cref{ass:eigen-scaling}. \cref{ass:eigen-scaling} also implies $\|\beta^{-1}\|_2 \asymp \frac{1}{n\rho_n}$ and $\|\widehat \beta^{-1}\|_2 \asymp \frac{1}{n\rho_n}$.

Thus, substituting \cref{ds1} into \cref{ds} yields
\[
\|\hat \beta_1^{-1} \hat \beta_2^{-1} - \beta_1^{-1} \beta_2^{-1}\|_F
\lesssim
\frac{k^{1/2}}{n^{5/2}\rho_n^{5/2}} .
\]
Inserting this bound into \cref{intereq} and invoking \cref{ass:eigen-scaling,ass:incoherence} along with the incoherence of $\widehat V$, guaranteed by \cref{lem:Very Good Set}, we conclude that
$$
\begin{aligned}
    \bigl\| (\hat{\beta}_1^{-1}\hat{\beta}_2^{-1})\circ (\hat{\beta}_1^{-1}\hat{\beta}_2^{-1}) - (\beta_1^{-1}\beta_2^{-1})\circ (\beta_1^{-1}\beta_2^{-1}) \bigr\|_F
    &\lesssim \max\!\left( \|\hat{\beta}_1^{-1}\hat{\beta}_2^{-1}\|_{\infty}, \|\beta_1^{-1}\beta_2^{-1}\|_{\infty} \right) \|\hat \beta_1^{-1} \hat \beta_2^{-1} - \beta_1^{-1} \beta_2^{-1}\|_F \\
    &\lesssim \frac{k^{3/2}}{n^{11/2}\rho_n^{9/2}}.
\end{aligned}
$$
This is because
$$
\begin{aligned}
    \big|(\beta_1^{-1}\beta_2^{-1})_{ij}\big| &\le \|e_i\t V_1\|_2 \|\Lambda_1^{-1}\|_2 \|V_1\t V_2\|_2 \|\Lambda_2^{-1}\|_2 \|V_2\t e_j\|_2 \\
    &\lesssim \frac{k}{n^3 \rho_n^2}.
\end{aligned}
$$
Here, we used \cref{ass:eigen-scaling,ass:incoherence} to obtain the bound. Similarly, in order to provide the same bound for $\|\hat{\beta}_1^{-1}\hat{\beta}_2^{-1}\|_{\infty}$, we will use the incoherence of $\widehat V$, guaranteed by \cref{lem:Very Good Set}.

Summing over all index combinations yields the desired bound
\[
\|\hat{G} \circ \hat{G} - G \circ G\|_F
=O_p \left(
\frac{k^{3/2}}{n^{11/2}\rho_n^{9/2}}\right) .
\]
\\ \ \\
\textbf{Proof of \cref{inter-g}.}
Once again, invoking \cref{dseries} yields
\begin{equation}\label{dseriesnew10}
\begin{aligned}
\bigg\|
\widehat V &\widehat\Lambda^{-2}\widehat V^{\top}
-
V\Lambda^{-2}V^{\top}
\bigg\|_F \\
&\le
\Big\|
V\Lambda^{-2}V^{\top} S_1(X)
+
S_1(X)V\Lambda^{-2}V^{\top}  
-
V\Lambda^{-2}V^{\top}(XP+PX)V\Lambda^{-2}V^{\top}
\Big\|_F \\[4pt]
&\quad+
\Big\|
S_1(X)
\big(
\widehat V \widehat\Lambda^{-2}\widehat V^{\top}
-
V\Lambda^{-2}V^{\top}
\big)
\Big\|_F \\[4pt]
&\quad+
\Big\|
V\Lambda^{-2}V^{\top}(XP+PX)
\big(
\widehat V \widehat\Lambda^{-2}\widehat V^{\top}
-
V\Lambda^{-2}V^{\top}
\big)
\Big\|_F \\[4pt]
&\quad+
\Big\|
\sum_{j\ge2} V\Lambda^{-2}V^{\top} S_j(X)
+
\sum_{j\ge2} S_j(X)\widehat V \widehat\Lambda^{-2}\widehat V^{\top}  
-
V\Lambda^{-2}V^{\top} X^2 \widehat V \widehat\Lambda^{-2}\widehat V^{\top}
\Big\|_F .
\end{aligned}
\end{equation}
By \cref{ass:asymptotic1,ass:eigen-scaling}, together with bound on \(\|X\|_2\) from \cref{lem:Good Set} and the fact that \(\|S_1(X)\|_2 \lesssim \frac{1}{\sqrt{n \rho_n}}\), the first term on the right-hand side of \cref{dseriesnew10} satisfies
\begin{align}\label{smllbd1}
    \Big\|
V\Lambda^{-2}V^{\top} S_1(X)
+
S_1(X)V\Lambda^{-2}V^{\top}
-
V\Lambda^{-2}V^{\top}(XP+PX)V\Lambda^{-2}V^{\top}
\Big\|_F
= O_p \left(
\frac{\sqrt k}{n^{5/2}\rho_n^{5/2}}\right).
\end{align}
The second and third terms can be bounded using submultiplicativity of the Frobenius norm and the operator-norm control of \(S_1(X)\) and \(XP+PX\), yielding
\begin{align}\label{smllbd2}
    \Big\|
S_1(X)
\big(
\widehat V \widehat\Lambda^{-2}\widehat V^{\top}
-
V\Lambda^{-2}V^{\top}
\big)
\Big\|_F
+
\Big\|
V\Lambda^{-2}V^{\top}(XP+PX)
\big(
\widehat V \widehat\Lambda^{-2}\widehat V^{\top}
-
V\Lambda^{-2}V^{\top}
\big)
\Big\|_F\\ \notag
\lesssim
\frac{1}{\sqrt{n\rho_n}}
\bigg\|
\widehat V \widehat\Lambda^{-2}\widehat V^{\top}
-
V\Lambda^{-2}V^{\top}
\bigg\|_F .
\end{align}
Finally, for the last term in \cref{dseriesnew10}, \cref{lem:Good Set} implies
\begin{align}\label{smllbd3}
    \Big\|
\sum_{j\ge2} V\Lambda^{-2}V^{\top} S_j(X)
\Big\|_F
= O_p \left(
\frac{k}{n^{3}\rho_n^{3}}\right),
\end{align}
and the same bound holds for
\(
\big\|
\sum_{j\ge2} S_j(X)\widehat V \widehat\Lambda^{-2}\widehat V^{\top}
\big\|_F
\)
and
\(
\big\|
V\Lambda^{-2}V^{\top} X^2 \widehat V \widehat\Lambda^{-2}\widehat V^{\top}
\big\|_F
\). Collecting \cref{smllbd1,smllbd2,smllbd3} completes the proof.
\end{proof}
\subsection{Proof of \cref{inter4}}
\label{sec:interproof4}
\begin{proof}[Proof of \cref{inter4}]
Just like in \cref{sec:interproof3}, we will assume that \(A\) (in the one-sample case) and \({A^{(i)}}\) (in the two-sample case) both belong to $\mathcal{E}_{{\sf very  \ good}}$, as defined in \cref{lem:Very Good Set}, which happens with high probability as already proved.
\\ \ \\
\textbf{Proof of \cref{inter-f1}}.
We prove the claim directly. Using the triangle inequality we obtain
\begin{equation*}
\begin{aligned}
\|{\hat P}^2 - P^2\|_F
&\le \|{\hat P} \|_F \|{\hat P}-P\|_2 +  \|{\hat P}-P\|_2 \| P \|_F\\
&\lesssim k^{1/2}\,n\rho_n\,\|X\|_2\\
&= O_p \left(
k^{1/2}\,n^{3/2}\rho_n^{3/2}\right),
\end{aligned}
\end{equation*}
where the last inequality follows from \cref{lem:series-expansion,lem:Good Set}.
\\ \ \\
\textbf{Proof of \cref{inter-f2}.}
We begin by decomposing the difference as
\[
{\hat P}({\hat P}\circ{\hat P})-P(P\circ P)
=
({\hat P}-P)({\hat P}\circ{\hat P})
+
P\bigl[({\hat P}\circ{\hat P})-(P\circ P)\bigr].
\]
Note that $\|{\hat P}\circ{\hat P}\|_F^2 = \sum_{i,j} \hat P_{ij}^4 = O_p \left(n^2 \rho_n^4 \right)$. Thus, for the first term on the right hand side,
\begin{align*}
    \|({\hat P}-P)({\hat P}\circ{\hat P})\|_F
    &\le k^{1/2}\|{\hat P}-P\|_2\,\|{\hat P}\circ{\hat P}\|_F \\
    &=O_p \left( k^{1/2}\,n^{3/2}\rho_n^{5/2}\right),
\end{align*}
where the bound follows from \cref{ass:asymptotic1,lem:series-expansion,lem:Good Set}. 

For the second term, observe that
\begin{equation}
\begin{aligned} \label{smllbd4}
    \|({\hat P}\circ{\hat P})-(P\circ P)\|_F
    &= \|({\hat P}-P)\circ({\hat P}+P)\|_F \\
    &\lesssim k^{1/2} \|{\hat P}-P\|_2 \|{\hat P}+P\|_{\infty} \\
    &=O_p \left( k^{1/2}\,n^{1/2}\rho_n^{3/2}\right).
\end{aligned}
\end{equation}
Consequently,
\(
\|P\bigl[({\hat P}\circ{\hat P})-(P\circ P)\bigr]\|_F
=O_p\left(
k^{1/2}\,n^{3/2}\rho_n^{5/2}\right).
\)
Combining the above bounds yields the claim of \cref{inter-f2}.
\\ \ \\
\textbf{Proof of \cref{inter-f3}.}
We again use the inequality
\begin{align} \label{smllbd5}
    \|({\hat P}\circ{\hat P})^2-(P\circ P)^2\|_F
\le
\|({\hat P}\circ{\hat P})-(P\circ P)\|_F\,
\|({\hat P}\circ{\hat P})+(P\circ P)\|_F.
\end{align}
The first factor is bounded by \(k^{1/2}\,n^{1/2}\rho_n^{3/2}\) as shown in the last part,
while \cref{ass:asymptotic1,ass:eigen-scaling} imply
\(
\|({\hat P}\circ{\hat P})+(P\circ P)\|_F = O_p \left( n\rho_n^2 \right)
\) because of $\|{\hat P}\circ{\hat P}\|_F^2 = \sum_{i,j} \hat P_{ij}^4 = O_p \left(n^2 \rho_n^4 \right)$ as already shown in the last part.
Combining this bound with \cref{smllbd4,smllbd5}, we obtain
\[
\|({\hat P}\circ{\hat P})^2-(P\circ P)^2\|_F
=O_p \left(
k^{1/2}\,n^{3/2}\rho_n^{7/2}\right).
\]
\end{proof}

\section{Proofs of Lemmas from \cref{sec:testconpf}}
\subsection{Proof of \cref{lem:cross-term-moments}} \label{sec:cross-term-momentsproof}
\begin{proof}
Throughout this proof we assume that each \({A^{(i)}}\) lies in the set $\mathcal{E}_{{\sf very \ good}}$ of 
\cref{lem:Very Good Set}, which holds with high probability.  
From Theorem 1 of \citet{xia_normal_2021}, we have
\(
(\hat V^{(i)}) (\hat V^{(i)})\t - {V^{(i)}} {V^{(i)}}\t 
= \sum_{l \ge 1} S_{i,l}({X^{(i)}}),
\)
where \(S_{i,l}({X^{(i)}})\) is defined exactly as in \cref{Sk_2sample}.  
Consequently, the second and third order terms of \({X^{(i)}}\) in 
\(\mathbb{E}\big\langle \hat {V^{(1)}}\hat {V^{(1)}}\t - {V^{(1)}}{V^{(1)}}\t, 
                \hat {V^{(2)}}\hat {V^{(2)}}\t - {V^{(2)}}{V^{(2)}}\t\big\rangle\)
are
\[
\mathbb{E}\left( \tr\big(S_{1,1}({X^{(1)}}) S_{2,1}({X^{(2)}})\big) \right)
\quad\text{and}\quad
\mathbb{E}\left( \tr\big(S_{1,2}({X^{(1)}}) S_{2,1}({X^{(2)}}) 
      + S_{1,1}({X^{(1)}}) S_{2,2}({X^{(2)}})\big) \right),
\]
both of which vanish immediately, as shown in the proof of 
\cref{lem:3rdmean2}.  Moreover, again from that proof, we have the bound
\begin{align}\label{smllbd6}
    \big\langle S_{1,l_1}({X^{(1)}}),\, S_{2,l_2}({X^{(2)}}) \big\rangle
\le 
\|S_{1,l_1}({X^{(1)}})\|_F \, \|S_{2,l_2}({X^{(2)}})\|_F
=O_p \left( 
\frac{k^2}{(n\rho_n)^{\,l/2}}\right),
\end{align}
where \(l = l_1 + l_2\).  
Next, by \cref{lem:second-order2} and \cref{thm:2sampleconsis}, it holds that 
\(
\hat\sigma_2^2 \asymp k^2/(n^3 \rho_n^2).
\)
We therefore obtain
\begin{equation}
\begin{aligned}\label{note05}
\frac{1}{\hat \sigma_2}
\big\langle 
      {V^{(1)}}{V^{(1)}}\t - {V^{(2)}}{V^{(2)}}\t, 
      \hat {V^{(2)}}\hat {V^{(2)}}\t - {V^{(2)}}{V^{(2)}}\t
\big\rangle
&=
\frac{1}{\hat \sigma_2}
\displaystyle\sum_{l=2,3}\ \sum_{l_1 + l_2 = l}
\big\langle S_{1,l_1}({X^{(1)}}),\, S_{2,l_2}({X^{(2)}}) \big\rangle
\\
&\qquad
+
\frac{1}{\hat \sigma_2}
\displaystyle\sum_{l \ge 4}\ \sum_{l_1 + l_2 = l}
\big\langle S_{1,l_1}({X^{(1)}}),\, S_{2,l_2}({X^{(2)}}) \big\rangle.
\end{aligned}
\end{equation}
We will show that the second term in \cref{note05} is \(o_p(1)\) under \cref{ass:dense}. This is because from \cref{smllbd6,lem:second-order2}, we have the following bound with probability \(1-O(n^{-c})\) for a constant \(c>1\), 
\begin{align*}
    \frac{1}{\hat \sigma_2}
\displaystyle\sum_{l \ge 4}\ \sum_{l_1 + l_2 = l}
\big\langle S_{1,l_1}({X^{(1)}}),\, S_{2,l_2}({X^{(2)}}) \big\rangle &\lesssim \frac{n^{3/2} \rho_n}{k}\sum_{l \ge 4}\ 2^l \frac{k^2}{(n\rho_n)^{\,l/2}}\\
&\lesssim \frac{k}{n^{1/2} \rho_n}=o_p(1).
\end{align*}
For the first term on the right hand side of \cref{note05}, we invoke Chebyshev's inequality. As argued earlier,
\[
\E\left(
\sum_{l=2,3}\ \sum_{l_1+l_2=l}
\big\langle S_{1,l_1}({X^{(1)}}),\, S_{2,l_2}({X^{(2)}}) \big\rangle
\right)=0.
\]
We therefore compute the order of
\(
\Var\!\left(
\sum_{l=2,3}\ \sum_{l_1+l_2=l}
\big\langle S_{1,l_1}({X^{(1)}}),\, S_{2,l_2}({X^{(2)}}) \big\rangle
\right).
\)
The first term in the above sum is
\[
\big\langle S_{1,1}({X^{(1)}}),\, S_{2,1}({X^{(2)}}) \big\rangle=
\tr\!\left(
\beta_1^{-1} {X^{(1)}} \beta^{\perp} {X^{(2)}} \beta_2^{-1}
+
\beta^{\perp} {X^{(1)}} \beta_1^{-1}\beta_2^{-1} {X^{(2)}}
\right)
=
2\,\tr\!\left(\beta_1^{-1}\beta_2^{-1} {X^{(1)}} \beta^{\perp} {X^{(2)}}\right).
\]
Hence it suffices to determine the order of
\(
4\,\E\!\left[
\bigl(\tr(\beta_1^{-1}\beta_2^{-1} {X^{(1)}} \beta^{\perp} {X^{(2)}})\bigr)^2
\right]
\)
since it is mean-zero.  

Conditioning on \({X^{(1)}}\) and writing \(M_1 := \beta_1^{-1}\beta_2^{-1} {X^{(1)}} \beta^{\perp}\), we obtain
\begin{equation}
\begin{aligned}\label{eq:e1a}
\E_{{X^{(2)}}}\!\left[
\bigl(\tr(M_1 {X^{(2)}})\bigr)^2
\right] &= \sum_{i < j} \E \left(X^{(2)}_{ij}\right)^2 \bigl( (M_1)_{ij} + (M_1)_{ji} \bigr)^2 + \sum_{i=1}^n \E \left(X^{(2)}_{ii}\right)^2 (M_1)^2_{ii}\\
&= \sum_{i,j} \sigma_{ij}^2 (M_1)_{ij}^2 + \sum_{i \neq j} \sigma_{ij}^2 (M_1)_{ij}(M_1)_{ji}.
\end{aligned}
\end{equation}
Taking expectation with respect to \({X^{(1)}}\), we will bound the first term on the right hand side of \cref{eq:e1a}:
\begin{equation}
\begin{aligned}\label{eq:e1b}
\sum_{i,j} \sigma_{ij}^2 \mathbb{E}_{X^{(1)}}[(M_1)_{ij}^2] &\asymp \rho_n \sum_{i,j} \mathbb{E}_{X^{(1)}}[(M_1)_{ij}^2]\\
&= \rho_n \mathbb{E}_{X^{(1)}}\bigl[\|M_1\|_F^2\bigr] = \rho_n \mathbb{E}_{X^{(1)}}\bigl[\text{tr}(M_1 M_1^\top)\bigr].
\end{aligned}
\end{equation}
For the second term of \cref{eq:e1a}, we apply the Cauchy-Schwarz inequality to the sum:
\begin{equation}
\begin{aligned}\label{eq:e1c}
\left| \sum_{i \neq j} \sigma_{ij}^2 \mathbb{E}_{X^{(1)}}[(M_1)_{ij}(M_1)_{ji}] \right| &\le \sum_{i \neq j} \sigma_{ij}^2 \mathbb{E}_{X^{(1)}}\bigl[|(M_1)_{ij}| |(M_1)_{ji}|\bigr]\\
&\le \frac{1}{2} \sum_{i \neq j} \sigma_{ij}^2 \mathbb{E}_{X^{(1)}}\bigl[ (M_1)_{ij}^2 + (M_1)_{ji}^2 \bigr] \lesssim \rho_n \mathbb{E}_{X^{(1)}}\bigl[\text{tr}(M_1 M_1^\top)\bigr].
\end{aligned}
\end{equation}
Thus, \cref{eq:e1b,eq:e1c} imply that to find the order of \(\E_{{X^{(2)}}}\!\left[
\bigl(\tr(M_1 {X^{(2)}})\bigr)^2
\right]\) as in \cref{eq:e1a}, it suffices to evaluate that of  \(\rho_n \mathbb{E}_{X^{(1)}}\bigl[\text{tr}(M_1 M_1^\top)\bigr]\).
To evaluate the expectation rigorously, note that both $\beta_2^{-1}\beta_1^{-2}\beta_2^{-1}$ and $\beta^\perp$ are deterministic, symmetric, and positive semi-definite.
\begin{equation}
\begin{aligned}
\mathbb{E}_{X^{(1)}}\!\left[\tr(M_1M_1^\top)\right]
&=
\sum_{i,j=1}^{n}
\sigma_{ij}^{2}\,
\left(\beta_2^{-1}\beta_1^{-2}\beta_2^{-1}\right)_{ii}\,
(\beta^\perp)_{jj}
+
\sum_{i,j=1}^{n}
\sigma_{ij}^{2}\,
\left(\beta_2^{-1}\beta_1^{-2}\beta_2^{-1}\right)_{ij}\,
(\beta^\perp)_{ij}.
\end{aligned}
\end{equation}

For the first term, since $\sigma_{ij}^{2}\asymp\rho_n$,
\begin{equation}
\begin{aligned}
\sum_{i,j=1}^{n}
\sigma_{ij}^{2}
\left(\beta_2^{-1}\beta_1^{-2}\beta_2^{-1}\right)_{ii}
(\beta^\perp)_{jj}
&\asymp
\rho_n
\sum_{i,j=1}^{n}
\left(\beta_2^{-1}\beta_1^{-2}\beta_2^{-1}\right)_{ii}
(\beta^\perp)_{jj} \\
&=
\rho_n
\left(
\sum_{i=1}^{n} \left(\beta_2^{-1}\beta_1^{-2}\beta_2^{-1}\right)_{ii}
\right)
\left(
\sum_{j=1}^{n} (\beta^\perp)_{jj}
\right) \\
&=
\rho_n\,\tr\!\left(
\beta_1^{-2}\beta_2^{-2}
\right)\,\tr(\beta^\perp)\\
&\asymp \frac{k}{n^3\rho_n^3}
\end{aligned}
\end{equation}

It remains to show that the second term is smaller. By Cauchy--Schwarz,
\begin{equation}
\begin{aligned}
\sum_{i,j=1}^{n}
\sigma_{ij}^{2}
\left(\beta_2^{-1}\beta_1^{-2}\beta_2^{-1}\right)_{ij}
(\beta^\perp)_{ij}
&\lesssim
\rho_n
\sum_{i,j=1}^{n}
|\left(\beta_2^{-1}\beta_1^{-2}\beta_2^{-1}\right)_{ij}|\,
|(\beta^\perp)_{ij}| \\
&\lesssim
\rho_n
\|\beta_2\inv \beta_1^{-2} \beta_2\inv \|_F
\|\beta^\perp\|_F \\
&\lesssim \frac{\sqrt{k}}{n^{7/2}\rho_n^3} \\
&=o\!\left(
\frac{k}{n^3\rho_n^3}
\right).
\end{aligned}
\end{equation}
Therefore,
\[
\mathbb{E}_{X^{(1)}}\!\left[
\tr(M_1M_1^\top)
\right]
\asymp
\frac{k}{n^3\rho_n^3}.
\]
Consequently,
\begin{equation}
\begin{aligned}\label{need1}
\Var\Bigl(
\bigl\langle
S_{1,1}(X^{(1)}),
S_{2,1}(X^{(2)})
\bigr\rangle
\Bigr)
&=
4\,\mathbb{E}
\Bigl[
\bigl\{
\tr(
\beta_1^{-1}\beta_2^{-1}
X^{(1)}
\beta^\perp
X^{(2)}
)
\bigr\}^2
\Bigr] \\
&\asymp
\rho_n\,
\mathbb{E}_{X^{(1)}}
\!\left[
\tr(M_1M_1^\top)
\right] \\
&\asymp
\rho_n
\cdot
\frac{k}{n^3\rho_n^3}
=
\frac{k}{n^3\rho_n^2}.
\end{aligned}
\end{equation}
We now proceed to bound
\(
\operatorname{Var}\big( \langle S_{1,2}({X^{(1)}}),\, S_{2,1}({X^{(2)}}) \rangle \big)
\)
using analogous arguments. We expand the inner product as
\begin{equation} \label{eq:e2a}
\begin{split}
\langle S_{1,2}({X^{(1)}}),\, S_{2,1}({X^{(2)}}) \rangle 
    &= \operatorname{tr}\left( \beta^\perp {X^{(1)}} \beta^\perp {X^{(1)}} \beta_1^{-2} \beta_2^{-1} {X^{(2)}} \right) \\ 
    &\quad - \operatorname{tr}\left( \beta^\perp {X^{(1)}} \beta_1^{-1} {X^{(1)}} \beta_1^{-1} \beta_2^{-1} {X^{(2)}} \right) \\ 
    &\quad + \operatorname{tr}\left( \beta_1^{-2} {X^{(1)}} \beta^\perp {X^{(1)}} \beta^\perp {X^{(2)}} \beta_2^{-1} \right) \\ 
    &\quad - \operatorname{tr}\left( \beta_1^{-1} {X^{(1)}} \beta_1^{-1} {X^{(1)}} \beta^\perp {X^{(2)}} \beta_2^{-1} \right).
\end{split}
\end{equation}
To analyze the variance, it suffices to determine the order of the second moment of the first term on the right hand side of \cref{eq:e2a} because all other terms can be tackled in the same manner:
\(
\mathbb{E}\left[\operatorname{tr}\left( \beta^\perp {X^{(1)}} \beta^\perp {X^{(1)}} \beta_1^{-2} \beta_2^{-1} {X^{(2)}} \right)^2 \right].
\)
Conditioning on \({X^{(1)}}\) as in \cref{eq:e1a} and defining \(N_1 := \beta^\perp {X^{(1)}} \beta^\perp {X^{(1)}} \beta_1^{-2} \beta_2^{-1}\), we obtain
\begin{equation} \label{eq:e2b}
\begin{split}
\mathbb{E}_{{X^{(2)}}}\left[ \bigl(\operatorname{tr}(N_1 {X^{(2)}})\bigr)^2 \right] \asymp 
    \sum_{i,j} \sigma_{ij}^2 (N_1)_{ij}^2 + \sum_{i \neq j} \sigma_{ij}^2 (N_1)_{ij}(N_1)_{ji}.
\end{split}
\end{equation}
Taking expectation with respect to $X^{(1)}$, the same Cauchy--Schwarz argument used in
\cref{eq:e1b,eq:e1c} gives
\begin{equation}
\begin{aligned}
\mathbb{E}\left[
\left\{
\tr\left(
\beta^\perp X^{(1)}\beta^\perp X^{(1)}
\beta_1^{-2}\beta_2^{-1}X^{(2)}
\right)
\right\}^2
\right]
&\asymp
\rho_n\,
\mathbb{E}_{X^{(1)}}\left[
\tr(N_1N_1^\top)
\right]\\
&\asymp \rho_n
\mathbb{E}_{X^{(1)}}\left[
\tr\left(
\beta_1^{-2}\beta_2^{-2}\beta_1^{-2}
X^{(1)}
(\beta^\perp X^{(1)})^3
\right)
\right].
\end{aligned}
\end{equation}
We now evaluate the last expectation. Expanding the trace gives
\begin{align*}
&\sum_{i_1,\dots,i_8}
\E
\Big[
(\beta_1^{-2}\beta_2^{-2}\beta_1^{-2})_{i_1,i_2}
(X^{(1)})_{i_2,i_3}
(\beta^\perp)_{i_3,i_4}
(X^{(1)})_{i_4,i_5}
(\beta^\perp)_{i_5,i_6}
(X^{(1)})_{i_6,i_7}
(\beta^\perp)_{i_7,i_8}
(X^{(1)})_{i_8,i_1}
\Big].
\end{align*}
The leading nonzero contributions arise from pairings of the four $X^{(1)}$
entries into two equal unordered edge-pairs. Consider, for instance, the pairing
\begin{align} \label{pair100}
\{i_2,i_3\}\sim \{i_4,i_5\},
\qquad
\{i_6,i_7\}\sim \{i_8,i_1\}.
\end{align}
The contribution from this pairing is
\begin{equation} \label{spl}
\begin{split}
    \sum_{\substack{i_1,\dots,i_8 \\ \{i_2, i_3\} \sim \{i_4, i_5\} \\ \{i_6, i_7\} \sim \{i_8, i_1\}}}
    \E &\Bigg[
    (\beta_1^{-2}\beta_2^{-2}\beta_1^{-2})_{i_1,i_2} (X^{(1)})_{i_2,i_3} (\beta^\perp)_{i_3,i_4} (X^{(1)})_{i_4,i_5}(\beta^\perp)_{i_5,i_6} (X^{(1)})_{i_6,i_7} (\beta^\perp)_{i_7,i_8} (X^{(1)})_{i_8,i_1}
    \Bigg]
    \\
    &= \sum_{\substack{i_1,\dots,i_8 \\ \{i_2, i_3\} \sim \{i_4, i_5\} \\ \{i_6, i_7\} \sim \{i_8, i_1\}}}
    \Bigg[
    (\beta_1^{-2}\beta_2^{-2}\beta_1^{-2})_{i_1,i_2}
    \underbrace{\E \left( (X^{(1)})_{i_2,i_3}^2 (\beta^\perp)_{i_3,i_2} + (X^{(1)})_{i_2,i_3}^2 (\beta^\perp)_{i_3,i_3} \right)}_{\text{Block 1}} \\
    &\qquad\qquad\qquad\qquad \times (\beta^\perp)_{i_5,i_6}
    \underbrace{\E \left( (X^{(1)})_{i_6,i_7}^2 (\beta^\perp)_{i_7,i_6} + (X^{(1)})_{i_6,i_7}^2 (\beta^\perp)_{i_7,i_7} \right)}_{\text{Block 2}}
    \Bigg] .
\end{split}
\end{equation}
The cases $\{i_2=i_5,\,i_3=i_4,\,i_6=i_8,\,i_7=i_1\}$ and
$\{i_2=i_4,\,i_3=i_5,\,i_6=i_1,\,i_7=i_8\}$ force \cref{spl} to vanish, since
\[
\sum_{i_5,i_6} (\beta_1^{-2}\beta_2^{-2}\beta_1^{-2})_{i_1,i_2} (\beta^\perp)_{i_5,i_6}=0.
\]
Thus the remaining pairings from \cref{pair100} are
\begin{align}\label{pair200}
    &\{i_2=i_5,\,i_3=i_4,\,i_6=i_1,\,i_7=i_8\}, \\ \notag
    &\{i_2=i_4,\,i_3=i_5,\,i_6=i_8,\,i_7=i_1\}.
\end{align}
The two corresponding blocks in \cref{spl} satisfy, uniformly in $i_2$ and $i_6$,
\begin{equation}
\begin{aligned}\label{spl1}
\sum_{i_3}
\E\left[
(X^{(1)})_{i_2,i_3}^{2}
\left\{
(\beta^\perp)_{i_3,i_2}
+
(\beta^\perp)_{i_3,i_3}
\right\}
\right]
&=
\sum_{i_3}
\sigma_{i_2i_3}^{2}
\left\{
(\beta^\perp)_{i_3,i_2}
+
(\beta^\perp)_{i_3,i_3}
\right\} \\
&\asymp
n\rho_n,
\end{aligned}
\end{equation}
and similarly
\begin{equation}
\begin{aligned}\label{spl2}
\sum_{i_7}
\E\left[
(X^{(1)})_{i_6,i_7}^{2}
\left\{
(\beta^\perp)_{i_7,i_6}
+
(\beta^\perp)_{i_7,i_7}
\right\}
\right]
\asymp
n\rho_n.
\end{aligned}
\end{equation}
We also have the trace identity
\begin{align}\label{spl3}
\tr(\beta_1^{-2}\beta_2^{-2}\beta_1^{-2})
\asymp
\frac{k}{n^6\rho_n^6}.
\end{align}
Therefore, under the first pairing of \cref{pair200}, combining
\cref{spl,spl1,spl2,spl3} we obtain
\begin{align*}
&\sum_{\substack{
i_1,\dots,i_8\\
\{i_2,i_3\}\sim\{i_4,i_5\}\\
\{i_6,i_7\}\sim\{i_8,i_1\}
}}
\E\Big[
\left(\beta_1^{-2}\beta_2^{-2}\beta_1^{-2}\right)_{i_1,i_2}
(X^{(1)})_{i_2,i_3}
(\beta^\perp)_{i_3,i_4}
(X^{(1)})_{i_4,i_5}
(\beta^\perp)_{i_5,i_6}
(X^{(1)})_{i_6,i_7}
(\beta^\perp)_{i_7,i_8}
(X^{(1)})_{i_8,i_1}
\Big] \\
&= \sum_{i_1,i_2}
    \Bigg[
    (\beta_1^{-2}\beta_2^{-2}\beta_1^{-2})_{i_1,i_2}
    \sum_{i_3}\E \left( (X^{(1)})_{i_2,i_3}^2 (\beta^\perp)_{i_3,i_2} + (X^{(1)})_{i_2,i_3}^2 (\beta^\perp)_{i_3,i_3} \right) \\
    &\qquad\qquad\qquad\qquad \times (\beta^\perp)_{i_2,i_1}
    \sum_{i_7}\E \left( (X^{(1)})_{i_6,i_7}^2 (\beta^\perp)_{i_7,i_6} + (X^{(1)})_{i_6,i_7}^2 (\beta^\perp)_{i_7,i_7} \right)
    \Bigg] \\
&\quad \lesssim
n^2\rho_n^2
\bigg|\sum_{i_1}
\left(\beta_1^{-2}\beta_2^{-2}\beta_1^{-2}\right)_{i_1,i_1}\bigg| \\
&\quad \lesssim
n^2\rho_n^2
\tr\left(\beta_1^{-2}\beta_2^{-2}\beta_1^{-2}\right)\\
&\quad \lesssim
\frac{k}{n^4\rho_n^4}.
\end{align*}
Similarly, under the second pairing of \cref{pair200}, we have
\begin{equation}
\begin{aligned}
&\sum_{\substack{
i_1,\dots,i_8\\
\{i_2,i_3\}\sim\{i_4,i_5\}\\
\{i_6,i_7\}\sim\{i_8,i_1\}
}}
\E\Big[
\left(\beta_1^{-2}\beta_2^{-2}\beta_1^{-2}\right)_{i_1,i_2}
(X^{(1)})_{i_2,i_3}
(\beta^\perp)_{i_3,i_4}
(X^{(1)})_{i_4,i_5}
(\beta^\perp)_{i_5,i_6} \\
&\hspace{4.5cm}\times
(X^{(1)})_{i_6,i_7}
(\beta^\perp)_{i_7,i_8}
(X^{(1)})_{i_8,i_1}
\Big] \\
&=
\sum_{i_6,i_2}
\Bigg[
\sum_{i_3}
\E\left(
(\beta^\perp)_{i_3,i_6}
(X^{(1)})_{i_2,i_3}^2
(\beta^\perp)_{i_3,i_2}
+
(\beta^\perp)_{i_3,i_6}
(X^{(1)})_{i_2,i_3}^2
(\beta^\perp)_{i_3,i_3}
\right)
\Bigg] \\
&\qquad\qquad\qquad\qquad\times
\Bigg[
\sum_{i_7}
\E\left(
(\beta_1^{-2}\beta_2^{-2}\beta_1^{-2})_{i_7,i_2}
(X^{(1)})_{i_6,i_7}^2
(\beta^\perp)_{i_7,i_6}
+
(\beta_1^{-2}\beta_2^{-2}\beta_1^{-2})_{i_7,i_2}
(X^{(1)})_{i_6,i_7}^2
(\beta^\perp)_{i_7,i_7}
\right)
\Bigg]\\
&\lesssim \sum_{i_6,i_2}
\left|
\sum_{i_3}
\sigma_{i_2i_3}^2
(\beta^\perp)_{i_3,i_6}
\left\{
(\beta^\perp)_{i_3,i_2}
+
(\beta^\perp)_{i_3,i_3}
\right\}
\right| \left|
\sum_{i_7}
\sigma_{i_6i_7}^2
(\beta_1^{-2}\beta_2^{-2}\beta_1^{-2})_{i_7,i_2}
\left\{
(\beta^\perp)_{i_7,i_6}
+
(\beta^\perp)_{i_7,i_7}
\right\}
\right|\\
&\lesssim \sum_{i_2,i_6} \rho_n^2
\left(\sum_{i_3}
\big|(\beta^\perp)_{i_3,i_6}\big|
\left(
\big|(\beta^\perp)_{i_3,i_2}\big|
+
\big|(\beta^\perp)_{i_3,i_3}\big|
\right)\right) \left(\sum_{i_7}
\big|(\beta_1^{-2}\beta_2^{-2}\beta_1^{-2})_{i_7,i_2}\big|
\left(
(\beta^\perp)_{i_7,i_6}
+
(\beta^\perp)_{i_7,i_7}
\right)\right)\\
&\lesssim
\rho_n^2
\sum_{i_6,i_2,i_7}
\left|
(\beta_1^{-2}\beta_2^{-2}\beta_1^{-2})_{i_7,i_2}
\right|
\left(
|(\beta^\perp)_{i_7,i_6}|
+
|(\beta^\perp)_{i_7,i_7}|
\right)\\
&\lesssim
n^2 \rho_n^2
\left\|
\beta_1^{-2}\beta_2^{-2}\beta_1^{-2}
\right\|_F \\
&\lesssim \frac{k}{n^4 \rho_n^4}.
\end{aligned}
\end{equation}
The two other nonzero pairings,
\[
\{i_2,i_3\}\sim \{i_6,i_7\},
\qquad
\{i_4,i_5\}\sim \{i_8,i_1\},
\]
and
\[
\{i_2,i_3\}\sim \{i_8,i_1\},
\qquad
\{i_4,i_5\}\sim \{i_6,i_7\},
\]
are handled identically and have the same order. The configurations in which three or four $X^{(1)}$-entries collapse onto the same edge have fewer free indices and contribute at most
\[
O\left(
\frac{k}{n^5\rho_n^5}
\right)
=
o\left(
\frac{k}{n^4\rho_n^4}
\right),
\]
Hence
\[
\mathbb{E}_{X^{(1)}}\left[
\tr\left(
\beta_1^{-2}\beta_2^{-2}\beta_1^{-2}
X^{(1)}
(\beta^\perp X^{(1)})^3
\right)
\right]
\lesssim
\frac{k}{n^4\rho_n^4}.
\]
implying
\begin{align}\label{need2}
    \Var\big(
\langle S_{1,2}(X^{(1)}),S_{2,1}(X^{(2)})\rangle
\big)
\lesssim
\frac{k}{n^4\rho_n^3}.
\end{align}
Similarly, we have
\begin{equation} \label{need3}
    \operatorname{Var}\big( \langle S_{1,1}({X^{(1)}}),\, S_{2,2}({X^{(2)}}) \rangle \big) \lesssim \frac{k}{n^{4}\rho_n^{3}}.
\end{equation}
Using \cref{need1,need2,need3} and applying Cauchy--Schwarz inequality to the covariances, we bound the variance of the total sum:
\begin{align*}
    \operatorname{Var}\left( \sum_{l=2,3}\ \sum_{l_1+l_2=l} \langle S_{1,l_1}({X^{(1)}}),\, S_{2,l_2}({X^{(2)}}) \rangle \right) &\le \sum_{l=2,3}\ \sum_{l_1+l_2=l} \operatorname{Var}\left(  \langle S_{1,l_1}({X^{(1)}}),\, S_{2,l_2}({X^{(2)}}) \rangle \right) \\
    & \quad + \Cov \left(\langle S_{1,1}({X^{(1)}}),\, S_{2,2}({X^{(2)}}) \rangle, \langle S_{1,2}({X^{(1)}}),\, S_{2,1}({X^{(2)}}) \rangle \right) \\
    & \quad + \Cov \left(\langle S_{1,1}({X^{(1)}}),\, S_{2,2}({X^{(2)}}) \rangle, \langle S_{1,1}({X^{(1)}}),\, S_{2,1}({X^{(2)}}) \rangle \right) \\
    & \quad + \Cov \left(\langle S_{1,1}({X^{(1)}}),\, S_{2,1}({X^{(2)}}) \rangle, \langle S_{1,1}({X^{(1)}}),\, S_{2,2}({X^{(2)}}) \rangle \right) \\
    &\le \sum_{l=2,3}\ \sum_{l_1+l_2=l} \operatorname{Var}\left(  \langle S_{1,l_1}({X^{(1)}}),\, S_{2,l_2}({X^{(2)}}) \rangle \right) \\
    & \quad + \sqrt {\Var \left(\langle S_{1,1}({X^{(1)}}),\, S_{2,2}({X^{(2)}}) \rangle\right) \Var \left(\langle S_{1,2}({X^{(1)}}),\, S_{2,1}({X^{(2)}}) \rangle \right)} \\
    & \quad + \sqrt {\Var \left(\langle S_{1,1}({X^{(1)}}),\, S_{2,2}({X^{(2)}}) \rangle\right) \Var \left(\langle S_{1,1}({X^{(1)}}),\, S_{2,1}({X^{(2)}}) \rangle \right)} \\
    & \quad + \sqrt {\Var \left(\langle S_{1,1}({X^{(1)}}),\, S_{2,1}({X^{(2)}}) \rangle\right) \Var \left(\langle S_{1,1}({X^{(1)}}),\, S_{2,2}({X^{(2)}}) \rangle \right)}\\
    &\lesssim \frac{k}{n^3\rho_n^2}.
\end{align*}
Therefore, by Chebyshev's inequality, we conclude
\[
\frac{\sum_{l=2,3}\ \sum_{l_1 + l_2 = l} \langle S_{1,l_1}({X^{(1)}}),\, S_{2,l_2}({X^{(2)}}) \rangle}{\hat\sigma_2} \asymp 1,
\]
which completes the proof.
\end{proof}

\subsection{Proof of \cref{lem:delta-linearization}}
\label{sec:delta-linearization-proof}
\begin{proof}
Assume that \({A^{(i)}}\) belongs to $\mathcal{E}_{{\sf very \ good}}$, as defined in \cref{lem:Very Good Set}, which happens with high probability as already proved.
Using the expansion in \cref{Sk_2sample}, we have
\begin{equation}
\begin{aligned}
    L(\Delta)
    &= \langle (\hat V^{(1)}) (\hat V^{(1)})\t - {V^{(1)}} {V^{(1)}}\t,\, \Delta \rangle
       - \langle (\hat V^{(2)}) (\hat V^{(2)})\t - {V^{(2)}} {V^{(2)}}\t,\, \Delta \rangle \\
    &= \sum_{k \geq 1} \langle S_{1,k}({X^{(1)}}) - S_{2,k}({X^{(2)}}),\, \Delta \rangle.
\end{aligned}
\end{equation}
We first determine the order of \(\langle S_{1,1}({X^{(1)}}),\,\Delta\rangle\).
Consider the scalar quantity
\(
T := \tr\big(\Delta\,\beta_1^{-1}{X^{(1)}}\beta_1^\perp\big).
\)
By cyclicity of the trace,
\(
T = \tr\big(\beta_1^\perp\Delta\beta_1^{-1}{X^{(1)}}\big)
    = \langle W,\,{X^{(1)}}\rangle,\)
where
\(W := \beta_1^\perp\Delta\beta_1^{-1}.
\)
Expanding over unordered indices gives
\(
T = \sum_{p} W_{pp}X_{pp} + \sum_{p<q} (W_{pq}+W_{qp})X_{pq}\)
with
\(a_{pp}=W_{pp}, a_{pq}=W_{pq}+W_{qp}(p<q).
\)
Applying Bernstein’s inequality to the independent, mean–zero summands \(a_{pq}X_{pq}\) yields
\[
\Pr(|T|\ge u)
\le
2\exp\Big(
-\frac{u^2}{2\sigma^2 + \tfrac{2}{3}Mu}
\Big),\]
where
\(\sigma^2 = \sum_{p\le q}\Var(X_{pq})a_{pq}^2,\) and
\(M = \max_{p\le q}|a_{pq}|\).
Since \(\Var(X_{pq})\lesssim \rho_n\) and \(\sum_{p\le q}a_{pq}^2\asymp \|W\|_F^2\), we have
\(
\sigma^2 \lesssim \rho_n \|W\|_F^2.
\) By submultiplicativity of norms,
\(
\|W\|_F
= \|\beta_1^\perp \Delta \beta_1^{-1}\|_F
\lesssim \frac{\|\Delta\|_F}{n\rho_n},
\)
which gives
\(
\sigma^2 \lesssim \rho_n \Big(\frac{\|\Delta\|_F}{n\rho_n}\Big)^2
= \frac{\|\Delta\|_F^2}{n^2\rho_n},
\) and
\(M \le \|W\|_F \lesssim \frac{\|\Delta\|_F}{n\rho_n}.
\)
Choosing \(u = C\sigma \sqrt {\log n}\) with a sufficiently large constant \(C\),
we obtain \(T = O_p(\sigma\sqrt{\log n} )\). Hence,
\begin{align} \label{start01}
    &\langle S_{1,1}({X^{(1)}}),\,\Delta\rangle = O_p\left( \frac{\|\Delta\|_F \sqrt{\log n}}{n \sqrt \rho_n} \right),
\end{align}
which implies that 
\begin{align*}
    &\langle S_{1,1}({X^{(1)}}) - S_{2,1}({X^{(2)}}),\, \Delta \rangle
= O_p\left(\frac{\|\Delta\|_F \sqrt{\rho_n\log n}}{n \rho_n}\right).
\end{align*}
For higher–order terms (\(l \ge 2\)), Theorem 1 of \citet{xia_normal_2021} implies
\[
\bigg|\langle S_{i,l}({X^{(i)}}),\, \Delta \rangle \bigg| \le \| S_{i,l}({X^{(i)}})\|_F \|\Delta\|_F
\lesssim \left(\frac{\|X^{(i)}\|}{n \rho_n} \right)^l\|\Delta\|_F
=O_p \left(\frac{\|\Delta\|_F}{(n\rho_n)^{l/2}}\right).
\]
Combining the sharp Bernstein bound for the first–order term from \cref{start01} with the rate for higher–order terms yields
\(
L(\Delta)
= O_p\left(\max \left(\frac{\|\Delta\|_F}{n\rho_n},\frac{\|\Delta\|_F \sqrt{\rho_n\log n}}{n \rho_n}\right)\right).
\)
\end{proof}

\subsection{Proof of \cref{lem:balanced-incoherence}}
\label{sec:balanced-incoherenceproof}
\begin{proof}
Because \(\col(P)=\col(Z)\) and \(\rank(P)=k\), the \(k\) nonzero population eigenvectors of \(P\) span the column space of \(Z\). Hence there exists a full-rank matrix \(W\in\mathbb{R}^{k\times k}\) such that
\(
V = ZW.
\)
By construction each row of \(Z\) is a standard basis vector: if node \(i\) belongs to community \(a\) then the \(i\)-th row of \(Z\) equals \(e_a\t\). Therefore the rows of \(V\) are constant on communities, namely for any node \(i\) in community \(a\) we have
\(
V_{i,\cdot} = W_{a,\cdot}.
\)
Orthonormality of the columns of \(V\) gives
\(
I_k= V\t V = W\t Z\t Z\, W =: W\t N W.
\)
Define \(U := N^{1/2} W\). Then
\(
U\t U = W\t N W = I_k,
\)
so \(U\) is a \(k\times k\) orthonormal matrix. Consequently
\(
W = N^{-1/2} U.
\)
Examining a row \(a\) of \(W\) yields
\(
\|W_{a,\cdot}\| = \frac{1}{\sqrt{n_a}} \,\|U_{a,\cdot}\|.
\)
Since \(U\) is square and orthonormal, each row of \(U\) has Euclidean norm equal to \(1\). Hence for every community \(a\),
\(
\|W_{a,\cdot}\| = \frac{1}{\sqrt{n_a}}.
\)
Therefore for any node \(i\) in community \(a\),
\(
\|V_{i,\cdot}\| = \|W_{a,\cdot}\| = \frac{1}{\sqrt{n_a}}.
\)
Under the balanced-size assumption \( n_a \asymp \,n/k\) we obtain
\(
\|V\|_{2,\infty} \asymp \sqrt{\frac{k}{c\,n}}.
\)
\end{proof}

\subsection{Proof of \cref{lem:balanced-mmsb-incoherence}} \label{sec:balanced-mmsbm-incoherenceproof}
\begin{proof}
The proof follows the same high-level steps as \cref{lem:balanced-incoherence} for the SBM. Because \(\col(P)=\col(Z)\) and \(\rank(P)=k\), there exists a full–rank matrix \(W\in\mathbb{R}^{k\times k}\) such that
\(
V = ZW.
\)
Orthonormality of the columns of \(V\) yields the same relation
\(
I_k= W\t N W.
\)
Setting \(U:=N^{1/2}W\) we have
\(
W = N^{-1/2} U
\) implying \(
V = Z N^{-1/2} U.
\) Right multiplication by the orthonormal matrix \(U\) preserves row Euclidean norms. Therefore for each node \(i\),
\(
\|V_{i,\cdot}\| = \|\,Z_{i,\cdot} N^{-1/2} U\,\| = \|\,Z_{i,\cdot} N^{-1/2}\,\|.
\) By the operator norm inequality,
\(
\|V_{i,\cdot}\| \le \|Z_{i,\cdot}\| \,\|N^{-1/2}\| = \frac{\|Z_{i,\cdot}\|}{\sqrt{\lambda_{\min}(N)}}.
\)
Since each row \(Z_{i,\cdot}\) is a probability vector we have \(\|Z_{i,\cdot}\|\le\|Z_{i,\cdot}\|_1=1\), and thus
\(
\|V_{i,\cdot}\| \le \frac{1}{\sqrt{\lambda_{\min}(N)}}\) for all \(i.\) Finally applying the assumed spectral bound \(\lambda_{\min}(N)\ge c_1\,n/k\)   completes the proof.
\end{proof}

\section{Proof of \cref{prop:GRDPG-equivalence,prop:subspace-vs-membership,prop:eig-scaling-sbm-mmsbm,rem:Minimax2} and \cref{thm:twosample_dense}}\label{sec:extra}
\subsection{Proof of \cref{prop:GRDPG-equivalence}}\label{prop1}

\begin{proof}
Since \(V^{(i)}\) contains the eigenvectors corresponding to the non-zero eigenvalues of \(P^{(i)}\), we have \(\col(V^{(i)}) = \col(P^{(i)})\).
Substituting the structure of the GRDPG probability matrix, we have
\[
\col(P^{(i)}) = \col(\Gamma^{(i)} I_{p_i,q_i} {\Gamma^{(i)}}\t).
\]
\(I_{p_i,q_i}\) is a diagonal matrix with entries \(\pm 1\) and hence, is non-singular. Furthermore, because \(\Gamma^{(i)}\) has full column rank \(k\), the product \(I_{p_i,q_i} {\Gamma^{(i)}}\t\) has full row rank \(k\). Right-multiplication by a full row-rank matrix preserves the column space of the left matrix. Therefore,
\(
\col(P^{(i)}) = \col(\Gamma^{(i)}).
\)
Consequently, the condition \(V^{(1)} {V^{(1)}}\t = V^{(2)} {V^{(2)}}\t\) holds if and only if \(\col(\Gamma^{(1)}) = \col(\Gamma^{(2)})\).
For two full-rank matrices \(\Gamma^{(1)}, \Gamma^{(2)} \in \mathbb{R}^{n \times k}\), their column spaces are identical if and only if there exists an invertible matrix \(W \in \mathbb{R}^{k \times k}\) such that \(\Gamma^{(1)} = \Gamma^{(2)} W\).
\end{proof}

\subsection{Proof of \cref{prop:subspace-vs-membership}}\label{prop2}
\begin{proof}
First note that, for each \(i\), \(\col(P^{(i)})\subseteq\col({Z^{(i)}})\) because \(P^{(i)}={Z^{(i)}}B^{(i)}{Z^{(i)}}\t\). Since \(B^{(i)}\) is full-rank and \(\rank(P^{(i)})=k\) by assumption, we must have \(\rank({Z^{(i)}})=k\) and therefore
\(
\col(P^{(i)})=\col({Z^{(i)}})\) for \( i=1,2.
\)
If \({V^{(1)}}{V^{(1)}}\t={V^{(2)}}{V^{(2)}}\t\), then the two projection (equivalently column) spaces coincide:
\[
\col({P^{(1)}})=\col({V^{(1)}})=\col({V^{(2)}})=\col({P^{(2)}}).
\]
Using the above equation, this is equivalent to \(\col(Z^{(1)})=\col(Z^{(2)})\). Conversely, if \(\col(Z^{(1)})=\col(Z^{(2)}),\) then \(\col({P^{(1)}})=\col({P^{(2)}})\) and hence \({V^{(1)}}{V^{(1)}}\t={V^{(2)}}{V^{(2)}}\t\). Because \(\col(Z^{(1)})=\col(Z^{(2)})\) and both \(Z^{(1)}\) and \(Z^{(2)}\) have full column rank, there exists an invertible \(k\times k\) matrix \(R\) such that
\(
Z^{(2)} = Z^{(1)} R.
\)

For the SBM, each canonical basis vector \(e_j\t\) appears as a row of \(Z^{(1)}\); applying \(R\) yields that the rows of \(Z^{(2)}\) are the vectors \(e_j\t R\). But each \(e_j\t R\) must itself be a canonical basis vector (since rows of \(Z^{(2)}\) are one-hot); hence \(R\) sends the standard basis to the standard basis and therefore \(R\) is a permutation matrix. Thus \(Z^{(2)} = Z^{(1)}\Pi\) for some permutation \(\Pi\), i.e., the membership matrices coincide up to relabeling of community indices. 

For the MMSBM case, consider the pure rows of $Z^{(1)}$. If $i_a$ is an index such that $Z^{(1)}_{i_a \cdot} = e_a^\top$, then evaluating $Z^{(2)}=Z^{(1)}R$ at row $i_a$ yields
$$
    Z^{(2)}_{i_a \cdot} = e_a^\top R.
$$
This implies that the rows of $R$ exactly coincide with the rows of $Z^{(2)}$ corresponding to the pure nodes of $Z^{(1)}$. Because the rows of $Z^{(2)}$ lie in the $(k-1)$-simplex, every row of $R$ must also reside in the simplex; in particular, $R$ is element-wise non-negative and its rows sum to one.

$Z^{(2)}$ also contains pure nodes for each community. Since $R$ is invertible, we can express the latent position relationship as $Z^{(1)} = Z^{(2)} R^{-1}$. Evaluating this equation at the pure nodes of $Z^{(2)}$ demonstrates that the rows of $R^{-1}$ similarly correspond to rows of $Z^{(1)}$, which implies that $R^{-1}$ is also an element-wise non-negative matrix.

Theorem 5.1 of \cite{ding2014matrix} establishes that a matrix and its inverse are both non-negative if and only if the matrix can be expressed as the product of a diagonal matrix with strictly positive diagonal entries and a permutation matrix. This result implies that $R$ must be a generalized permutation matrix. Because we have already established that the rows of $R$ must sum to one, its non-zero entries are forced to be exactly one. Hence, $R$ is a standard permutation matrix, and consequently, $Z^{(2)}=Z^{(1)}\Pi$ for some permutation $\Pi$.
\end{proof}

\subsection{Proof of \cref{prop:eig-scaling-sbm-mmsbm}}\label{prop3}
\begin{proof}
Write \(B=\rho_n\widetilde B\) and \(N=Z\t Z\). Then
\(
P = Z B Z\t = \rho_n\, Z\widetilde B Z\t.
\)
The nonzero eigenvalues of \(Z\widetilde B Z\t\) coincide with the eigenvalues of the \(k\times k\) matrix
\(
\widetilde B^{1/2} N \widetilde B^{1/2},
\)
so that for each \(1\le j\le k\),
\(
\lambda_j(P) = \rho_n\cdot \lambda_j\big(\widetilde B^{1/2} N \widetilde B^{1/2}\big).
\)
Suppose \(v_j\) is the eigenvector of \(\widetilde B^{1/2} N \widetilde B^{1/2}\) corresponding to the eigenvalue \(\lambda_j\big(\widetilde B^{1/2} N \widetilde B^{1/2}\big).\) Since \(\widetilde B\) and \(N\) are symmetric positive definite, we have
\begin{align*}
    \lambda_j\big(\widetilde B^{1/2} N \widetilde B^{1/2}\big) &=  v_j\t \big(\widetilde B^{1/2} N \widetilde B^{1/2}\big)v_j. \\
\intertext{which yields}
    \lambda_{\min}(N) \|\widetilde B^{1/2} v_j\|^2 &\le \lambda_j\big(\widetilde B^{1/2} N \widetilde B^{1/2}\big) \le \lambda_{\max}(N) \|\widetilde B^{1/2} v_j\|^2, \\
\intertext{and further bounding $\|\widetilde B^{1/2} v_j\|^2$ gives}
    \lambda_{\min}(N) \lambda_{\min}(\widetilde B) &\le\lambda_j\big(\widetilde B^{1/2} N \widetilde B^{1/2}\big) \le \lambda_{\max}(N) \lambda_{\max}(\widetilde B).
\end{align*}
Combining this result with \cref{ass:asymptotic1}, the eigenvalue bounds on $B$, and the assumption of balanced community sizes (which imply there exist constants \(b_1,b_2>0\) such that $ b_1 n \lesssim \lambda_{\min}(N) \le \lambda_{\max}(N) \lesssim b_2 n$), we obtain
\begin{align*}
    a_1 b_1 \, n &\le \lambda_j\big(\widetilde B^{1/2} N \widetilde B^{1/2}\big) \le a_2 b_2 \, n. \\
\intertext{This implies that}
    a_1 b_1 \, n\rho_n &\le \lambda_j(P) \le a_2 b_2 \, n\rho_n, \\
\intertext{which upon dividing by $n \rho_n$ yields}
    a_1 b_1 &\le \frac{\lambda_j(P)}{n \rho_n} \le a_2 b_2,
\end{align*}
completing the proof.
\end{proof}
\subsection{Proof of \cref{thm:twosample_dense}}\label{thm:twosample_densepf}
\begin{proof}
The conclusion follows directly from \cref{thm:twosample} together with \cref{lem:second-order2}.  
The latter establishes that 
\(\tilde \sigma_2^{2} \asymp k^{2}/(n^{3}\rho_{n}^{2})\).  
Moreover, under \cref{ass:dense} we have
\(
\frac{k}{n^{3/2}\rho_{n}}
\gg
\frac{k}{n^{2}\rho_{n}^{2}},
\)
which, in combination with the preceding variance characterization, yields the desired result.
\end{proof}

\subsection{Proof of \cref{rem:Minimax2}}\label{prop4}
We will break down the proof into several steps.
\subsubsection{Step 1: Initial Decomposition.} \label{part01}
 Assume that \(A\) belongs to $\mathcal{E}_{{\sf very \ good}}$ as defined in \cref{lem:Very Good Set}.
We slightly refine the proof of \cref{thm:1sampleconsis} where it is already shown that \(
|\hat{\mu}_1 - \mu_1| =O_p \left(\max \left(\frac{k}{n^2 \rho_n^2},\frac{k \sqrt{\log n}}{n^2 \rho_n^{3/2}}\right)\right).
\) We will now show that \(
|\hat{\mu}_1 - \mu_1| \asymp \max \left(\frac{k}{n^2 \rho_n^2},\frac{k \sqrt{\log n}}{n^2 \rho_n^{3/2}}\right).
\) We will decompose \(\mu_1^{(2)}\) and \(\hat \mu_1^{(2)}\) (defined in \cref{prf:FinEst1}) in a slightly different way for this. We have that
\begin{equation}
\begin{aligned}\label{inite_dec01}
\mu_1^{(2)} - \hat{\mu}_1^{(2)} 
&= \tr\left( V \Lambda^{-2} V\t\Diag(\Sigma \cdot d) - V \Lambda^{-2} V\t\Diag(\Sigma \cdot \mathbf{1}_n) \right) \\
&\quad + \tr\left( V \Lambda^{-2} V\t\Diag(\Sigma \cdot \mathbf{1}_n) - \hat V \hat \Lambda^{-2} \hat V\t \Diag(\hat{\Sigma} \cdot \mathbf{1}_n) \right) \\
&\quad + \tr\left(  \hat{V} \hat{\Lambda}^{-2} \hat{V}\t \Diag(\hat{\Sigma} \cdot \mathbf{1}_n) -  \hat{V} \hat{\Lambda}^{-2} \hat{V}\t \Diag(\hat{\Sigma} \cdot \hat{d}) \right).
\end{aligned}
\end{equation}
All the terms of \cref{inite_dec01} except the second one \(\tr\left( V \Lambda^{-2} V\t\Diag(\Sigma \cdot \mathbf{1}_n) - \hat V \hat \Lambda^{-2} \hat V\t \Diag(\hat{\Sigma} \cdot \mathbf{1}_n) \right)\) have already been shown to be \( o_p \left(\max \left(\frac{k}{n^2 \rho_n^2},\frac{k \sqrt{\log n}}{n^2 \rho_n^{3/2}}\right) \right)\) in the proof of \cref{thm:1sampleconsis} in \cref{prf:FinEst1}.
Our target will be to prove that
$$
    \tr\left( V \Lambda^{-2} V\t\,\Diag(P \cdot \mathbf{1}_n) - \hat V \hat \Lambda^{-2} \hat V\t\,\Diag({\hat P} \cdot  \mathbf{1}_n) \right) \asymp \max \left(\frac{k}{n^2 \rho_n^2},\frac{k \sqrt{\log n}}{n^2 \rho_n^{3/2}}\right),
$$
because we can show that 
$$
    \tr\left( V \Lambda^{-2} V\t\,\Diag((P \circ P) \cdot \mathbf{1}_n) - \hat V \hat \Lambda^{-2} \hat V\t\,\Diag(({\hat P} \circ {\hat P}) \cdot \mathbf{1}_n) \right) = o_p \left(\max \left(\frac{k}{n^2 \rho_n^2},\frac{k \sqrt{\log n}}{n^2 \rho_n^{3/2}}\right) \right)
$$
by an argument entirely analogous to that used earlier in the proof of \cref{inter1} in \cref{sec:interproof1}.

We have that
\begin{equation}
\begin{aligned} \label{nseries1}
    \tr\left( 
V \Lambda^{-2} V\t\,\Diag(P \cdot \mathbf{1}_n)
-
\hat V \hat \Lambda^{-2} \hat V\t\,\Diag({\hat P} \cdot \mathbf{1}_n)
\right)
    &=  \tr\left( \left(
V \Lambda^{-2} V\t
-
\hat V \hat \Lambda^{-2} \hat V\t \right)\,\Diag(P \cdot \mathbf{1}_n)
\right) \\
& \quad - \tr\left( 
\left(\hat V \hat \Lambda^{-2} \hat V\t - V \Lambda^{-2} V\t \right)\,\Diag(( {\hat P}-P) \cdot \mathbf{1}_n)
\right)\\
& \quad - \tr\! 
\left( V  \Lambda^{-2}  V\t \,\Diag(( {\hat P}-P) \cdot \mathbf{1}_n)
\right).
\end{aligned}
\end{equation}

We now expand the terms in \cref{nseries1} one by one using \cref{lem:series-expansion}, \cref{dseries}, and Theorem~1 of \citet{xia_normal_2021}. This yields the following decompositions:
\begin{equation}
\begin{aligned} \label{nseries2}
\tr\bigg(&
\bigl(\widehat V \widehat \Lambda^{-2} \widehat V\t
-
V \Lambda^{-2} V\t\bigr)\,\Diag(P \cdot \mathbf{1}_n)
\bigg) \\
&=
\tr\Big(
\bigl[
V\Lambda^{-2}V\t S_1(X)
+
S_1(X)V\Lambda^{-2}V\t 
-
V\Lambda^{-2}V\t(XP+PX)V\Lambda^{-2}V\t
\bigr]\Diag(P \cdot \mathbf{1}_n)
\Big) \\[4pt]
&\quad+
\tr\Big(
S_1(X)\bigl(\widehat V\widehat\Lambda^{-2}\widehat V\t
-
V\Lambda^{-2}V\t\bigr)
\Diag(P \cdot \mathbf{1}_n)
\Big) \\[4pt]
&\quad-
\tr\Big(
V\Lambda^{-2}V\t(XP+PX)
\bigl(\widehat V\widehat\Lambda^{-2}\widehat V\t
-
V\Lambda^{-2}V\t\bigr)
\Diag(P \cdot \mathbf{1}_n)
\Big) \\[4pt]
&\quad+
\tr\Big(
\bigl(
\sum_{j\ge2}V\Lambda^{-2}V\t S_j(X)
+
\sum_{j\ge2}S_j(X)\widehat V\widehat\Lambda^{-2}\widehat V\t -
V\Lambda^{-2}V\t X^2\widehat V\widehat\Lambda^{-2}\widehat V\t
\bigr)\Diag(P \cdot \mathbf{1}_n)
\Big).
\end{aligned}
\end{equation}
Similarly,
\begin{equation}
\begin{aligned} \label{nseries3}
\tr\bigg( &
\bigl(\widehat V \widehat \Lambda^{-2} \widehat V\t
-
V \Lambda^{-2} V\t\bigr)\,
\Diag\bigl(({\widehat P}-P)\cdot \mathbf{1}_n\bigr)
\bigg) \\
&=
\sum_{l \ge 1}
\tr\Big(
\bigl[
V\Lambda^{-2}V\t S_1(X)
+
S_1(X)V\Lambda^{-2}V\t
-
V\Lambda^{-2}V\t(XP+PX)V\Lambda^{-2}V\t
\bigr]\Diag(T_l(X)\cdot \mathbf{1}_n)
\Big) \\[4pt]
&\quad+
\sum_{l \ge 1}
\tr\Big(
S_1(X)\bigl(\widehat V\widehat\Lambda^{-2}\widehat V\t
-
V\Lambda^{-2}V\t\bigr)
\Diag(T_l(X)\cdot \mathbf{1}_n)
\Big) \\[4pt]
&\quad-
\sum_{l \ge 1}
\tr\Big(
V\Lambda^{-2}V\t(XP+PX)
\bigl(\widehat V\widehat\Lambda^{-2}\widehat V\t
-
V\Lambda^{-2}V\t\bigr)
\Diag(T_l(X)\cdot \mathbf{1}_n)
\Big) \\[4pt]
&\quad+
\sum_{l \ge 1}
\tr\Big(
\bigl(
\sum_{j\ge2}V\Lambda^{-2}V\t S_j(X)
+
\sum_{j\ge2}S_j(X)\widehat V\widehat\Lambda^{-2}\widehat V\t -
V\Lambda^{-2}V\t X^2\widehat V\widehat\Lambda^{-2}\widehat V\t
\bigr)\Diag(T_l(X)\cdot \mathbf{1}_n)
\Big).
\end{aligned}
\end{equation}
Finally,
\begin{equation}
\begin{aligned} \label{nseries4}
\tr\left(
V\Lambda^{-2}V\t\,
\Diag\bigl(({\widehat P}-P)\cdot \mathbf{1}_n\bigr)
\right)
=
\sum_{l \ge 1}
\tr\left(
V\Lambda^{-2}V\t\,
\Diag(T_l(X)\cdot \mathbf{1}_n)
\right).
\end{aligned}
\end{equation}
Thus, we have shown that
\begin{align} \label{note06}
    \mu_1^{(2)} - \hat\mu_1^{(2)} &= \tr\left( 
V \Lambda^{-2} V\t\,\Diag(P \cdot \mathbf{1}_n)
-
\hat V \hat \Lambda^{-2} \hat V\t\,\Diag({\hat P} \cdot \mathbf{1}_n)
\right) +  o_p \left(\max \left(\frac{k}{n^2 \rho_n^2},\frac{k \sqrt{\log n}}{n^2 \rho_n^{3/2}}\right) \right).
\end{align}
As in the proofs of \cref{inter-d1,inter-a}, one can verify that all linear in $X$ contributions of the first term in \cref{note06} (expanded term-wise in \cref{nseries1,nseries2,nseries3,nseries4}) are of order
\(
O_p\left(\frac{k \log n}{n^{2}\rho_n^{3/2}}\right).
\)
Next, we will find the quadratic contributions of the same trace term in \cref{note06}.


\subsubsection{Step 2: Concentration of order-two terms around their mean.}\label{part03}
We now demonstrate that the combined contribution of the second-order terms in \cref{nseries1} is of order $\max \left(\frac{k}{n^2 \rho_n^2},\frac{k \sqrt{\log n}}{n^2 \rho_n^{3/2}}\right)$. Recall the term-by-term expansions of \cref{nseries1} provided in \cref{nseries2,nseries3,nseries4}. Focusing first on \cref{nseries2} and isolating all terms that depend quadratically on $X$, we obtain:
\begin{equation}
\begin{aligned} \label{nseries5}
\text{Term 1}
&=
\tr\Big(
S_1(X)^2\, V\Lambda^{-2}V\t
\Diag(P \cdot \mathbf{1}_n)
\Big)
-
\tr\Big(
S_1(X)V\Lambda^{-2}V^\top (PX+XP)V\Lambda^{-2}V^\top
\Diag(P \cdot \mathbf{1}_n)
\Big)\\
&\quad+
\tr\Big(
S_1(X) V\Lambda^{-2}V\t S_1(X)
\Diag(P \cdot \mathbf{1}_n)
\Big)
-
\tr\Big(
V\Lambda^{-2}V\t(PX+XP) S_1(X) V\Lambda^{-2}V\t
\Diag(P \cdot \mathbf{1}_n)
\Big)\\
&\quad+
\tr\Big(
V\Lambda^{-2}V\t(PX+XP) V\Lambda^{-2}V^\top (PX+XP) V\Lambda^{-2}V^\top
\Diag(P \cdot \mathbf{1}_n)
\Big)\\
&\quad-
\tr\Big(
V\Lambda^{-2}V\t(PX+XP) V\Lambda^{-2}V\t S_1(X)
\Diag(P \cdot \mathbf{1}_n)
\Big)\\
&\quad+
\tr\Big(
\bigl(
V\Lambda^{-2}V\t S_2(X)
+
S_2(X)V\Lambda^{-2}V\t
-
V\Lambda^{-2}V\t X^2V\Lambda^{-2}V\t
\bigr)\Diag(P \cdot \mathbf{1}_n)
\Big)\\
&=
\tr\Bigg( 
\Big(
\beta^{\perp} X \beta^{-2} X \beta^{-2}
-
\beta^{\perp} X \beta^{-3} X \beta^{-1}
+
\beta^{\perp} X \beta^{-4} X \beta^{\perp}
-
\beta^{-1} X \beta^{\perp} X \beta^{-3}
+
\beta^{-2} X \beta^{-1} X \beta^{-1}
\\
&\qquad\quad
-
\beta^{-2} X \beta^{-2} X \beta^{\perp} -
\beta^{-2} X \beta^{\perp} X \beta^{-2}
+
\beta^{-1} X \beta^{-1} X \beta^{-2}
+
\beta^{-1} X \beta^{-2} X \beta^{-1}
-
\beta^{-1} X \beta^{-3} X \beta^{\perp} \\
&\qquad\quad
+
\beta^{-4} X \beta^{\perp} X \beta^{\perp}
-
\beta^{-3} X \beta^{\perp} X \beta^{-1}
-
\beta^{-3} X \beta^{-1} X \beta^{\perp}
+
\beta^{\perp} X \beta^{\perp} X \beta^{-4}
-
\beta^{\perp} X \beta^{-1} X \beta^{-3}
\Big)
\Diag(P \cdot \mathbf{1}_n)
\Bigg)\\
&=
\text{Term 0}  + 
\tr \left(\beta^{-4} X \beta^{\perp} X \beta^{\perp} \Diag(P \cdot \mathbf{1}_n)\right)  
+ 
\tr \left(\beta^{\perp} X \beta^{\perp} X \beta^{-4} \Diag(P \cdot \mathbf{1}_n)\right) \\
& \qquad \qquad + \tr \left(\beta^{\perp} X \beta^{-4} X \beta^{\perp} \Diag(P \cdot \mathbf{1}_n)\right),
\end{aligned}
\end{equation}
where we group the remaining components into \(\text{Term 0}\), defined as
\begin{equation*}
\begin{aligned}
    \text{Term 0} \coloneqq \tr\Bigg( \Big(
    &\beta^{\perp} X \beta^{-2} X \beta^{-2}
    - \beta^{\perp} X \beta^{-3} X \beta^{-1}
    - \beta^{-1} X \beta^{\perp} X \beta^{-3}
    + \beta^{-2} X \beta^{-1} X \beta^{-1}
    -
\beta^{-2} X \beta^{-2} X \beta^{\perp} \\
    &+ \beta^{-1} X \beta^{-1} X \beta^{-2}
    + \beta^{-1} X \beta^{-2} X \beta^{-1}
    - \beta^{-1} X \beta^{-3} X \beta^{\perp}
    - \beta^{-3} X \beta^{\perp} X \beta^{-1} \\
    &- \beta^{-3} X \beta^{-1} X \beta^{\perp}
    - \beta^{\perp} X \beta^{-1} X \beta^{-3} -
\beta^{-2} X \beta^{\perp} X \beta^{-2}
    \Big) \Diag(P \cdot \mathbf{1}_n) \Bigg).
\end{aligned}
\end{equation*}
We now proceed to bound all terms in \cref{nseries5} starting first with
\[
    \tr \left(\beta^{-4} X \beta^{\perp} X \beta^{\perp} \Diag(P \cdot \mathbf{1}_n)\right) = \sum_{j,k,l,m} w^{(1)}_{jklm}\, X_{jk} X_{lm},
\]
where the coefficients are defined as
\begin{align*}
    w^{(1)}_{jklm} = \sum_{i,q} (\beta^{-4})_{ij} (\beta^\perp)_{kl} (\beta^\perp)_{mi} P_{iq}.
\end{align*}
Under \cref{ass:asymptotic1,ass:eigen-scaling,ass:incoherence}, these coefficients satisfy the bounds
\[
    \bigl| w^{(1)}_{jklm} \bigr| \lesssim
    \begin{cases}
        \dfrac{k^2}{n^4 \rho_n^3}, & k=l, \\[10pt]
        \dfrac{k^2}{n^5 \rho_n^3}, & k\neq l.
    \end{cases}
\]
Define 
\begin{align} \label{e.1}
    E_1 = \E \left(\tr \left(\beta^{-4} X \beta^{\perp} X \beta^{\perp} \Diag(P \cdot \mathbf{1}_n)\right) \right).
\end{align}
Consequently, it is straightforward to verify that the expectation satisfies
\begin{align*} 
    \big|E_1\big| \lesssim \frac{k^2}{n^2 \rho_n^2}.
\end{align*}
Next, we determine the asymptotic order of the corresponding variance:
\begin{equation}
\begin{aligned}\label{varref01}
    \Var \left(\sum_{j,k,l,m} w^{(2)}_{jklm} X_{jk} X_{lm}\right) 
    &= \sum_{j,k,l,m} \left(w^{(2)}_{jklm} \right)^2 \Var \left(X_{jk} X_{lm} \right) + \sum_{(j,k,l,m) \neq \atop (j',k',l',m')} w^{(2)}_{jklm} w^{(2)}_{j' k' l' m'} \Cov \left(X_{jk} X_{lm},X_{j'k'} X_{l'm'} \right) \\
    &= \sum_{j,k,l,m \atop k \neq l} \left(w^{(2)}_{jklm} \right)^2 \Var \left(X_{jk} X_{lm} \right) + \sum_{j,k,l,m \atop k = l} \left(w^{(2)}_{jklm} \right)^2 \Var \left(X_{jk} X_{lm} \right) \\
    &\quad + \sum_{(j,k,l,m) \neq \atop (j',k',l',m')} w^{(2)}_{jklm} w^{(2)}_{j' k' l' m'} \Cov \left(X_{jk} X_{lm},X_{j'k'} X_{l'm'} \right) \\
    &\lesssim \frac{k^4}{n^5 \rho_n^4}.
\end{aligned}
\end{equation}
Therefore, an application of Chebyshev's inequality yields that, for some constant \(C_1>0\),
\begin{align} \label{nseries6.4}
    \mathbb{P} \left(\left|\tr \left(\beta^{-4} X \beta^{\perp} X \beta^{\perp} \Diag(P \cdot \mathbf{1}_n)\right)- E_1 \right| \le \frac{k^2}{n^{5/2} \rho_n^2} \right) \ge C_1.
\end{align}
By analogous reasoning as in \cref{nseries6.4}, there exists a constant \(C_2>0\) such that
\begin{align} \label{nseries6.5}
    \mathbb{P} \left(\left|\tr \left(\beta^{\perp} X \beta^{\perp} X \beta^{-4} \Diag(P \cdot \mathbf{1}_n)\right)- E_2 \right| \le \frac{k^2}{n^{5/2} \rho_n^2} \right) \ge C_2,
\end{align}
where  
\begin{align} \label{e.2}
    E_2 = \E \left(\tr \left(\beta^{\perp} X \beta^{\perp} X \beta^{-4} \Diag(P \cdot \mathbf{1}_n)\right)\right).
\end{align}
Next, we will bound the term  \(\tr \left(\beta^{\perp} X \beta^{-4} X \beta^{\perp} \Diag(P \cdot \mathbf{1}_n)\right)\), where we we again do our routine decomposition:
\[\tr \left(\beta^{\perp} X \beta^{-4} X \beta^{\perp} \Diag(P \cdot \mathbf{1}_n)\right)=\sum_{j,k,l,m} w^{(4)}_{jklm}\, X_{jk} X_{lm}.\]
The coefficients are defined as
\begin{align*}
    w^{(3)}_{jklm} = \sum_{i,q} (\beta^\perp)_{ij} (\beta^{-4})_{kl} (\beta^\perp)_{mi} P_{iq}.
\end{align*}
Under \cref{ass:asymptotic1,ass:eigen-scaling,ass:incoherence}, \[
    \bigl| w^{(3)}_{jklm} \bigr| \lesssim
    \begin{cases}
        \dfrac{k^2}{n^4 \rho_n^3}, & j=m, \\[10pt]
        \dfrac{k^2}{n^5 \rho_n^3}, & j\neq m.
    \end{cases}
\]
By analogous reasoning as in \cref{nseries6.4}, there exists a constant \(C_3>0\) such that
\begin{align} \label{nseries6.6}
    \mathbb{P} \left(\left|\tr \left(\beta^{\perp} X \beta^{-4} X \beta^{\perp} \Diag(P \cdot \mathbf{1}_n)\right)- E_3 \right| \lesssim \frac{k^2}{n^{5/2} \rho_n^2} \right) \ge C_3,
\end{align}
where 
\begin{align}\label{e.3}
    E_3 = \E \left(\tr \left(\beta^{\perp} X \beta^{-4} X \beta^{\perp} \Diag(P \cdot \mathbf{1}_n)\right)\right).
\end{align}
Identical asymptotic bounds hold for Term 0 in \cref{nseries5}. Suppose we take $\tr \left(\beta^\perp X \beta^{-3}X \beta^{-1} \Diag(P \mathbf{1}_n) \right)$ as a representative term. Expanding the term, we get
$$
    \tr \left(\beta^\perp X \beta^{-3}X \beta^{-1} \Diag(P \mathbf{1}_n) \right) = \sum_{j,u,v,l} w^{(0)}_{juvl}X_{ju}X_{vl},
$$
where $w^{(0)}_{juvl}=\sum_{i,s} (\beta^\perp)_{ij}(\beta^{-3})_{uv}(\beta^{-1})_{li}P_{is}$. We can check that under \cref{ass:asymptotic1,ass:eigen-scaling,ass:incoherence}, 
$$
    \bigl| w^{(0)}_{juvl} \bigr| \lesssim \frac{k^2}{n^5 \rho_n^3}.
$$
Thus,
$$
    \left|\E \left[ \tr \left(\beta^\perp X \beta^{-3}X \beta^{-1} \Diag(P \mathbf{1}_n) \right) \right] \right| = O \left(\frac{k^2}{n^3 \rho_n^2}\right),
$$
and, proceeding similarly to \cref{varref01}, we have
$$
    \Var \left(\tr \left(\beta^\perp X \beta^{-3}X \beta^{-1} \Diag(P \mathbf{1}_n) \right) \right) = O \left(\max \left(\frac{k^2}{n^3 \rho_n^2}, \frac{k^2}{n^4 \rho_n^{5/2}}\right)\right).
$$
Applying Chebyshev's inequality, we obtain 
\begin{equation}\label{nseries6.2}
    \mathbb{P} \left(\big| \text{Term 0} \big| \lesssim \frac{k^2}{n^3 \rho_n^2}\right) \ge C_4
\end{equation}
for some constant $C_4>0$.
Therefore, combining the decomposition of \cref{nseries5} with the concentration bounds established in \cref{nseries6.4,nseries6.5,nseries6.6,nseries6.2}, there exists a constant $C_5 > 0$ such that
\begin{equation}\label{nseries6.9}
    \mathbb{P} \left( \big| \text{Term 1} - (E_1 + E_2 + E_3) \big| \ll \frac{k^2}{n^2 \rho_n^2} \right) \ge C_5.
\end{equation}

Next, we move on to the second-order contributions from \cref{nseries3}:
\begin{equation}
\begin{aligned} \label{nseries6}
\text{Term 2}
&=
\tr\Big(
\bigl[
V\Lambda^{-2}V\t S_1(X)
+
S_1(X)V\Lambda^{-2}V\t 
-
V\Lambda^{-2}V\t(XP+PX)V\Lambda^{-2}V\t
\bigr]\Diag(T_1(X)\cdot \mathbf{1}_n)
\Big) \\
&=
\tr\Big(
\bigl[
\beta^{-3} X \beta^{\perp}
+
\beta^{\perp} X \beta^{-3}
-
\beta^{-2} X \beta^{-1}
-
\beta^{-1} X \beta^{-2}
\bigr]
\Diag(T_1(X)\cdot \mathbf{1}_n)
\Big).
\end{aligned}
\end{equation}
Focusing on the first term of \cref{nseries6}, we expand the trace as
\begin{equation}
\begin{aligned} \label{nseries6.1}
    \tr \left( \beta^{-3} X \beta^{\perp} \Diag(T_1(X)\cdot \mathbf{1}_n) \right) & = \sum_{i,j,k,l} (\beta^{-3})_{ij} (\beta^\perp)_{ki} X_{jk} (T_1(X))_{il}\\
    & \lesssim \sum_{i,j,k,m,l} (\beta^{-3})_{ij} (\beta^\perp)_{ki} (VV\t)_{im} X_{jk} X_{ml}\\
    & \lesssim \sum_{j,k,m,l} \underbrace{\left(\sum_i(\beta^{-3})_{ij} (\beta^\perp)_{ki} (VV\t)_{im} \right)}_{w^{(4)}_{jklm}} X_{jk} X_{ml}.
\end{aligned}
\end{equation}
Under \cref{ass:asymptotic1,ass:eigen-scaling,ass:incoherence}, it follows that \(w^{(4)}_{jklm} \lesssim \frac{k^2}{n^2(n \rho_n)^3}\). Consequently, it is straightforward to verify the expectation
\[
    \left|\E \left(\sum_{j,k,m,l} w^{(4)}_{jklm} X_{jk} X_{ml}\right) \right| \lesssim \frac{k^2}{n^3 \rho_n^2}.
\]
Next, we determine the order of the corresponding variance:
\begin{align*}
    \Var \left(\sum_{j,k,m,l} w^{(4)}_{jklm} X_{jk} X_{ml}\right) & = \sum_{j,k,m,l} \left(w^{(4)}_{jklm} \right)^2 \Var \left(X_{jk} X_{ml} \right) + \sum_{(j,k,m,l) \neq \atop (j',k',m',l')} w^{(4)}_{jklm} w^{(4)}_{j'k'l'm'} \Cov \left(X_{jk} X_{ml},X_{j'k'} X_{m'l'} \right) \\
    & \lesssim \frac{k^4}{n^6 \rho_n^4}.
\end{align*}
Therefore, an application of Chebyshev's inequality to \cref{nseries6.1} implies that with constant probability
\begin{align} \label{nseries6.20}
    \big | \tr \left( \beta^{-3} X \beta^{\perp} \Diag(T_1(X)\cdot \mathbf{1}_n) \right) \big| \lesssim \frac{k^2}{n^3 \rho_n^2}.
\end{align}
By analogous reasoning, identical asymptotic bounds hold for the remaining components in \cref{nseries6}, which ultimately yields 
\begin{align}\label{nseries6.3}
    \mathbb{P} \left(\big| \text{Term 2} \big| \lesssim \frac{k^2}{n^3 \rho_n^2}\right) \ge C_6.
\end{align}
for some constant \(C_6>0\).
\\ \ \\
Finally, the sole second-order contribution from \cref{nseries4} is given by
\[
    \text{Term 3} = \tr\left(\beta^{-2}\Diag(T_2(X)\cdot \mathbf{1}_n)\right).
\]
Invoking \cref{lem:series-expansion}, we obtain
\[
    T_2(X) = \bigl(\beta^{\perp} X \beta^{\perp} X \beta^{-1} + \beta^{-1} X \beta^{\perp} X \beta^{\perp} + \beta^{\perp} X \beta^{-1} X \beta^{\perp}\bigr),
\]
which implies that
\[
    \text{Term 3} = \sum_{a,b,c,d} w^{(5)}_{abcd}\, X_{ab} X_{cd},
\]
where the coefficients are defined as
\begin{align*}
    w^{(5)}_{abcd} &= \sum_{i,j} (\beta^{-2})_{ii} \Big[(\beta^\perp)_{ia} (\beta^\perp)_{bc} (\beta^{-1})_{dj} + (\beta^\perp)_{ia} (\beta^{-1})_{bc} (\beta^\perp)_{dj} + (\beta^{-1})_{ia} (\beta^\perp)_{bc} (\beta^\perp)_{dj}\Big].
\end{align*}
Under \cref{ass:asymptotic1,ass:eigen-scaling,ass:incoherence}, these coefficients satisfy the bounds
\[
    \bigl| w^{(5)}_{abcd} \bigr| \lesssim
    \begin{cases}
        \dfrac{k^2}{n^4 \rho_n^3}, & b=c, \\[10pt]
        \dfrac{k^2}{n^5 \rho_n^3}, & b\neq c.
    \end{cases}
\]
Define 
\begin{align}\label{e.4}
    E_4 &= \E \left(\tr\left(\beta^{-2}\Diag(T_2(X)\cdot \mathbf{1}_n)\right) \right) \\ \notag
    &= \E \left(\tr\left(\beta^{-2}\Diag(\left( \beta^{\perp} X \beta^{\perp} X \beta^{-1} + \beta^{-1} X \beta^{\perp} X \beta^{\perp} + \beta^{\perp} X \beta^{-1} X \beta^{\perp}\right)\cdot \mathbf{1}_n)\right) \right).
\end{align}
Consequently, it is straightforward to verify that the expectation satisfies
\[
    \left |E_4 \right| = \left |\E \left(\sum_{a,b,c,d} w^{(5)}_{abcd} X_{ab} X_{cd}\right) \right| \lesssim \frac{k^2}{n^2 \rho_n^2}.
\]
Next, we determine the asymptotic order of the corresponding variance:
\begin{align*}
    \Var \left(\sum_{a,b,c,d} w^{(5)}_{abcd} X_{ab} X_{cd}\right) 
    &= \sum_{a,b,c,d} \left(w^{(5)}_{abcd} \right)^2 \Var \left(X_{ab} X_{cd} \right) + \sum_{(a,b,c,d) \neq \atop (a',b',c',d')} w^{(5)}_{abcd} w^{(5)}_{a'b'c'd'} \Cov \left(X_{ab} X_{cd},X_{a'b'} X_{c'd'} \right) \\
    &= \sum_{a,b,c,d \atop b \neq c} \left(w^{(5)}_{abcd} \right)^2 \Var \left(X_{ab} X_{cd} \right) + \sum_{a,b,c,d \atop b = c} \left(w^{(5)}_{abcd} \right)^2 \Var \left(X_{ab} X_{cd} \right) \\
    &\quad + \sum_{(a,b,c,d) \neq \atop (a',b',c',d')} w^{(5)}_{abcd} w^{(5)}_{a'b'c'd'} \Cov \left(X_{ab} X_{cd},X_{a'b'} X_{c'd'} \right) \\
    &\lesssim \frac{k^4}{n^5 \rho_n^4}.
\end{align*}
Therefore, an application of Chebyshev's inequality yields that, for some constant \(C_7>0\),
\begin{align} \label{nseries6.7}
    \mathbb{P} \left(\left|\text{Term 3}- E_4  \right| \le \frac{k^2}{n^{5/2} \rho_n^2} \right) \ge C_7.
\end{align} 
Combining all the second-order terms from \cref{nseries6.9,nseries6.3,nseries6.7}, we conclude that there exists a constant \(C_8>0\) such that 
\begin{align} \label{nseries6.8}
    \mathbb{P} \left(\left|\text{Term 4}- E \right| \ll \frac{k^2}{n^{2} \rho_n^2} \right) \ge C_8,
\end{align} 
where \(\text{Term 4}=\text{Term 1}-\text{Term 2}-\text{Term 3}\) contains all the quadratic terms of \(X\) in \cref{nseries1}, and \(E=E_1+E_2+E_3-E_4\) as defined earlier in \cref{e.1,e.2,e.3,e.4}.

Combining these arguments, we have therefore shown that with at least constant probability,
\begin{align}\label{note07}
    \mu_1^{(2)} - \widehat{\mu}_1^{(2)} &= E_1 + E_2 + E_3 - E_4 + (\text{terms of degree $\ge 3$ in $X$}) \\
    &\quad + o_p \left(\max \left(\frac{k}{n^2 \rho_n^2},\frac{k \sqrt{\log n}}{n^2 \rho_n^{3/2}}\right) \right).\notag
\end{align}

\subsubsection{Step 3: Computing the order of order-two mean.}\label{part04}
Using the identity 
\[
    \E \left(X \beta^\perp X \right)= \Sigma \circ \beta^\perp + \Diag \left( \left(\Sigma - \Diag \left(\Sigma \right) \right) \diag \left(\beta^\perp \right) \right),
\]
we can evaluate \(E_1, E_2, E_3\), and \(E_4\):
\begin{equation}
\begin{aligned} \label{e123}
    E_1 &= \tr \left(\beta^{-4} \left[ \Sigma \circ \beta^\perp + \Diag \left( \left(\Sigma - \Diag \left(\Sigma \right) \right) \diag \left(\beta^\perp \right) \right)\right] \beta^{\perp} \Diag \left(P \cdot \mathbf{1}_n \right) \right), \\
    E_2 &= \tr \left(\beta^{\perp} \left[ \Sigma \circ \beta^{\perp} + \Diag \left( \left(\Sigma - \Diag \left(\Sigma \right) \right) \diag \left(\beta^{\perp} \right) \right)\right] \beta^{-4} \Diag \left(P \cdot \mathbf{1}_n \right) \right),\\
    E_3 &= \tr \left(\beta^{\perp} \left[ \Sigma \circ \beta^{-4} + \Diag \left( \left(\Sigma - \Diag \left(\Sigma \right) \right) \diag \left(\beta^{-4} \right) \right)\right] \beta^{\perp} \Diag \left(P \cdot \mathbf{1}_n \right) \right), 
\end{aligned}
\end{equation}
and
\begin{equation}
\begin{aligned} \label{e4}
    E_4 &= \tr \left(\beta^{-2} \Diag \left( \beta^\perp \left[ \Sigma \circ \beta^\perp + \Diag \left( \left( \Sigma - \Diag \left( \Sigma \right) \right) \diag(\beta^\perp) \right) \right] \beta^{-1} \cdot \mathbf{1}_n \right) \right)  \\
    &\quad + \tr \left(\beta^{-2} \Diag \left( \beta^\perp \left[ \Sigma \circ \beta^{-1} + \Diag \left( \left( \Sigma - \Diag \left( \Sigma \right) \right) \diag(\beta^{-1}) \right) \right] \beta^{\perp} \cdot \mathbf{1}_n \right) \right) \\
    &\quad + \tr \left(\beta^{-2} \Diag \left( \beta^{-1} \left[ \Sigma \circ \beta^\perp + \Diag \left( \left( \Sigma - \Diag \left( \Sigma \right) \right) \diag(\beta^\perp) \right) \right] \beta^{\perp} \cdot \mathbf{1}_n \right) \right).
\end{aligned}
\end{equation}
For the specific case of the SBM family, we show that $E_1=E_3=E_4=0$ and $E_2 \asymp \frac{k}{n^2\rho_n^2}$.

\paragraph{Calculation for $E_1, \,E_2$:} 
Here we prove that $\beta^{\perp} \Diag \left(P \cdot \mathbf{1}_n \right) = \Diag \left(P \cdot \mathbf{1}_n \right) \beta^{\perp}$. We first analyze the structure of $\beta^\perp$ under the SBM:
$$
    \beta^\perp = I - VV\t = I - Z(Z\t Z)^{-1}Z\t.
$$
Thus, $\beta^\perp$ is a block-diagonal matrix with the $m$th block given by $(\beta^\perp)^{(m)}=I_{n_m}-\frac{1}{n_m}J_{n_m}$, where $n_m$ is the size of the $m$th community. Similarly, we have 
\begin{align} \label{last1}
    \Diag(P \cdot \mathbf{1}_n) = \Diag( Z B Z\t \mathbf{1}_n) = \Diag( Z \tilde{d} ),
\end{align}
where $\tilde{d}=B Z\t \mathbf{1}_n$ is a $k \times 1$ vector with elements $\tilde{d}_m$ for $m=1,\dots,k$. Thus, \cref{last1} implies that $\Diag(P \cdot \mathbf{1}_n)$ also has a block structure, with the $m$th block being $\Diag(P \cdot \mathbf{1}_n)^{(m)}=\tilde{d}_m I_{n_m}$. Since both $\beta^\perp$ and $\Diag(P \cdot \mathbf{1}_n)$ are block-diagonal matrices with matching block sizes, their product is also block-diagonal. The $m$th block of $\beta^{\perp} \Diag \left(P \cdot \mathbf{1}_n \right)$ is therefore $\tilde d_m \left(I_{n_m}-\frac{1}{n_m}J_{n_m}\right)$, which is identical to the $m$th block of $ \Diag \left(P \cdot \mathbf{1}_n \right)\beta^{\perp}$. This proves that $\beta^{\perp} \Diag \left(P \cdot \mathbf{1}_n \right) = \Diag \left(P \cdot \mathbf{1}_n \right) \beta^{\perp}$.

Consequently, for $E_1$, we obtain
\begin{equation}
\begin{aligned} \label{last02}
    E_1 &= \tr \left(\beta^{-4} \left[ \Sigma \circ \beta^\perp + \Diag \left( \left(\Sigma - \Diag \left(\Sigma \right) \right) \diag \left(\beta^\perp \right) \right)\right] \beta^{\perp} \Diag \left(P \cdot \mathbf{1}_n \right) \right)\\
    &= \tr \left(\left[ \Sigma \circ \beta^\perp + \Diag \left( \left(\Sigma - \Diag \left(\Sigma \right) \right) \diag \left(\beta^\perp \right) \right)\right] \beta^{\perp} \Diag \left(P \cdot \mathbf{1}_n \right)\beta^{-4} \right)\\
    &= \tr \left(\left[ \Sigma \circ \beta^\perp + \Diag \left( \left(\Sigma - \Diag \left(\Sigma \right) \right) \diag \left(\beta^\perp \right) \right)\right] \Diag \left(P \cdot \mathbf{1}_n \right) \beta^{\perp} \beta^{-4} \right) \\
    &= 0.
\end{aligned}
\end{equation}
The exact same argument applies to $E_2$ as well.
\paragraph{Calculation for $E_4$:} 
For the SBM, since $Z \cdot \mathbf{1}_k=\mathbf{1}_n$, we have $\mathbf{1}_n \in \col(Z)=\col(V)$ because the block matrix $B$ has full rank. This directly forces the second and third terms of \cref{e4} to vanish since $\beta^\perp \cdot \mathbf{1}_n = 0$.

Next, we evaluate the first term:
$$
    \tr \left(\beta^{-2} \Diag \left( \beta^\perp \left[ \Sigma \circ \beta^\perp + \Diag \left( \left( \Sigma - \Diag \left( \Sigma \right) \right) \diag(\beta^\perp) \right) \right] \beta^{-1} \cdot \mathbf{1}_n \right) \right).
$$
Our target is to show that the components $\Diag \left( \beta^\perp \left[ \Sigma \circ \beta^\perp \right] \beta^{-1} \cdot \mathbf{1}_n \right)$, $\Diag \left( \beta^\perp \left[ \Diag \left( \Sigma \, \diag(\beta^\perp) \right) \right] \beta^{-1} \cdot \mathbf{1}_n \right)$, and $\Diag \left( \beta^\perp \left[ \Diag \left( \Diag \left( \Sigma \right) \diag(\beta^\perp) \right) \right] \beta^{-1} \cdot \mathbf{1}_n \right)$ vanish separately. 

We established in the calculation for $E_1$ that $(\beta^\perp)^{(m)}=I_{n_m}-\frac{1}{n_m}J_{n_m}$. The matrix $\Sigma=P \circ (J_n-P)$ also possesses a block structure. Thus, $\Sigma \circ \beta^\perp$ is a block-diagonal matrix with the $m$th block given by $(\Sigma \circ \beta^\perp)^{(m)}=P_{mm}(1-P_{mm})\left(I_{n_m}-\frac{1}{n_m}J_{n_m}\right)$. Furthermore, $\beta^{-1}\cdot \mathbf{1}_n=V \Lambda^{-1} V\t \cdot \mathbf{1}_n \in \col(Z)$ is a block-constant vector, with $q_m$ denoting the value of the $m$th block. Evaluating the $m$th block of the vector $\left[ \Sigma \circ \beta^\perp \right] \beta^{-1} \cdot \mathbf{1}_n$, we find
\begin{equation}
\begin{aligned} \label{last03}
    \left(\left[ \Sigma \circ \beta^\perp \right] \beta^{-1} \cdot \mathbf{1}_n \right)^{(m)} &= \left(\left[ \Sigma \circ \beta^\perp \right]\right)^{(m)} q_m \mathbf{1}_{n_m}\\
    &= P_{mm}(1-P_{mm})q_m\left(I_{n_m}-\frac{1}{n_m}J_{n_m}\right)\mathbf{1}_{n_m} = 0.
\end{aligned}
\end{equation}
Thus, \cref{last03} proves that $\left[ \Sigma \circ \beta^\perp \right] \beta^{-1} \cdot \mathbf{1}_n = 0$.

Next, $\diag(\beta^\perp)$ is a block-constant vector with its $m$th block equal to $(1-\frac{1}{n_m})$. Consequently, $\Sigma \, \diag(\beta^\perp)$ is also a block-constant vector with its $m$th block equal to $P_{mm}(1-P_{mm})(1-\frac{1}{n_m})$. This yields $$\left(\Diag \left( \Sigma \, \diag(\beta^\perp) \right) \right)^{(m)} = P_{mm}(1-P_{mm})(1-\frac{1}{n_m})I_{n_m}.$$ We then have 
\begin{equation}
\begin{aligned} \label{last04}
    \left(\beta^\perp \left[ \Diag \left( \Sigma \, \diag(\beta^\perp) \right) \right] \beta^{-1} \cdot \mathbf{1}_n\right)^{(m)} 
    &= (\beta^\perp)^{(m)} \left(\Diag \left( \Sigma \, \diag(\beta^\perp) \right) \right) ^{(m)} \left(\beta^{-1} \cdot \mathbf{1}_n\right)^{(m)}\\
    &= P_{mm}(1-P_{mm}) q_{m} \left(1-\frac{1}{n_m}\right) (\beta^\perp)^{(m)}\mathbf{1}_{n_m} = 0.
\end{aligned}
\end{equation}
This proves that $\Diag \left( \beta^\perp \left[ \Diag \left( \Sigma \, \diag(\beta^\perp) \right) \right] \beta^{-1} \cdot \mathbf{1}_n \right) = 0$. By a nearly identical argument, one can show that $\Diag \left( \beta^\perp \left[ \Diag \left( \Diag \left( \Sigma \right) \diag(\beta^\perp) \right) \right] \beta^{-1} \cdot \mathbf{1}_n \right) = 0$. Combining this result with \cref{last03,last04} confirms that $E_4 = 0$.
\paragraph{Calculation for $E_3$:} 
From \cref{e.3}, we know that \[E_3 = \E \left(\tr \left(\beta^{\perp} X \beta^{-4} X \beta^{\perp} \Diag(P \cdot \mathbf{1}_n)\right)\right) = \E \left(\|\beta^{-2}X \beta^{\perp} \left(\Diag(P \cdot \mathbf{1}_n)\right)^{1/2}\|_F^2 \right)>0.\] Now, we evaluate the exact asymptotic order of $E_3$. We already know that $\Sigma$ is a block-constant matrix. Additionally, $\beta^{-4}= V \Lambda^{-4} V\t \in \col(Z)$ is also a block-constant matrix. Thus, $\Sigma \circ \beta^{-4}$ is a block-constant matrix, meaning that $\col(\Sigma \circ \beta^{-4}) \subseteq \col(Z)$. This implies that $\beta^\perp (\Sigma \circ \beta^{-4})=0$. Thus, from \cref{e123}, we have
\begin{equation}
\begin{aligned} \label{last05}
    E_3 &= \tr \left(\beta^{\perp} \left[ \Diag \left( \left(\Sigma - \Diag \left(\Sigma \right) \right) \diag \left(\beta^{-4} \right) \right)\right] \beta^{\perp} \Diag \left(P \cdot \mathbf{1}_n \right) \right)\\
    &= \sum_{i=1}^n (\Diag(P\cdot \mathbf{1}_n))_{ii} \cdot \left(\Diag \left( \left(\Sigma - \Diag \left(\Sigma \right) \right) \diag \left(\beta^{-4} \right) \right)\right)_{ii} \cdot (\beta^\perp)_{ii}.
\end{aligned}
\end{equation}
Now,
\begin{align*}
    \left(\Diag \left( \left(\Sigma - \Diag \left(\Sigma \right) \right) \diag \left(\beta^{-4} \right) \right)\right)_{ii}&= \sum_{j \neq i} \Sigma_{ij} \beta_{jj}^{-4}\\
    & \asymp \rho_n \tr(\beta^{-4})\\ 
    &\asymp \frac{k}{n^4 \rho_n^3}
\end{align*}
and
$$(\Diag(P\cdot \mathbf{1}_n))_{ii}=\sum_{j} P_{ij} \asymp n \rho_n.
$$
Thus, \cref{last05} implies that $E_3 \asymp \frac{k}{n^2 \rho_n^2}$.

Consequently, we have established that under the SBM, the combined term $E = E_1 + E_2 + E_3 - E_4 \asymp \frac{k}{n^2 \rho_n^2}$. Combining this result with \cref{note07} yields
\begin{align}\label{note08}
    \mu_1^{(2)} - \widehat{\mu}_1^{(2)} &\asymp \frac{k}{n^2 \rho_n^2} + (\text{terms of degree $\ge 3$ in $X$})  + o_p \left(\max \left(\frac{k}{n^2 \rho_n^2},\frac{k \sqrt{\log n}}{n^2 \rho_n^{3/2}}\right) \right).
\end{align}
In the final step, we bound the remaining higher-order terms to complete the proof.
\subsubsection{Step 4: Bounding higher-order terms and completing the proof.}
To bound the higher-order terms (those with power of \(X\) greater than or equal to \(3\)), we begin with the final term in \cref{nseries1} which has been expanded in \cref{nseries4}. Using \cref{lem:series-expansion}, we have
\begin{align*}
\tr 
\left( V  \Lambda^{-2}  V\t \,\Diag(( {\hat P}-P) \cdot \mathbf{1}_n)
\right) &= \sum_{l \ge 3} \tr\left( V \Lambda^{-2} V\t \Diag(T_l(X)\cdot \mathbf{1}_n) \right) \\
&\lesssim
\sum_{l \ge 3}
\frac{k}{(n\rho_n)^2}
\left(\frac{4}{n\rho_n}\right)^{l/2}
\, n\rho_n \\
&\lesssim
\frac{k}{(n\rho_n)^{2.5}}.
\end{align*}

To bound the remaining higher-order contributions from \cref{nseries1}, we revisit the decomposition in \cref{nseries2} and follow an argument entirely analogous to the proof of \cref{inter-d1} in \cref{sec:interproof3}. The sole distinction is that in \cref{sec:interproof3}, after applying Bernstein's inequality to the linear terms, we established a high-probability bound for terms of degree two or higher in $X$; here, however, we restrict our attention to terms of degree three or higher. The remainder of the proof remains unchanged, yielding a bound of $O_p\left(\frac{k}{(n\rho_n)^{2.5}}\right)$.
The higher-order terms in \cref{nseries3} can be handled analogously using the same reasoning as in the proof of \cref{inter-d1}. Again, the only modification is that we now discard all terms of order \(3\) or higher in \(X\). This gives the same bound,
\(
O_p\left(\frac{k}{(n\rho_n)^{2.5}}\right).
\)

Combining the above bounds for terms of order \(3\) or higher in \(X\) expanded in \cref{nseries2,nseries3,nseries4}, we conclude that the total contribution of all terms with power of \(X\) greater than or equal to \(3\) in \cref{nseries1} is
\(
O_p\left(\frac{k}{(n\rho_n)^{2.5}}\right)
=
o_p\left(\frac{k}{n^2\rho_n^2}\right).
\)
Thus, \cref{note08} now can be updated to conclude that
\begin{align*}
    \big|\mu_1^{(2)} - \widehat{\mu}_1^{(2)}\big| &\asymp \frac{k}{n^2 \rho_n^2}+ o_p\left(\frac{k}{n^2\rho_n^2}\right) + o_p \left(\max \left(\frac{k}{n^2 \rho_n^2},\frac{k \sqrt{\log n}}{n^2 \rho_n^{3/2}}\right) \right) \\
    &\asymp \frac{k}{n^2 \rho_n^2},
\end{align*}
which completes the proof.

\bibliographystyle{plainnat_JA}
\bibliography{two_sample_hyp_subspace,references}

@article{li2018two,
  title={Two-sample test of community memberships of weighted stochastic block models},
  author={Li, Yezheng and Li, Hongzhe},
  journal={arXiv preprint arXiv:1811.12593},
  year={2018}
}

@book{bai2010spectral,
  title={Spectral analysis of large dimensional random matrices},
  author={Bai, Zhidong and Silverstein, Jack W},
  volume={20},
  year={2010},
  publisher={Springer}
}

@misc{bts_flight_data,
  author       = {{Bureau of Transportation Statistics}},
  title        = {Airline On-Time Performance Data},
  year         = {2026},
  howpublished = {\url{https://www.transtats.bts.gov/DatabaseInfo.asp?QO_VQ=EEE}},
  note         = {Accessed: 2026-05-13}
}

@article{lei2015consistency,
title={Consistency of spectral clustering in stochastic block models},
author={Lei, Jing and Rinaldo, Alessandro},
journal={The Annals of Statistics},
volume={43},
number={1},
pages={215--237},
year={2015},
doi={10.1214/14-AOS1274},
url={https://arxiv.org/abs/1312.2050}
}

@article{bandeira2016sharp,
  title={Sharp nonasymptotic bounds on the norm of random matrices with independent entries},
  author={Bandeira, Afonso S. and van Handel, Ramon},
  journal={The Annals of Probability},
  volume={44},
  number={4},
  pages={2479--2506},
  year={2016},
  publisher={Institute of Mathematical Statistics},
  doi={10.1214/15-AOP1025}
}

@book{kolaczyk2009statistical,
  title={Statistical analysis of network data: methods and models},
  author={Kolaczyk, Eric D},
  year={2009},
  publisher={Springer}
}

@article{lazega2001collegial,
  title={The collegial phenomenon: The social mechanisms of cooperation among peers in a corporate law partnership},
  author={Lazega, Emmanuel},
  journal={Oxford University Press},
  year={2001}
}

@article{bullmore2009complex,
  title={Complex brain networks: graph theoretical analysis of structural and functional systems},
  author={Bullmore, Ed and Sporns, Olaf},
  journal={Nature reviews neuroscience},
  volume={10},
  number={3},
  pages={186--198},
  year={2009},
  publisher={Nature Publishing Group}
}

@article{girvan2002community,
  title={Community structure in social and biological networks},
  author={Girvan, Michelle and Newman, Mark EJ},
  journal={Proceedings of the national academy of sciences},
  volume={99},
  number={12},
  pages={7821--7826},
  year={2002},
  publisher={National Acad Sciences}
}

@article{zhu2006automatic,
  title={Automatic dimensionality selection from the scree plot via the use of profile likelihood},
  author={Zhu, Mu and Ghodsi, Ali},
  journal={Computational Statistics \& Data Analysis},
  volume={51},
  number={2},
  pages={918--930},
  year={2006},
  publisher={Elsevier}
}

@article{ding2014matrix,
  title={When a matrix and its inverse are nonnegative},
  author={Ding, Jiu and Rhee, Noah H},
  journal={Missouri Journal of Mathematical Sciences},
  volume={26},
  number={1},
  pages={98--103},
  year={2014},
  publisher={University of Central Missouri, Department of Mathematics, Actuarial Science~…}
}

@inproceedings{chatterjee_two-sample_2023,
	title = {Two-{Sample} {Tests} for {Inhomogeneous} {Random} {Graphs} in \${L}\_r\$ {Norm}: {Optimality} and {Asymptotics}},
	shorttitle = {Two-{Sample} {Tests} for {Inhomogeneous} {Random} {Graphs} in \${L}\_r\$ {Norm}},
	url = {https://proceedings.mlr.press/v206/chatterjee23a.html},
	language = {en},
	urldate = {2025-11-18},
	booktitle = {Proceedings of {The} 26th {International} {Conference} on {Artificial} {Intelligence} and {Statistics}},
	publisher = {PMLR},
	author = {Chatterjee, Sayak and Saha, Dibyendu and Dan, Soham and Bhattacharya, Bhaswar B.},
	month = apr,
	year = {2023},
	note = {ISSN: 2640-3498},
	pages = {6903--6911},
}

@article{tang_nonparametric_2017,
	title = {A nonparametric two-sample hypothesis testing problem for random graphs},
	volume = {23},
	issn = {1350-7265},
	url = {https://projecteuclid.org/journals/bernoulli/volume-23/issue-3/A-nonparametric-two-sample-hypothesis-testing-problem-for-random-graphs/10.3150/15-BEJ789.full},
	doi = {10.3150/15-BEJ789},
	number = {3},
	urldate = {2025-11-18},
	journal = {Bernoulli},
	author = {Tang, Minh and Athreya, Avanti and Sussman, Daniel L. and Lyzinski, Vince and Priebe, Carey E.},
	month = aug,
	year = {2017},
	note = {Publisher: Bernoulli Society for Mathematical Statistics and Probability},
	pages = {1599--1630},
}

@article{athreya_statistical_2018,
	title = {Statistical {Inference} on {Random} {Dot} {Product} {Graphs}: a {Survey}},
	volume = {18},
	issn = {1533-7928},
	shorttitle = {Statistical {Inference} on {Random} {Dot} {Product} {Graphs}},
	url = {http://jmlr.org/papers/v18/17-448.html},
	number = {226},
	urldate = {2025-11-18},
	journal = {Journal of Machine Learning Research},
	author = {Athreya, Avanti and Fishkind, Donniell E. and Tang, Minh and Priebe, Carey E. and Park, Youngser and Vogelstein, Joshua T. and Levin, Keith and Lyzinski, Vince and Qin, Yichen and Sussman, Daniel L.},
	year = {2018},
	pages = {1--92},
}

@article{jin_optimal_2025,
	title = {Optimal {Network} {Pairwise} {Comparison}},
	volume = {120},
	issn = {0162-1459},
	url = {https://doi.org/10.1080/01621459.2024.2393471},
	doi = {10.1080/01621459.2024.2393471},
	number = {550},
	urldate = {2025-11-18},
	journal = {Journal of the American Statistical Association},
	author = {Jin, Jiashun and Ke, Zheng Tracy and Luo, Shengming and Ma, Yucong},
	month = apr,
	year = {2025},
	note = {Publisher: ASA Website
\_eprint: https://doi.org/10.1080/01621459.2024.2393471},
	pages = {1048--1062},
}

@article{ghoshdastidar_two-sample_2020,
	title = {Two-{Sample} {Hypothesis} {Testing} for {Inhomogeneous} {Random} {Graphs}},
	volume = {48},
	issn = {0090-5364},
	url = {https://www.jstor.org/stable/26931555},
	number = {4},
	urldate = {2025-11-18},
	journal = {The Annals of Statistics},
	author = {Ghoshdastidar, Debarghya and Gutzeit, Maurilio and Carpentier, Alexandra and von Luxburg, Ulrike},
	year = {2020},
	note = {Publisher: Institute of Mathematical Statistics},
	pages = {2208--2229},
}

@article{chakraborty_scalable_2025,
	title = {Scalable {Estimation} and {Two}-{Sample} {Testing} for {Large} {Networks} via {Subsampling}},
	volume = {34},
	issn = {1061-8600},
	url = {https://doi.org/10.1080/10618600.2024.2432974},
	doi = {10.1080/10618600.2024.2432974},
	number = {3},
	urldate = {2025-11-18},
	journal = {Journal of Computational and Graphical Statistics},
	author = {Chakraborty, Kaustav and Sengupta, Srijan and Chen, Yuguo},
	month = jul,
	year = {2025},
	note = {Publisher: ASA Website
\_eprint: https://doi.org/10.1080/10618600.2024.2432974},
	pages = {1127--1139},
}

@article{saxena_lost_2025,
	title = {Lost in the shuffle: {Testing} power in the presence of errorful network vertex labels},
	volume = {204},
	issn = {0167-9473},
	shorttitle = {Lost in the shuffle},
	url = {https://www.sciencedirect.com/science/article/pii/S0167947324001750},
	doi = {10.1016/j.csda.2024.108091},
	urldate = {2025-11-18},
	journal = {Computational Statistics \& Data Analysis},
	author = {Saxena, Ayushi and Lyzinski, Vince},
	month = apr,
	year = {2025},
	pages = {108091},
}

@misc{qi_multivariate_2024,
	title = {Multivariate {Inference} of {Network} {Moments} by {Subsampling}},
	url = {http://arxiv.org/abs/2409.01599},
	doi = {10.48550/arXiv.2409.01599},
	urldate = {2025-11-18},
	publisher = {arXiv},
	author = {Qi, Mingyu and Li, Tianxi and Zhou, Wen},
	month = sep,
	year = {2024},
	note = {arXiv:2409.01599 [stat]},
}

@article{agterberg_joint_2025,
	title = {Joint {Spectral} {Clustering} in {Multilayer} {Degree}-{Corrected} {Stochastic} {Blockmodels}},
	volume = {120},
	issn = {0162-1459},
	url = {https://doi.org/10.1080/01621459.2025.2516201},
	doi = {10.1080/01621459.2025.2516201},
	number = {551},
	urldate = {2025-11-18},
	journal = {Journal of the American Statistical Association},
	author = {Agterberg, Joshua and Lubberts, Zachary and Arroyo, Jesús},
	month = jul,
	year = {2025},
	note = {Publisher: ASA Website
\_eprint: https://doi.org/10.1080/01621459.2025.2516201},
	pages = {1607--1620},
}

@article{agterberg_estimating_2025,
	title = {Estimating {Higher}-{Order} {Mixed} {Memberships} via the \${\textbackslash}ell\_\{2,{\textbackslash}infty\}\$ {Tensor} {Perturbation} {Bound}},
	volume = {120},
	issn = {0162-1459},
	url = {https://doi.org/10.1080/01621459.2024.2404265},
	doi = {10.1080/01621459.2024.2404265},
	number = {550},
	urldate = {2025-11-18},
	journal = {Journal of the American Statistical Association},
	author = {Agterberg, Joshua and Zhang, Anru R.},
	month = apr,
	year = {2025},
	note = {Publisher: ASA Website
\_eprint: https://doi.org/10.1080/01621459.2024.2404265},
	pages = {1214--1224},
}

@misc{agterberg_nonparametric_2020,
	title = {Nonparametric {Two}-{Sample} {Hypothesis} {Testing} for {Random} {Graphs} with {Negative} and {Repeated} {Eigenvalues}},
	url = {http://arxiv.org/abs/2012.09828},
	doi = {10.48550/arXiv.2012.09828},
	urldate = {2025-11-18},
	publisher = {arXiv},
	author = {Agterberg, Joshua and Tang, Minh and Priebe, Carey},
	month = dec,
	year = {2020},
	note = {arXiv:2012.09828 [math]},
}

@misc{agterberg_overview_2025,
	title = {An {Overview} of {Asymptotic} {Normality} in {Stochastic} {Blockmodels}: {Cluster} {Analysis} and {Inference}},
	shorttitle = {An {Overview} of {Asymptotic} {Normality} in {Stochastic} {Blockmodels}},
	url = {http://arxiv.org/abs/2305.06353},
	doi = {10.48550/arXiv.2305.06353},
	urldate = {2025-11-18},
	publisher = {arXiv},
	author = {Agterberg, Joshua and Cape, Joshua},
	month = jan,
	year = {2025},
	note = {arXiv:2305.06353 [math]},
}

@misc{macdonald_mesoscale_2024,
	title = {Mesoscale two-sample testing for network data},
	url = {http://arxiv.org/abs/2410.17046},
	doi = {10.48550/arXiv.2410.17046},
	urldate = {2025-11-18},
	publisher = {arXiv},
	author = {MacDonald, Peter W. and Levina, Elizaveta and Zhu, Ji},
	month = oct,
	year = {2024},
	note = {arXiv:2410.17046 [stat]},
}

@misc{fan_simple-rc_2022,
	title = {{SIMPLE}-{RC}: {Group} {Network} {Inference} with {Non}-{Sharp} {Nulls} and {Weak} {Signals}},
	shorttitle = {{SIMPLE}-{RC}},
	url = {http://arxiv.org/abs/2211.00128},
	doi = {10.48550/arXiv.2211.00128},
	urldate = {2025-11-18},
	publisher = {arXiv},
	author = {Fan, Jianqing and Fan, Yingying and Lv, Jinchi and Yang, Fan},
	month = oct,
	year = {2022},
	note = {arXiv:2211.00128 [stat]},
}

@article{shao_higher-order_nodate,
	title = {Higher-{Order} {Accurate} {Two}-{Sample} {Network} {Inference} and {Network} {Hashing}},
	volume = {0},
	issn = {0162-1459},
	url = {https://doi.org/10.1080/01621459.2025.2520459},
	doi = {10.1080/01621459.2025.2520459},
	number = {0},
	urldate = {2025-11-18},
	journal = {Journal of the American Statistical Association},
	author = {Shao, Meijia and Xia, Dong and Zhang, Yuan and Wu, Qiong and Chen, Shuo},
	note = {Publisher: ASA Website},
	pages = {1--13},
}

@article{xia_confidence_2019,
	title = {Confidence {Region} of {Singular} {Subspaces} for {Low}-{Rank} {Matrix} {Regression}},
	volume = {65},
	issn = {1557-9654},
	url = {https://ieeexplore.ieee.org/abstract/document/8754764},
	doi = {10.1109/TIT.2019.2924900},
	number = {11},
	urldate = {2025-11-18},
	journal = {IEEE Transactions on Information Theory},
	author = {Xia, Dong},
	month = nov,
	year = {2019},
	pages = {7437--7459},
}

@article{bao_singular_2021,
	title = {Singular vector and singular subspace distribution for the matrix denoising model},
	volume = {49},
	issn = {0090-5364, 2168-8966},
	url = {https://projecteuclid.org/journals/annals-of-statistics/volume-49/issue-1/Singular-vector-and-singular-subspace-distribution-for-the-matrix-denoising/10.1214/20-AOS1960.full},
	doi = {10.1214/20-AOS1960},
	number = {1},
	urldate = {2025-11-18},
	journal = {The Annals of Statistics},
	author = {Bao, Zhigang and Ding, Xiucai and Wang, {and} Ke},
	month = feb,
	year = {2021},
	note = {Publisher: Institute of Mathematical Statistics},
	pages = {370--392},
}

@article{chen_asymmetry_2021,
	title = {Asymmetry helps: {Eigenvalue} and eigenvector analyses of asymmetrically perturbed low-rank matrices},
	volume = {49},
	issn = {0090-5364, 2168-8966},
	shorttitle = {Asymmetry helps},
	url = {https://projecteuclid.org/journals/annals-of-statistics/volume-49/issue-1/Asymmetry-helps--Eigenvalue-and-eigenvector-analyses-of-asymmetrically-perturbed/10.1214/20-AOS1963.full},
	doi = {10.1214/20-AOS1963},
	number = {1},
	urldate = {2025-11-18},
	journal = {The Annals of Statistics},
	author = {Chen, Yuxin and Cheng, Chen and Fan, Jianqing},
	month = feb,
	year = {2021},
	note = {Publisher: Institute of Mathematical Statistics},
	pages = {435--458},
}

@article{cheng_tackling_2021,
	title = {Tackling {Small} {Eigen}-{Gaps}: {Fine}-{Grained} {Eigenvector} {Estimation} and {Inference} {Under} {Heteroscedastic} {Noise}},
	volume = {67},
	issn = {1557-9654},
	shorttitle = {Tackling {Small} {Eigen}-{Gaps}},
	url = {https://ieeexplore.ieee.org/abstract/document/9535122},
	doi = {10.1109/TIT.2021.3111828},
	number = {11},
	urldate = {2025-11-18},
	journal = {IEEE Transactions on Information Theory},
	author = {Cheng, Chen and Wei, Yuting and Chen, Yuxin},
	month = nov,
	year = {2021},
	pages = {7380--7419},
}

@article{xia_inference_2022,
	title = {Inference for low-rank tensors—no need to debias},
	volume = {50},
	issn = {0090-5364, 2168-8966},
	url = {https://projecteuclid.org/journals/annals-of-statistics/volume-50/issue-2/Inference-for-low-rank-tensorsno-need-to-debias/10.1214/21-AOS2146.full},
	doi = {10.1214/21-AOS2146},
	number = {2},
	urldate = {2025-11-18},
	journal = {The Annals of Statistics},
	author = {Xia, Dong and Zhang, Anru R. and Zhou, Yuchen},
	month = apr,
	year = {2022},
	note = {Publisher: Institute of Mathematical Statistics},
	pages = {1220--1245},
}

@article{chen_spectral_2021,
	title = {Spectral {Methods} for {Data} {Science}: {A} {Statistical} {Perspective}},
	volume = {14},
	issn = {1935-8237, 1935-8245},
	shorttitle = {Spectral {Methods} for {Data} {Science}},
	url = {http://arxiv.org/abs/2012.08496},
	doi = {10.1561/2200000079},
	number = {5},
	urldate = {2025-11-18},
	journal = {Foundations and Trends® in Machine Learning},
	author = {Chen, Yuxin and Chi, Yuejie and Fan, Jianqing and Ma, Cong},
	year = {2021},
	note = {arXiv:2012.08496 [stat]},
	pages = {566--806},
}

@article{yan_coherence-free_2024,
	title = {Coherence-free {Entrywise} {Estimation} of {Eigenvectors} in {Low}-rank {Signal}-plus-noise {Matrix} {Models}},
	volume = {37},
	url = {https://proceedings.neurips.cc/paper_files/paper/2024/hash/e4cd50120b6d7e8daff1749d6bbaa889-Abstract-Conference.html},
	doi = {10.52202/079017-4021},
	language = {en},
	urldate = {2025-11-18},
	journal = {Advances in Neural Information Processing Systems},
	author = {Yan, Hao and Levin, Keith},
	month = dec,
	year = {2024},
	pages = {126566--126619},
}

@article{fan_asymptotic_2022,
	title = {Asymptotic {Theory} of {Eigenvectors} for {Random} {Matrices} {With} {Diverging} {Spikes}},
	volume = {117},
	issn = {0162-1459},
	url = {https://doi.org/10.1080/01621459.2020.1840990},
	doi = {10.1080/01621459.2020.1840990},
	number = {538},
	urldate = {2025-11-18},
	journal = {Journal of the American Statistical Association},
	author = {Fan, Jianqing and Fan, Yingying and Han, Xiao and Lv, Jinchi},
	month = apr,
	year = {2022},
	pmid = {36060554},
	note = {Publisher: ASA Website
\_eprint: https://doi.org/10.1080/01621459.2020.1840990},
	pages = {996--1009},
}

@article{xia_normal_2021,
	title = {Normal approximation and confidence region of singular subspaces},
	volume = {15},
	issn = {1935-7524, 1935-7524},
	url = {https://projecteuclid.org/journals/electronic-journal-of-statistics/volume-15/issue-2/Normal-approximation-and-confidence-region-of-singular-subspaces/10.1214/21-EJS1876.full},
	doi = {10.1214/21-EJS1876},
	number = {2},
	urldate = {2025-11-18},
	journal = {Electronic Journal of Statistics},
	author = {Xia, Dong},
	month = jan,
	year = {2021},
	note = {Publisher: Institute of Mathematical Statistics and Bernoulli Society},
	pages = {3798--3851},
}

@misc{wang_analysis_2024,
	title = {Analysis of singular subspaces under random perturbations},
	url = {http://arxiv.org/abs/2403.09170},
	doi = {10.48550/arXiv.2403.09170},
	urldate = {2025-11-18},
	publisher = {arXiv},
	author = {Wang, Ke},
	month = mar,
	year = {2024},
	note = {arXiv:2403.09170 [math]},
}

@article{han_universal_2023,
	title = {Universal rank inference via residual subsampling with application to large networks},
	volume = {51},
	issn = {0090-5364, 2168-8966},
	url = {https://projecteuclid.org/journals/annals-of-statistics/volume-51/issue-3/Universal-rank-inference-via-residual-subsampling-with-application-to-large/10.1214/23-AOS2282.full},
	doi = {10.1214/23-AOS2282},
	number = {3},
	urldate = {2025-11-18},
	journal = {The Annals of Statistics},
	author = {Han, Xiao and Yang, Qing and Fan, Yingying},
	month = jun,
	year = {2023},
	note = {Publisher: Institute of Mathematical Statistics},
	pages = {1109--1133},
}

@article{li_minimax_2025,
	title = {Minimax {Estimation} of {Linear} {Functions} of {Eigenvectors} in the {Face} of {Small} {Eigen}-{Gaps}},
	volume = {71},
	issn = {1557-9654},
	url = {https://ieeexplore.ieee.org/abstract/document/10789229},
	doi = {10.1109/TIT.2024.3514795},
	number = {2},
	urldate = {2025-11-18},
	journal = {IEEE Transactions on Information Theory},
	author = {Li, Gen and Cai, Changxiao and Poor, H. Vincent and Chen, Yuxin},
	month = feb,
	year = {2025},
	pages = {1200--1247},
}

@article{liu_asymptotic_2025,
	title = {Asymptotic limits of spiked eigenvalues and eigenvectors of signal-plus-noise matrices with weak signals and heteroskedastic noise},
	volume = {31},
	issn = {1350-7265},
	url = {https://projecteuclid.org/journals/bernoulli/volume-31/issue-3/Asymptotic-limits-of-spiked-eigenvalues-and-eigenvectors-of-signal-plus/10.3150/24-BEJ1808.full},
	doi = {10.3150/24-BEJ1808},
	number = {3},
	urldate = {2025-11-18},
	journal = {Bernoulli},
	author = {Liu, Xiaoyu and Liu, Yiming and Pan, Guangming and Zhang, Lingyue and Zhang, Zhixiang},
	month = aug,
	year = {2025},
	note = {Publisher: Bernoulli Society for Mathematical Statistics and Probability},
	pages = {2351--2376},
}

@misc{chang_extreme_2025,
	title = {Extreme value theory for singular subspace estimation in the matrix denoising model},
	url = {http://arxiv.org/abs/2507.19978},
	doi = {10.48550/arXiv.2507.19978},
	urldate = {2025-11-18},
	publisher = {arXiv},
	author = {Chang, Junhyung and Cape, Joshua},
	month = jul,
	year = {2025},
	note = {arXiv:2507.19978 [math]},
}

@misc{fan_asymptotic_2025,
	title = {Asymptotic {Theory} of {Eigenvectors} for {Latent} {Embeddings} with {Generalized} {Laplacian} {Matrices}},
	url = {http://arxiv.org/abs/2503.00640},
	doi = {10.48550/arXiv.2503.00640},
	urldate = {2025-11-18},
	publisher = {arXiv},
	author = {Fan, Jianqing and Fan, Yingying and Lv, Jinchi and Yang, Fan and Yu, Diwen},
	month = mar,
	year = {2025},
	note = {arXiv:2503.00640 [stat]},
}

@misc{agterberg_distributional_2024,
	title = {Distributional {Theory} and {Statistical} {Inference} for {Linear} {Functions} of {Eigenvectors} with {Small} {Eigengaps}},
	url = {http://arxiv.org/abs/2308.02480},
	doi = {10.48550/arXiv.2308.02480},
	urldate = {2025-11-18},
	publisher = {arXiv},
	author = {Agterberg, Joshua},
	month = oct,
	year = {2024},
	note = {arXiv:2308.02480 [math]},
}

@article{tang_semiparametric_2017,
	title = {A {Semiparametric} {Two}-{Sample} {Hypothesis} {Testing} {Problem} for {Random} {Graphs}},
	volume = {26},
	issn = {1061-8600},
	url = {https://doi.org/10.1080/10618600.2016.1193505},
	doi = {10.1080/10618600.2016.1193505},
	number = {2},
	urldate = {2025-11-18},
	journal = {Journal of Computational and Graphical Statistics},
	author = {Tang, Minh and Athreya, Avanti and Sussman, Daniel L. and Lyzinski, Vince and Park, Youngser and Priebe, Carey E.},
	month = apr,
	year = {2017},
	note = {Publisher: ASA Website
\_eprint: https://doi.org/10.1080/10618600.2016.1193505},
	pages = {344--354},
}

@misc{fu_two-sample_2022,
	title = {Two-{Sample} {Test} for {Stochastic} {Block} {Models} via {Maximum} {Entry}-wise {Deviation}},
	url = {http://arxiv.org/abs/2211.08668},
	doi = {10.48550/arXiv.2211.08668},
	urldate = {2025-11-18},
	publisher = {arXiv},
	author = {Fu, Kang and Hu, Jianwei and Keita, Seydou and Liu, Hao},
	month = dec,
	year = {2022},
	note = {arXiv:2211.08668 [stat]},
}

@article{auerbach_testing_2022,
	title = {Testing for {Differences} in {Stochastic} {Network} {Structure}},
	volume = {90},
	copyright = {© 2022 The Econometric Society},
	issn = {1468-0262},
	url = {https://onlinelibrary.wiley.com/doi/abs/10.3982/ECTA18093},
	doi = {10.3982/ECTA18093},
	language = {en},
	number = {3},
	urldate = {2025-11-18},
	journal = {Econometrica},
	author = {Auerbach, Eric},
	year = {2022},
	note = {\_eprint: https://onlinelibrary.wiley.com/doi/pdf/10.3982/ECTA18093},
	pages = {1205--1223},
}

@misc{bhadra_bootstrap-based_2025,
	title = {A {Bootstrap}-based {Method} for {Testing} {Network} {Similarity}},
	url = {http://arxiv.org/abs/1911.06869},
	doi = {10.48550/arXiv.1911.06869},
	urldate = {2025-11-18},
	publisher = {arXiv},
	author = {Bhadra, Somnath and Chakraborty, Kaustav and Sengupta, Srijan and Lahiri, Soumendra},
	month = feb,
	year = {2025},
	note = {arXiv:1911.06869 [stat]},
}

@misc{jin_two-sample_2024,
	title = {Two-{Sample} {Hypothesis} {Testing} for {Large} {Random} {Graphs} of {Unequal} {Size}},
	url = {http://arxiv.org/abs/2402.11133},
	doi = {10.48550/arXiv.2402.11133},
	urldate = {2025-11-18},
	publisher = {arXiv},
	author = {Jin, Xin and Chan, Kit and Barnett, Ian and Ghosh, Riddhi Pratim},
	month = feb,
	year = {2024},
	note = {arXiv:2402.11133 [stat]},
}

@misc{nguen_network_2024,
	title = {Network two-sample test for block models},
	url = {http://arxiv.org/abs/2406.06014},
	doi = {10.48550/arXiv.2406.06014},
	urldate = {2025-11-18},
	publisher = {arXiv},
	author = {Nguen, Chung Kyong and Padilla, Oscar Hernan Madrid and Amini, Arash A.},
	month = jun,
	year = {2024},
	note = {arXiv:2406.06014 [math]},
}

@misc{wu_two-sample_2023,
	title = {Two-sample test of sparse stochastic block models},
	url = {http://arxiv.org/abs/2304.00739},
	doi = {10.48550/arXiv.2304.00739},
	urldate = {2025-11-18},
	publisher = {arXiv},
	author = {Wu, Qianyong and Hu, Jiang},
	month = apr,
	year = {2023},
	note = {arXiv:2304.00739 [stat]},
}

@misc{shen_combinatorial-probabilistic_2020,
	title = {Combinatorial-{Probabilistic} {Trade}-{Off}: {Community} {Properties} {Test} in the {Stochastic} {Block} {Models}},
	shorttitle = {Combinatorial-{Probabilistic} {Trade}-{Off}},
	url = {http://arxiv.org/abs/2010.15063},
	doi = {10.48550/arXiv.2010.15063},
	urldate = {2025-11-18},
	publisher = {arXiv},
	author = {Shen, Shuting and Lu, Junwei},
	month = oct,
	year = {2020},
	note = {arXiv:2010.15063 [math]},
}

@article{bhattacharjee_change_2020,
	title = {Change {Point} {Estimation} in a {Dynamic} {Stochastic} {Block} {Model}},
	volume = {21},
	issn = {1533-7928},
	url = {http://jmlr.org/papers/v21/18-814.html},
	number = {107},
	urldate = {2025-11-18},
	journal = {Journal of Machine Learning Research},
	author = {Bhattacharjee, Monika and Banerjee, Moulinath and Michailidis, George},
	year = {2020},
	pages = {1--59},
}

@article{wang_optimal_2021,
	title = {Optimal {Change} {Point} {Detection} and {Localization} in {Sparse} {Dynamic} {Networks}},
	volume = {49},
	issn = {0090-5364},
	url = {https://www.jstor.org/stable/27028771},
	number = {1},
	urldate = {2025-11-18},
	journal = {The Annals of Statistics},
	author = {Wang, Daren and Yu, Yi and Rinaldo, Alessandro},
	year = {2021},
	note = {Publisher: Institute of Mathematical Statistics},
	pages = {203--232},
}

@misc{zheng_limit_2024,
	title = {Limit results for distributed estimation of invariant subspaces in multiple networks inference and {PCA}},
	url = {http://arxiv.org/abs/2206.04306},
	doi = {10.48550/arXiv.2206.04306},
	urldate = {2025-11-18},
	publisher = {arXiv},
	author = {Zheng, Runbing and Tang, Minh},
	month = may,
	year = {2024},
	note = {arXiv:2206.04306 [math]},
}

@article{du_hypothesis_2023,
	title = {Hypothesis testing for equality of latent positions in random graphs},
	volume = {29},
	issn = {1350-7265},
	url = {https://projecteuclid.org/journals/bernoulli/volume-29/issue-4/Hypothesis-testing-for-equality-of-latent-positions-in-random-graphs/10.3150/22-BEJ1581.full},
	doi = {10.3150/22-BEJ1581},
	number = {4},
	urldate = {2025-11-18},
	journal = {Bernoulli},
	author = {Du, Xinjie and Tang, Minh},
	month = nov,
	year = {2023},
	note = {Publisher: Bernoulli Society for Mathematical Statistics and Probability},
	pages = {3221--3254},
}

@misc{agterberg_statistical_2024,
	title = {Statistical {Inference} for {Low}-{Rank} {Tensors}: {Heteroskedasticity}, {Subgaussianity}, and {Applications}},
	shorttitle = {Statistical {Inference} for {Low}-{Rank} {Tensors}},
	url = {http://arxiv.org/abs/2410.06381},
	doi = {10.48550/arXiv.2410.06381},
	urldate = {2025-11-18},
	publisher = {arXiv},
	author = {Agterberg, Joshua and Zhang, Anru},
	month = oct,
	year = {2024},
	note = {arXiv:2410.06381 [math]},
}

@article{zhang_edgeworth_2022,
	title = {Edgeworth expansions for network moments},
	volume = {50},
	issn = {0090-5364, 2168-8966},
	url = {https://projecteuclid.org/journals/annals-of-statistics/volume-50/issue-2/Edgeworth-expansions-for-network-moments/10.1214/21-AOS2125.full},
	doi = {10.1214/21-AOS2125},
	number = {2},
	urldate = {2025-11-18},
	journal = {The Annals of Statistics},
	author = {Zhang, Yuan and Xia, Dong},
	month = apr,
	year = {2022},
	note = {Publisher: Institute of Mathematical Statistics},
	pages = {726--753},
}

@article{lunde_subsampling_2023,
	title = {Subsampling sparse graphons under minimal assumptions},
	volume = {110},
	issn = {1464-3510},
	url = {https://doi.org/10.1093/biomet/asac032},
	doi = {10.1093/biomet/asac032},
	number = {1},
	urldate = {2025-11-18},
	journal = {Biometrika},
	author = {Lunde, Robert and Sarkar, Purnamrita},
	month = mar,
	year = {2023},
	pages = {15--32},
}

@article{zu_local_2025,
	title = {Local bootstrap for network data},
	volume = {112},
	issn = {1464-3510},
	url = {https://doi.org/10.1093/biomet/asae046},
	doi = {10.1093/biomet/asae046},
	number = {1},
	urldate = {2025-11-18},
	journal = {Biometrika},
	author = {Zu, Tianhai and Qin, Yichen},
	month = feb,
	year = {2025},
	pages = {asae046},
}

@article{bravo-hermsdorff_quantifying_2023,
	title = {Quantifying {Network} {Similarity} using {Graph} {Cumulants}},
	volume = {24},
	issn = {1533-7928},
	url = {http://jmlr.org/papers/v24/21-082.html},
	number = {187},
	urldate = {2025-11-18},
	journal = {Journal of Machine Learning Research},
	author = {Bravo-Hermsdorff, Gecia and Gunderson, Lee M. and Maugis, Pierre-André and Priebe, Carey E.},
	year = {2023},
	pages = {1--27},
}

@article{deng_subsampling-based_2024,
	title = {Subsampling-based modified {Bayesian} information criterion for large-scale stochastic block models},
	volume = {18},
	issn = {1935-7524, 1935-7524},
	url = {https://projecteuclid.org/journals/electronic-journal-of-statistics/volume-18/issue-2/Subsampling-based-modified-Bayesian-information-criterion-for-large-scale-stochastic/10.1214/24-EJS2309.full},
	doi = {10.1214/24-EJS2309},
	number = {2},
	urldate = {2025-11-18},
	journal = {Electronic Journal of Statistics},
	author = {Deng, Jiayi and Huang, Danyang and Chang, Xiangyu and Zhang, Bo},
	month = jan,
	year = {2024},
	note = {Publisher: Institute of Mathematical Statistics and Bernoulli Society},
	pages = {4724--4766},
}

@misc{shao_u-statistic_2023,
	title = {U-{Statistic} {Reduction}: {Higher}-{Order} {Accurate} {Risk} {Control} and {Statistical}-{Computational} {Trade}-{Off}, with {Application} to {Network} {Method}-of-{Moments}},
	shorttitle = {U-{Statistic} {Reduction}},
	url = {http://arxiv.org/abs/2306.03793},
	doi = {10.48550/arXiv.2306.03793},
	urldate = {2025-11-18},
	publisher = {arXiv},
	author = {Shao, Meijia and Xia, Dong and Zhang, Yuan},
	month = jun,
	year = {2023},
	note = {arXiv:2306.03793 [stat]},
}

@misc{li_assumption-lean_2025,
	title = {Assumption-lean {Inference} for {Network}-linked {Data}},
	url = {http://arxiv.org/abs/2510.00287},
	doi = {10.48550/arXiv.2510.00287},
	urldate = {2025-11-18},
	publisher = {arXiv},
	author = {Li, Wei and Chakraborty, Nilanjan and Lunde, Robert},
	month = sep,
	year = {2025},
	note = {arXiv:2510.00287 [stat]},
}

@article{levin_bootstrapping_2025,
	title = {Bootstrapping networks with latent space structure},
	volume = {19},
	issn = {1935-7524, 1935-7524},
	url = {https://projecteuclid.org/journals/electronic-journal-of-statistics/volume-19/issue-1/Bootstrapping-networks-with-latent-space-structure/10.1214/25-EJS2347.full},
	doi = {10.1214/25-EJS2347},
	number = {1},
	urldate = {2025-11-18},
	journal = {Electronic Journal of Statistics},
	author = {Levin, Keith and Levina, Elizaveta},
	month = jan,
	year = {2025},
	note = {Publisher: Institute of Mathematical Statistics and Bernoulli Society},
	pages = {745--791},
}


\end{document}